\documentclass[11pt,a4paper]{article}

\parskip=\baselineskip

\usepackage{tikz}
\usetikzlibrary{arrows,decorations.markings,shapes.arrows,patterns, decorations.pathmorphing}
\usepackage{amsmath,amsthm}
\usepackage{color,slashed}
\usepackage[colorlinks=true,linkcolor=blue,citecolor=red]{hyperref}
\usepackage{amssymb,amsfonts,verbatim}
\usepackage{lscape}

\tikzset{
  ->-/.style={decoration={markings, mark=at position 0.5 with {\arrow{to}}},
              postaction={decorate}},
}

\tikzset{-<-/.style={decoration={markings,mark=at position 0.5 with %
    {\arrow[scale=1.5,>=stealth]{<}}},postaction={decorate}}}

\numberwithin{equation}{section}

\newcommand{\mc}{\mathcal}
\newcommand{\mf}{\mathfrak}

\newcommand{\zbar}{\br{z}}

\newcommand{\dpa}[1]{\frac{\partial}{\partial #1}}

\newcommand{\eps}{\epsilon}
\newcommand{\g}{\mathfrak{g}}

\newcommand{\xto}{\xrightarrow}
\newcommand{\what}{\widehat}

\newcommand{\til}{\widetilde}

\newcommand{\br}{\overline}
\newcommand{\iso}{\cong}
\newcommand{\C}{\mathbb C}

\newcommand{\norm}[2]{d( #1 , #2)}

\newcommand{\op}{\operatorname}
\newcommand{\mbf}{\mathbf}
\newcommand{\mbb}{\mathbb}

\newcommand{\ip}[1]{\left\langle #1 \right\rangle}
\newcommand{\abs}[1]{\left| #1 \right|}

\newcommand{\R}{\mbb R}
\renewcommand{\d}{\mathrm{d}}

\DeclareMathOperator{\Sym}{Sym}

\def\i{\mathsf{i}}
\def\ii{\mathsf{i}} 
\def\vepsilon{\varepsilon}
\def\Tr{{\mathrm{Tr}}}

\def\veps{\varepsilon}
\def\O{{\mathcal O}}
\def\sh{{\sf h}}
\def\ad{{\mathrm{ad}}}
\def\M{{\mathcal M}}
\def\W{{\mathcal W}}
\def\c{{\sf c}}
\def\V{{\mathcal V}}
\def\L{{\mathcal L}}
\def\CP{{\mathbb{CP}}}
\def\i{{\mathrm i}}
\def\l{{\mathfrak l}}
\def\Z{\mathbb{Z}}

\newcommand{\vin}{\rotatebox[origin=c]{90}{$\in$}}

\def\eqn#1{eqn.\ #1}
\def\fig#1{Fig.\ #1}

\newtheorem{theorem}{Theorem}[section]
\newtheorem{proposition}[theorem]{Proposition}
\newtheorem{lemma}[theorem]{Lemma}

\newtheorem{corollary}[theorem]{Corollary}

\def\be{\begin{equation}}
\def\ee{\end{equation}}

\begin{document}

\pagenumbering{Alph} 
\begin{titlepage}
\begin{flushright}

\end{flushright}
\vskip 1.5in
\begin{center}
{\bf\Large{Gauge Theory And Integrability, I}}
\vskip 0.5cm 
{Kevin Costello$^1$, Edward Witten$^2$ and Masahito Yamazaki$^3$} \vskip 0.05in 
{\small{ 
\textit{$^1$Perimeter Institute for Theoretical Physics, }\vskip -.4cm
\textit{Waterloo, ON N2L 2Y5, Canada}
\vskip 0 cm 
\textit{$^2$School of Natural Sciences, Institute for Advanced Study,}\vskip -.4cm
\textit{Einstein Drive, Princeton, NJ 08540 USA}
\vskip 0 cm 
\textit{$^3$Kavli Institute for the Physics and Mathematics of the Universe (WPI), }\vskip -.4cm
\textit{University of Tokyo, Kashiwa, Chiba 277-8583, Japan}
}}

\end{center}

\vskip 0.5in
\baselineskip 16pt
\begin{abstract}
Several years ago, it was proposed that the usual solutions of the Yang-Baxter
equation associated to Lie groups can be deduced in a systematic way from four-dimensional gauge theory.
In the present paper, we extend this picture, fill in many details, and present the arguments in a concrete and down-to-earth way.
Many interesting effects, including the leading nontrivial contributions to the $R$-matrix, the operator product expansion of line operators,
the framing anomaly, and the quantum deformation that leads from $\g[[z]]$ to the Yangian, are computed explicitly via Feynman diagrams.
We explain  how rational, trigonometric, and elliptic solutions of the Yang-Baxter equation arise in this framework, along with a generalization
that is known as the dynamical Yang-Baxter equation.
\end{abstract}
\date{September, 2017}
\end{titlepage}
\pagenumbering{arabic} 

\newpage

\tableofcontents

\newpage

\section{Introduction}

Integrable systems of $1+1$-dimensional many-body physics and two-dimensional statistical mechanics first emerged in Bethe's discovery of the Bethe
Ansatz \cite{Bethe} and Onsager's solution of the two-dimensional Ising model \cite{Onsager}, respectively.  Subsequent study led to remarkable generalizations
and new discoveries, continuing to the present day.

Much of the wisdom about integrable models can be distilled into the Yang-Baxter equation  \cite{McGuire:1964zt,Yang:1967bm,Baxter:1971cr,Zamolodchikov*2} and its interpretation via quantum groups \cite{Drinfeld_ICM,Chari-Pressley}.
Many of the classic papers on this subject are reprinted in \cite{Jimbo}.  See also, for example, \cite{PY} for an introduction.

In the present paper, we will aim to explain, simplify, and further develop a new approach to the Yang-Baxter equation and related integrable systems
of statistical mechanics that was proposed several years ago by one of us
\cite{CostelloA,Costello:2013zra}.  In this approach, the solutions of the Yang-Baxter equation and their properties are deduced from a four-dimensional
gauge theory that can be regarded as a $T$-dual version of three-dimensional Chern-Simons gauge theory.   An informal introduction to this approach can
be found in     \cite{Witten:2016spx}.  In spirit, this approach is in keeping with a vision that was proposed long ago by Atiyah \cite{Atiyah}: the Yang-Baxter
equation in two dimensions is deduced by starting with a theory in higher dimensions.  

The purpose of the present paper is to further develop this approach, filling in many details and hopefully presenting the arguments
in a concrete and down-to-earth way.  In section 2, we review the basic facts about integrability and the Yang-Baxter equation that will be needed.  In section 3, we introduce
the relevant four-dimensional gauge theory.  We explain why it leads automatically to solutions of the Yang-Baxter equation and more
specifically why possible choices of compactification to two dimensions
  lead to   rational, trigonometric, and elliptic solutions of that  equation.  We explain why in this theory, one can at the classical level introduce Wilson line operators
associated not just to representations of the gauge group $G$ but to representations of an infinite-dimensional algebra $\g[[z]]$. This generalization turns
out to be crucial in understanding the theory.   Quantum mechanically, as we learn
later, $\g[[z]]$ will
be promoted to the Yangian deformation of $\g[[z]]$ (or its trigonometric or elliptic generalization).   In section 3, we also review some simple examples of solutions of the
Yang-Baxter equation and see how one can deduce from these simple examples that there must be a framing anomaly for Wilson line operators.  

In sections \ref{lowestorder} - \ref{section_2loop}, we study some increasingly subtle quantum effects in this theory.  In section \ref{lowestorder}, we compute
directly the first nontrivial term in the quantum $R$-matrix. General theorems \cite{Drinfeld_ICM,Chari-Pressley} actually determine the whole structure in terms of this
lowest order term together with formal properties of the theory.   However, our goal in the present paper is to see everything as explicitly as possible rather than
relying on abstract arguments.  

In section \ref{parallel}, we study the first nontrivial quantum correction to the operator product expansion (OPE) of Wilson line operators.  We show that to get a closed
OPE, one has to consider line operators associated to representations of $\g[[z]]$, not just representations of the underlying finite-dimensional gauge group $G$.
We also explain that this first quantum correction to the classical OPE implies that in yet higher orders, there will have to be further deformations, which will deform $\g[[z]]$
to the Yangian (or one of its generalizations).  

In section \ref{framinganomaly}, we compute the framing anomaly for Wilson line operators in this theory, recovering from a Feynman diagram calculation the result
that was predicted on more abstract grounds in section 3.  The framing anomaly found here is somewhat analogous to the framing anomaly
for Wilson operators in three-dimensional Chern-Simons theory, but its consequences are more far-reaching.  

In section \ref{networks}, we generalize the analysis from Wilson line operators to networks of Wilson lines -- graphs that can be drawn in the plane in which the
line segments are Wilson line operators and the ``vertices'' are invariant couplings that describe (for example) the ``fusion'' of two Wilson line operators to make a single one.
In this context, there is a quantum anomaly that generalizes and can largely be deduced from the framing anomaly.

In section \ref{section_2loop}, we reconsider, following the elementary considerations in section \ref{parallel}, 
the deformation from $\g[[z]]$ to the Yangian.
At the two-loop level, that is, in order\footnote{There does not seem to be a standard terminology for counting loops in Feynman diagrams that contain Wilson
operators.  We refer to a contribution that is of order $\hbar^n$ relative to a leading order contribution as an $n$-loop effect.} $\hbar^2$, there is a potential anomaly in the coupling of two gauge bosons to a Wilson line operator.  To avoid or cancel the
anomaly, Wilson lines must be associated (in the rational case) to representations of a quantum deformation of $\g[[z]]$ known as the Yangian.   A surprising consequence
of this is that an ordinary  Wilson operator associated to a finite-dimensional representation of $G$ may be anomalous and hence absent in the quantum theory.
For example, for $G=SO(N)$ (or any simple Lie group other than $SU(N)$), there is no Wilson line operator associated to the adjoint representation.

In sections \ref{trigonometric} and \ref{elliptic}, we analyze the variants of the construction  that lead to trigonometric and elliptic solutions
of the Yang-Baxter equation, respectively.  In section \ref{dybe}, we explain how a generalization of the Yang-Baxter equation known as the ``dynamical Yang-Baxter
equation''  \cite{GN,Felder,FelderTwo,Etinghof} fits in this framework.  In brief, one finds an ordinary Yang-Baxter equation when one expands around a classical
gauge theory solution that has no moduli; moduli lead to a dynamical Yang-Baxter equation.  

All of our explicit computations in the present paper are in the lowest nontrivial order in $\hbar$ in which some quantum effect occurs.   In a companion paper \cite{Part2},
we will explain how to construct the Yangian algebra, and its trigonometric and elliptic generalizations, ``exactly'' in the present framework, not just in lowest
order of perturbation theory.   We have put the word ``exactly'' in quotes because the theory, in the form in which it has been developed so far, is a perturbative theory,
so ``exactly'' really means ``to all orders in perturbation theory.''  It is anticipated that the D4-NS5 system of string theory would provide the framework for a nonperturbative description, along the lines of the study of the D3-NS5 system in \cite{Witten:2011zz}, but this has not yet been developed.  

Finally, we recall the existence of another and superficially quite different relationship between integrable spin systems and four-dimensional gauge theory
\cite{NS,NS2}.  A connection between the two approaches is not yet known.

\section{Review of Integrability}

In this section, we review some standard  facts about integrable models, aiming
just to explain what is needed for the purposes of this paper.  
The goal of the rest of the paper will be to explain these facts from the standpoint of four-dimensional gauge theory.

\subsection{The Yang-Baxter Equation}\label{YBeq}

We consider a system of particles whose internal quantum numbers take values in some vector space $V$.
It is often convenient to pick a basis  $\{e_i \}_{i=1}^{\textrm{dim} V}$ of $V$.  A particle is also characterized
by a complex parameter $z$ that in the context of integrable systems is known as the spectral parameter.  It will play a crucial role in
what follows.  The particles live in a two-dimensional spacetime and travel on (possibly curved) one-dimensional worldlines in this spacetime.  
When two worldlines cross (\fig \ref{Rmatrix}), their spectral
parameters are assumed to be unchanged, but their internal state is transformed by a matrix that in general depends on the spectral
parameters.

\begin{figure}[htbp]
\centering{\includegraphics[scale=.3]{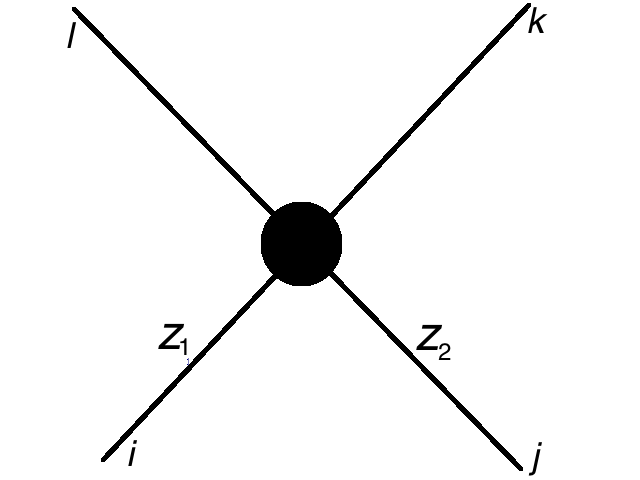}}
\caption{\small{Crossing of two worldlines in a two-dimensional spacetime.  The ``blob'' indicates a scattering process the amplitude for which 
will (in the context of the present paper) ultimately be computed via gauge theory.  When two particles cross, their spectral parameters $z_1$
and $z_2$ are unchanged but their ``internal'' state is transformed.}}
\label{Rmatrix}
\end{figure}

We write this matrix as $R(z_1,z_2):V\otimes V\to V\otimes V$, or in more detail, in the chosen basis, as $R_{ij}^{kl}(z_1,z_2)$.  However,
although there are interesting $R$-matrices that lack this property,\footnote{The basic example  is the chiral Potts model \cite{AuYang:1987zc,Baxter:1987eq}, in which
the spectral parameter takes values in a curve of genus greater than 1. Interestingly this model arises as a root-of-unity degeneration of the model of \cite{Bazhanov:2010kz}, which in turn arises from 
supersymmetric indices of four-dimensional $\mathcal{N}=1$ quiver gauge theories \cite{Spiridonov:2010em,Yamazaki:2012cp,Yamazaki:2013nra}.} we will be concerned in this paper with the case that the $R$-matrix
depends only on the difference $z=z_1-z_2$ of the two spectral parameters.  (In some applications of the Yang-Baxter equation, the spectral parameter
is interpreted as a particle momentum or rapidity, and the fact that the $R$-matrix depends only on the difference of spectral parameters is interpreted as a consequence of Galilean invariance or Lorentz invariance.)  

\begin{figure}[htbp]
\centering{\includegraphics[scale=.4]{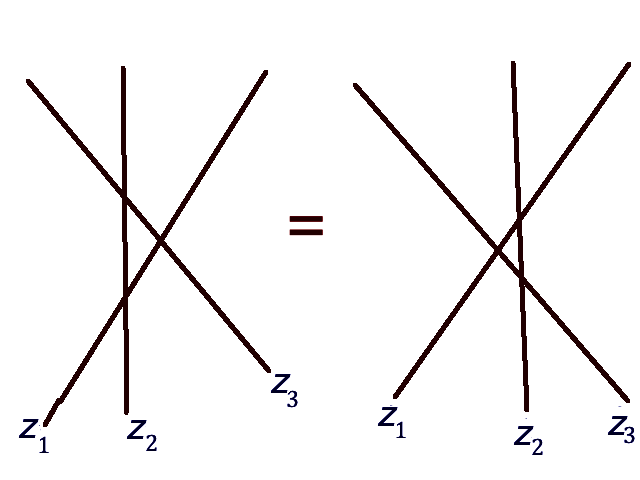}}
\caption{\small{The Yang-Baxter equation asserts the equivalence between these two pictures.}}
\label{YB2}
\end{figure}

The Yang-Baxter equation says that when three worldlines cross in a pairwise fashion, the arrangement in which they cross does not matter 
(\fig \ref{YB2}).
We denote the three particles as $a,b,c$, and write, for example, $V_a$ for the vector space of internal states of particle $a$, $z_a$
for its spectral parameter,  and $R_{ab}(z_a-z_b):V_a\otimes V_b\to V_a\otimes V_b$ for the corresponding $R$-matrix.\footnote{We also denote $R_{ab}\otimes 1:V_a\otimes V_b\otimes V_c\to V_a\otimes V_b\otimes V_c$ simply as $R_{ab}$.}  Then the Yang-Baxter
equation reads\footnote{The general form of this equation without assuming that the spectral parameter depends only on the difference of
rapidities is simply
\begin{align}\notag
R_{12}(z_1, z_2) R_{13}(z_1, z_3) R_{23}(z_2, z_3) = 
R_{23}(z_2, z_3) R_{13}(z_1, z_3) R_{12}(z_1, z_2)\;.
\end{align}}
\be\label{YBE}
R_{12}(z_1-z_2) R_{13}(z_1- z_3) R_{23}(z_2- z_3) = 
R_{23}(z_2- z_3) R_{13}(z_1- z_3) R_{12}(z_1- z_2) \;.
\ee
In terms of the basis $\{e_i\}$ of $V$, the equation takes the imposing form 
\begin{align}\label{longforum} 
\begin{split}
&
\sum_{o,p,q} R_{12}(z_1-z_2)_{qo}^{nm} R_{13}(z_1-z_3)_{ip}^{ql} R_{23}(z_2-z_3)_{jk}^{op}\\ &= 
\sum_{r,s,t} R_{23}(z_2-z_3)_{rt}^{ml} R_{13}(z_1-z_3)_{sk}^{nt} R_{12}(z_1-z_2)_{ij}^{sr} \;,
\end{split}
\end{align}
where the meaning of the indices is more clear in a picture (\fig \ref{FigureYBE}).

\begin{figure}[htbp]
\centering{\includegraphics[scale=0.95]{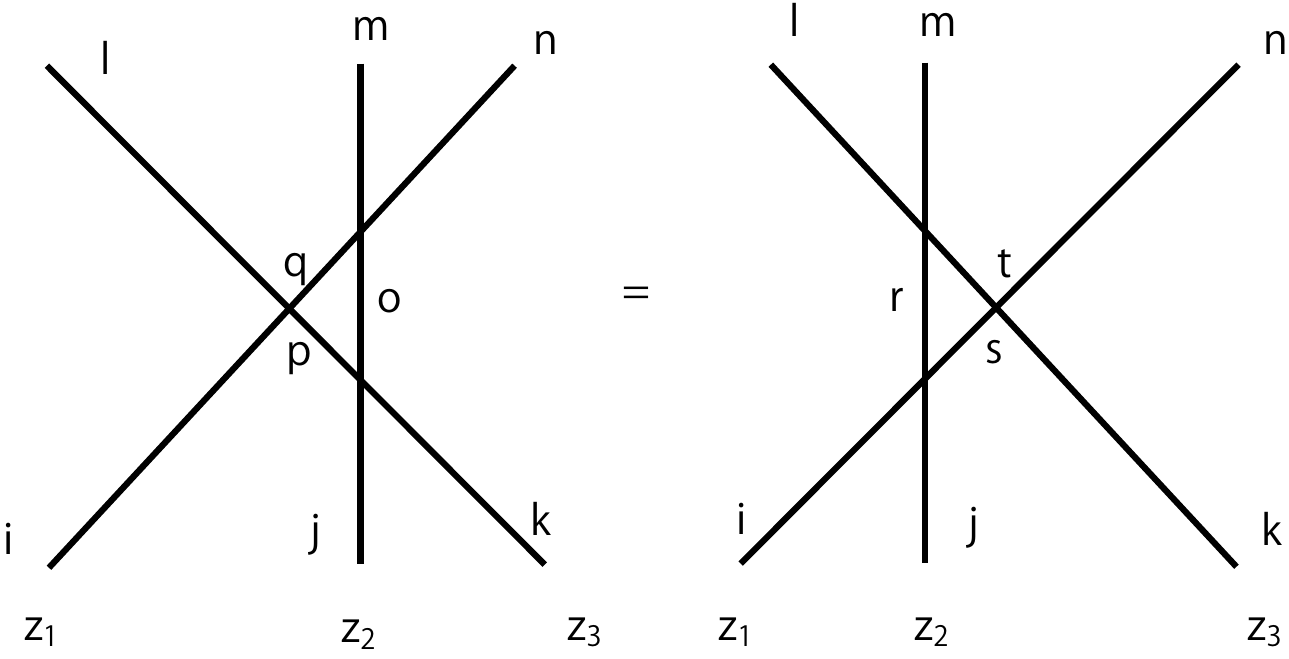}}
\caption{\small{A graphical representation of the Yang-Baxter equation.  On the left, one sums over labels $p,q,o$, and on the
right one sums over $r,s,t$.  An appropriate $R$-matrix element is attached to each vertex.}}
\label{FigureYBE}
\end{figure}

\begin{figure}[htbp]
\centering{\includegraphics[scale=.35]{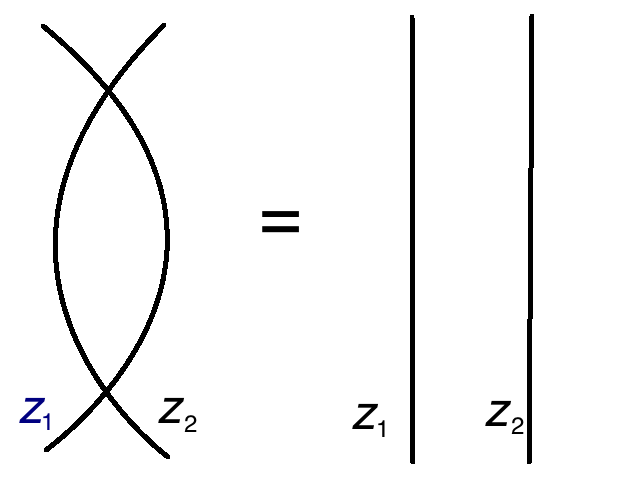}}
\caption{\small{In the context of the Yang-Baxter equation, ``unitarity'' is the equivalence of these two pictures.}}
\label{unitarity}
\end{figure}
The Yang-Baxter equation can be naturally supplemented with a condition sometimes called ``unitarity,''
which asserts that a picture in which two worldlines cross and then cross back is equivalent to one in which they
do not cross at all (\fig \ref{unitarity}).  In formulas, the relation is
\be\label{zunitarity}
R_{21}(z_2-z_1)R_{12}(z_1-z_2)=1 \;.
\ee
All solutions of the Yang-Baxter equation studied in this paper satisfy such a unitarity condition.  There is also a crossing
relation, which we will come to in sections  \ref{elementary} and \ref{firstlook}.

The Yang-Baxter equation is highly over-constrained, especially if the dimension of $V$ is large:
the $R$-matrix  has $\mathcal{O}((\dim V)^4)$ coefficients, while the 
Yang-Baxter equation has   $\mathcal{O}((\dim V)^6)$ components.
Moreover, the presence of the spectral parameter further constrains the possible solutions to the Yang-Baxter equation.
Nevertheless, the Yang-Baxter equations does have solutions and these lead to a remarkably rich theory.

\subsection{\texorpdfstring{Quasi-Classical $R$-matrix}{Quasi-Classical R-matrix}}\label{sec.quasi_classical}

While no complete classification is known of the general solution of the Yang-Baxter equation,
there are more complete results for the case of a so-called quasi-classical $R$-matrix,\footnote{
There are known solutions of Yang-Baxter equations which are not quasi-classical. 
The chiral Potts model \cite{AuYang:1987zc,Baxter:1987eq} is again a basic example.
}
a concept that we now explain.

A quasi-classical $R$-matrix is a solution $R_\hbar(z)$ of the Yang-Baxter
equation that depends on another continuous parameter $\hbar$ as well as on the spectral
parameter $z$, and that is holomorphic near $\hbar=0$ with $R_0(z)=1$. 
Thus $R_\hbar(z)$ has an expansion near $\hbar=0$ that begins
\begin{align}
R_{\hbar}(z)= I + \hbar \, r(z) + \mathcal{O}(\hbar^2) \;.
\label{eq.R_expand}
\end{align}
Here $r(z)$ is called the classical $r$-matrix.

By considering  the $\mathcal{O}(\hbar^2)$ term of the Yang-Baxter equation \eqref{YBE}, one learns
that the classical $r$-matrix obeys an equation that is known as the classical Yang-Baxter equation:
\begin{align}
[r_{12}(z_1-z_2), r_{13}(z_1-z_3) +r_{23}(z_2-z_3)]+[r_{13}(z_1-z_3), r_{23}(z_2-z_3)]=0 \;.
\label{eq.classical_YBE}
\end{align}
Note that this equation is  quadratic in the classical $r$-matrix $r$, whereas
the original Yang-Baxter equation was cubic in the $R$-matrix $R$.

Now, Belavin and Drinfeld \cite{Belavin-Drinfeld} classified solutions of the classical Yang-Baxter equation \eqref{eq.classical_YBE}, modulo  trivial equivalences,\footnote{The
classical Yang-Baxter equation is invariant under conjugation and under adding to $r$ a multiple of the identity.  The latter possibility reflects the fact that the
Yang-Baxter equation is invariant under multiplying $R(z)$ by a function of $z$.} under certain assumptions.  The assumptions were motivated by the examples
which were known at that time.   The solution is assumed to be associated to the Lie algebra $\mathfrak{g}$ of a semi-simple\footnote{In gauge theory,
it is natural to consider the somewhat larger class of Lie groups consisting of those whose Lie algebra
admits an invariant, non degenerate bilinear form.   An important example is a  reductive group, which is locally the product of a semi-simple group and a torus (an abelian group).
We will find at least two reasons to consider reductive groups in this paper.  One reason is that the simplest example for some purposes is actually the case
$G=GL_N(\C)$, which is reductive but not semi-simple.  Another reason is that in the framework we will follow in this paper, trigonometric solutions of the Yang-Baxter equation are most easily understood starting with a gauge group that is reductive but not semi-simple.}
 Lie group $G$.
No reality condition will be important in the present paper, so we consider $G$ to be a complex semi-simple Lie group with complex Lie algebra $\mathfrak{g}$.
The classical $r$-matrix 
$r$ is assumed to be an element of $\mathfrak{g}\otimes \mathfrak{g}$:
\be
r(z)=\sum_{a,b} r_{ab}(z) (t^a \otimes t^b) \;,
\label{BD_Ansatz}
\ee
where $t^a$ are a basis of $\mf{g}$.
The classical $r$-matrix is assumed to be non-degenerate, namely $\det_{a,b}(r_{ab}(z))\not \equiv 0$. 

Then the result shows that the poles of $r(z)$ in the complex plane spans a lattice,
which is either of rank $0$, $1$ or $2$. 
Solutions of the classical Yang-Baxter equation of any of the three types are almost uniquely  determined\footnote{In the elliptic case, there is a discrete choice to be made.  For $\mathfrak{g}=\mathfrak{sl}_N$, one is free to pick a generator of the finite group $\mathbb{Z}_n$.  There also are some subtleties in the trigonometric case, involving the possibility of an ``external field.'' We explain in sections \ref{trigonometric} and \ref{elliptic} what these issues mean from the point of view of four-dimensional gauge theory.}
by the choice of $\mathfrak{g}$ and the rank of the lattice.  The solutions for rank 0, 1, and 2
can be written explicitly in terms of rational, trigonometric, and elliptic functions and are known as rational, trigonometric, and elliptic solutions of the Yang-Baxter equation.

In the language of representation theory, rational, trigonometric and elliptic solutions have their algebraic counterparts, namely the Yangian $Y_{\hbar}(\mathfrak{g})$, 
the quantum affine algebra $U_{q, \hbar}(\mathfrak{g})$
and the elliptic algebra $E_{q, \tau, \hbar}(\mathfrak{g})$.   Solutions of the full (rather than classical) Yang-Baxter equation depend on the 
choice of a representation of one of these algebras and 
the $R$-matrix $R_{\hbar}(z)\in \textrm{End}(V \otimes V)$ is then an intertwiner for the 
tensor products of these representations.   In the algebraic approach, the spectral parameter enters as part of the data needed to specify a representation.

Elliptic solutions of the classical Yang-Baxter equation exist only for  $\mathfrak{g}=\mathfrak{sl}_N$,
whereas trigonometric and rational solutions exist for 
any semisimple $\mathfrak{g}$.   Rational solutions of the classical Yang-Baxter equation have $G$ as a group of symmetries,
while trigonometric solutions admit only the maximal torus of $G$ as a symmetry group and elliptic solutions for $\mathfrak{sl}_N$ have
only a finite group $\Z_N\times \Z_N$ of automorphisms.

In this paper, we will see how quasi-classical $R$-matrices with these properties can emerge from four-dimensional gauge theory.

\section{Four-Dimensional Gauge Theory}

\subsection{The Starting Point}\label{starting}
The  four-dimensional gauge theory that is relevant to our subject \cite{CostelloA,Costello:2013zra} may be described as follows.

The theory in question is only defined on 
a four-manifold with some additional structure. We start with the basic case, which is 
  a product 4-manifold $\R^2 \times \C$ with real coordinates\footnote{When convenient, we will denote $x$ and $y$ as $x^1$ and $x^2$.}
   $x,y$ on $\R^2$ and a holomorphic  coordinate $z$  on $\C$.

\def\g{\mf{g}}
The fundamental field of our theory is a $3$-component partial connection 
\be 
A = A_x \d x + A_y \d y + A_{\zbar} \d \zbar \;,
\ee 
where we did not include the component $A_z \d z$ that is of type $(1,0)$ along $\C$.
The fields $A_x$, $A_y$, and $A_{\bar z}$ all depend nontrivially on $z$ and $\bar z$ as well as $x$ or $y$;
that is, they are not constrained to vary holomorphically or antiholomorphically on $\C$.   Since $A_z$ is missing,
it would not be possible to place a reasonable reality condition on this space of fields.  Instead we take the gauge group
to be a complex Lie group $G$ with complex
 Lie algebra
$\mf{g}$, and view $A_x$, $A_y$, and $A_{\bar z}$ as independent complex fields. The construction will make
use of an invariant and nondegenerate bilinear form on $\g$, which we will denote as $\Tr$. 
The notation is motivated by the fact that if $\mathfrak g$ is semisimple or more generally if it is reductive (the direct sum of
a semi simple Lie algebra with an abelian one), then an invariant quadratic form can be defined as the trace in a suitable
representation.  However, our discussion in this paper applies whether $\Tr$ actually has this interpretation or not.  
For a simple summand of $\mathfrak{g}$, the Killing form of $\mathfrak{g}$ gives an invariant nondegenerate quadratic
form.   We choose an orthonormal basis $t_a$ of $\mathfrak{g}$ with respect to this Killing form and normalize $\Tr$ by
\begin{align}
\label{Killing_normalization}
\textrm{Tr}(t_a t_b)= \delta_{ab} \;.
\end{align}

The action of our theory is given by 
\begin{align}
S
=\frac{1}{2\pi } \int_{\R^2 \times \C} \d z \wedge  \textrm{CS}(A)
\;,
\label{eq.action}
\end{align} 
where $\textrm{CS}(A)$ is the Chern-Simons three-form
\begin{align}
\textrm{CS}(A):=\textrm{Tr}\left(A \wedge \d A +\frac{2}{3} A\wedge A\wedge A\right)=
\vepsilon^{ijk} \textrm{Tr}\left(A_i \partial_j A_k +\frac{2}{3} A_i A_j A_k\right) \;.
\end{align} 
Here and afterwards the indices $i, j, \ldots$ run over $x, y$ and $\bar{z}$
($\vepsilon$ is a totally antisymmetric tensor with $\vepsilon^{x y \bar{z}}=1$).
The VEV (vacuum expectation value) of an observable $\mathcal{O}$ is given by the path-integral
\begin{align}
\label{path-integral}
\langle \mathcal{O} \rangle= \frac{ \displaystyle \int \mathcal{D} A \, \mathcal{O} \exp\left(\frac{\i S}{\hbar}\right)}{ \displaystyle \int \mathcal{D} A \, \exp\left(\frac{\i S}{\hbar}\right)} \;.
\end{align}

The action $S$ is obviously not invariant under four-dimensional diffeomorphisms, because the use of the 1-form $\d z$ spoils the four-dimensional symmetry.
Nor does it have the three-dimensional diffeomorphism symmetry of three-dimensional Chern-Simons theory; this is the symmetry that enables one to define quantum invariants
of knots.  But we still have two-dimensional diffeomorphism symmetry -- invariance under orientation-preserving diffeomorphisms of $\R^2$ (or of its generalization $\Sigma$ that
will be introduced later).  This will ultimately lead to the Yang-Baxter equation and the unitarity relation.

We understand the action $S$ as  a holomorphic function of complex variables $A_x,$ $A_y$, $A_{\bar z}$, and this implies
that  the construction that we will be describing is somewhat formal.  There is no difficulty
in formally carrying out perturbation theory in such a holomorphic theory.  That approach was taken in \cite{CostelloA,Costello:2013zra} and it is the approach that we will follow here.  (We expect that a nonperturbative definition of the theory can be given by considering the D4-NS5 system of string theory, along the lines of the study of the D3-NS5 system in \cite{Witten:2011zz},  but we will not pursue this in the present paper.)   The parameter $\hbar$ that appears in the action is, at the quantum level,
 the loop-counting parameter.  In the semi-classical limit $\hbar\to 0$, this parameter will be identified with the parameter
of the same name that appears in the quasi-classical $R$-matrix
 \eqref{eq.R_expand}.   The parameter $\hbar$ has dimensions of length, in the sense that for  $C=\C$, the theory is invariant under a common
 rescaling of $z$ and $\hbar$.  
The factor of  $1/(2\pi)$ in the action is included here to match with the literature on integrable models. 

A reflection of the fact that the construction is formal and leads (in the form we present here) only to a perturbative theory is the following.
There is no quantization condition for $\hbar$ that will ensure that the action is gauge-invariant mod $2\pi \Z$.  This contrasts
with three-dimensional Chern-Simons theory, which is defined with such a condition. 

The action is invariant, modulo surface terms that are irrelevant in perturbation theory, under gauge transformations acting in the usual way.
\be 
A_{i} \mapsto g^{-1} A_i g +  g^{-1} \partial_i g \;,\quad (i=x,y, \bar{z}) \;.
\label{eq.gauge_transf}
\ee 
This is true because the Chern-Simons three-form is gauge-invariant modulo an exact form.  Alternatively, we can integrate by parts
to put the action in a manifestly gauge-invariant form,  after discarding surface terms
that are irrelevant in perturbation theory: 
\be \label{varthaction}
S=-\frac{1}{2\pi} \int_{\R^2 \times \C} z\,\Tr\, F\wedge F  
\;.
\ee 
This is the standard topological term of the Yang-Mills theory, where the $\theta$-angle now depends linearly on $z$.   

Some readers might be more comfortable starting with a standard $4$-component connection
\be 
A = A_x \d x + A_y \d y + A_{\zbar} \d \zbar + A_z \d z \;,
\label{4_component}
\ee 
with gauge transformations acting in the usual way on all four components, and again with the action (\ref{eq.action}).
In this case, one finds that due to the presence of the differential form $\d z$ in the action \eqref{eq.action},
the $A_z$ component drops out from the action, and hence we have an extra gauge symmetry
\be 
A \mapsto A + \chi \, \d z \;.
\label{eq.newgauge}
\ee 
We can then fix this extra gauge symmetry by choosing a gauge 
$A_z=0$. The $4$-component gauge transformation for the $4$-component gauge field \eqref{eq.newgauge}
is not consistent with this gauge since it 
will in general generate a non-trivial $A_z$ component. However,
a combination of the $4$-component gauge transformation, with
the extra gauge symmetry \eqref{eq.newgauge}
with $\chi=-A_z$, remains as a residual gauge symmetry. This is the $3$-component 
gauge transformation \eqref{eq.gauge_transf}.

In the following, we will always choose $A_z=0$, so that $A$ is the $3$-component connection
and the only remaining 
gauge symmetry is the   conventional
gauge transformation \eqref{eq.gauge_transf}.

A possibly more familiar theory that is defined in a similar way with a partial connection is holomorphic
Chern-Simons theory.  This theory is defined on a Calabi-Yau threefold $X$ with holomorphic 3-form $\Omega$.
The dynamical variable is a $(0,1)$ connection $A=\sum_{i=1}^3 A_{i} \d \zbar^{i}$ and the action
is the integral of the Chern-Simons $(0,3)$-form, wedged with $\Omega$:
\be \label{elbo}
S= \int_{X} \Omega \wedge  \textrm{CS}(A^{(0,1)}) \;.
\ee 
The definition of this action  depends only on the complex structure and holomorphic volume form $\Omega$ of the 3-fold $X$.

In this light, the four-dimensional theory of \eqref{eq.action} is intermediate between  ordinary Chern-Simons theory in three dimensions
and holomorphic Chern-Simons theory on a Calabi-Yau threefold. These theories arise as effective theories of branes in the topological
A-model and B-model respectively \cite{Witten:1992fb} and are related by mirror symmetry.  The four-dimensional theory
that we will be studying here is intermediate between the two cases and on an appropriate four-manifold
can be related by $T$-duality -- mirror symmetry in some but not all dimensions of spacetime  --
to either one of them.

The classical equations of  motion of the theory  read
\be \label{zelbo}
F_{xy}=0 \; , \quad
F_{x \bar{z}} = F_{y \bar{z}} =0 \;.
\ee 
This means that the gauge field defines a flat bundle on $\mathbb{R}^2$, which then varies holomorphically as we 
move along $\mathbb{C}$.

The equations (\ref{zelbo}) imply that all local gauge-invariant quantities  that can be constructed from the field $A$ actually vanish.
This is the reason that the theory works at the quantum level.  Because the loop-counting parameter $\hbar$ has dimensions
of length or inverse mass, the theory is unrenormalizable by power-counting.  But this does not cause difficulty because all conceivable counterterms
actually vanish by the equations of motion.  The theory thus can be quantized in perturbation theory \cite{CostelloA,Costello:2013zra}.  However, it
is affected by framing anomalies somewhat similar to those of three-dimensional Chern-Simons theory, but more subtle.  

The fact that the theory is unrenormalizable by power counting actually leads to a very important simplification.   After gauge-fixing, when
one concretely constructs the theory in perturbation theory, it is infrared-free.  The fact that the theory is infrared-free makes it straightforward,
once one introduces Wilson line operators, to deduce a local procedure to compute their expectation values.  From this local procedure,
one then can immediately recover the Yang-Baxter equation of an integrable system. (This will be explained in detail in section \ref{ybeu}.)  By contrast, three-dimensional Chern-Simons theory is
renormalizable by power counting and does not lead as directly to a local picture.  

\subsection{Generalization}
We  comment next on replacing $\mathbb{R}^2\times \mathbb{C}$ by a more general $4$-manifold.

In this paper, we will exclusively study the special case that the $4$-manifold is a product of two Riemann surfaces, 
\be 
M = \Sigma \times C \;,
\label{eq.product_mfd}
\ee 
where $\Sigma$ is a smooth oriented 2-manifold,  and  $C$ is a complex manifold endowed with a holomorphic (or sometimes meromorphic,
as discussed shortly) 1-form $\omega$,
which plays a role similar to $\Omega$ in the holomorphic Chern-Simons theory of \eqn (\ref{elbo}).
We will  sometimes refer to $\Sigma$ as the ``topological plane'' and $C$ as the ``holomorphic plane,''
though in general neither one of them is really a plane.  

We can then define a natural generalization of the action \eqref{eq.action} by 
\be
I=\frac{1}{2\pi} \int_M \omega \wedge \textrm{CS}(A)
 \;.
\label{eq.action_2}
\ee
As long as $\omega$ is closed, this action is gauge-invariant modulo total derivatives that do not affect perturbation theory.

Now we should discuss the possible role of zeroes and poles of $\omega$.   Naively, since the action involves  only the ratio $\omega/\hbar$,
a zero of $\omega$ corresponds to a point at which $\hbar\to \infty$.  Thus, in a theory that one only knows how to define perturbatively,
it should not be straightforward to make sense of the behavior near a zero of $\omega$.   We expect that essentially new ingredients are needed
to make sense of that behavior.  We will not explore this issue in the present paper.

Conversely, near a pole of $\omega$, $\hbar$ is effectively going to zero and perturbation theory should be within reach.  However, poles
of $\omega$ are still subtle for the following reason.  If $\omega$ has a pole at a point $p\in C$, then $\d\omega$ does not vanish near $p$ but is a distribution
supported at $p$.  Accordingly, gauge-invariance will fail unless we place some suitable conditions  near $p$ on the gauge field $A$ and the gauge parameter
$g$.  If $\omega$ has a double pole at $p$, one can restore gauge invariance by asking that $A=0$ at $p$ and $g=1$ at $p$.  What one
has do if $\omega$ has a simple pole at $p$ is more subtle and will be described in section \ref{trigonometric}.

Only simple and double poles are relevant, as one sees if one considers
the possibilities for a complex Riemann surface $C$ with a holomorphic one-form $\omega$ that is allowed to have poles but not  zeroes.
By the Riemann-Roch theorem, the number of zeroes of any meromorphic differential $\omega$ minus the number of its poles is $2g-2$, where $g$ is the genus of $C$.
Thus if $\omega$ has no zeroes, $C$ must have genus 0 or 1.
Moreover,  for $g=0$
we have either (1) a single pole with multiplicity $2$, which corresponds to $C=\mathbb{C}$ with differential $\d z$ (which has a double
pole at $\infty$)
or (2) two simple poles,  in which case we can take  $C=\mathbb{C}^{\times}=\mathbb{C}/\mathbb{Z}$ with differential $\omega=\frac{\d z}{z}$, which has
simple poles at 0 and $\infty$.
For $g=1$, there are no poles at all; $C$ is a complex torus or  elliptic curve $\mathbb{C}/(\mathbb{Z}+\tau \mathbb{Z})$ (with modulus $\tau$) with the holomorphic
differential $\d z$.  In each case, the choice of $\omega$ is unique up to a normalization constant that can be absorbed in rescaling $\hbar$.

Summarizing, we have  the following three possibilities for $C$:
\begin{align}
\begin{aligned}
\label{mirf}
 &C=\mathbb{C} \;,&& \omega=\d z \;,&& \textrm{double pole at }\{ \infty\}\;, && \textrm{(rational)}  \;, \\
  &C=\mathbb{C}^{\times} \;,&& \omega=\frac{\d z}{z} \;,&& \textrm{poles at }\{ 0, \infty\}  \;,&& \textrm{(trigonometric)} \;, \\
   &C=E=\mathbb{C}/(\mathbb{Z}+\tau \mathbb{Z}) \;,&&\omega=\d z \;,&&\textrm{no poles}\;,&&\textrm{(elliptic)} \;.
\end{aligned}
\end{align}
As indicated, the three choices of $C$ match the three broad classes  of quasi-classical R-matrices that were summarized in section \ref{sec.quasi_classical}, 
if we assume that $C$ parametrizes the spectral
parameter of the classical $r$-matrix.  Developing this relationship is the purpose of the present paper.

A notable fact is that the three examples are all abelian groups.  This is no coincidence, of course.  Since the holomorphic differential $\omega$ on $C$
has no zeroes, its inverse is a holomorphic vector field $\zeta=\omega^{-1}$ that generates an abelian group symmetry.    For the three cases,
in the coordinates used in \eqn (\ref{mirf}), the group action is $z\to z+a$ in the case that $C$ is the complex plane or an elliptic curve, or $z\to \lambda z$
in the case of $\C^\times$.  It is because of this group action, which is a symmetry of the action (\ref{eq.action}) and the theory constructed from it,
that the $R$-matrix $R(z_1,z_2)$ that we eventually construct is a function only of the difference $z_1-z_2$ or the ratio  $z_1/z_2$, as the case may be.  

Though this will not be developed in the rest of the paper, we will  briefly describe a more general possible choice of 
$4$-manifold. 
Suppose that the $4$-manifold $M$
admits a complex-valued closed $1$-form $\omega$.
We require that $\op{Re} \omega$ and $\op{Im} \omega$ are everywhere linearly independent.
This means that locally $\omega=\d f+\i \, \d g$, where $f$ and $g$ are real-valued functions and $\d f$ and $\d g$ are linear independent.  $M$
can then locally be foliated by the smooth two-manifolds that are defined by setting $f$ and $g$ to constants.  Thus $M$ has a two-dimensional
integrable foliation.   The gauge field $A$ is a 3-component partial $\g$-valued connection, or alternatively it
is an ordinary connection with the 
extra gauge symmetry
\be 
A \mapsto A + \chi \,\omega \;.
\label{eq.newgauge2}
\ee 
The action is still given by \eqref{eq.action_2}.

The considerations to this point have been purely classical, but there are important quantum corrections.
As we discuss briefly in section \ref{firstlook} (see \cite{CostelloA,Costello:2013zra} for a detailed account), at the quantum level 
there is a framing anomaly which means that we can only define the theory on some, but not all, $4$-manifolds $M$
of the type mentioned above. For the product manifold of \eqref{eq.product_mfd}, the framing anomaly implies that
$\Sigma$ must be equipped with a framing.  In particular, if compact, $\Sigma$ must be a two-torus.
For integrable lattice models associated to solutions of the Yang-Baxter equation, 
the most important examples are that $\Sigma$ is $\R^2$ or a two-torus.

\subsection{Wilson Lines}\label{subsec.Wilson}

\def\h{\widehat}

Now let us consider the gauge-invariant operators of the theory.  There are no local ones because they all vanish by the equations of motion.
The simplest gauge-invariant operators -- and the only ones that we will study in the present paper -- are Wilson line operators.  

In ordinary gauge theory, a natural gauge-invariant quantity is the trace, in some representation $\rho$ of the gauge group,
of the holonomy of the connection around a closed loop $K$.  Quantum field theorists usually write this quantity as
\be 
 W_{\rho}(K) =
 \Tr_{\rho} P  \exp\left( \oint_K A_i(x^1,x^2,z,\bar z)\d x^i\right),
\label{eq.W_rho}
\ee 
where $P$ denotes  path-ordering along the loop $K$, and $\Tr_{\rho}$ the trace
in the representation $\rho$.
In the present context, $K$ cannot be an arbitrary loop in the four-manifold $\Sigma\times C$.  On the contrary,
because we only have a partial connection with no $\d z$ term, there is no notion of parallel transport in the $C$ direction.\footnote{Either there is no $A_z$ and no way to
define parallel transport along a path on which $z$ is not constant, or there is an $A_z$ but also an extended gauge invariance (\ref{eq.newgauge2}), and parallel
transport in the $z$ direction is not gauge-invariant.}  Accordingly,
we are restricted to the case that $K$ is a loop in the topological plane $\Sigma$, at a specified point\footnote{This classical statement 
will later be subject to some revision because of the framing anomaly.}
 $z=z_0$ in $C$.  This already makes contact in a preliminary way
 with some aspects of the standard Yang-Baxter picture that we reviewed in section \ref{YBeq}.  A Wilson operator is supported on a 1-manifold $K$ in the two-manifold $\Sigma$
 (which one can think of as the worldline of a particle in a two-dimensional spacetime), and it is a labeled by a spectral parameter, that is, by a point in $C$, and by a choice of a representation $\rho$ of $G$.  Here $\rho$ will play the role of the vector space $V$ of internal states
 of a particle, introduced at the beginning of section \ref{YBeq}.  

\def\t{\text{\sf t}}
We can also introduce more general Wilson operators, which do not have analogs in standard gauge theories.  The existence of these operators
is related to the fact that the loop $K$ is highly restricted, as described in the last paragraph.
At the classical level, these Wilson loops are labeled by a representation $\what{\rho}$ of the infinite-dimensional Lie algebra $\g[[z]] =\prod_{n\ge 0}  (\g \otimes z^n)$ of series in $z$ whose coefficients are in the finite-dimensional Lie algebra $\g$. 
(The same algebra $\g[[z]]$ will appear regardless of the choice of $C$ because the considerations will be local along $C$.)   To be more exact, Wilson
operators supported at $z=0$ will be associated to representations of $\g[[z]]$.  Wilson operators supported at $z=z_0$ are similarly associated
to representations of $\g[[z-z_0]]$ (which is obtained from $\g[[z]]$ by $z\to z-z_0$).

Since the relevant concepts may be unfamiliar,
we pause for an explanation.   Roughly speaking, an element of $\g[[z]]$ is
a $\g$-valued function of $z$.   If $\g$ has a basis $t_a$, $a=1,\dots,\mathrm{dim}\,\g$, then a basis, in the relevant sense, of the space of $\g$-valued functions of $z$
is provided by
\be\label{melz}
\t_{a,n}(z)=t_az^n,~~~ n\geq 0 \;.
\ee
So to define a representation of $\g[[z]]$, we need to give, for every $a$ and $n$, a matrix (or operator) $t_{a,n}$ that represents the action of $\t_{a,n}(z)$.
Concretely, if $\g$ has a basis $t_a$ with $[t_a,t_b]=f_{ab}{}^c t_c$, then the natural commutation relation for $\g$-valued functions of $z$ is
$[t_a z^n, t_b z^m]=f_{ab}{}^c t_c z^{n+m}$.  Therefore, the corresponding representation matrices should obey
\be\label{welz} 
[t_{a,n},t_{b,m}]=f_{ab}{}^ct_{c,n+m} \;. 
\ee
It is important that there is no central extension here and that this algebra has finite-dimensional representations.  
We will be primarily interested in finite-dimensional representations, and more specifically representations with the property
that   there is some $n_0$ such that $t_{a,n}=0$ for $n\geq n_0$.   To orient the reader,
we consider the  first nontrivial
example, which arises for $n_0=2$.   The nonzero generators are just $t_{a,0}$, which generates the finite-dimensional algebra $\g$, and $t_{a,1}$, which commutes
with itself and transforms in the adjoint representation of $\g$. To construct a representation $\h\rho$ of this algebra, we can take a direct sum $\rho\oplus \rho$ of two copies
of any representation $\rho$ of $\g$.  If $t_a$ are the representation matrices of $\rho$, then a representation of $\g[[z]]$ is given by
\be\label{zelx}
t_{a,0}=\begin{pmatrix}t_a& 0 \cr 0& t_a\end{pmatrix} \;, \quad
 t_{a,1}=\begin{pmatrix}0&t_a\cr 0&0 \end{pmatrix} \;. 
 \ee  
 This representation is indecomposable
as a representation of $\g[[z]]$, though it is decomposable as a representation of $\g$.
There are many elaborations on this theme with $n_0\geq 2$.  

Because the line $K$ in (\ref{eq.W_rho}) is supported
at a point in $C$, its  definition depends only on the components $A_i(x^1,x^2,z,\bar z)$ of the gauge field, with $i=1,2$. 
We will preserve this fact in defining (at the classical level) more general Wilson operators.
 For fixed $x^1$ and $x^2$,
$A_i(x^1,x^2,z,\bar z)$ is a $\g$-valued function of $z$ and $\bar z$.   Formally setting $\bar z=0$, we get (for each point in the topological
plane and each $i$) a $\mathfrak{g}$-valued function of $z$,
namely $\h A_i(x^1,x^2,z)=A_i(x^1,x^2,z,0)$.   A precise definition of what we mean by setting $\bar z$ to 0 with $z$ fixed is that we define
\be 
\what{A}_i (x^1,x^2,z):=\sum_{k\ge 0} \frac{z^k}{k!}\left.\frac{\partial^k}{\partial z^k}A_i(x^1,x^2,z,\bar z)\right|_{z=\bar z=0}. 
\ee 
We do not need to worry about convergence of this series, because we consider representations that are annihilated by a sufficiently high power of $z$.  Thus for any given
representation, we can terminate the series after finitely many terms and consider $\h A_i$ to have a polynomial dependence on $z$.

Next we consider gauge transformations.   The generator of such a gauge transformation is a $\g$-valued function $u(x^1,x^2,z,\bar z)$.  We can restrict such a function to $\bar z=0$,
in the same sense described in the last paragraph, and extract a $\g$-valued function $\h u(x^1,x^2,z)=u(x^1,x^2,z,0)$, which we can interpret as a $\g[[z]]$-valued function of $x^1,x^2$.
The theory we are studying is invariant under a gauge transformation of $A_i$ generated by $u(x^1,x^2,z,\bar z)$.  When we restrict to $\bar z=0$, the action of $u(x^1,x^2,z,\bar z)$ on $A_i$
becomes an action of $\h u(x^1,x^2,z)$ on $\h A_i(x^1,x^2,z)$.  Here we can view $\h u(x^1,x^2,z)$ as a $\g[[z]]$-valued function on $\Sigma$, and its action of $\h A_i$ is the natural
action of such a function on $\h A_i$, viewed as a $\g[[z]]$-valued gauge field on $\Sigma$.   

So finally if $\h\rho$ is a representation of $\g[[z]]$ of the allowed class, we can define a corresponding Wilson operator 
by modifying eqn.\ (\ref{eq.W_rho}) in an almost trivial way:
\be
 W_{\h\rho}(K) =
 \op{Tr}_{\h\rho} P  \exp\left( \oint_K A_i(x^1,x^2,z,0)\d x^i\right) .
\label{eqW}
\ee
Here $A_i(x^1,x^2,z,0)$ is expanded around $z=0$ with fixed $\bar z=0$.   Wilson lines supported at some other point $z=z_0$ in $C$ are
similarly defined by expanding around $z=z_0$ with fixed $\bar z=\bar z_0$.   (In this case, one considers representations of $\g[[z-z_0]]$.)

However, here we should point out a crucial subtlety that is important for applications of  the extension to $\g[[z]]$.  It is very undesirable to take a trace in
eqn.\ (\ref{eqW}) because this causes much of the interesting structure to disappear.  We can see that by going back to eqn.\ (\ref{zelx}).  This representation is indecomposable
as a representation of $\g[[z]]$, and it does not just come from a representation of the finite-dimensional algebra $\g$.  The holonomy operator along a given path in this representation is different from
what it would be if one sets $t_{a,1}=0$, which would give a decomposable representation of $\g[[z]]$.  But if we take the {\it trace} of the holonomy, then $t_{a,1}$ will play no role
because it is strictly upper triangular, and we cannot distinguish the given representation from its decomposable cousin.

Because the theory is infrared-free, as explained at the end of section \ref{starting}, the holonomy itself rather than its trace is a meaningful observable.  To see this, we take $\Sigma=\R^2$
and we consider a 1-manifold $K\subset \Sigma$ (supported at a point in $C$) that is not compact but has its ends at infinity along $\Sigma$.  Because the theory is infrared-free,
the gauge field can be considered to vanish at infinity along $\Sigma$ and then the holonomy along $K$ is a gauge-invariant observable, with no need to take its trace.  It is really this that gives power to the fact that the theory has Wilson operators associated to a large class of
representations of $\g[[z]]$.   The importance of the extension to $\g[[z]]$ will not be fully clear until we analyze the operator product expansion of Wilson
operators in section \ref{parallel}.

\subsection{The Yang-Baxter Equation and Unitarity}\label{ybeu}
\begin{figure}[htbp]
\centering{\includegraphics[scale=.35]{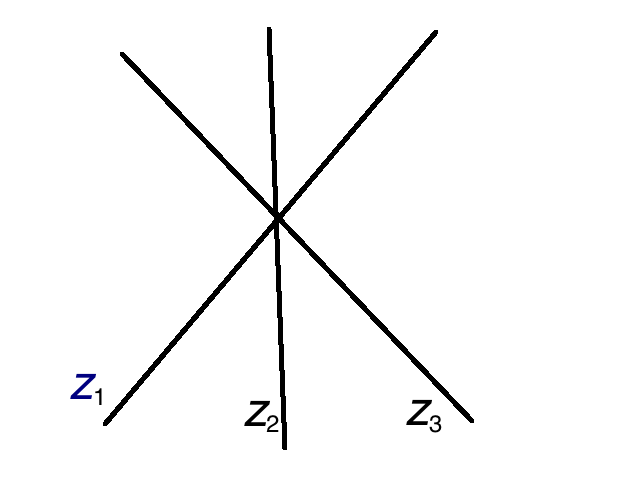}}
\caption{\small{Three lines meeting at a point in $\Sigma$, but with distinct values of the spectral parameters $z_i$.}}
\label{singular}
\end{figure}
Having defined Wilson operators, we can now return to Figs. \ref{YB2} and \ref{unitarity} in which the Yang-Baxter equation and the associated unitarity relation are illustrated.
We now interpret these figures as representing configurations of Wilson lines  $K_i\in \Sigma$ (supported at the indicated points  $z_i\in C$).  It is now not hard to argue for the equivalence of the left and right hand sides of these pictures.
  Two-dimensional diffeomorphism symmetry means that we are free to move the $K_i$ around as long as we do not change
the topology of the situation.  However, this alone is not quite enough to prove the equivalences suggested in the pictures.
  For example, in \fig \ref{YB2}, we are free to move the ``middle'' Wilson line to the left or right, as long as we do not try to pass through
a configuration  (\fig \ref{singular}) in which the three lines all meet at a point; such a configuration  is not equivalent by a diffeomorphism of $\Sigma$ to a configuration without
a triple intersection.  In trying to prove the  equivalences between the left and right of \fig \ref{YB2}
 by moving the middle Wilson line from left to right, we have to ask whether there is a discontinuity in the
path integral at the moment that a triple intersection occurs.  However, as long as the spectral parameters $z_1,z_2,z_3$ are all distinct, none of the lines are meeting in four
dimensions and it is manifest that the configuration that has a triple intersection when projected to $\Sigma$ is not associated to any singularity.\footnote{An interpretation of the
Yang-Baxter equation and the spectral parameter somewhat along these lines was conjectured by M.~F.~Atiyah in the 1980's \cite{Atiyah}.}  In
particular there is no discontinuity and the pictures on the left and right of
\fig \ref{YB2} are equivalent as long as the $z_i$ are distinct (in fact, it is enough that they are not all equal). Likewise the pictures on the left and right of \fig \ref{unitarity}
are equivalent as long as $z_1\not=z_2$.

It takes more than this to argue that the theory has an $R$-matrix that satisfies the Yang-Baxter equation and the unitarity relation.  The usual 
$R$-matrix formalism, as summarized in section \ref{YBeq}, involves a much more specific interpretation of the pictures.    Each line
in  \fig \ref{Rmatrix} is supposed have associated to it a space $V$ of ``internal states'' accessible to a particle.  In the present
framework, the meaning of this is clear: a line is a Wilson line associated to some representation $\rho$ of $\g$ (or more generally of
$\g[[z]]$), and $\rho$ corresponds to $V$.  But the usual $R$-matrix picture is much simpler than one would expect in 
quantum field theory in general.  
In the usual $R$-matrix picture, each line segment between two crossings is labeled
by a basis vector $e_i$ of $V$, and to a crossing one associates a local factor, the $R$-matrix element $R_{ij}^{kl}(z_1-z_2)$, which depends only on the data at a particular crossing and not on any other details in which the local picture is embedded.  Moreover, this
$R$-matrix element depends only on the difference $z_1-z_2$.   In quantum field theory in general, one would not expect a local
picture like this. 
Finally, though in standard presentations of $R$-matrix theory one might
take this for granted and skip it over, it is noteworthy that in $R$-matrix theory,
the two-dimensional regions bounded by the lines do not carry any labels.  
This is a nontrivial point and in fact there is a generalization of the Yang-Baxter equation (the dynamical Yang-Baxter equation \cite{GN,Felder,FelderTwo,Etinghof})
in which the bulk regions do carry labels, above and beyond the labels carried by the line segments.

In trying to explain these facts in the present context, the most basic question is why there is a local picture of any sort.  The reason for
this is that the theory is infrared-free, as was noted at the end of section \ref{starting}.  Concretely, in constructing perturbation theory,
as we will do starting in section \ref{lowestorder}, one picks a Riemannian metric on $\Sigma\times C$.  If one scales up the metric on $\Sigma$
by a large factor, so that different crossings are very far apart (compared to the distances between the points in $C$ at which a given
set of Wilson line operators are supported), 
then the infrared-free nature of the theory guarantees that some kind of local picture will be possible.

\def\H{{\mathcal H}}
To explain more, let us first ask what would happen in the absence of any line operators.  The theory under study is topological in the
$\Sigma$ direction, so (ignoring further subtleties that arise because we are dealing with a theory whose action is a holomorphic function
of complex variables) in general we would expect the theory to have 
 a space $\H$ of quantum states.  These would roughly correspond to vacua of a standard
quantum field theory.  In general, one would expect to label the regions between the lines -- that is, any region of $\Sigma$ that is
not near one of the Wilson operators -- by a basis vector of $\H$.   Accordingly, if $\H$ has dimension bigger than 1, we would get
something like the dynamical Yang-Baxter equation \cite{GN,Felder,FelderTwo,Etinghof}, with labels for regions as well as line segments,
rather than the standard Yang-Baxter equation in which regions between the lines are unlabeled.
We discuss this situation in section \ref{dybe}.  

To get something as simple as the standard Yang-Baxter equation, we want $\H$ to be one-dimensional, which will happen if the space
of classical solutions of the theory, modulo gauge transformations, is a point.   This is also the condition that eliminates the subtleties
associated with having a holomorphic action; perturbation theory is straightforward in principle if there is only one classical solution to
expand around, and it has only a finite group of automorphisms.\footnote{If there is a unique classical solution up to gauge transformation,
but it has a nontrivial automorphism group $H$, then in developing perturbation theory one wants to divide by the volume of $H$.
If $Hf$ is not a finite group, this volume might be hard to interpret.  However, this issue involves only an overall constant factor in
the path integral, independent of what collection of Wilson lines one considers.}
 The simplest case is that $C=\C$.  In quantizing the theory on $\Sigma\times \C$ for any $\Sigma$, we require
that the gauge field $A$ and the generator of a gauge transformation both vanish at infinity.  With this choice, the only (stable\footnote{There are many
classical solutions on $\Sigma\times \C$
 that correspond to bundles on $\C$ (trivialized at infinity) that are unstable in the sense of algebraic geometry.  This likely makes them unsuitable as a starting point for perturbation theory.
At any rate, the fact that is really important for us is that the trivial connection on $\Sigma\times \C$ is a classical solution that has no infinitesimal
deformations or gauge automorphisms.  In perturbation theory around this solution, we do not meet unstable bundles.}) classical solution,
up to a gauge transformation, is $A=0$.  So we are in the situation in which the bulk regions do not carry labels and perturbation theory
is straightforward in principle.   

It is likewise possible when $C$ is $\C^\times$ or an elliptic curve to ensure that the classical phase space is a point, leading
to straightforward perturbation theory and (as we argue shortly) a conventional Yang-Baxter equation.  The details are more involved and
we defer a discussion to sections \ref{trigonometric} and \ref{elliptic}.  

\begin{figure}[htbp]
\centering{\includegraphics[scale=.35]{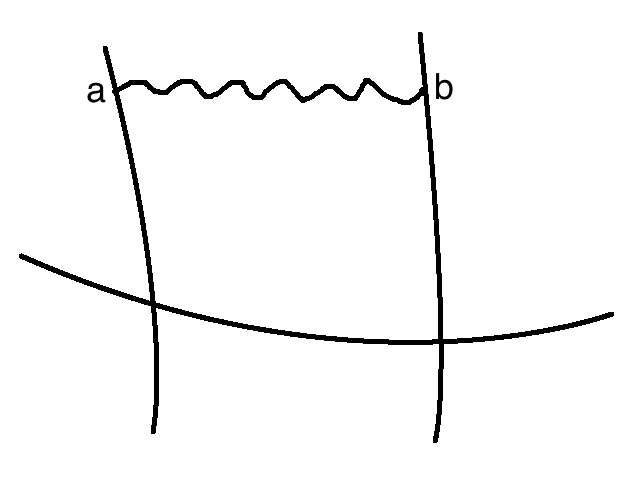}}
\caption{\small{The infrared-free nature of the theory under study means that when we scale up the metric of $\Sigma$ after gauge-fixing,
gluon exchange between Wilson lines that are not crossing becomes irrelevant.}}
\label{Irrelevant}
\end{figure}

Let us now imagine doing perturbation theory in an infrared-free theory in the presence of a configuration of Wilson lines.  What
sort of perturbative corrections are significant?   A typical example of an effect that is not significant
 is gauge boson exchange between two Wilson
lines that are not crossing (\fig \ref{Irrelevant}).  By scaling  up the metric of $\Sigma$, the points $a$ and $b$ in the figure
can be made arbitrarily far apart, regardless of where they lie on the Wilson lines in question, and the contribution of gluon exchange between them goes to zero.  So this can be ignored.

\begin{figure}[htbp]
\centering{\includegraphics[scale=.35]{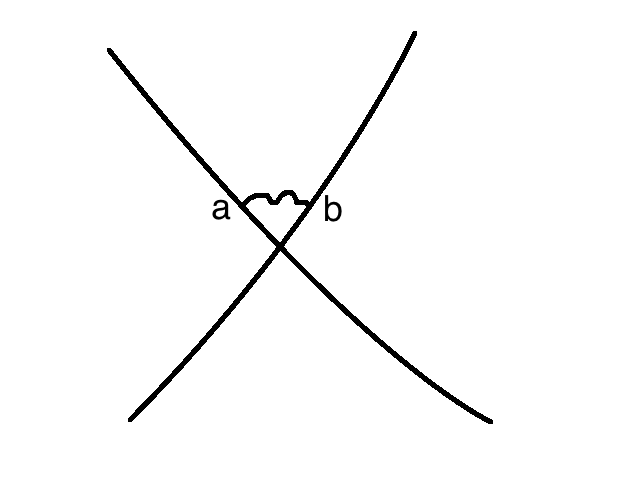}}
\caption{\small{The infrared-free nature of the theory under study means that when we scale up the metric of $\Sigma$ after gauge-fixing,
gluon exchange between Wilson lines that are not crossing becomes irrelevant.}}
\label{CrossingLines}
\end{figure}
A typical contribution that cannot be ignored is a gauge boson exchange between two lines that are crossing (\fig \ref{CrossingLines}).
In this case, the points $a$ and $b$ can be near the crossing point, so they cannot  be assumed to be far away in $\Sigma$.  
We can still scale up the metric in the picture in order to exploit
the infrared-free nature of the theory.  But all that happens when we do this is that the lines that are crossing turn into straight lines
near the point in $\Sigma$ where they cross, and they become very widely separated from any other crossing points.  So the diagram
of \fig \ref{CrossingLines} may be nontrivial -- and it is nontrivial, as we will calculate in section \ref{lowestorder} -- but it  will be local:
it will not depend on the details of a larger picture in which the crossing of \fig \ref{CrossingLines} might be embedded.  

Now we can put the pieces together and explain why something along the lines of standard Yang-Baxter theory will emerge.
We assume a situation with a unique classical solution\footnote{In the trigonometric and elliptic cases, this argument needs to
be stated a little more carefully.  There is a unique classical solution, and although it is not gauge-equivalent globally to $A=0$, this is true
locally.  Given this, the argument proceeds as in the text.} $A=0$.  Moreover, we know that quantum effects are negligible except near crossings.
Thus away from crossings, we can assume that $A=0$ everywhere.  This means that a line segment between crossings just describes
a free particle in the relevant representation of $\mathfrak{g}$ (or of $\g[[z]]$) and can be labeled by a basis vector in that representation.  
And crucially, the amplitude associated to a given crossing can only depend on the local data at that crossing -- the representations
and labels of the lines that are crossing. 
Thus the equivalence of the two pictures of \fig \ref{YB2}, which follows from rather general arguments that were given above,
turns into the more precise numerical equivalence of \fig \ref{FigureYBE}, with a local $R$-matrix at each crossing.  
The  same reasoning applies to the unitarity relation of \fig \ref{unitarity}.  Because of the infrared-free nature of the theory,
it turns into the concrete unitarity relation
$R_{21}R_{12}=1$ of Yang-Baxter theory.

For the case that $C$ is $\C$ or an elliptic curve, the
 local $R$-matrix $R_{ij}^{kl}(z_1,z_2)$ is actually a function only of the difference $z_1-z_2$, because the classical action
(\ref{eq.action}) is invariant under shifting $z$ by a constant.  For the case of $C=\C^\times $ with differential $\d z/z$, the equivalent
statement is that the $R$-matrix (written in these multiplicative coordinates) is a function of the ratio $z_1/z_2$.

For the case that $C=\C$, related to the Yangian,
 a few further nice things happen which make this case particularly simple and elegant.  First of all,
the action is invariant under a common rescaling of $z$ and $\hbar$, so actually the $R$-matrix is a function only of 
a single variable $(z_1-z_2)/\hbar$.  Second, in quantizing the theory with $C=\C$, we divide only by gauge transformations that
are 1 at infinity along $C$.  But we are left with gauge transformations that are constant at infinity along $C$, and these behave
as global symmetries.  Thus the $R$-matrix for $C=\C$ has $G$ as an automorphism group. (This is not true
for the other choices of $C$, as we will see in sections \ref{trigonometric} and \ref{elliptic}.)  The properties stated in this paragraph
make it straightforward to understand some simple examples.  We present some of these elementary examples in the next section.
We present them because they are fun -- though probably well-known to many readers -- and also because they enable one to
see in a completely direct and elementary way why the theory must have a framing anomaly.

\subsection{Elementary Examples}\label{elementary}

For some elementary examples, we take $G=GL_N$ (or $SL_N$, which would be equivalent for the purposes of this analysis),
and we will consider the case that $\rho$ is the fundamental representation of $G$ or its dual.   We  denote these
representations as $V$ and  $V^*$, respectively.
In all cases, we will take $C=\C$, so that the $R$-matrix has $G$ symmetry.

\begin{figure}[htbp]
\centering{\includegraphics[scale=.9]{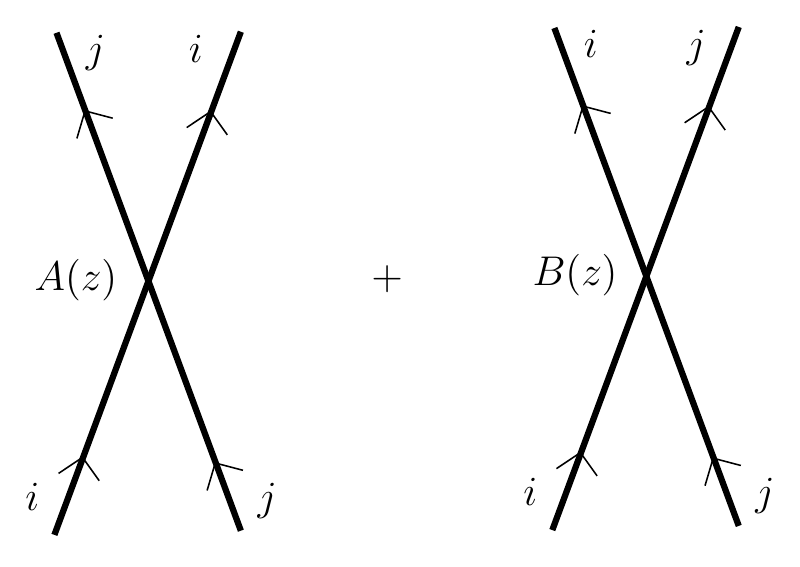}}
\caption{\small{The $R$-matrix for crossing of two Wilson lines in the fundamental representation $V$.  
The non-zero matrix elements of $R$ correspond to processes in which the final states equal the initial
states, possibly modulo a permutation.}}
\label{FigA}
\end{figure}

First we consider the $R$-matrix for crossing of two copies of $V$.  It will be a $G$-invariant linear map $R(z_1-z_2):V\otimes V\to
V\otimes V$.  Such an operator is a linear combination of the identity and the operator $P:V\otimes V\to V\otimes V$
that exchanges the two factors:
$R(z)=A(z)+PB(z)$, with $z=z_1-z_2$.   In this particular case,  $R(z)$  can also be written fairly conveniently  as a matrix with all its indices:
\be
\label{zelf}
R_{ij}^{i'j'}(z)= \delta_i^{i'}\delta_j^{j'} A(z)+\delta_i^{j'}\delta_j^{i'}B(z) \;. 
\ee
Here  $i$ and $j$ refer to ``incoming'' lines and $i'$ and $j'$ to outgoing ones.
The $A(z)$ term describes two lines crossing without ``charge exchange,'' while $B(z)$ describes
crossing with charge exchange.   The non-zero matrix elements of $R$ are depicted in Fig. \ref{FigA}.

\begin{figure}[htbp]
\centering{\includegraphics[scale=.7]{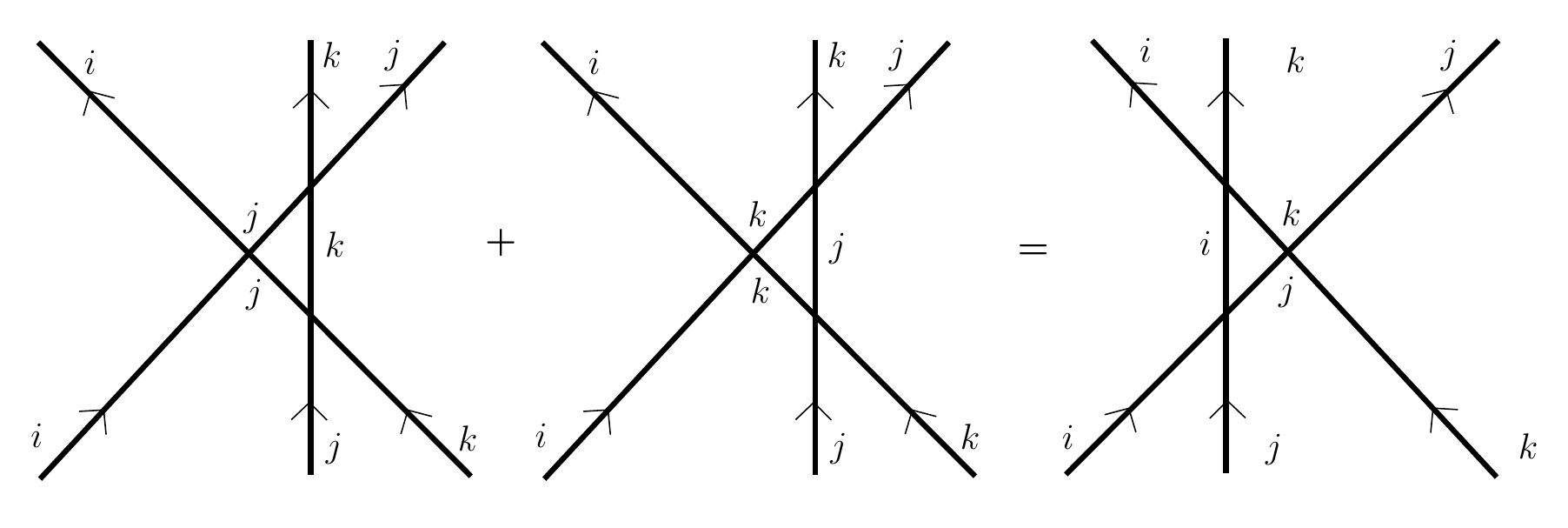}}
\caption{\small{The Yang-Baxter equation for crossing of three Wilson lines in the fundamental representation $V$.}}
\label{FigB}
\end{figure}
The Yang-Baxter equation is in general invariant under multiplying the $R$-matrix by a scalar function -- a $z$-dependent multiple of the identity.
So it is only sensitive to the ratio $U(z)=B(z)/A(z)$.  It is not difficult to work out the Yang-Baxter equation in this case (\fig \ref{FigB}) and to learn
that it is equivalent to 
\be\label{welf}U(z_1-z_3)U(z_2-z_3)+U(z_1-z_2)U(z_1-z_3) = U(z_1-z_2)U(z_2-z_3)  \;.   \ee
After dividing by the product $U(z_1-z_2)U(z_1-z_3)U(z_2-z_3)$, we learn that $1/U(z)$ is a multiple of $z$, so (remembering
that the $R$-matrix for $C=\C$ is a function of $\hbar/z$) $U$ must be a constant multiple of $\hbar/z$.  Determining the constant from 
the ${\mathcal O}(\hbar)$ contribution to the $R$-matrix (see section \ref{lowestorder}), we find
\be\label{merf}U(z)=\frac{\hbar}{z}.\ee

\begin{figure}[htbp]
\centering{\includegraphics[scale=.9]{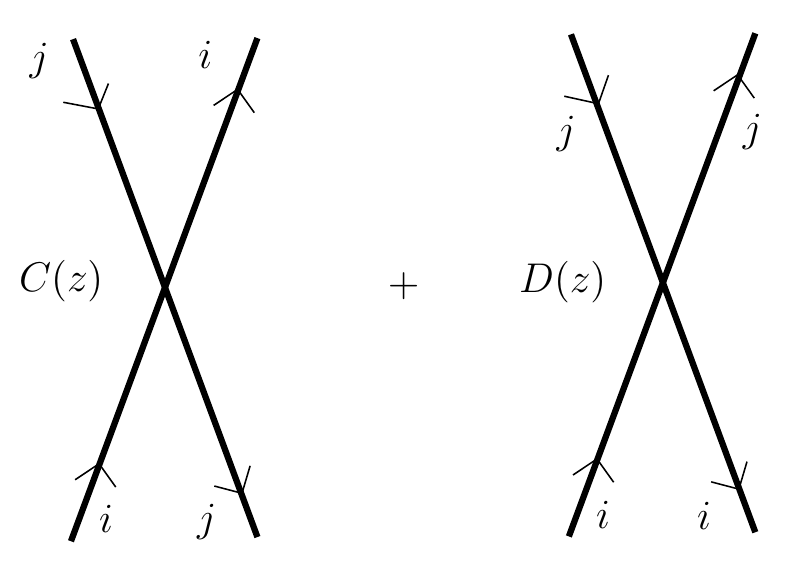}}
\caption{\small{The $R$-matrix for crossing of two Wilson lines in the dual representations $V$ and $V^*$.}}
\label{FigC}
\end{figure}
Now we consider the $R$-matrix for crossing of a copy of $V$ with a copy of the dual representation $V^*$.  The $R$-matrix is now
a linear map $R(z):V\otimes V^*\to V\otimes V^*$.  Again $R(z)$ is determined by the $G$ symmetry in terms of two functions:
$R(z)=C(z)+Q D(z)$.  Here $Q$ is the $G$-invariant projection operator from $V\otimes V^*$ to its $G$-invariant subspace.
A picture is rather clear (\fig \ref{FigC}), but now a formula analogous to eqn.\ (\ref{zelf}) is less transparent:
\be\label{relf} 
R_i^{j}{}^{i'}_{j'}(z)=\delta_i^{i'}\delta^j_{j'} C(z) +\delta_i^j \delta_{j'}^i  D(z) \;. 
\ee
(As before, lower indices refer to incoming lines and upper indices to outgoing ones.)  The Yang-Baxter equation will involve only the ratio $W(z)=D(z)/C(z)$.

\begin{figure}[htbp]
\centering{\includegraphics[scale=.78]{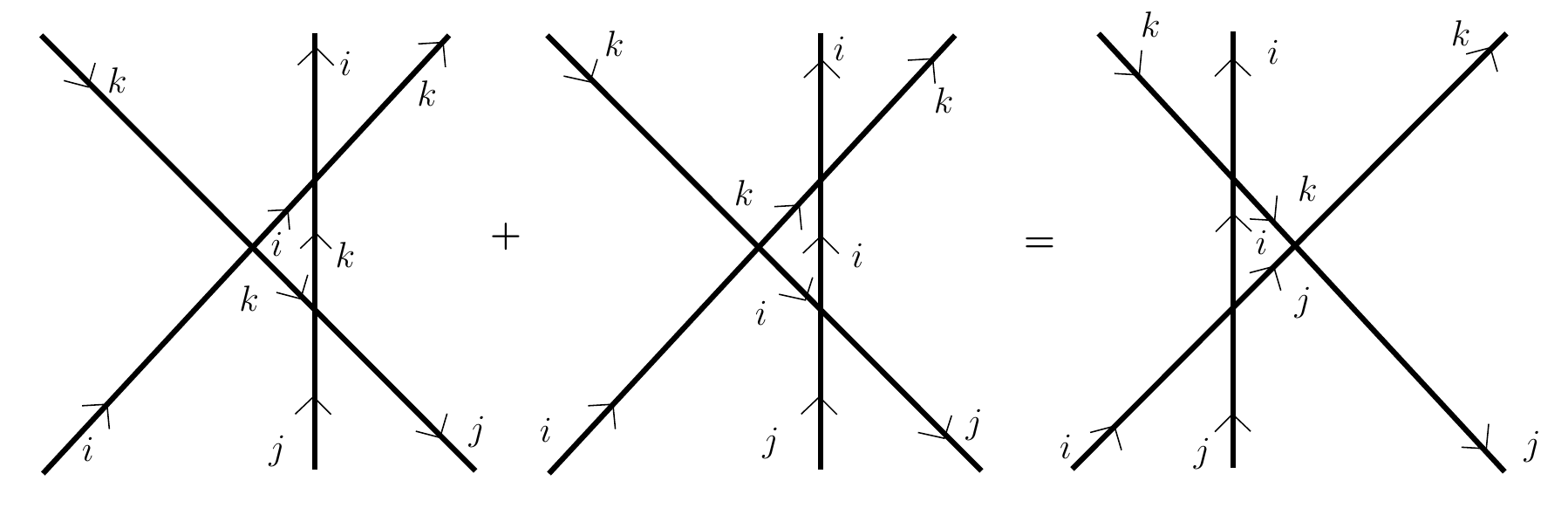}}
\caption{\small{The Yang-Baxter equation for three lines labeled by representations $V$, $V$, and $V^*$.}}
\label{FigD}
\end{figure}
It is again not difficult to write down the Yang-Baxter equation.  We learn (\fig \ref{FigD}) that
\be\label{nelf} 
U(z_1-z_2)W(z_2-z_3)+W(z_1-z_3)W(z_2-z_3)=U(z_1-z_2)W(z_1-z_3)\;. 
\ee
This is equivalent to $1/W(z_2-z_3)-1/W(z_1-z_3)=1/U(z_1-z_2)$, with the general solution 
\be\label{zerf} 
W(z)=-\frac{\hbar }{z-\hbar b}\;, 
\ee
with a constant $b$.
\begin{figure}[htbp]
\centering{\includegraphics[scale=.9]{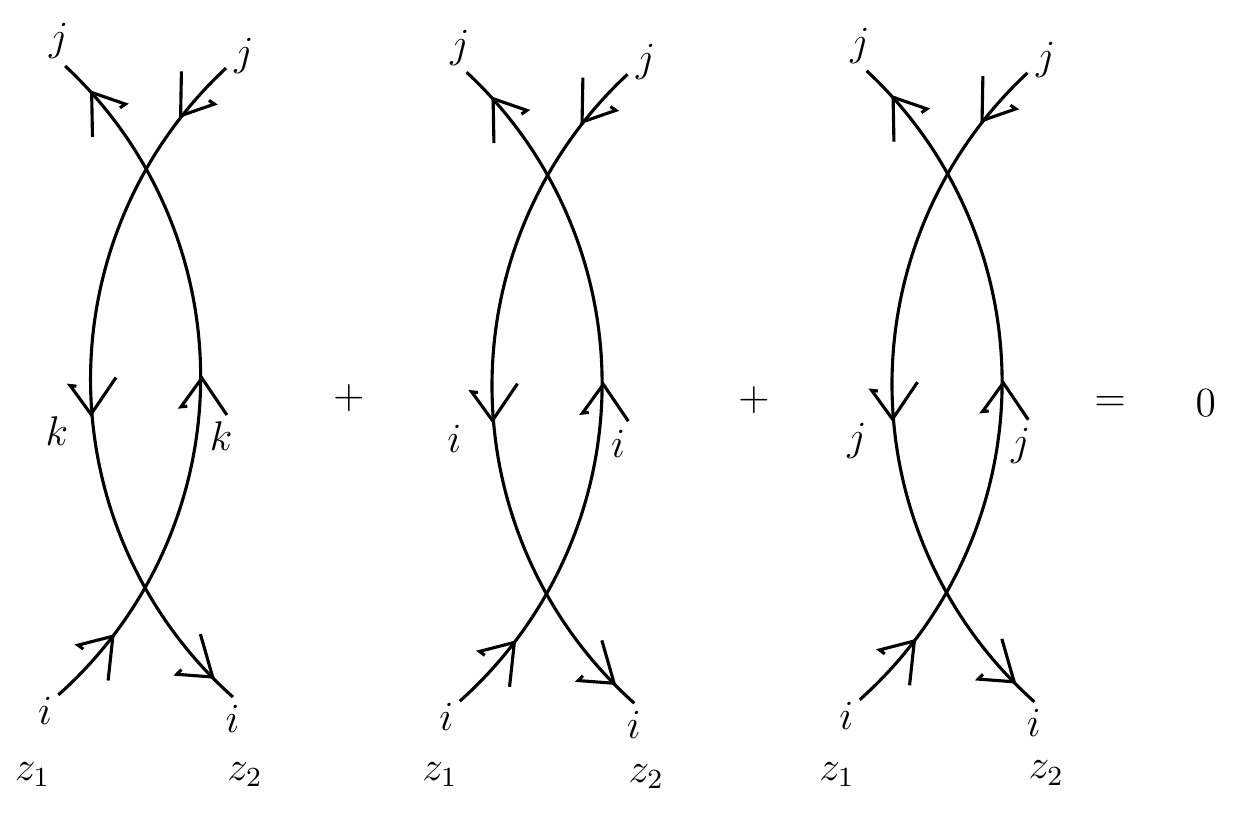}}
\caption{\small{The ``unitarity'' relation for lines labeled by $V$ and $V^*$ implies that the given sum vanishes for $i\not=j$.}}
\label{FigE}
\end{figure}
This constant $b$  can actually be determined by the unitarity relation $R_{21}(-z)R_{12}(z)=1$.  A specific matrix element of
this relation, for the case of crossing of $V$ and $V^*$ (see \fig \ref{FigE}) gives
\be\label{lefot} NW(z)W(-z)+W(z)+W(-z)=0\;,\ee
leading to
\be\label{perf} W(z)=-\frac{\hbar }{z+\frac{\hbar N}{2}} \;. \ee

We need not separately consider the $R$-matrix for crossing of two copies of $V^*$, because it simply equals the $R$-matrix for
crossing of two copies of $V$.   The reason is that the outer automorphism of $GL_N$ or $SL_N$ that exchanges $V$ and $V^*$
is a symmetry of the action (\ref{eq.action}), and of the theory derived from it.

Not determined by these arguments
are overall scalar functions in the $R$-matrices for $V\otimes V\to V\otimes V$ and $V\otimes V^*\to V\otimes V^*$.
These can be partly but not entirely determined by the unitarity relation; to some extent these overall scalar functions depend on arbitrary
choices in quantizing the theory.

\begin{figure}[htbp]
\centering{\includegraphics[scale=.35]{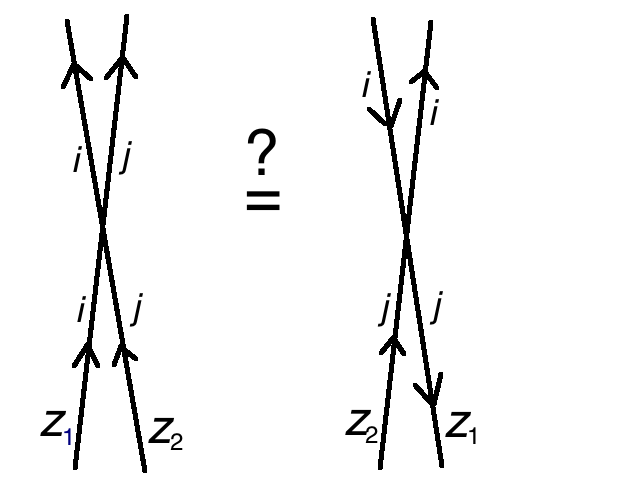}}
\caption{\small{On the left are two Wilson lines in the fundamental representation of $GL_N$ or $SL_N$; they differ from the vertical
by small angles $\pm \alpha$.  The one on the
left has spectral parameter $z_1$ and the one on the right has spectral parameter $z_2$.  To get the picture on the right, the 
$z_2$ Wilson line has been rotated clockwise by the small angle $2\alpha$, and the $z_1$ Wilson line has been rotated clockwise
by  the much larger angle $\pi-2\alpha$.   Naively the two pictures should be equivalent, implying that $U(z_1-z_2)=W(z_2-z_1)$, a claim
that turns out to be false.}}
\label{Anomaly}
\end{figure}

In the theory as we have developed it so far, the angles at which two Wilson lines cross are of no consequence.  We can use the above
formulas to test this expectation in an interesting way.  We consider (see the left of 
\fig \ref{Anomaly}) two Wilson lines in the representation $V$,
with spectral parameters $z_1$ and $z_2$, and differing from the vertical by small angles $\pm \alpha$.  Rotating one line clockwise
by a small angle $2\alpha$ and the other one clockwise by a larger angle $\pi-2\alpha$, we arrive at the right hand side of the figure,
which depicts the crossing of a pair of nearly vertical Wilson lines in the representations $V$ and $V^*$.  The rotation converts
a charge exchange for $VV$ crossing to an ``annihilation'' process for $VV^*$.  Naively, the two parts of \fig \ref{Anomaly} should
be equivalent.  This would imply
\be\label{orfom} U(z)\overset{?}{=}W(-z)\;.\ee
A look back at our previous formulas shows, however, that this is false.  What is true instead is that 
\be\label{worfom} U\left(z-\frac{\hbar N}{2}\right)=W(-z)\;. \ee

\subsection{First Look at the Framing Anomaly}\label{firstlook}

What accounts for this discrepancy?  The answer is that the theory has a framing anomaly for Wilson operators, which will be explored
from another point of view in section \ref{framinganomaly}.  The framing anomaly is analogous to the perhaps
familiar framing anomaly for Wilson operators in Chern-Simons theory, but more subtle.  It can be formulated in different but topologically
equivalent ways.  However, in an approach natural in perturbation theory, the framing anomaly can be formulated as follows.  To begin
with we take the topological plane to really be a plane $\Sigma=\R^2$, and we quantize with a gauge choice that uses a flat metric
on $\R^2$.  (See section \ref{lowestorder}.) We consider a Wilson operator supported on a general curve $K\subset
\Sigma$, and we let $\varphi(p)$ be the angle between the tangent vector to $K$ at a given point $p\in K$ and some chosen direction in $\R^2$ (e.g.\ the vertical).
(We define this angle to increase if $K$ bends in a clockwise direction.)
Thus $\varphi$ is not quite well-defined as a function on $K$, but it is well-defined up to an additive constant, and its differential
$\d\varphi$ is well-defined.

The framing anomaly means that what is constant along $K$ is not the spectral parameter $z$, as one would expect from a classical
analysis, but $z-\hbar\, \sh^\vee \varphi/(2 \pi)$, where $\sh^\vee$ is the dual Coxeter number of the gauge group.  As a perhaps surprising
 example of the
implications of this statement, a Wilson operator whose support is a simple closed curve in $\Sigma$ is anomalous and does not
exist in the quantum theory, because in going around a simple closed loop, $\varphi$ increases by $2\pi$.  

\begin{figure}[htbp]
\centering{\includegraphics[scale=1]{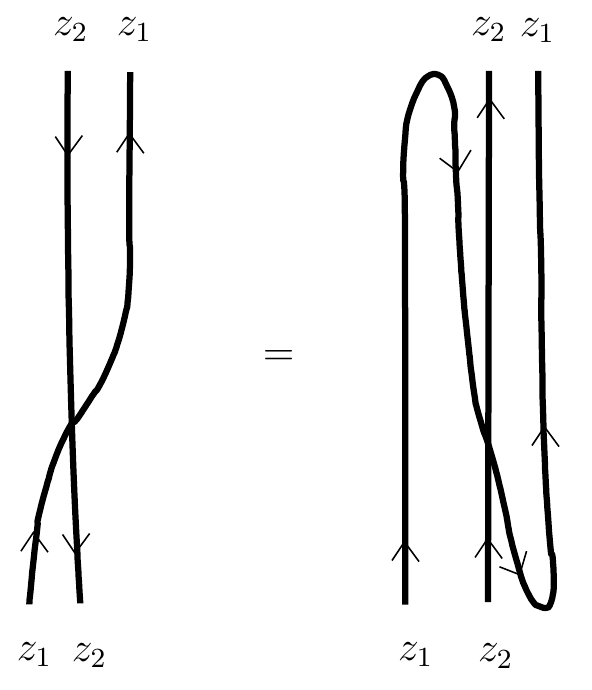}}
\caption{\small{An improved view of the unitarity relation that was explored in \fig \ref{Anomaly}, formulated to take the framing anomaly
into account.  Two lines both labeled by $V$ enter at the bottom and emerge at the top at nearly vertical angles.  If they
cross at nearly vertical angles, we get an $R$-matrix element for a $VV$ crossing, while if one of them turns by an angle very close
to $\pi$ before the crossing, then we get a $VV^*$ crossing, but with $z_1$ replaced by $z_1-N\hbar/2$ because of the framing anomaly.
Thus the naive relation $U(z)=W(-z)$ suggested by \fig \ref{Anomaly} is corrected to $U(z-N\hbar/2)=W(-z)$. }}
\label{FigF}
\end{figure}

Now in our problem, we can formulate the comparison between the two parts of \fig \ref{Anomaly} in a slightly different
way.  To go from the left to the right of the figure, we make the $z_1$ Wilson line bend in the plane by an angle $\Delta\varphi=\pi$ before
crossing the $z_2$ Wilson line, as in \fig \ref{FigF}.  But because of the framing anomaly, when we do this $z_1$ is shifted to 
$z_1-\hbar\, \sh^\vee\Delta \varphi/(2\pi) =z_1-\hbar N/2$, where $N$ is the dual Coxeter number of $SL_N$ (or $GL_N$).
Thus what should coincide with $W(-z)$ is not $U(z)$ but $U(z-\hbar N/2)$,
and this is precisely what we found in eqn.\ (\ref{worfom}).

In light of the framing anomaly, one might ask the following question.  In the usual formulation of the Yang-Baxter equation for
crossing of three lines, what are the angles at which the lines cross?  The answer is clear if one considers the case that the three
representations involved are all the same.  It is usually then assumed that the three $R$-matrices $R_{12}$, $R_{23}$, and $R_{13}$
are given by the same matrix-valued function of $z$.  For this to be true, the relative angles must be the same at all three crossings.
Since the actual relation is that the 13 crossing angle is the sum of the 12 and 23 crossing angles, the three angles are
equal only in the limit that they all  go to zero.  So the usual formulation
of the Yang-Baxter equation refers to the case of nearly parallel lines in the limit that the crossing angles vanish.  That is actually
why \fig \ref{Anomaly} has been drawn with lines at small angles $\pm\alpha$ to the vertical, and it was implicitly assumed in our discussion of \fig \ref{FigF}.

The reader might be slightly perplexed that we began our explanation of the framing anomaly
by restricting to the special case $\Sigma=\R^2$.  To understand this point
properly, one has to analyze the framing anomaly for four-manifolds as well as the framing anomaly for Wilson operators.  This
is somewhat beyond the scope of the present paper.  However, the upshot is that to avoid an anomaly, the two-manifold $\Sigma$
must be ``framed,'' meaning that its tangent bundle must be trivialized.  (This is actually analogous to the framing anomaly of
three-dimensional Chern-Simons theory, which is defined on a framed three-manifold.)  On a framed two-manifold $\Sigma$, one
can define along any embedded oriented one-manifold $K\subset \Sigma$ a function with the properties of the $\varphi$ used above.  For
this, recall that a ``framing'' is a pair of everywhere linearly independent vector fields $v_1$ and $v_2$ on $\Sigma$.  Given a framing,
one can pick a metric such that $v_1$ and $v_2$ are everywhere orthonormal, and then one can define $\varphi$ as the angle between
the tangent vector to $K$ and the $v_1$ direction.  Note that the condition that $\Sigma$ should be framed is very restrictive; for example,
a compact framed two-manifold  must have zero Euler characteristic and hence must have genus 1.  Thus -- as we also saw from the
anomaly for a closed loop in the plane -- the framing anomaly in the four-dimensional theory discussed here is much more restrictive
than its three-dimensional cousin.

If we think of the vertical direction in the above figures as the Euclidean ``time,''  then the reader
will note that we have interpreted a $V^*$-valued particle moving forward in time as a $V$-valued particle moving backward in time.
This is reminiscent of the relation between particles and antiparticles in relativistic quantum field theory, and indeed in the application
of $R$-matrix theory to integrable models of relativistic quantum field theory, this operation becomes crossing symmetry whereby
an $S$-matrix element with a particle in the initial state, after analytic continuation to negative energy, is interpreted as an $S$-matrix element
with an antiparticle in the final state.   That is why in $R$-matrix theory, the relation between $R$-matrix elements associated
to a pair of dual representations $\rho$ and $\rho^*$ is often called ``crossing.''

\section{\texorpdfstring{$R$-Matrix from Crossing Wilson Lines}{R-Matrix from Crossing Wilson Lines}}\label{lowestorder}

Here and in the next two sections, we will perform concrete Feynman diagram calculations to
compute (1) the $\O(\hbar)$ term in the $R$-matrix; (2) the $\O(\hbar)$ quantum correction to the operator
product expansion (OPE) for Wilson line operators; (3) the framing anomaly.  In fact, in the case of the framing anomaly,
the $\O(\hbar)$ term that we compute gives the complete answer; in the other cases, there are higher order contributions
to the effects that we calculate, although they can be determined by general principles (such as the Yang-Baxter equation and
associativity of the operator product expansion) once the lowest order terms are known.  We will perform independent Feynman diagram
calculations of the three effects, but actually the three effects can be deduced from each other to a large extent.  We already deduced
the framing anomaly from the $\O(\hbar)$ term in the $R$-matrix in section \ref{firstlook}, and we will explain in section \ref{parallel}
why the quantum correction to the OPE is inevitable given the quantum correction to the $R$-matrix.  

In all cases, we take $C=\C$, corresponding to a rational solution of the Yang-Baxter equation.  In the case of the framing anomaly,
the considerations are manifestly local, so the choice of $C$ does not matter.  For the $R$-matrix and the OPE, 
once the result is known for $\C$, it can be deduced from global considerations for the other choices of $C$. We leave this
for sections \ref{trigonometric} and \ref{elliptic}.

We will compute in a way that involves a choice of metric on $\Sigma\times C$.  As explained in section \ref{ybeu},
for the output of Feynman diagrams to have a straightforward interpretation in the usual language of $R$-matrix theory,
we have to scale up the metric on $\Sigma$ by a large factor.  In the limit, $\Sigma$ becomes $\R^2$ near the crossing
and the supports of the two Wilson lines that are crossing become straight lines in $\R^2$.  In $\O(\hbar)$, the angle
at which the lines cross does not matter (the reader can verify this by a slight generalization of the calculation that
we will describe), so we can take the two lines to be the $x$-axis and the $y$-axis in the $xy$ plane.  In higher orders,
the crossing angle would matter via the framing anomaly.

\begin{figure}[htbp]
\centering{\includegraphics[scale=0.9]{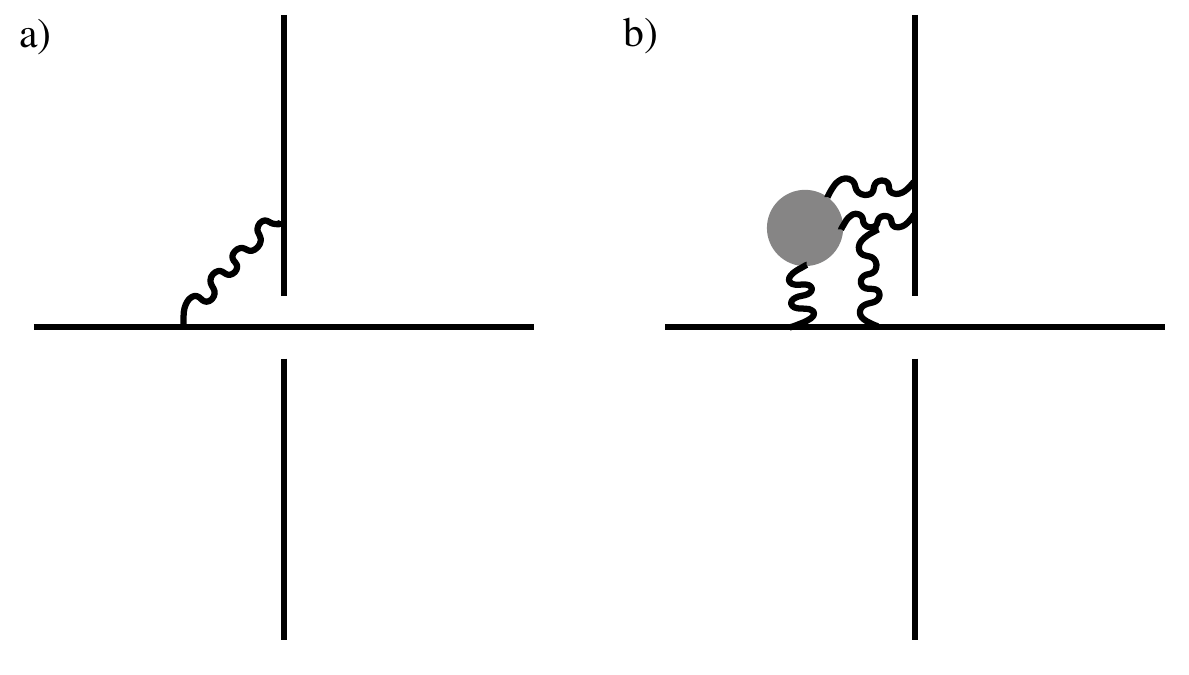}}
\caption{\small{(a) The leading order Feynman diagram for gluon exchange between crossing Wilson lines.
(b) A higher order diagram, with further complications hidden in the ``blob.''}}
\label{FigureGluon_Exchange}
\end{figure}

The $\O(\hbar)$ contribution to the $R$-matrix involves one gluon exchange between the two lines, as sketched in \fig \ref{FigureGluon_Exchange}.
We evaluate the contribution of this diagram for the case that the Wilson lines are associated to representations $\rho$ and $\rho'$
of $\mathfrak{g}$ supported respectively at $z=z_1$ and $z=z_2$. 

The metric that we will use on $\R^2\times C$ is $\d s^2=\d x^2+\d y^2+ \d z \d\bar{z}=g_{\mu \nu} \d x^{\mu} \d x^{\nu}$,
or in another language \begin{align}
g_{xx}=g_{yy}=1 \;, \quad g_{z\bar{z}}=g_{\bar{z}z}=\frac{1}{2}\; , \quad \textrm{other components zero} \;.
\end{align} 
The corresponding inverse metric is
\begin{align}
g^{xx}=g^{yy}=1\;, \quad g^{z\bar{z}}=g^{\bar{z}z}=2\;, \quad \textrm{other components zero} \;.
\end{align}
For a gauge-fixing condition, we pick
\be \label{gfix}
0=\dpa{x} A_x + \dpa{y} A_y + 4 \dpa{z} A_{\zbar} \;.
\ee 
This is the closest analog of the usual  Lorentz gauge for this theory with a partial gauge connection. The factor of four is explained by noting that if the gauge field $A$ satisfies this equation and also the linearized equations of motion $\d z \wedge \d A= 0$, then each component of $A$ is harmonic for the metric we have chosen.  

In this gauge, the four-dimensional propagator for the gauge field is then given by
\begin{align}
\begin{split}
\langle A^a_x(x,y,z,\zbar) A^b_y (x',y',z',\bar{z}') \rangle &= 
-\delta^{ab}\frac{4}{4 \pi}  \dpa{z}  \frac{1}{(x-x')^2 + (y-y')^2 + |z-z'|^2} \\  
&=\delta^{ab}\frac{1}{2 \pi }  \frac{2(\zbar-\zbar')}{\left( (x-x')^2 + (y-y')^2 + |z-z'|^2 \right)^2} \;,\\  
 \langle A^a_{\zbar}(x,y,z,\zbar) A^b_{x} (x',y',z',\bar{z}') \rangle &=
 -\delta^{ab}\frac{1}{4 \pi} \dpa{y}  \frac{1}{(x-x')^2 + (y-y')^2 + |z-z'|^2} \\  
 &=  \delta^{ab}\frac{1}{2 \pi}   \frac{y-y'}{\left( (x-x')^2 + (y-y')^2 + |z-z'|^2 \right)} \;,\\ 
 \langle A^a_y(x,y,z,\zbar) A^b_{\zbar} (x',y',z',\bar{z}') \rangle
 &=  - \delta^{ab}\frac{1}{4 \pi} \dpa{x}  \frac{1}{(x-x')^2 + (y-y')^2 + |z-z'|^2}  \\
 &=   \delta^{ab}\frac{1}{2 \pi }   \frac{x-x'}{\left( (x-x')^2 + (y-y')^2 + |z-z'|^2 \right)^2}  \;.
\end{split}
\end{align}

For later purposes, we can reinterpret the propagator as a two-form on two copies of $\R^2\times C$.  Setting
$x=x'-x''$, $y=y'-y''$, $z=z'-z''$, $\bar z=\bar z'-\bar z''$, we define  
\begin{align}\label{eq.propagator}
\begin{split}
P^{ab}(x,y, z, \bar{z})&:=
\frac{1}{2} \sum_{i,j=x,y, \bar{z}} \langle A^a_i (x', y', z', \bar{z}') A^b_j(x'', y'', z'', \bar{z}'')\rangle \d x^i \wedge \d x^j 
\\
&=-\frac{\delta^{ab}}{4 \pi}\left(\d y \wedge \d \zbar \frac{\partial }{\partial x} +\d\bar{z}\wedge \d x \frac{\partial }{\partial y} + 4 \d x \wedge \d y \frac{\partial}{\partial z}\right)
\frac{1}{(x^2+y^2+z \bar{z})} \\
&= \frac{\delta^{ab}}{2 \pi}\left( x \d y \wedge \d \zbar+y \d\bar{z}\wedge \d x + 2 \zbar \d x \wedge \d y\right)\frac{1}{(x^2+y^2+z \bar{z})^2} 
 \;.
\end{split}
\end{align}
In the following, we often use the propagator two-form with 
adjoint indices stripped off:
\begin{align}
\begin{split}
P^{ab}(x,y, z, \bar{z})&=\delta^{ab} P(x, y, z, \bar{z}) \;.
\end{split}
\end{align}

The defining equations of the $2$-form $P(x,y,z,\zbar)$ are
\begin{align} 
\frac{\ii}{2\pi} \d z \wedge \d  P (x,y,z,\zbar) &= \delta_{x,y,z,\zbar = 0} \label{propagator_greens} \;,\\
\left( \partial_x \iota_{\partial_x} + \partial_y \iota_{\partial_y} + 4 \partial_{z} \iota_{\partial_{\zbar}}\right) P(x,y,z,\zbar) &= 0 \;. \label{propagator_gauge} 
\end{align}
where $\delta_{x,y,z,\zbar = 0}$ is a delta-function distribution localized at $x=y=z=\zbar = 0$, and $\iota_V$ indicates contraction with a vector field $V$.  

To verify the normalizations, we can check these equations explicitly.  It is easy to verify that $\d z \wedge \d P = 0$ away from the origin $x = y = z = \zbar = 0$, and that \eqn \eqref{propagator_gauge} is satisfied.  

Note that, when restricted to the unit three-sphere,
\begin{equation}
 P(x,y,z,\zbar) =
   \frac{1}{2 \pi} \left(  x  \d y  \wedge \d \zbar -  y  \d x \wedge \d \zbar + 2 \zbar  \d x\wedge \d y\right)  
  \;.
\end{equation}
Since the forms on the two sides of this equation are the same on the unit three-sphere, Stokes' theorem tell us that applying the operator $\d z\wedge \d $ to both sides and integrating over the ball of radius $1$ will give the same answer. Therefore we have the identity (using the coordinates $u,v$ with $z = u + \ii v$) 
\begin{align}
\begin{split}
\frac{\ii}{2\pi}\int_{x^2 +y^2 + z \zbar \le 1} \d z \wedge \d P(x,y,z,\zbar) 
&= \frac{\ii}{2\pi} \frac{1}{2\pi} \int_{x^2 +y^2 +z \zbar \le 1} 4 \d x \d y \d z \d \zbar\\
&=   \frac{\ii}{2\pi} \frac{1}{2\pi} \int_{x^2 + y^2 + z \zbar \le 1} (-8 \ii) \d x \d y  \d u \d v \\
&= \frac{\ii}{2\pi} \frac{1}{2\pi} (-8 \ii) \frac{\pi^2}{2}=1 \;.
\end{split}
\end{align}
The four-dimensional bulk gauge field $A^a$ couples to a Wilson line in the representation $\rho$
by a factor $t_{a,\rho}$, which is the matrix by which the Lie algebra element $t_a$ acts in the representation $\rho$.
For the case of a single gluon coupling to a Wilson line, this factor does not depend on where on the line the gluon is inserted.
For gluon exchange between two Wilson lines, as in 
\fig \ref{FigureGluon_Exchange}, we simply get such a factor on each line.
The propagator between a Wilson line supported on the $x$ axis at $z=z_1$ and one supported on the $y$ axis at $z=z_2$ then
evaluates to 
\begin{align}
\begin{split}
I_1&=\hbar \left((t_{a,\rho} \otimes  t_{b,\rho'}\right) \int \d x \d y' 
\, P^{ab}(x-x', y-y', z_1-z_2, \bar{z}_1-\bar{z}_2) \\
&= \hbar \, c_{\rho, \rho'} \int \d x \d y' 
\, P(x-x', y-y', z_1-z_2, \bar{z}_1-\bar{z}_2)   \;,
\label{eq.Axy}
\end{split}
\end{align}
where the color factor reads
\be 
c_{\rho, \rho'}=\sum_a t_{a,\rho}\otimes t_{a,\rho'} \;.
\label{eq.color}
\ee 
Here  $c_{\rho,\rho'}$ can be viewed as  the image of an element $c=\sum_a t_a\otimes t_a$ of $\g\otimes \g$ in the representation
$\rho\otimes \rho'$ of $\g\otimes \g$.  
The factor of $\hbar$ is the loop counting parameter.

It is straightforward to evaluate the integral in  \eqref{eq.Axy}, with the result 
\begin{align}
I_1=\hbar  \, c_{\rho, \rho'} 
\frac{1}{2\pi}\int \d x \d y \frac{2(\zbar_1-\zbar_2)}{(x^2+y^2+|z_1-z_2|^2)^2} 
= \frac{\hbar\, c_{\rho, \rho'}}{z_1-z_2}  \;.
\end{align}
This reproduces the standard semi-classical expansion of the rational $R$-matrix
\begin{align}
R=I+\hbar \, r +\mathcal{O}(\hbar^2)= I + \frac{\hbar\, c_{\rho,\rho'} }{z_1-z_2}   +\mathcal{O}(\hbar^2) \;.
\end{align}

General theorems (see \cite[p.\ 814]{Drinfeld_ICM} or \cite[p.\ 418] {Chari-Pressley}) imply that the full rational $R$-matrix is determined up to a prefactor by the general conditions
that it obeys together with the leading order term that we have just computed.  Some interesting special cases of this statement are
rather easy, as we have reviewed in section \ref{elementary}. 

\section{OPE of Parallel Wilson Lines}\label{parallel}

\subsection{Overview}\label{noverview}

\begin{figure}[htbp]
\centering{\includegraphics[scale=0.37]{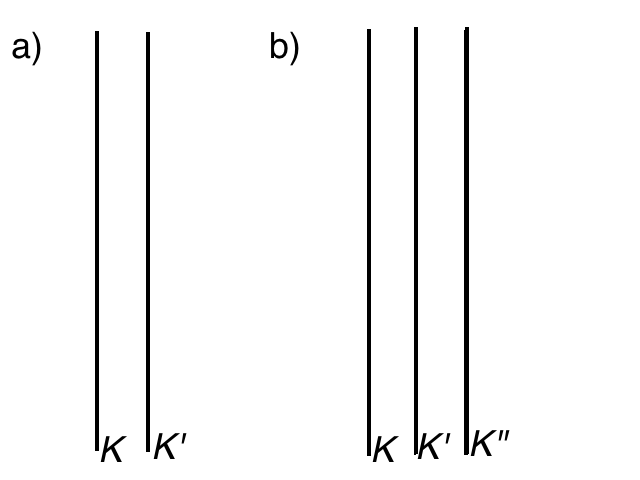}}
\caption{\small{ (a) In a diffeomorphism-invariant theory, two parallel line operators can be considered to be arbitrarily close, so they
behave as a single line operator.  This is the operator product expansion for line operators.  (b) In two dimensions, this operator product expansion is not necessarily
commutative, but it is always associative, because given three parallel line operators, there is no natural sense of one pair being ``closer'' than the other.  The triple product
$KK'K''$ can be identified as either $(KK')K''$ or $K(K'K'')$ by moving together one pair of adjacent line operators or the other.
}}
\label{OPE}
\end{figure}

The next topic that we will consider is the operator product expansion (OPE) of Wilson lines.   We begin with generalities about line operators in diffeomorphism-invariant theories. 

Consider parallel line operators $K$ and $K'$ in a theory with diffeomorphism invariance in any dimension $\geq 2$ (\fig \ref{OPE}(a)).
Diffeomorphism invariance means that there is no natural notion of whether $K$ and $K'$ are ``near'' or ``far'' and therefore that we can think of
them as being arbitrarily near.  This implies that it must be possible in any diffeomorphism invariant  theory to interpret a product $KK'$ of two
parallel line operators as a single line operator $K''$.  This is the operator product expansion for line operators.

Although there is in a diffeomorphism invariant theory no natural notion of $K$ and $K'$  being ``near'' or ``far,'' something special
happens in two dimensions:  there can be a natural notion
of whether $K$ is to the left or right of $K'$.  The product $KK'$ with $K$ to the left of $K'$ may be different from the product $K'K$ with $K$ to the right.
Thus, in two dimensions, the OPE of line operators is not necessarily commutative.\footnote{In three dimensions, there is a more subtle analog of this:
$KK'$ and $K'K$ are always isomorphic, but there can be different isomorphisms between them, depending on the direction in which $K$ and $K'$ are moved around each other, leading to a notion of braiding of line operators.}

Although not necessarily commutative in two dimensions, the OPE of line operators is always associative.  That is because (\fig \ref{OPE}(b)) given three
parallel line operators $K$, $K'$, and $K''$, there is no natural notion of $K$ and $K'$ being closer or farther than $K'$ and $K''$.  There is just
one product of line operators that depends only on how they are arranged from left to right.

Our problem has a few special features.   The line operators that we will be studying are indeed supported on a line $K$ in the smooth two-manifold $\Sigma$
which has diffeomorphism symmetry,\footnote{This diffeomorphism symmetry is mildly broken by a framing anomaly, but not in a way that affects the
present discussion.}
but they are also supported at a point in the complex Riemann surface $C$.  The OPE for line operators supported at distinct points in $C$ is trivial -- and
in particular commutative -- as they can pass through each other in $\Sigma$ without any singularity.  They interesting case is the OPE for line operators
that are supported at the same point in $C$.  We may as well take this point to be $z=0$.   

Although abstractly the product $KK'$ of line operators $K$ and $K'$ will always be another line operator, if we try to consider too small a class
of line operators, we might find that the product $KK'$ is not in the class that we started with.  That is actually what happens in the theory described
in the present paper if we consider only the most obvious class of line operators: Wilson line operators associated to representations of the finite-dimensional
Lie algebra $\mathfrak{g}$.  We will see in section \ref{lowest} that this class of line operators is not closed under operator products.  To get an OPE for Wilson
line operators, we have to consider the more general class of Wilson operators associated to representations of $\mathfrak{g}[[z]]$.  

Given two parallel Wilson line operators, on general grounds their product in a diffeomorphism-invariant theory is another line operator.
But it is nontrivial to exhibit this new line operator as another Wilson operator for some representation of $\mathfrak{g}[[z]]$.  To exhibit
this, we have to do a calculation, starting with two parallel line operators, taking the limit, in a concrete quantization scheme, as they approach each other,
and searching for a single Wilson operator that will reproduce the effects of the two Wilson operators that we started with.

At the classical level, the OPE for Wilson operators is trivial.  Given Wilson operators associated to representations $\rho$ and $\rho'$ of $\g[[z]]$,
their product is the Wilson line associated to the tensor product representation $\rho\otimes \rho'$.  There is a quantum correction to this and
we will compute it to lowest order in $\hbar$.

We will perform the computation assuming that $\rho$ and $\rho'$ are representations of the finite-dimensional algebra
$\g$, or equivalently that they are representations of $\g[[z]]$ in which the generators $t_{a,n}$ vanish for $n>0$.
What we will show is that in order $\hbar$,  the tensor product representation acquires a nonzero $t_{a,1}$ (there is no correction
to $t_{a,n}$ for any $n\not=1$).  It is fundamentally because of this fact that it is important to consider Wilson operators associated
to representations of $\g[[z]]$ that do not come from representations of $\g$. 

\subsection{Lowest Order Computation}\label{lowest}

We start with a Wilson line in the representation $\rho'$ of $\mf{g}$ at $y = 0$, and one in the representation $\rho''$ at $y = \eps$.  As just
explained, we assume that these are ``ordinary'' Wilson lines, associated to representations of $\g$.

The leading order Feynman diagram is given in \fig \ref{FigureOPE}. This diagram represents the coupling of an external gauge field to the two Wilson lines.
We will calculate what happens, to leading order in $\hbar$, when we put these lines beside each other and send $\eps \to 0$ from above. Modulo $\hbar$, the 
result will simply be the Wilson line associated to the tensor product representation $\rho=\rho' \otimes \rho''$ of $\mf{g}$.  There is an order $\hbar$ correction, in which the 
$z$-derivative of the gauge field $A$ is coupled to the Wilson line in a non-trivial way.   

\begin{figure}[htbp]
\centering{\includegraphics[scale=0.7]{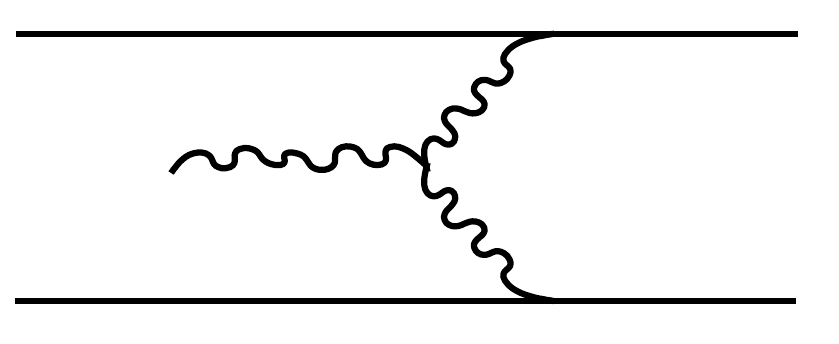}}
\caption{\small{The Feynman diagram representing the $\O(\hbar)$ correction to the OPE of Wilson lines.  The two lines are at $y=0$ and $y=\epsilon$,
respectively. We consider the limit $\epsilon\to 0$.}
} \label{FigureOPE}
\end{figure} 

For the evaluation of this diagram we need an extra Feynman rule not needed so far,
namely the bulk interaction vertex,  away from any Wilson lines.
This can be read off from the action \eqref{eq.action}, and is
\begin{align}
\frac{\ii}{2\pi} f^{abc} \d z \;.
\end{align}
Note that this vertex contains  a 1-form $\d z$, since this is present in the action.
To evaluate the diagram of \fig \ref{FigureOPE}, we have to integrate one interaction vertex over the line $y=0$, one over the
line $y=\epsilon$, and one over $\R^4=\R^2\times \C$.  Connecting the vertices by propagators, we have to integrate
\begin{align}
\begin{split}
I_2&=\frac{\ii}{2\pi} \hbar (t^a \otimes t^b f_{abc}) \\
&\qquad\times \int_{x_1,x_2 = -\infty}^{\infty} \int_{x,y,z, \bar{z}} P(x_1-x,y,z) \wedge\d z\wedge  A^c(x,y,z) \wedge P(x_2-x,y-\eps,z)
\;,
\end{split}
\end{align}
where $P$ is the four-dimensional propagator two-form introduced previously in \eqref{eq.propagator}.
Since we have two propagators and one bulk vertex, we  have a factor of 
$\hbar^{2-1}=\hbar$ in front, as in the discussion of the $R$-matrix.  (The meaning of this is that the diagram is of order $\hbar$
compared to a contribution in which the gauge boson couples directly to one of the Wilson operators, with the bulk interaction playing no role.) 

Integrating first over $x_1$ and $x_2$,  we obtain
\begin{align}\label{usefull}
I_2=\frac{\ii}{2\pi} \hbar  (t^a \otimes t^b f_{abc})\int_{x, y,z, \bar{z}}\d z\wedge A^c(x,y,z, \bar{z}) \wedge P'(y,z, \bar{z})  \wedge P'(y-\eps,z, \bar{z}) \;,
\end{align}
where we defined a  three-dimensional one-form propagator with color indices stripped off on the plane $\mathbb{R}^3$ parametrized by $(y,z, \bar{z})$:
\begin{align}
P'(y, z, \bar{z}):=\int \d x \,P(x,y, z, \bar{z})
= -\frac{1}{4} \left(- \d\bar{z} \frac{\partial }{\partial y} +4 \d y \frac{\partial}{\partial z}\right)
\frac{1}{(y^2+z \bar{z})^{\frac{1}{2}}}
 \;.
\label{3dP}
\end{align}
Note that we obtained a numerical factor of $\pi$ from the $x$ integral.

The most fundamental thing to explain about eqn.\ (\ref{usefull}) is why, for $\epsilon\to 0$, it  produces a local coupling of $A$
to a line operator at $y=z=0$.  The reason for this is simply that for $\epsilon\to 0$, the integrand in eqn.\ (\ref{usefull}) vanishes
as long as $y$ and $z$ are not both 0.  That is true simply because the integrand for $\epsilon=0$ is proportional to
$P'(y,z,\bar z)\wedge P'(y,z,\bar z)$, and this trivially vanishes because $P'$ is a 1-form.

Accordingly, the small $\epsilon$ limit of $\d z \wedge P'(y,z, \bar{z})  \wedge P'(y-\eps,z, \bar{z}) $ is a distribution supported at $y=z=0$.  
By dimensional analysis and rotation symmetry in the $z$ plane, this distribution must be a multiple of $\partial_z \delta^3(y,z,\bar z)$, where $\delta^3(y,z,\bar z)$ is a
three-form delta function satisfying $\int_{\R\times \C}\delta^3(y,z,\bar z)=1$.  We   will now show that
\begin{align}
 \lim_{\eps\to 0}
\d z \wedge P'(y,z, \bar{z})  \wedge P'(y-\eps,z, \bar{z}) =
\frac{\pi}{ \ii} 
\frac{\partial }{\partial z} \delta^3(y,z, \bar{z}) \;.
\label{PP}
\end{align}
This will imply that
\begin{align}
I_2 &= \frac{\ii}{2\pi} \hbar (t^a\otimes t^b f_{abc})\int_{x, y, \bar{z}}  A^c(x,y,z, \bar{z}) 
\left(  \frac{\pi}{\ii} \frac{\partial }{\partial z} \delta^3(y,z, \bar{z}) \right)\nonumber\\
&= - \hbar \frac{1}{2} (t^a \otimes t^b f_{abc}) \int \d x  \frac{\partial }{\partial z} A^c(x,y,z, \bar{z})  
 \;.
\label{OPE_result}
\end{align}
The coupling to $\partial_z A$ shows that the composite Wilson line operator obtained by bringing two such operators together has $t_{a,1}\not=0$, that
is it is associated to a representation of $\g[[z]]$ but not to a representation of $\g$.  Moreover, it is easy to see from this formula that although we will take $\epsilon>0$
in our calculation, the result is actually proportional to the sign of $\epsilon$.  Changing the sign of $\epsilon$ would lead, after replacing $y$ with $y+\epsilon$,
 to the same calculation but with the roles
of the two line operators exchanged, replacing $t^a\otimes t^b$ in
eqn.\ (\ref{OPE_result})  with  $t^b\otimes t^a$.  Since this changes the sign of $f_{abc} t^a\otimes t^b$, eqn.\ (\ref{OPE_result}) implies that the quantum OPE for
Wilson operators is noncommutative: it gives a result that depends on which of the two Wilson operators is on the ``left'' of the other.  

To show \eqref{PP}, let us express the three-dimensional one-form propagator \eqref{3dP}
as 
\begin{align}
P'(y, z, \bar{z})
= - \d\bar{z} \, \partial_y Q(y,z, \bar{z}) + 4 \d y\, \partial_{z} Q(y,z, \bar{z})%
\end{align}
with
\begin{align}
	Q(y,z, \bar{z}):=
	-\frac{1}{4}\frac{1}{(y^2+z\bar{z})^{\frac{1}{2}}} \;.
\end{align}
Then the left hand side of \eqref{PP} reads
(we here drop $z, \bar{z}$ from the arguments of $Q$, to simplify the expressions)
\begin{align}
\label{PP_comp}
&- 4 \d y \wedge \d z \wedge \d\bar{z} \, \left( \partial_y Q(y)\, \partial_z Q(y-\epsilon)- \partial_z Q(y)\, \partial_y Q(y-\epsilon)\right)  
=   \d y \wedge\d  z \wedge\d \bar{z} \left[
\frac{\partial }{\partial z}(a)
+\frac{\partial }{\partial y}(b)
\right]
 ,
\end{align}
with 
\begin{align}
\label{longform}
	(a)&:=-2 \partial_y Q(y)\,  Q(y-\epsilon) +2Q(y) \, \partial_y Q(y-\epsilon) 
	=
	\frac{2}{ 4^2}
\frac{\epsilon \left(y (\epsilon -y)+|z|^2\right)}{\left(y^2+|z|^2\right)^{\frac{3}{2}}
   \left((y-\epsilon )^2+|z|^2\right)^{\frac{3}{2}}}    \;,
\\\notag
(b)&:= - 2Q(y)\, \partial_z Q(y-\epsilon)+2\partial_z Q(y)\, Q(y-\epsilon) 
	=	\frac{1}{ 4^2}
\frac{ - \epsilon \bar z(\epsilon-2y) }{\left(y^2+|z|^2\right)^{\frac{3}{2}}
   \left((y-\epsilon )^2+|z|^2\right)^{\frac{3}{2}}}   \;.
\end{align}

These are all explicitly proportional to $\epsilon$, and thus, as claimed earlier, everything vanishes
for $\epsilon\to 0$ as long as $y,z$ are not both zero.   We expect to extract from (\ref{PP_comp}) in the limit $\epsilon\to 0$
a multiple of
$\partial_z\delta^3(y,z)$.  Since the $(a)$ term appears in eqn.\ (\ref{PP_comp}) in 
 the form $\partial_z(a)$, we expect that $(a)$ by itself without this derivative
would simply produce a multiple of $\delta^3(y,z)$.

This means that the contribution of the $(a)$ term to the coefficient of $\partial_z\delta^3(y,z)$ can be obtained
by evaluating the integral of $(a)$:
\begin{align}
\label{tmp4}
C:=\int \d y \d z \d \bar{z} \, (a) \;.
\end{align}
To demonstrate \eqn \eqref{PP}, it suffices to show that 
 $C = \frac{\pi}{ \ii}$.  
This integral is absolutely convergent.  A simple scaling argument shows that it is independent of $\epsilon$, so we may as well set $\epsilon=1$.

In polar coordinates ($\d z \d \bar{z}=- 2 \ii r \d r \d \theta$), after integration over $\theta$, we get
\begin{align}
	C=
	\frac{-8 \ii \pi }{ 4 ^2} 
	\int _{-\infty}^\infty \d y \int_0^{\infty} \d r \,  \frac{ r\left(y (1 -y)+r^2\right)}{\left(y^2+r^2\right)^{\frac{3}{2}}
   \left((y-1 )^2+r^2\right)^{\frac{3}{2}}} \;.
   \label{tmp2}
\end{align}
We can evaluate this integral with a formula familiar from
the evaluation of Feynman diagrams:
\begin{align}
\label{Feyn_param}
\frac{1}{A^{\alpha}  B^{\beta}}
= \frac{\Gamma(\alpha+\beta)}{\Gamma(\alpha) \Gamma(\beta)}
\int_0^{1} \d t  \frac{t^{\alpha-1} (1-t)^{\beta-1}}
{(tA+(1-t)B)^{\alpha+\beta}} \;.
\end{align}
Choosing $\alpha=\beta=3/2$ and remembering 
$\Gamma(3)/\Gamma(3/2)^2=8/\pi$, we have 
\begin{align}
C &= 
\frac{-64 \ii }{4^2 } 
\int_0^{1} \d t \, \sqrt{t(1-t)}  \int _{-\infty}^\infty\d y  \int_0^{\infty} \d r \frac{r \left(y (1 -y)+r^2\right) }{\left((1- t)y^2
 +t (y-1)^2+  r^2\right)^{3}} \;.
 \label{tmp1}
\end{align}
Using
\begin{align}
\int_0^{\infty} \d r  \frac{r (a+r^2)  }{(b+r^2)^3} = \frac{a+b}{4 b^2} \;,
\end{align}
we obtain
\begin{align}
	C &= \frac{-16 \ii }{4^2} 
	\int_0^{1} \d t \, \sqrt{t(1-t)}  \int \d y \,
\frac{y(1-y)+(1-t)y^2+t(y-1)^2}{((1-t)y^2+t(y-1)^2)^2} \nonumber \\
	&=
	\frac{-16 \ii }{4^2}
	 \int_0^{1} \d t \, \sqrt{t(1-t)}  \frac{\pi}{\sqrt{t(1-t)}} 
	 =\frac{\pi}{\i } 
	 \;.
\end{align}

We can do a similar analysis for $(b)$, but $(b)$ takes the form
\begin{align}
\bar{z} F(y, |z|)
\end{align}
and hence vanishes when integrated over the angle $\theta$.
This means that there is no singular contribution from $(b)$. This concludes the verification of  \eqref{PP}.

\subsection{\texorpdfstring{Relation to the $R$-Matrix}{Relation to the R-Matrix}}\label{relr}

\begin{figure}[htbp]
\centering{\includegraphics[scale=0.5]{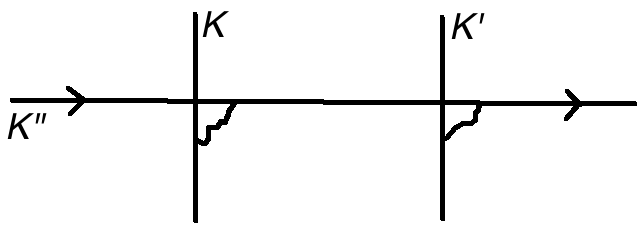}}
\caption{\small{A Wilson operator $K''$ crossing two parallel Wilson operators $K$ and $K'$.  $K$ and $K'$ are supported at $z=0$ and $K''$ at some other
point $z=z_0$.  Drawn is a Feynman diagram with a gluon exchanged between $K$ and $K''$ as well as one between $K'$ and $K''$.}
} \label{OPEtwo}
\end{figure} 
Though we have performed an independent computation, the result can actually be deduced from a knowledge of the $\O(\hbar)$ contribution to
the $R$-matrix, computed in section \ref{lowestorder}.  Consider (\fig \ref{OPEtwo}) two parallel Wilson lines $K$ and $K'$ both supported
at $z=0$ and a third Wilson line $K''$ at some other point $z=z_0$ that is crossing them.    Assuming that $K$ and $K'$ are far apart (compared to $|z_0|$) there is a unique lowest order Feynman diagram in which $K''$ interacts nontrivially with both $K$ and $K'$.  This is the diagram sketched in the figure
with a gluon exchanged from $K''$ to $K$ and another from $K''$ to $K'$.  The group theory factors associated to the two gluon exchanges
are respectively $t_a\otimes 1\otimes t_a$ and $1\otimes t_b\otimes t_b$ (where the three factors refer respectively to $K$, $K'$, and $K''$).  
The product of these is $t_a\otimes t_b\otimes t_bt_a$, where the factors acting on $K''$ are ordered as $t_bt_a$ because $K'$ is to the right of $K$
(so that path ordering along $K''$ puts $t_b $ to the left of $t_a$).  

Now suppose that $K$ and $K'$ are brought closer together.  Two-dimensional diffeomorphism invariance means that we have
 to get the same result after summing over all diagrams, but the contributions of
individual diagrams can depend on the distance from $K$ to $K'$.  In particular, the diagram that we considered before still contributes when
$K$ and $K'$ are close compared to $|z_0|$, but their contribution is not the same as before.  That is because (\fig \ref{OPEthree}(a)) even though
$K'$ is to the right of $K$, the gluon emitted from $K'$ might be absorbed on $K''$ to the left of the gluon emitted from $K$.  So the contribution
of the diagram that we considered previously is now modified by a term proportional to $t_a\otimes t_b\otimes [t_a,t_b]=f_{abc} t_a\otimes t_b\otimes t_c$.
So there must be another diagram that is significant when $K$ and $K'$ are nearby and that gives a group theory factor of this form.  That
diagram is shown in \fig \ref{OPEthree}(b).  

\begin{figure}[htbp]
\centering{\includegraphics[scale=0.55]{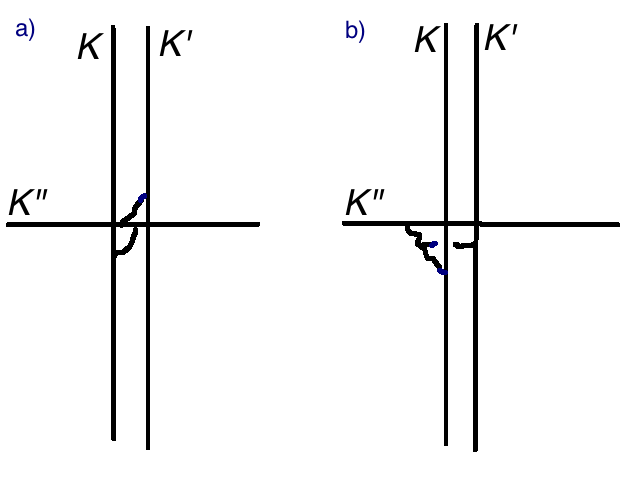}}
\caption{\small{When $K$ and $K'$ are brought nearby, the contribution of the diagram in (a) is modified relative to \fig \ref{OPEtwo}.  To maintain two-dimensional topological
invariance, this is balanced by a new contribution in (b), which is related to the existence of a quantum correction to the OPE.} }\label{OPEthree}
\end{figure} 

But this last diagram is just the one (\fig \ref{FigureOPE}) that we studied to find the quantum correction to the OPE, except that now instead
of considering an arbitrary external $A$ field, as in the previous discussion, we have provided a third Wilson operator that is the source of this $A$ field. The analysis of \fig \ref{OPEthree}(b) for $K\to K'$ is essentially the same as the analysis that we have already performed of \fig \ref{FigureOPE}.

The upshot is that the quantum correction to the OPE is an inevitable consequence of the quantum correction to the $R$-matrix, and vice-versa.

\subsection{Interpretation of the Result}\label{interpretation}

In the analysis of \fig \ref{FigureOPE}, we started with two
representations $\rho'$ and $\rho''$ of $\g$, that is representations of $\g[[z]]$  with generators $t_{a,n}'=t_{a,n}''=0$ for $n>0$.  Let us write just $t_a'$ and $t_a''$ for $t_{a,0}'$ and $t''_{a,0}$.  The above computation showed
that the  Wilson operator
obtained by fusing the two we started with couples to $\partial_z A$ and therefore has a nonzero $t_{a,1}$.  In fact $t_{a,1}$ of the fused Wilson line can be read off from eqn.\ (\ref{OPE_result}):
\be
\label{monof} t_{a,1}=
-\frac{\hbar}{2}\sum_{b,c}f_{abc}t_b'\otimes t_c'' \;.  
 \ee
On the other hand, the computation gave no contributions to generators $t_{a,n}$ of the composite Wilson operator for $n\not=1$.  So to this order, $t_{a,0}$ is given
by the classical formula
\be\label{onof}
t_{a,0}=t_a'\otimes 1+1\otimes t_a'' \;,
\ee
and the higher generators vanish, $t_{a,n}=0$, $n>1$.  Here $t_{a,0}$ are the generators of $\g$ in the classical tensor product $\rho=\rho'\otimes \rho''$.

This result certainly shows that to get a closed OPE, we must consider Wilson operators derived from representations of  $\g[[z]]$, not just $\g$.
But its consequences  go far beyond that.  The formula (\ref{monof}) actually implies that quantum corrections to the theory actually deform
$\g[[z]]$ itself -- or to be more precise that at the quantum level, Wilson operators of the theory correspond to representations not of $\g[[z]]$ itself
but of a quantum deformation of this algebra (or more accurately, of its universal enveloping algebra, as we will see).

The basic reason for this is that the formula (\ref{monof}) is not consistent with the commutation relations of $\g[[z]]$, so it implies further deformations.
In $\g[[z]]$, one has the commutation relation
\be
[t_{a,1},t_{b,1}]=f_{abc}t_{c,2} \;. 
\ee
Recalling the Jacobi identity
\be\label{nof} 
f_{uva}f_{abc} +f_{vba}f_{auc}+f_{bua}f_{avc}=0 \;,
\ee
we deduce from this
that
\be\label{plof}
f_{uva}[t_{a,1},t_{b,1}]+(\dots)=0 \;,
\ee
where the omitted terms are obtained by cyclic permutations of the indices $uvb$.

This is an identity in $\g[[z]]$, but a short calculation will show that $t_{a,1}$ as defined in eqn.\ (\ref{monof}) does not obey this identity.
Instead it satisfies a deformed version of this identity that we will describe presently.  

Because $\g[[z]]$ is being deformed, what did we mean in claiming that  to get a closed OPE for line operators, we
should start with representations of $\g[[z]]$?  The precise statement is not that line operators in the quantum theory correspond to representations
of $\g[[z]]$, but that had  we started at the classical level with arbitrary 
representations of the $\g[[z]]$, then consideration of products of line operators would not have forced us to consider new objects.  By contrast,
if we start at the classical level with representations of $\g$ only, then getting a closed OPE does require introducing many more line operators.

Now let us go back to the question of how to interpret the failure of the quantum-induced $t_{a,1}$ to obey the commutation relations of $\g[[z]]$.
Since $t_{a,1}$ in eqn.\ (\ref{monof}) is bilinear in the generators $t_{a}'$ and $t_{a}''$, the left hand side of eqn.\ (\ref{nof}), after evaluating
the commutator, is cubic in these generators; more specifically it is a sum of terms of bidegree $(2,1)$ and $(1,2)$ in $t_a'$ and $t_a''$.

The upshot is that to account for the failure of eqn.\ (\ref{plof}), we have to deform this commutation relation by adding, in order $\hbar^2$,
a certain cubic polynomial in $t_{a,0}$.  The generators $t_{a,0}$ and $t_{a,1}$ of the deformed algebra satisfy
\be\label{noof} 
f_{uva}[t_{a,1},t_{b,1}]+(\dots)=\hbar^2 Q_{uvb}(t_{a,0})\;,
\ee
where $Q_{uvb}(t_{a,0})$, which is completely antisymmetric in its indices $uvb$, is for each $uvb$ a homogeneous symmetric cubic polynomial in the
$t_{a,0}$.  

What polynomial is needed can be deduced by trying to make sure that eqn.\ (\ref{noof}) is consistent with eqns. (\ref{monof}) and (\ref{onof}).
However, when one tries to do this, one runs into a snag.  No matter what $Q_{uvb}$ may be, if $t_{a,0}=t_a'\otimes 1+1\otimes t_a''$,
then $Q_{uvb}(t_{a,0})$ is a sum of terms of bidegree $(3,0)$, $(2,1)$, $(1,2)$, and $(0,3)$ in $t_a'$ and $t_a''$.
We can pick $Q_{uvb}$ so that the terms in eqn.\ (\ref{noof}) of bidegree (2,1) and (1,2) work out correctly.  But there is nothing we can do about
the terms of bidegree $(3,0)$ and $(0,3)$, unless they vanish by themselves.

The interpretation of this is as follows.  The quantum deformed algebra satisfies (\ref{noof}) (with analogous deformations of the commutation relations
involving other generators).    In contrast to $\g[[z]]$, a representation of $\g$ does not automatically lift or extend to a representation of the deformed
algebra.   If we start with matrices $t_a$ that represent $\g$, we can always get a representation of $\g[[z]]$ by setting $t_{a,0}=t_a$ and $t_{a,n}=0$ for
$n>0$.  But this only works in the deformed algebra if $Q_{uvb}(t_a)=0$.  Otherwise, the deformed commutation relations force $t_{a,1}$ to be nonzero.
So quantum mechanically, the Wilson operators whose product we are trying to determine are anomalous, and need to be modified\footnote{If they cannot
be so modified, they are simply anomalous and do not have counterparts in the quantum theory.}  with a contribution to $t_{a,1}$ of $\O(\hbar)$,
unless $Q_{uvb}(t_a')=Q_{uvb}(t_a'')=0$.  Once we impose this restriction, we do not need to worry about the terms in eqn.\ (\ref{noof}) of bidegree $(3,0)$
or $(0,3)$.  We need consider only the $(2,1)$ and $(1,2)$ terms in that equation.

The explicit polynomial $Q_{uvb}(t_{a,0})$, though not very illuminating, can be worked out by analyzing those terms.\footnote{We will also derive it by a direct Feynman diagram analysis in section \ref{section_2loop}. Section \ref{matching} contains a fairly thorough analysis of this polynomial. See also  \cite{Chari-Pressley}, p.\ 376.}  A notable fact is that because this polynomial is of degree greater than 1, the deformation from $\g[[z]]$ by including
$\hbar^2 Q_{uvb}$ on the right hand side of eqn.\ (\ref{noof}) cannot be understood as a deformation of the Lie algebra $\g[[z]]$.  It can instead be
understood as part of an associative algebra deformation of the universal enveloping algebra of $\g[[z]]$, denoted $U(\g[[z]])$.  This makes sense because a polynomial
in elements  of $U(\g[[z]])$ is itself an element of $U(\g[[z]])$, so $Q_{uvb}(t_{a,0})$ is an element of $U(\g[[z]])$.  Eqn.\ (\ref{noof}) (and its analogs for other
components that we have not calculated) can be understood as giving an associative algebra deformation (which in particular entails a Lie algebra deformation)
of $U(\g[[z]])$.

\subsection{The Yangian}

The deformation of $U(\g[[z]])$ that was uncovered in section \ref{interpretation} is known as the Yangian.  
In fact, eqn.\ (\ref{monof}) is a standard formula describing the difference in lowest order between the tensor product of representations of the Yangian and the
tensor product of representations of $\g[[z]]$.

How can we know that the relevant deformation of $U(\g[[z]])$ is the Yangian, given that we have only performed some simple computations in lowest nontrivial order?
The answer to this question is that according to general theorems, the Yangian is the only deformation of the OPE that agrees with the lowest order deformation that
we have found in eqn.\ (\ref{monof}) and possesses certain general properties.
Thus (as already stated in \cite{CostelloA,Costello:2013zra}),  the associativity of the OPE, together with the fact that the OPE 
 in $\O(\hbar)$ receives a non-trivial correction, determines the OPE uniquely to all orders in perturbation theory, up to changes of variables. 
 This is guaranteed by a theorem of Drinfeld, 
which also says that the resulting 1-parameter deformation of the $U(\g[[z]])$
is the Yangian $Y_{\hbar}(\g)$, the algebra underlying the rational solution of the Yang-Baxter equation associated to $\g$.

To make these statements  more precise and closer in spirit to those in \cite{CostelloA,Costello:2013zra}, we can use the language of  category theory to say that the ``category of Wilson lines''  $\mathcal{C}_{\hbar}$ is given by the category of representations of $Y_{\hbar}(\g)$.

This category $\mathcal{C}_{\hbar}$ has the following four properties.

First, $\mathcal{C}_{\hbar}$ is a non-trivial one-parameter deformation of the 
category of representations\footnote{At the classical level, representations of $\g[[z]]$ are the same as representations of $U(\g[[z]])$, so Wilson operators can
be identified with representations of either of these algebras.  But the statement in the
text must be formulated for $U(\g[[z]])$, because the Yangian is a deformation of $U(\g[[z]])$, not of $\g[[z]]$.} of $U(\g[[z]])$.

Second, the OPE of the Wilson lines defines a functor, 
by fusing two objects to produce a third object
\be 
\mc{C}_\hbar \times \mc{C}_\hbar \to \mc{C}_\hbar \;.
\ee 
Moreover, the OPE is associative, though not commutative.
These two conditions define a \emph{monoidal category}, and hence
$\mc{C}_{\hbar}$ is monoidal.\footnote{By the same argument, the category of line operators in a two-dimensional topological quantum field theory (TQFT) is always a monoidal category, as is the category of boundary line operators in a $3$-dimensional TQFT with a topological boundary condition.
}

Third, the category $\mc{C}_\hbar$ allows for a kind of braiding\footnote{This is different from the braiding familiar from line operators in three-dimensional topological field theories.}, because it has certain additional structure.
Here we again consider two Wilson lines parallel in the topological plane,
but now at different points $z$ and $z+\lambda$ of the complex plane.
Since $z\ne z+\lambda$, the Wilson lines never coincide 
with each other, even when they  coincide in the topological plane.
Moreover, since we have two directions at our disposal in the holomorphic plane $C$,
we can move the two Wilson lines around in $C$, to replace $z$ with $z+\lambda$,
without encountering any singularities.
In the language of \cite{CostelloA,Costello:2013zra}, the translation $z \mapsto z + \lambda$ of the Wilson line
by a parameter $\lambda$ defines a functor $T_\lambda : \mc{C}_\hbar \to \mc{C}_\hbar$.
If $W,V$ are two Wilson lines supported at $0$, then 
the explanations above means that we have a natural isomorphism
\be 
R_{V,W} : T_\lambda W \otimes V \iso V \otimes T_\lambda W \;,
\ee 
where $\otimes$ indicates the monoidal structure on $\mc{C}_\hbar$ coming from the OPE of parallel Wilson lines.

A theorem of Drinfeld \cite{Drinfeld_ICM} now implies that $\mc{C}_\hbar$ is uniquely fixed by these properties to be  the category of representations of the Yangian $Y_\hbar(\g)$. 
The monoidal structure on $\mc{C}_\hbar$ comes from the coproduct on the Yangian, and the map $R_{V,W}$ comes from the $R$-matrix of the Yangian $Y_{\hbar}(\mathfrak{g})$. 

These arguments are admittedly somewhat abstract.  
In a companion paper \cite{Part2}, we will explain, in a concrete and down to earth way,  
how to extract a representation of the Yangian -- exactly, not just to order $\hbar$ -- from any Wilson
line operator of this theory.
The ability to do this is somewhat analogous to the ability in section \ref{elementary} to determine in an elementary way
some solutions of the Yang-Baxter
equation for $GL_N$.  Hopefully this will help convince the reader that the conclusions from the abstract arguments really apply to the quantum field theory.

\section{The Framing Anomaly and Curved Wilson Lines}\label{framinganomaly}

In section \ref{firstlook}, we deduced from some explicit elementary examples that there must be a framing anomaly for Wilson lines.  The framing anomaly says that what
is constant along a Wilson line is not $z$ but $z-\hbar\, \sh^\vee \varphi/\pi$, where $\sh^\vee$ is the dual Coxeter number of the gauge group, and $\varphi$ is the angle
between the tangent vector to the Wilson line and some chosen direction in the topological plane.   The purpose of the present section is to perform a direct Feynman diagram
computation exhibiting this effect.

In contrast to the calculations that we performed for the $R$-matrix and for the OPE of Wilson operators, for the framing anomaly, a lowest order computation gives the
exact result.  For the case of the Yangian -- that is for $C=\C$ -- this is clear from the fact that this theory is invariant under adding a constant to $z$ and
under a common rescaling of $z$ and $\hbar$.
The assertion that $z-\hbar\, \sh^\vee\varphi/\pi$ is constant along a Wilson operator is consistent with those symmetries, but a modified statement with additional terms
of higher order in $\hbar$ would not be.  The fact that the lowest order framing anomaly gives the full answer is also evident in the examples considered in section \ref{firstlook}.

For other choices of $C$ related to trigonometric or elliptic solutions of the Yang-Baxter equation, there is no scaling symmetry.  However, because the analysis of the framing
anomaly is local along $C$, the form of the anomaly is the same for the other cases as long as one picks variables so that locally along $C$ the action takes the same form
as for $C=\C$.  (In the coordinates of eqn.\ (\ref{mirf}), this is automatically true if $C$ is an elliptic curve, and it is true for $C=\C^\times$ after replacing $z$ with $\log z$.)

\begin{figure}[htbp]
\centering{\includegraphics[scale=.4]{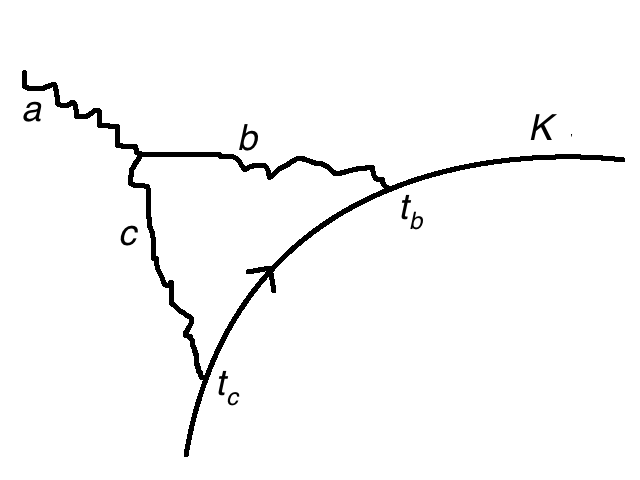}}
\caption{\small{The one-loop diagram that leads to a framing anomaly for Wilson operators.}}
\label{Diagram1}
\end{figure}

The relevant diagram that will give the anomaly is shown in \fig \ref{Diagram1}.  This diagram makes a convergent and well-defined contribution   to the effective
action of the theory.  However, this contribution is not gauge-invariant.  A failure of gauge-invariance will arise from the behavior when all three interaction vertices coincide at a common point, which will inevitably be along the support $K$ of the Wilson operator.

\def\g{{\mathfrak g}}
\def\R{\mathbb{R}}
\def\C{\mathbb{C}}
\def\cl{{\mathrm{cl}}}

Before describing what form the anomaly will take, we first do a small exercise in group theory.  We write the gauge field $A$ explicitly
as $A=\sum_a A_a t_a$, and recall that the trilinear vertex involving
a coupling of three gauge fields is  proportional to $f_{abc}$.  A gauge boson of type $b$ or $c$  couples to the Wilson line via a factor
of $t_b$ or $t_c$ (understood here as a matrix acting in the appropriate representation).  So as far as group theory is concerned, the indicated
diagram generates a coupling of an external gauge field $A_a$ to the Wilson operator via a matrix 
\be\label{yft} 
\sum_{a,b,c} A_a f_{abc}t_bt_c\;.
\ee
Here because of antisymmetry of $f_{abc}$, we can replace $t_bt_c$ by $\frac{1}{2}[t_b,t_c]=\frac{1}{2}f_{bcd}t_d$.  Then 
as $\sum_{b,c} f_{abc}f_{bcd}=2\sh^\vee \delta _{ad}$, where $\sh^\vee$ is known as the dual Coxeter number\footnote{Recall our normalization of the Killing form in \eqn \eqref{Killing_normalization}. The factor of $2$ here is inserted to match with the standard definition of $\sh^\vee$ in the literature.} of $\g$.  
Thus just from the point of view of
group theory, we get a coupling  $\sh^\vee \sum_a A_a t_a$ of the external gauge field to the Wilson line.  Thus the dependence on $\g$ is only
an overall factor of $\sh^\vee$.  

To calculate the anomaly, it suffices to consider the case of a Wilson line  supported on a curve $K\subset \R^2$ that is nearly a straight line,
and at a point  $z=z_0$ in $\C$.  In fact, since a simple description in terms of Yangians and integrable systems only arises in a limit in which
the metric of $\R^2$ is scaled up, we are really supposed to consider only the case that $K$ has  a very large radius
of curvature and therefore can everywhere be locally well-approximated as a straight line.  
We can pick coordinates so that $K$, in the region of interest, is very close to the $x$-axis in $\R^2$, and is described by a curve
$y=y(x)$, with $y$ everywhere small.   We will find that the amplitude $I_1$ that comes from the diagram of \fig \ref{Diagram1} is not gauge-invariant.
Under a gauge transformation $A\to A+D\veps$, with $\veps=\sum_at_a\veps_a(x,y,z)$ a gauge parameter, the variation of $I_1$ will be
\be\label{firstvar}\delta I_1=-\frac{\hbar\, \sh^\vee}{2\pi} \int_K \d x \left( \frac{\d^2y}{\d x^2}\partial_z\veps(x,y(x),z_0)\right) .\ee
Here as above $\sh^\vee$ is the dual Coxeter number.
(A factor of $\hbar$ is present because this is the anomaly in a diagram with two propagators, each giving a factor of $\hbar$, and one bulk vertex, proportional to $1/\hbar$.)
This formula for $\delta I_1$ as an integral over $K$ is written in a way that is only valid in a portion of $K$ that is close to the $x$-axis.
  Since we will assume $y(x)$ small, we can replace $\veps(x,y(x),z_0)$ by $\veps(x,0,z_0)$,
which we will write simply as $\veps(x,z_0)$. 
Now let us discuss how to cancel this anomaly by a correction to the classical definition of the Wilson operator.  We recall that, in linear order,
an external gauge field $A$ couples to the Wilson operator via
\be\label{birstvar}
I_\cl =\int_K \d x A_x(x,z_0) \;.  
\ee
This is gauge-invariant, as long as $z_0$ is constant.  To cancel the anomaly $\delta I_1$, we must modify the classical coupling $I_\cl$ so 
that it is no longer gauge-invariant.  We do this by replacing $z_0$ by $z_0-\hbar\, \sh^\vee \frac{1}{2\pi} \d y/\d x$.  Thus we replace $I_\cl$ with
\be\label{formvar}
I_\cl'=\int_K\d x A_x\left(x,z_0- \frac{\hbar\, \sh^\vee}{2\pi} \frac{\d y}{\d x}\right). 
\ee  
Under a gauge transformation $\delta A_x(x,y,z)=\partial_x \veps(x,y,z)$,
the variation of $I_\cl'$ is\footnote{We use $\d \veps/\d x=\partial_x\veps +(\partial_x z) \partial_z\veps.$
We work to first order in $y$ and replace $\partial_z\veps(x,z_0-\hbar\, \sh^\vee \d y/\d x)$ with $\partial_z\veps(x,z_0)$.}
\be\label{ormvar}
\delta I_\cl' =-\frac{\hbar\,  \sh^\vee}{2\pi}  \int_K \d x \left(\frac{\d^2 y}{\d x^2} \partial_z \veps(x,z_0)\right). 
\ee
This cancels the anomaly $\delta I_1$.

We learn, then, that including the anomaly, we must set not $z=z_0$ but $z=z_0-\tfrac{1}{2\pi}\hbar\, \sh^\vee  \d y/\d x$.  
The formulas as written, however, are only valid under the assumption that $\d y/\d x$ is small.  To get a more general formula,
we should re-express $\d y/\d x$ in terms of the angle $\varphi$ between the tangent vector to $K$ and some chosen direction in the $xy$ plane, for
instance the direction of increasing $x$.  If we define the sign of $\varphi$ to increase when $K$ rotates in a clockwise direction, then $\d  y/\d x$ can be identified
(when it is small) with $-\varphi$.   So a more general form of the relation between $z$ and the slope of $K$ is $z=z_0+\hbar\, \sh^\vee \varphi/(2\pi)$.  A different way
to express this result is to say that along $K$, what is constant is not $z$ but $z-\hbar\, \sh^\vee\varphi/(2\pi)$.   This is the way the anomaly was stated in section \ref{firstlook}.

\begin{figure}[htbp]
\centering{\includegraphics[scale=.4]{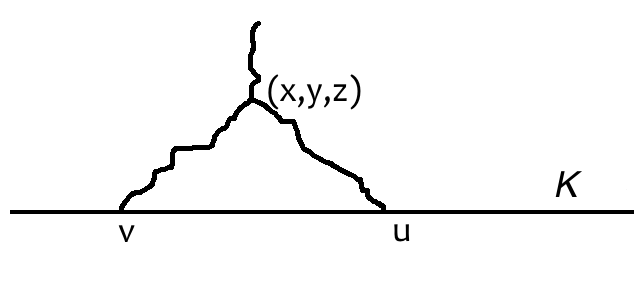}}
\caption{\small{The one-loop diagram with $K$ taken to be a straight line -- the $x$ axis.}}
\label{Diagram2}
\end{figure}
Before trying to explain how this anomaly comes from the diagram of \fig \ref{Diagram1}, let us first explain how 
the amplitude coming from that diagram is defined.
The only possible problem in defining the integral associated to this diagram comes from the region of integration in which all three interaction
vertices -- one in bulk and two along the curve $K$ -- coincide in $\R^4$.  In this region, $K$ can be approximated as a straight line.  Let
us take the line to be the $x$-axis, that is the curve $y=0$, and discuss the integral in this case.  For simplicity, we choose $z_0=0$. As in
the figure, we take the vertices along $K$ to be at $x=u,v$ with $v\leq u$, and we place the bulk vertex at a generic point $(x,y,z)$.
 Making use of the explicit form of the propagators,
the necessary integral turns out to be
(up to a constant factor, about which we will be more precise later)
\be\label{ziffo}
\int_{-\infty <v\leq u<\infty} {\d u\,\d v} \int_{\R^4} \d x \d y \d z \d \bar z  \,\, A_x(x,y,z) \frac{ 2y \bar z}{((x-v)^2+y^2+|z|^2)^2
((x-u)^2+y^2+|z|^2)^2} \;. 
\ee
Here $A=A_x \d x + A_y \d y +A_{\bar z}d\bar z$ is a background field. However, for the case of a straight
line in a flat $\R^2$ (and only for this case), only the $A_x$ component contributes, as is indicated in eqn.\ (\ref{ziffo}).  
This depends on the detailed form of the propagator.  

Any possible divergence in the integral will be a local effective action along $K$ that will be represented as an integral along $K$.
Thus in assessing the convergence of the integral, we should leave one variable unintegrated, say $x$.  (Alternatively, we could ensure
that there is no divergence in the integral over $x$ by taking $A$ to have compact support in the $x$ direction.)  So we set, for example,
$x=0$.  If $A_x(0,y,z)$ is independent of $y$ and $z$, then the integral (\ref{ziffo}) is linearly divergent by power counting.  (We have to integrate
over five variables $u,v,y,z,\bar z$, and the integrand scales as $(\mathrm{length})^{-4}$.)   However, the integrand is odd under $y\to -y$
and separately under $z\to -z$.  Accordingly, the integral is well-defined by a sort of principal value prescription.  We take $\R^4_\eta$
to be the subspace of $\R^4$ defined by $(y^2+|z|^2)^{1/2}\geq \eta$, and we replace the integral by 
\be\label{iffo}\lim_{\eta\to 0}
\int_{-\infty <v\leq u<\infty} {\d u\,\d v} \int_{\R^4_\eta} \d x \d y \d z \d \bar z  \,\, A_x(x,y,z) \frac{ 2y \bar z}{((x-v)^2+y^2+|z|^2)^2
((x-u)^2+y^2+|z|^2)^2}\;. 
\ee
This integral is well-defined. Concretely, if we expand $A_x(0,y,z)$ in powers of $y$, $z$ and $\bar z$ near $y=z=0$, the first term that can contribute
is proportional to $yz$ and leads to a convergent integral.  (Other nonvanishing contributions come from terms with additional factors of $y^2$
or $z\bar z$.)

We have described this for the case that $K$ is a straight line, but since any curve can be approximated locally as a straight line, the general
case is similar.  The one-loop integral can always be defined by a principal value procedure in which one constrains the bulk vertex to be a distance
at least $\eta$ from $K$ and then takes the limit $\eta\to 0$.   The only change from the discussion in the last paragraph is that instead of
breaking the symmetry under $y\to -y$ by picking out a term in $A(0,y,z)$ that is linear in $y$, we could make use of the curvature of $K$ to
break the symmetry.  This again leads to a convergent integral. 

Now that we have defined the 1-loop amplitude that we want to study, we can assess its gauge-invariance.  For this, it is convenient to perform
the $u$ and $v$ integrals in (\ref{iffo}), leading to an expression of the general form
\be\label{ffo}
\int_{\R^4} A\wedge \Theta_0 \;, 
\ee
where $\Theta_0$ is 3-form on $\R^4$.  To be more precise, $\Theta_0$ is a distributional 3-form on $\R^4$,
 that is, it is  a distribution on smooth 1-forms $A$.  We note that although we have
described a specific procedure to define $\Theta_0$,  there was nothing really distinguished about this procedure and we could
have used a different one.  The effect of using a different procedure would be to add  to $\Theta_0$  a distribution $\Theta'_0$
supported on $K$.  But
\be\label{ziffi} 
\int_{\R^4}A\wedge \Theta'_0=\int_K A \beta \;,
\ee
where  $\beta$ is  some  0-form on $K$, which depends on $\Theta'_0$.  A shift of this type in the effective action can be interpreted
as a change in the classical line operator whose quantum properties  we are trying to study. We are only interested in anomalies modulo
those that can be removed by such a redefinition of the underlying classical line operator.

We want to assess gauge-invariance of eqn.\ (\ref{iffo}) under $A\to A+D\veps$, where $D=\d x \partial_x+\d y\partial_y+\d\bar z\partial_{\bar z}$.    It suffices to consider the case $\veps=z \veps'$ where $\veps'$ is regular at $z=0$.
In fact, any $\veps$ can be written $\veps=z\veps'(x,y,z,\bar z) +\veps''(x,y,\bar z)$, where $\veps''$ is holomorphic in $\bar z$.  In the
following computation, any contribution from $\veps''$ will disappear after integration over the phase or argument of $z$.  (It will also become clear
that we can assume $\veps'$ to be independent of $z$ and $\bar z$ since terms proportional to $z$ or $\bar z$ are not singular
enough to be relevant, and that  likewise only the restriction of $\veps'$ to $K$ matters.)   Since $D$ commutes
with $z$, it is equivalent to replace $\Theta_0$ with $\Theta=z \Theta_0$ and to assess the invariance of
\be\label{fffo}\int_{\R^4} A\wedge \Theta \ee
under $A\to A+D\veps'$.    After integration by parts, this means that we need to study the distributional four-form $D\Theta$.
This form vanishes away from $K$ because of the classical gauge-invariance of the theory.  So we expect $D\Theta$ to be a distribution
with support on $K$.
Finally, as $\Theta$ is proportional to an explicit factor of $\d z$ which comes from the bulk interaction vertex, we can replace $D\Theta$
by $\d\Theta$, 
where $\d$ is the ordinary exterior derivative $\d = D+\d z\partial_z$.

\def\hat{\widehat}
Thus, we are reduced to studying $\d\Theta$, which should equal a distribution supported along $K$.  In fact, we claim
\be\label{nurc} 
\d\Theta =\alpha\delta_K \;,
\ee
where $\delta_K$ is a three-form delta function that is Poincar\'e dual to $K$, and $\alpha$ is a one-form supported on $K$.  The anomaly
is then given by $\alpha$. 
({\it A priori}, instead of $\delta_K$, we might have gotten an expression involving normal derivatives of $\delta_K$.  Indeed, this would
have happened had we not included an explicit factor of $z$ in the definition of $\Theta$.  With that term extracted, we will see that the most
singular contribution is a delta function rather than a derivative of one.)  We will use the following simple procedure to study 
$\d\Theta$.  We assume that $K$ is close to the $x$-axis, so that it can be parametrized by $x$.  Then we expand the three-form $\Theta$
as the sum of two terms, one proportional to $\d x$ and one not:
\be\label{ffop} 
\Theta =\d x \Lambda +\Lambda' \;. 
\ee 
Here $\Lambda'$ cannot contribute to the anomaly.   The reason is that as $\Lambda'$, by definition, is a 3-form that vanishes if contracted
with $\partial_x$, and is smooth away from $K$, $\d\Lambda'$ cannot generate a delta function supported on $K$ unless $\Lambda'$ already
has a delta function supported on $K$.  But such a term is irrelevant; it could be eliminated by redefining $\Theta$ along the lines of eqn.
(\ref{ziffi}).  

\begin{figure}[htbp]
\centering{\includegraphics[scale=.4]{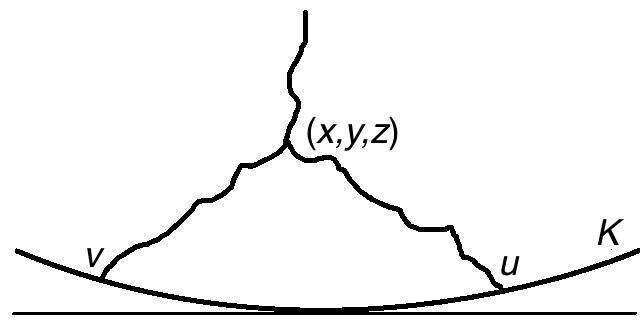}}
\caption{\small{Here we take $K$ to be tangent to the $x$-axis (the horizontal line in the figure) at, say, $x=0$.}}
\label{Diagram3}
\end{figure}

So in studying $\d\Theta$, we can replace $\Theta$ with $\d x \Lambda$, in other words with the part of $\Theta$ that is proportional to $\d x$.
We restrict $\Lambda$ to, say, $x=0$, and look for a delta function contribution in $\d\Lambda$ at $y=z=0$.  We can assume that $K$ is tangent
to the $x$-axis at $x=0$ and is described near $x=0$ by $y=\frac{1}{2}f x^2$, where $f=\d^2 y/\d x^2$. 
Note that $\Lambda=0$ for the case that $K$ is the $x$-axis.  (This is equivalent to the statement that the only component of $A$ that appears
in eqn.\ (\ref{ziffo}) above is $A_x$.)  So we expand $\Lambda$ in powers of $f$ near $f=0$.  We will see momentarily that the term in $\Lambda$ linear in $f$
does indeed lead to a delta function in $\d\Lambda$.  Higher order terms in $f$ are not singular enough to make such a contribution.

Explicit calculation (see Appendix \ref{apptwo}) shows that the contribution to $\Lambda$ that is linear in $f$ is at $x=0$,
\begin{align}\label{milp}
\Lambda= 
-\frac{3 f \ii }{32 \pi^2} \hbar\, \sh^\vee
\frac{ -y z \bar{z}  (\d z \wedge \d\bar{z} ) + 2z\bar{z}^2  (\d y \wedge \d z )}{\left(y^2+|z|^2 \right)^{\frac{5}{2}} } \;.
\end{align}
By Stokes's theorem, the coefficient of the delta function can be extracted as
\begin{align}
\int_{S_{\eta}} \Lambda \;,
\end{align}
where $S_{\eta}$ is the two-sphere $\left(y^2+|z|^2 \right)^{1/2}=\eta$
(the integral is independent of $\eta$, so we can choose $\eta=1$).
We then need to evaluate
\begin{align}
\int_{S_{\eta=1}}\left(
-y z \bar{z}  (\d z \wedge \d\bar{z} ) +2 z\bar{z}^2  (\d y \wedge \d z)\right) \;.
\end{align}
By choosing $y=\cos \theta, z=\sin \theta e^{\ii\phi}$, 
this becomes\footnote{We use
$\d z \wedge d\bar{z}=-2\ii \sin \theta \cos \theta d\theta \wedge d\phi$ ,
$dy \wedge \d z=-\ii \sin^2 \theta e^{- \ii\phi} d\theta \wedge d\phi$.
}
\begin{align}
2\pi  \int_0^{\pi} \d\theta
\left(-2\ii \sin^3 \theta \cos^2\theta
-2\ii \sin^5 \theta\right)
= (2\pi)(- 2\ii)  \left(\frac{4}{15}+\frac{16}{15}\right)
= -\frac{16\pi \ii}{3} \;.
\end{align}
Therefore the coefficient in front of the delta function is  $-f  \hbar\, \sh^\vee/(2\pi)$:
\be\label{omilp}
\d\Lambda'=-\frac{1}{2\pi} f  \hbar\, \sh^\vee \delta^3(y,z) \;,
\ee
where $\delta^3(y,z)$ is the Poincar\'e dual to the point $x=y=z=0$ in the hypersurface $x=0$.  

Since $f=\d^2y/\d x^2$, this result when inserted back in (\ref{fffo}) corresponds to an anomaly
\be
\label{dumbo}
- \frac{\hbar\, \sh^\vee}{2\pi} \int_K \d x \frac{\d^2 y}{\d x^2} \veps'. 
\ee
But $\veps'$ was defined by removing a factor of $z$ from $\veps$.  So an equivalent and more illuminating way to describe the result is that
the anomaly in the 1-loop amplitude under a gauge transformation generated by $\veps$ is
\be\label{dumbox}
-\frac{1}{2\pi}  \hbar\, \sh^\vee \int_K \d x \frac{\d^2 y}{\d x^2}\partial_z\veps \;. 
\ee
This is the form of the anomaly that was promised in eqn.\ (\ref{firstvar}).

\section{Networks of Wilson Lines}\label{networks}
\subsection{Overview}\label{netov}

\begin{figure}[htbp]
\centering{\includegraphics[scale=.4]{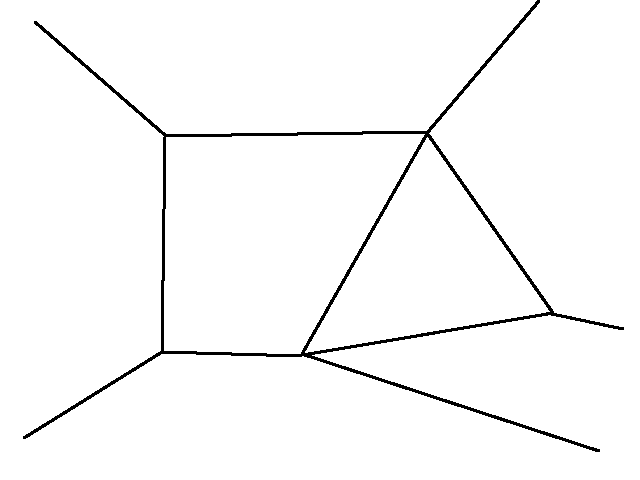}}
\caption{\small{A ``network'' of Wilson lines is a graph in which the line segments are Wilson lines, labeled by  representations of $G$, and a vertex represents a gauge-invariant coupling among the representations that label the lines that meet at that vertex. }}
\label{netw}
\end{figure}

By a ``network'' of Wilson lines, we mean simply a graph (\fig \ref{netw}) made of Wilson lines.  Each line segment in the graph is labeled by a representation of  the gauge
group $G$ (in general
a different representation for each segment), and a vertex in the graph represents, in physics language, a gauge-invariant coupling among the representations that
meet at that vertex.  At the classical level -- modulo the framing anomaly and its generalization for networks
-- the whole network will be at a fixed value\footnote{The analog of this in purely three-dimensional Chern-Simons theory is to consider
a not necessarily planar graph made from Wilson lines and embedded in spacetime in an arbitrary fashion.  For example, the quantum 6j symbol is
the expectation value of a tetrahedral graph  \cite{Wittenb}.} of the spectral parameter $z$.

Since we are already familiar with Wilson lines, the new ingredient in building such a network is the vertex.  So let us discuss this in more detail.
In general, if we are given a collection of representations $V_1,\dots,V_n$, with a $G$-invariant element 
\begin{equation}
 v \in  V_1 \otimes \dots \otimes V_n \;, 
\end{equation}  then 
we can at the classical level form a Wilson line vertex in which $n$ Wilson lines labeled by these $n$ representations meet,  as depicted in Fig. \ref{figure_network}.

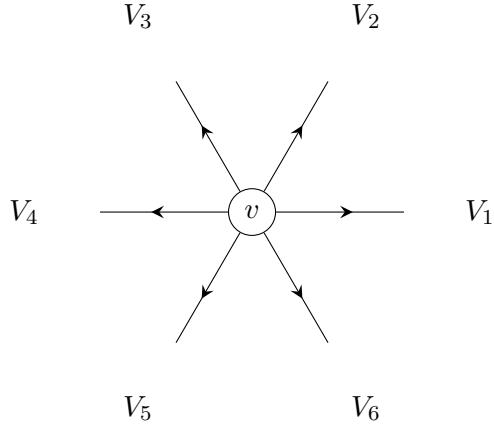
\begin{figure}
\begin{center}
\begin{tikzpicture}
\node(N1) at (0:3) {$V_1$};
\node(N2) at (60:3) {$V_2$};
\node(N3) at (120:3) {$V_3$};
\node(N4)at (180:3) {$V_4$};
\node(N5) at (240:3) {$V_5$};
\node(N6) at (300:3) {$V_6$};
\node[draw, circle] (central) at (0:0) {$v$}; 
\draw[-<-] (0:2) to (central);
\draw[-<-] (60:2) to  (central);
\draw[-<-] (120:2) to (central);
 \draw[-<-] (180:2) to (central);
 \draw[-<-] (240:2) to (central);
 \draw[-<-] (300:2) to (central); 
\end{tikzpicture}
\caption{\small{A vertex labeled by an invariant tensor $v \in  V_1 \otimes \dots V_6 $.  Here and later, lines are labeled by the corresponding representations.}}\label{figure_network}
\end{center}
\end{figure}  

This construction is possible, at least classically, in any gauge theory.  Classically, the vertex is defined as follows.  For each Wilson line, we define the holonomy operator as usual, starting at the vertex where all the Wilson lines meet.  The $G$-invariant tensor $v \in V_1 \otimes \dots \otimes V_n$ provides an initial state to which we apply the holonomy operator on each Wilson line. The resulting operator is invariant under gauge transformations which vanish at the other end of the Wilson line.  If each Wilson line 
has its second end either at another gauge-invariant vertex or at infinity (where we assume fields
and gauge parameters to be trivial), then we get a gauge-invariant network of Wilson lines, as in \fig \ref{netw}.   

In the theory under discussion in the present paper, not every classical Wilson line corresponds to a line operator of the quantum theory.  There is an obstruction that we first
encountered in section \ref{interpretation}.  Likewise, it turns out that not every classical vertex corresponds to a vertex in the quantum theory.

As we will see, the condition that a classical vertex should correspond to a vertex in the quantum theory is that a certain $\O(\hbar)$ anomaly should
vanish.  This anomaly comes from a lowest order Feynman diagram with one gluon exchange.

In this section, we will establish the result just stated; moreover, we will compute the $\O(\hbar)$ obstruction to ``quantizing'' a vertex
and find useful conditions under which it vanishes.

Finally, we will give concrete, interesting, and (as we will see in a companion paper) useful examples of quantum vertices.

Of course, to give examples of vertices, we first need examples of Wilson lines that we are allowed to work with. 
We will use the following sufficient condition that a Wilson line associated to a representation of $G$ (as opposed to a more general Wilson line associated
to a representation of $\mathfrak{g}[[z]]$) can be quantized.
Suppose the representation $V$ satisfies  the following algebraic condition:\footnote{An equivalent statement
is that the only irreducible representation of $G$ that appears both in $\wedge^2 \g$ and in $\op{End}(V)=V\otimes V^\ast$ is the adjoint representation.} 

\begin{minipage}{.9\linewidth}
$(\dagger)$ \ \ Every $G$-invariant map from $\wedge^2 \g$ to $\op{End}(V)$ factors through a copy of the adjoint representation.
\end{minipage} 

Then a Wilson line associated to the representation $V$ can be ``quantized,'' that is, it exists in the quantum theory.  Unfortunately, the proof of this
statement is rather technical, and we have relegated it to Appendix \ref{app.anomaly_Wilson}.   We should point out that the condition $(\dagger)$ is only a sufficient
criterion for a representation to be quantizable in this sense, but is far from being necessary.  For example, for $G=SL_n$, it is known from other arguments that
all representations of $G$ are quantizable, but in general the condition $(\dagger)$ is not satisfied.

We should perhaps remark that what we call vertices correspond, in the theory of integrable relativistic scattering, to couplings of external particles to 
bound states or poles of the $S$-matrix.
However, we will not try to make contact with the insight that comes from that point of view.

\subsection{Vanishing of Higher Order Anomalies}\label{vanishing}
Our first task is to show that an anomaly obstructing quantization of a vertex can arise only in lowest nontrivial order, from one-gluon exchange.

The anomaly that describes the failure of the configuration in \fig \ref{figure_network} to be gauge-invariant at the quantum level will be a local
operator of ghost number 1, made from the gauge field and the ghost field $\c$, supported at the vertex. It will be valued 
 in the vector space $\h V=\otimes V_i$ which lives at the vertex.
   As in our discussion of the framing anomaly, the anomaly is determined by local
 considerations, so the choice of the complex Riemann surface $C$ will not matter.  We may as well take $C=\C$ and take the vertex to be supported
 at $z=0$.   Likewise we can take the topological two-manifold $\Sigma$ to be the $xy$ plane with the vertex at the origin.

A key constraint is that the anomaly will be invariant under any symmetry of the theory that is also a symmetry of the classical vertex.  
A special case of diffeomorphism invariance in the $\Sigma$ directions is invariance under scaling of the $xy$ plane.  
 The configuration
in \fig \ref{figure_network} with $n$ Wilson lines emerging radially from a common vertex is invariant under this scaling, which therefore
must be a symmetry of the anomaly. 
This tells us that the anomaly cannot depend on the $x$ and $y$ components of the gauge field $A$, and it cannot involve any $x$ and $y$ derivatives. 

The classical theory is also invariant under the symmetry that simultaneously scales $z, \zbar$ and $\hbar$ by a real number, and under
the symmetry in which  $z$ is rotated through $\theta$, $\zbar$ through $-\theta$, and $\hbar$ through $\theta$. 
 It follows that the anomaly is also preserved by these symmetries.
From this it follows that the anomaly cannot have any $\zbar$ derivatives, nor can it depend on $A_{\zbar}$.  Thus the anomaly
must be constructed from the ghost field $\c$ only, and it is linear in $\c$ since it has ghost number 1.  
Moreover, in order $\hbar^k$, the anomaly must have $k$ $z$-derivatives, that is, it must be proportional to $\partial_z^k\c(0)$.

Finally, the anomaly must be invariant under constant gauge transformations.  Since $\partial_z^k\c$ transforms in the adjoint
representation, it must  be combined with a copy of the adjoint representation in $\h V=\otimes_i V_i$.  Thus, an order $\hbar^k$ contribution
to the anomaly must be an operator of the form 
\be\label{theno}\sum_a \partial_z^k \c^a (0) \alpha_a
\ee
 for some collection of elements $\alpha_a \in  \h V$ that transform in the adjoint representation of $\mf{g}$. 
 Here $\c^a$ are the components of the ghost field $\c$ relative to a basis $t_a$, $a=1,\dots, \mathrm{dim}\,\g$ of $\g$.

It remains to impose the condition that the anomaly must be BRST closed.   We will see that this condition is satisfied if and only if $k=1$.  
In implementing the condition of BRST invariance, we have to remember that the local operator that represents the anomaly lives at the endpoint
of $n$ Wilson lines.  In general, the BRST transformation of a Wilson operator $W(p,q)$ with ends at $p$ and $q$ is
\begin{equation}\label{polygo} 
\left\{Q, W(p,q) \right\}=\c(p)W-W \c(q)\;.
\end{equation}
Here $\c(p)$ and $\c(q)$ are operators acting in the representation carried by the Wilson line.
We are interested in the case that $q$ is the  location $x=y=z=0$ of the vertex under study, and we are only interested in the $\c(q)$ term in (\ref{polygo}) (the other term will
participate in a similar cancellation at the other end of the Wilson line in question).
For the $i^{th}$ Wilson line that ends at the vertex, we can write $\c(q)$ in more detail as 
\be\label{nolygo}
\c(q)=\sum_a \c^a(0) t_{a;i}\;,
\ee
where $t_{a;i}$ is the operator by which the Lie algebra generator $t_a$ acts in the representation $V_i$.

The quantity that must vanish for BRST invariance of the anomaly (\ref{theno}) is therefore
\be\label{toto}
-\sum_{i=1}^n \c^a(0)t_{a;i} \partial_z^k\c^b(0)\alpha_b +\{Q,\partial_z^k\c^b(0)\}\alpha_b\;.
\ee
The statement that the $\alpha_b$ transform in the adjoint representation means that
\be\label{oto}
\sum_{i=1}^n t_{a;i}\alpha_b=f_{ab}^c \alpha_c\;.
\ee
Using the standard commutator relation
\be\label{ofo}
\{Q,\c^a(0)\}=\frac{1}{2}f^a{}_{bc}\c^b(0)\c^c(0)\;,
\ee
from which it follows that
\begin{equation}
 \{Q,\partial_z^k\c^a(0)\}= \frac{1}{2} \sum_{r=0}^kf^{a}_{bc} \partial_z^r \c^b(0) \partial_z^{k-r} \c^c(0)\;,
\end{equation} 
we see that eqn.\ (\ref{toto}) is satisfied if and only if $k=1$.

We conclude that  anomalies can occur  only in order $\hbar$, that is in the lowest nontrivial order, due to one-gluon exchange.

\subsection{Calculating the  Anomaly}\label{calculating}
Let us now turn to calculating the anomaly. Any $\O(\hbar)$ anomaly must come from the Feynman diagrams  depicted in \fig \ref{figure_network_anomaly}.
\begin{figure}
\begin{center}
\begin{tikzpicture}
\node (N2) at (0:3) {$V_2$};
\node (N1) at (60:3) {$V_1$};
\node (N6) at (120:3) {$V_6$};
\node (N5) at (180:3) {$V_5$};
\node (N4) at (240:3) {$V_4$};
\node (N3) at (300:3) {$V_3$};

\draw[-<-] (N2) to (0:0);
\draw[-<-] (N1)  to  (0:0);
\draw[-<-] (N6)  to (0:0);
 \draw[-<-] (N5)  to (0:0);
 \draw[-<-] (N4) to (0:0);
 \draw[-<-] (N3) to (0:0);

\node(A) at (160:5) {$A$};
\draw[decorate,decoration=snake] (120:2.5) to (160:3);
\draw[decorate,decoration=snake] (160:3) to (A);
\draw [color=white, thick] (180:2.15) to (180:2.3);
 \draw[decorate,decoration=snake]  (240:2.5) to (160:3); 
\end{tikzpicture}
\caption{\small{One-loop anomaly to the Wilson network in \fig \ref{figure_network}.  Gluons are attached to two of the
outgoing Wilson lines, in this case the ones labeled by $V_4$ and $V_6$.  The full anomaly comes from a sum of such diagrams,
with gluons attached to any two distinct Wilson lines.}}
\label{figure_network_anomaly}
\end{center}
\end{figure}
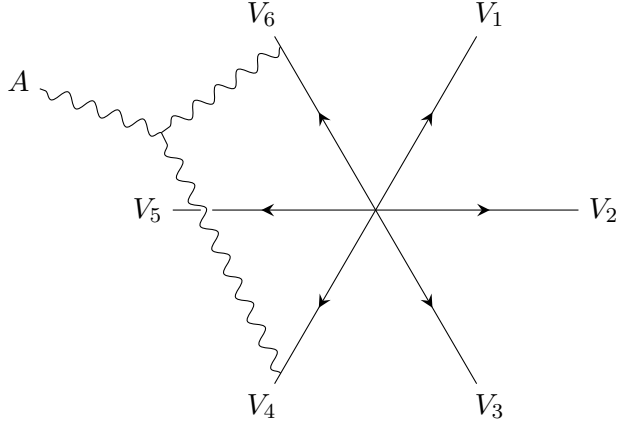

This diagram involves only two of the Wilson lines in the network, and the anomaly will receive a contribution from each pair of Wilson lines.  

For a fixed pair of Wilson lines, we have already done essentially this computation in section \ref{framinganomaly}: it is the same as the calculation leading to the framing anomaly.  There are a few small differences.  First, in the calculation we are now considering, both Wilson lines are labelled with arrows pointing away from the vertex.  At the classical level, this configuration is the same as one where one arrow is incoming, and the other is outgoing, but we use the representations $V_i$ and $V_j^\ast$ instead of $V_i$ and $V_j$.  

Also, in our calculation of the framing anomaly, we labelled both segments of the line by a single representation of the group, but now the two representations may be different.   This affects only the group theory factor, and even that factor can still be written similarly to what we had before.

Finally, in our previous analysis, the Wilson line was allowed to bend in an arbitrary fashion, but now the tangent direction is changing only by a delta function jump at the origin.

Taking these details into account, we can readily write down
 the anomaly associated to \fig \ref{figure_network_anomaly}, where the two gluons attach to Wilson lines $i$ and $j$. To get the full 
anomaly, we will have to sum over pairs $i,j$.  Let 
\begin{equation}
v \in V_1 \otimes \dots \otimes V_n
\end{equation}
denote the invariant tensor that we place at the vertex classically.  The group  theory factor associated to the diagram in which the gluons are attached to lines
$i$ and $j$ is \begin{equation}
\sum f_a{}^{bc} t_{b;i} t_{c;j} v \;.
\end{equation} 
(This reduces to eqn.\ (\ref{yft}) if $v$ is an invariant bilinear form that establishes an isomorphism between $V_i$ and $V_j^\ast$.) 
 
Let $\theta_{i}$ denote the angle between the $i^{th}$ Wilson line and the vertical, measured in the clockwise direction.  Let us assume that $\theta_i < \theta_j$ for $i < j$. This means that $V_1$ is the first line segment we encounter if we start at the vertical and move clockwise. Then, the anomaly associated to the $i$ and $j$ line segments with $i < j$ is proportional to $\theta_j - \theta_i - \pi$.  This follows from our calculation of the framing anomaly, which was the special case of this calculation when the representations $V_j$ was $V_i^\ast$. In that case, a straight line does not have an anomaly, and the anomaly is proportional to the deviation of the line from being straight. This deviation for the geometry
we are considering now is $\theta_j - \theta_i - \pi$. 

Thus, the total anomaly to the existence of the vertex at the quantum level, summing all pairs of Wilson lines and otherwise borrowing our previous result, is 
\begin{equation}
\frac{\hbar}{2\pi} \partial_z \c^a \left( \sum_{1 \le i < j \le n} (\theta_j - \theta_i - \pi) f_{a}{}^{bc} t_{b;i} t_{c;j} \right) v\;.  \label{equation_network_anomaly} 
\end{equation} 
This can be viewed as a $G$-invariant linear map from $\mf{g}$ to $V_1 \otimes \dots \otimes V_n$. 

A possibly  counter-intuitive feature of the anomaly is that it is discontinuous in the angles $\theta_i$. If $\theta_i$ increases so that it crosses $\theta_{i+1}$, there is a discontinuity in the anomaly of the form $- \hbar f_{a}{}^{bc} t_{b;i} t_{c;i+1} v$
(an exchange of $i$ and  $i+1$, .  This can be explained as follows. When $\theta_i = \theta_{i+1}$, what we are computing is the anomaly for a system with $n-1$ Wilson lines, where the lines $V_i$ and $V_{i+1}$ have been fused to a single Wilson line. 
 In section \ref{parallel}, we have calculated the $\O(\hbar)$  correction to the action $\g[[z]]$ on the fusion of two Wilson lines.  The level 1 generator acting on the fused Wilson line
 turned out to be 
\begin{equation}\label{morf}
 \mp \frac{1}{2}\hbar f_{a}{}^{bc} t_{b;i} t_{c;i+1} \;,
\end{equation} 
where the sign depends on whether $V_i$ is brought to $V_{i+1}$ from the left or the right.   Thus the configurations with $V_i$ approaching $V_{i+1}$ from left or right 
are different, and therefore they can have different anomalies.   

The above formula for the fusion of $V_i$ with $V_{i+1}$ shows that the anomaly for $n$ Wilson lines with $\theta_i$ tending to $\theta_{i+1}$ from below coincides with the anomaly for $n-1$ Wilson lines where $V_i$ and $V_{i+1}$ have  been fused, with $V_i$ merging with $V_{i+1}$ from the right.   (The fused Wilson line has a level one Yangian
generator.  A vertex in which it couples to ordinary Wilson lines that lack such a generator is not gauge-invariant at the classical level, and this contributes to the anomaly
for such a vertex.)

Similarly, the anomaly for $n$ Wilson lines where $\theta_i$ tends to $\theta_{i+1}$ from above coincides with the anomaly for $n-1$ Wilson lines where $V_i$ has been fused with $V_{i+1}$ from the left.  

Thus the fact that the anomaly has a discontinuity when $\theta_i$ crosses $\theta_{i+1}$ reflects the fact that the algebra of Wilson line operators is non-commutative: the result of fusing two parallel Wilson lines depends on whether they are brought together form the left or the right.

\subsection{Cancelling the Anomaly}\label{cancelling}
There are two ways to try to cancel this anomaly. We can shift the relative  positions of the Wilson lines in the $z$-plane by an amount of order $\hbar$.
Suppose we shift the $i^{th}$  Wilson line from $z_i=0$ to to $z_i = b_i \hbar$, for some constants $b_i$.   This introduces a new anomaly $\left( \sum b_i \hbar t_{a;i}\right) v$, and one can hope that with a judicious choice of $b_i$ this will cancel the anomaly.

The other thing we can try to do is to change the angles of the Wilson lines.   However, it turns out that moving the angle of the $i^{th}$ Wilson line has the same effect as shifting it in the $z$ plane.   This is clear from the framing anomaly, but we can derive it from \eqn \eqref{equation_network_anomaly}.  
Suppose that we change $\theta_i$ to $\theta_i + w_i$ for some fixed value of $i$, with $w_i$ a constant, leaving the other Wilson lines unchanged. 
Then, the anomaly shifts by
\begin{equation}
\frac{\hbar}{2\pi} \left( + \sum_{j < i} w_i  f_{a}{}^{bc} t_{b;j} t_{c;i} - \sum_{i < j} w_i  f_{a}{}^{bc} t_{b;i} t_{c;j} \right) v \;.  
\end{equation}  
Using the fact that $\sum_{j} t_{j;a} v = 0$, we can rewrite this as 
\begin{align}
\frac{\hbar}{2\pi} \left(  w_i  f_{a}{}^{bc} t_{b;i} t_{c;i} \right) v = \frac{\hbar\, \sh^\vee}{\pi}  w_i t_{a;i} v \;.    
\end{align}  
This is the same shift we find if we move the $i^{th}$ Wilson line by $w_i$ in the $z$ plane, up to a factor of the dual Coxeter 
number $\sh^\vee$ which arises from the framing anomaly.\footnote{If $\sh^\vee=0$, then shifting the angles has no effect and we have to try to cancel
the anomaly by shifting the values of $z$.  For simple Lie groups, $\sh^\vee$ is always positive, but there actually are interesting solutions of the Yang-Baxter equation for supergroups with $\sh^\vee=0$ \cite{Beisert}.}

Comparing the two formulas makes it clear that if it is possible to cancel the anomaly by adjusting the angles, then it is also possible to do
so by shifting the relative positions in the $z$-plane.  The converse is not quite true.  In our computation of the anomaly, we have made an assumption
about the cyclic order of the Wilson lines in the $xy$ plane.  This assumption only allows us to vary the angles while preserving certain inequalities.
In some cases, the anomaly can be canceled by shifting the $z_i$ but not by varying the angles, because the relevant inequalities would be violated.

\subsection{Indecomposable Vertices}\label{irred}

Before discussing specific examples, let us discuss a useful condition that can be satisfied by a classical vertex.

 For each $i$, the states $t_{a;i}v$ transform
in the adjoint representation of $G$, since $v$ is $G$-invariant.  This gives $n$ copies of the adjoint representation, but
$G$-invariance of $v$ is equivalent to one relation between them:
\be\label{morx}\sum_{i=1}^nt_{a;i}v=0 \;. 
\ee
Thus, these states form at most $n-1$ linearly independent copies of the adjoint representation.  We will call a vertex \emph{indecomposable} if
there are no further relations between the states $t_{a;i}v$, so that the number of linearly independent
copies of the adjoint representation made in this way is precisely $n-1$.
The motivation for the terminology is that if there is an additional relation, one can show by general considerations of group theory that we can decompose the set of representations
 into disjoint subsets $R$ and $S$  with the feature that the operators $\sum_{r\in R} t_{a; r}$ and $\sum_{s\in S} t_{a;s}$ each annihilate the element  $v$. This means that  $v$ is the tensor product (or a sum of tensor products) of invariant
 elements in $\otimes_{r\in R} V_{r}$ and in $\otimes_{s\in S} V_{s}$. Thus, our vertex decomposes as a tensor product of two vertices, which can be shifted
 to different values of $z$ and analyzed independently.
 
 Accordingly, it is reasonable to restrict our attention to indecomposable vertices.
 An indecomposable vertex has at least $n-1$ possible anomalies (since the $t_{a;i}v$ provide $n-1$ copies of the adjoint representation in
 $\h V=V_1\otimes V_2\otimes \cdots \otimes V_n$).   If there are only these $n-1$ linearly independent copies of the adjoint in $\h V$, then by shifting the $n-1$
 relative positions of the $n$ Wilson lines in the $z$-plane, there is a unique way to cancel the anomaly.   Indeed, anomalies of the form
 $t_{a;i}V$ are precisely the ones that can be eliminated by shifting the values of $z$.  
If certain inequalities are obeyed, the anomalies can also be cancelled by shifting the relative angles rather than the values of $z$.   
 
If the number of copies of the adjoint representation in $\h V$ is actually greater than $n-1$, one would expect that generically anomaly cancellation is 
 not possible. 

\subsection{Vertices Constrained by Symmetries}\label{consym}
We will describe various concrete examples of anomaly-free quantum vertices.\footnote{These detailed examples will not be needed in the rest of the present paper,
though some of them play an important role in the companion paper \cite{Part2}.} 
  The simplest examples, which do not require any computations, 
  arise when there are enough symmetry constraints to determine the angles between the Wilson lines and ensure that the anomaly vanishes.

For a simple and also useful case, 
suppose that  $V_1,\dots,V_n$ are all the same representation $V$, and consider a vertex associated to a $G$-invariant vector 
$v \in V^{\otimes n}$ that is either cyclically 
invariant or cyclically anti-invariant (in other words, assume that it either is invariant or changes sign under a cyclic permutation of the $n$ copies of $V$).
Such a  vertex is anomaly-free, assuming that the angles between successive Wilson lines are equal, so as to respect the symmetry. (Fig. \ref{figure_network} has been
drawn with equal angles, so it is invariant under a $2\pi/n$ rotation and  potentially represents a vertex with cyclic symmetry or antisymmetry.) The proof is simple.   Since the anomaly depends linearly on $v$, it is  cyclically invariant or anti-invariant if $v$ is. By assumption, there are no cyclically invariant or anti-invariant copies of the adjoint representation in $\h V$, so the anomaly is zero.

In most cases, a vertex with cyclic symmetry or antisymmetry has a further symmetry.  This happens because a cyclically
symmetric configuration of $n$ lines meeting at equal angles in the plane,
as
 in Fig. \ref{figure_network}, is actually invariant under suitable reflections of the plane.  In the gauge theory under study in the present paper, 
a reflection of the $xy$ plane is a symmetry if accompanied by $z\to -z$.  The latter is also a symmetry of ordinary Wilson lines (associated to representations of $\g$,
not $\g[[z]]$) that are supported at $z=0$.   So as long as the vector $v\in V^{\otimes n}$ used in constructing the vertex
 is either even or odd under the reflection symmetry (all examples we will consider will have this property), the corresponding cyclically symmetric classical vertex actually has dihedral symmetry, generated by $2\pi/n$
rotations and also certain reflections.   If the vertex is anomaly-free because of the conditions stated in the last paragraph, it will automatically possess the dihedral symmetry.

We will now describe explicit examples of anomaly-free vertices that are cyclically invariant or anti-invariant (and thus also dihedrally invariant).  
In each case, we have to first make sure that
the representation $V$ that we want to use is itself anomaly-free, in other words that there is a quantum Wilson line in this representation.
For this, we will use the criterion $(\dagger)$ that was stated at the end of section \ref{netov}: any $G$-invariant map from $\wedge^2\g$ to
$\mathrm{End}(V)=V\otimes V^\ast$ factors through the adjoint representation, or equivalently the only irreducible representation of $\g$ that
appears both in $\wedge^2\g$ and in $V\otimes V^\ast$ is the adjoint representation.

Here are some examples:
\begin{enumerate}
\item  $V$ is the fundamental representation of $SL_n$, which satisfies condition $(\dagger)$, and $v \in V^{\otimes n}$ is the  essentially unique\footnote{In such a statement, we always mean unique up to a constant multiple.} invariant tensor. There are $n-1$ copies of the adjoint in $V^{\otimes n}$, so this vertex can be quantized. Since $v$ is cyclically invariant or anti-invariant (depending on whether $n$ is odd or even), the vertex can be quantized so the Wilson lines all have the same value of $z$ and the angle between them is $2 \pi / n$.
\item Take $V$ to be the $\mathbf{7}$ of $G_2$, which satisfies condition $(\dagger)$, and take 
$$
v \in \wedge^3 \mathbf{7} \subset  \mathbf{7}^{\otimes 3}
$$
to be the essentially unique invariant tensor.  Using the tables on p.\ 298 of \cite{Ramond}, one finds that every map from $\wedge^2 \mathbf{14}$ to $\mathbf{7} \otimes \mathbf{7}$ factors through $\mathbf{14}$.  This implies that the Wilson line associated to $\mathbf{7}$ can be quantized. Further, the adjoint appears precisely twice in $\mathbf{7}^{\otimes 3}$.  Therefore, the vertex associated to the tensor $v$ can be quantized with angles $2\pi/3$ between the Wilson lines.
\item Consider the three $8$ dimensional representations $\mathbf{8}_v$, $\mathbf{8}_c$, $\mathbf{8}_s$ of $\op{Spin}(8)$, which are permuted by triality.  For each representation, criterion $(\dagger)$ holds, so that there are no anomalies to quantizing the corresponding Wilson lines.  Let 
$$
v \in \mathbf{8}_v \otimes \mathbf{8}_c \otimes \mathbf{8}_s
$$
be the essentially unique invariant tensor. Note that $v$ is also invariant under triality, together with a cyclic permutation of the representations.  There are only two copies of the adjoint in $\mathbf{8}_v \otimes \mathbf{8}_c \otimes \mathbf{8}_s$, so this vertex can be quantized.  
If the Wilson lines have relative angles $2\pi/3$, then a rotation through $2 \pi / 3$ together with an application of the triality symmetry of $\op{Spin}(8)$ is a symmetry of the configuration.  This tells us that the only consistent quantization is the one where the Wilson lines have relative angles $2 \pi/3$. 
\item  Take $V$ to be the $\mbf{26}$ of $F_4$.  Using table 45 of \cite{Slansky}, we find that this representation satisfies condition $(\dagger)$ and so the corresponding Wilson line can be quantized.  We consider the vertex associated to the essentially unique invariant tensor $v \in \Sym^3 \mbf{26}$.  Again using table 45 of \cite{Slansky}, one can check that there are only two copies of the adjoint in $\mbf{26}^{\otimes 3}$.  Therefore this vertex quantizes, with angles $2 \pi /3$ between adjacent lines.
\item Take $V$ to be the $\mbf{27}$ of $E_6$. Table 48 of \cite{Slansky} implies that this representation satisfies condition $(\dagger$). (See section \ref{vesix}.)
 We take $v \in \Sym^3 \mbf{27}$ to be the unique invariant element.  Again using table 48 of \cite{Slansky}, we find that there are only two copies of the adjoint in $\mbf{27}^{\otimes 3}$.  Therefore this vertex quantizes, with angles $2 \pi /3$ between adjacent lines.
\end{enumerate}

\subsection{A General Formula for the Angles at a Trivalent Vertex}\label{trivgen}

Now we turn our attention to general trivalent vertices.  Assuming that only two copies of the adjoint occur in $\h V=V_1\otimes V_2\otimes V_3$,
a classical trivalent vertex can be quantized.  Moreover, it is possible to find a simple general formula for the relative angles that are required.  
(We will express our results in terms of angles rather than in terms of shifts in $z$ because that corresponds to a simpler classical picture, but
when some relative angles come out to be negative and thus inconsistent with an assumed cyclic ordering of the vertices, one can
use the alternative approach in terms of shifting the relative values of the $z_i$.)

Suppose that $V_{\rho_1}$, $V_{\rho_2}$, $V_{\rho_3}$ are  three irreducible highest weight representations of a group $G$ of highest weights $\rho_1,\rho_2,\rho_3$, and suppose that the corresponding Wilson lines can be quantized.
 Consider a classical vertex associated to an invariant tensor $v \in V_{\rho_1} \otimes V_{\rho_2} \otimes V_{\rho_3}$.
As we have seen, if there are exactly two copies of the adjoint representation in $V_{\rho_1} \otimes V_{\rho_2} \otimes V_{\rho_3}$, then this vertex can
be quantized.

Suppose the Wilson lines are arranged in the plane with cyclic order $V_{\rho_1}, V_{\rho_2}, V_{\rho_3}$.  Let $\theta_{12}$, $\theta_{23}$, $\theta_{31}$ be the angles between the Wilson lines.  We will derive a formula for the angles $\theta_{ij}$. 

Let $c(\rho_i)$ denote the action of the quadratic Casimir of $\mf{g}$ on $V_{\rho_i}$. We define
\begin{align} 
\begin{split}
\label{beta_def}
\beta_{1} &=  c(\rho_1) - c(\rho_2) - c(\rho_3)\;, \\ 
\beta_{2} &= c(\rho_2) - c(\rho_3) - c(\rho_1)\;,\\
\beta_{3} &= c(\rho_3) - c(\rho_1) - c(\rho_2)\;.
\end{split}
\end{align}
We will show that the angles between the three Wilson lines are given by the formula
\begin{align}
\begin{split}
\theta_{12} &=  \pi - \pi \frac{ \beta_1 \beta_2  } {\beta_{1} \beta_{2} + \beta_{1} \beta_{3} + \beta_{2} \beta_{3}}\;, \\
\theta_{23} &=  \pi - \pi \frac{ \beta_2 \beta_3  } {\beta_{1} \beta_{2} + \beta_{1} \beta_{3} + \beta_{2} \beta_{3}}\;, \\ 
\theta_{31} &= \pi  - \pi \frac{ \beta_{1} \beta_{3} } {\beta_{1} \beta_{2} + \beta_{1} \beta_{3} + \beta_{2} \beta_{3}}\;. 
\end{split}
\label{theta_ans}
\end{align}

The derivation of this formula is as follows.  The anomaly to quantizing the vertex vanishes if
\begin{equation}
( (\theta_{12} - \pi) f_{a}{}^{bc} t_{b;1} t_{c;2} + (\pi-\theta_{31}) f_{a}{}^{bc} t_{b;1} t_{c;3} + (\theta_{23} - \pi) f_{a}{}^{bc} t_{b;2} t_{c;3}  ) v = 0 \;.\label{equation_trivalent}  
\end{equation}
This equation should hold for every value of $a$. 

Let us choose our basis $t_a$ to be orthonormal with respect to the chosen invariant pairing on the Lie algebra $\mf{g}$.   Applying the operator $t_{a;1}$ to \eqn (\ref{equation_trivalent}) and summing over $a$  we find
\begin{equation}
 \left(( \theta_{12} - \pi) f^{a}_{bc} t_{a;1} t_{b;1} t_{c;2} + (\pi-\theta_{31}) f^{a}_{bc}t_{a;1} t_{b;1} t_{c;3} + (\theta_{23} - \pi)f^{a}_{bc}t_{a;1} t_{b;2} t_{c;3}  \right) v = 0 \;. 
\end{equation}
Since
\begin{equation}
\sum_a f^a_{bc}t_{a;1} t_{b;1} = 2\sh^\vee t_{c;1} \;,
\end{equation}
we can rewrite this equation as 
\begin{equation}
\left(( \theta_{12} - \pi)\sh^\vee t_{c;1} t_{c;2} + (\pi-\theta_{31}) \sh^\vee t_{c;1} t_{c;3} + (\theta_{23} - \pi) f^{a}_{bc}t_{a;1} t_{b;2} t_{c;3}  \right) v = 0 \;. 
\end{equation}
Next, since 
\begin{equation}
 t_{c;3} v = - t_{c;1} v - t_{c;2} v    \;,
\end{equation}
we have
\begin{align} \begin{split}
f_{abc} t_{a;1} t_{b;2} t_{c;3} v &= - f_{abc} t_{a;1} t_{b;2} (t_{c;1} + t_{c;2})\\ 
        &= f_{acb} t_{a;1} t_{c;1} t_{b;2} - f_{bca} t_{a;1} t_{b;2} t_{c;2} \\ 
        &= t_{b;1} t_{b;2}\sh^\vee  - t_{a;1} t_{a;2} \sh^\vee\\ 
        &= 0 \;. 
\end{split}
\end{align}
Thus, our equation becomes
\begin{equation}
\left(( \theta_{12} - \pi)\sh^\vee t_{c;1} t_{c;2} + (\pi-\theta_{31}) \sh^\vee t_{c;1} t_{c;3} \right)v = 0 \;.
\end{equation}
Now,
\begin{equation}
\sum_c t_{c;1} t_{c;2} = \sum_c \tfrac{1}{2} (t_{c;1} + t_{c;2})^2 - \tfrac{1}{2} t_{c;1}^2 - \tfrac{1}{2} t_{c;2}^2 \;.
\end{equation}
Let $c(\rho_i)$ denote the eigenvalue of the quadratic Casimir $\sum t_a^2$ on the representation $V_{\rho_i}$. The operator $\sum_a t_{a;i}^2$ acts on $V_{\rho_1} \otimes V_{\rho_2} \otimes V_{\rho_3}$ by $c(\rho_i)$.

Acting on an element $v$, we have
\begin{align}
\begin{split}
\sum_c t_{c;1} t_{c;2} v &= \sum_c \tfrac{1}{2} (t_{c;1} + t_{c;2})^2 v - \tfrac{1}{2} t_{c;1}^2 v  - \tfrac{1}{2} t_{c;2}^2 v \\
        &= \tfrac{1}{2}t_{c;3}^2 v  - \tfrac{1}{2} t_{c;1}^2 v  - \tfrac{1}{2} t_{c;2}^2v  \\
        &= \tfrac{1}{2}c(\rho_3) v - \tfrac{1}{2} c(\rho_1)v - \tfrac{1}{2} c(\rho_i) v \;. 
\end{split}        
\end{align}
Thus (after dividing by $\sh^\vee/2$),    our equation becomes
\begin{equation}
(\theta_{12} - \pi)(c(\rho_3) - c(\rho_1) - c(\rho_2) )    = (\theta_{31}- \pi)(c(\rho_2) - c(\rho_1) - c(\rho_3)) \;.
\end{equation}
Similar arguments  give us two additional equations, which are the cyclic permutations of the equation we have just derived:
\begin{align} 
\begin{split}
 (\theta_{12} - \pi)(c(\rho_3) - c(\rho_1) - c(\rho_2) )&=  (\theta_{23} - \pi)(c(\rho_1) - c(\rho_2) - c(\rho_3))\; , \\ 
(\theta_{12} - \pi)(c(\rho_3) - c(\rho_1) - c(\rho_2) )   &= (\theta_{31}- \pi)(c(\rho_2) - c(\rho_1) - c(\rho_3)) \;.  
\end{split}
\end{align}
Evidently, one of these equations is redundant.

In terms of $\beta_i$ defined in \eqn \eqref{beta_def},
our equations become
\begin{equation} 
(\theta_{12} - \pi) \beta_{3} = (\theta_{23} - \pi)\beta_{1} 
\end{equation}
plus its cyclic permutations. 
Their unique solution is given by \eqref{theta_ans}.

As a consistency check, note that
$$
\theta_{12} + \theta_{23} + \theta_{31} = 2 \pi \;. 
$$
Note also that if the three representations are the same, then the angle between any two Wilson lines is $2 \pi /3$, as expected for symmetry reasons.

Let us specialize the formula to the case that the two representations $V_{\rho_1}$, $V_{\rho_2}$ are the same (or at least have the same eigenvalue of the quadratic Casimir).  In that case, letting $c(\rho) = c(\rho_1) = c(\rho_2)$, we have
\begin{align}
\begin{split}
\beta_{1} &= \beta_2 = - c(\rho_3) \;,\\
\beta_{3} &= c(\rho_3) - 2 c(\rho) \;. 
\end{split}
\end{align}
The angles between the Wilson lines become
\begin{align}
\begin{split}
\theta_{12} &= \pi - \pi \frac{c(\rho_3)}{4 c(\rho) -  c(\rho_3) }\;,\\
\theta_{31} &= \pi \frac{2 c(\rho) }{   4 c(\rho) - c(\rho_3)}\;, \\  
\theta_{23} &= \pi \frac{2 c(\rho) }{   4 c(\rho) - c(\rho_3)} \;. \label{formula_weights_angles_symmetric} 
\end{split}
\end{align}

As an example, let analyze the vertex connecting two copies of the fundamental representation of $\mf{sl}_n$ with the dual of the exterior square of the fundamental representation.  
 Let us normalize the quadratic Casimir so that its value on the fundamental representation $V$ is $1$ (the normalization plays no role in our formula).  Then, its value on $\wedge^2 V^\ast$ is $\tfrac{2(n-2)}{n-1}$.  If we take $V_{\rho_1}, V_{\rho_2}$ to be the two copies of the fundamental representation in the above calculation, then we have 
\begin{align}\begin{split}
\theta_{12} &= \pi - \pi \frac{\frac{2(n-2)}{n-1}}{4 - \frac{2(n-2)}{n-1} }
=  \pi \frac{2} {n}\;, \\ 
\theta_{31} &= \pi \frac{2}{4 - \frac{2 (n-2)}{n-1}} 
= \pi \frac{n-1}{n}\;, \\ 
\theta_{23} &= \pi \frac{n-1}{n}\;.  
\end{split}
\end{align}
This is a special case of a more general formula that we compute next. 

\subsection{\texorpdfstring{Trivalent Vertices Linking Fundamental Representations of $\mf{sl}_n$}{Trivalent Vertices Linking Fundamental Representations of sl(n)}}
We will describe trivalent vertices involving three of the fundamental representations of $\mf{sl}_n$, which are the  $k^{th}$ rank antisymmetric tensors
$\wedge^kV$, where $V$ is the fundamental representation and
$1\leq k\leq n-1$.   These representations are all quantizable (since in fact all representations of $\mf{sl}_n$ are quantizable), though criterion 
$(^\dagger)$ generally does not hold.

In what follows, it is useful to recall that, in an appropriate normalization, the eigenvalue $c(V)$ of the quadratic Casimir on a representation $V = V_{\rho}$ satisfies
\begin{equation}
c(V) = \frac{l(v)}{\op{dim}(V)} \; ,
\end{equation} 
where $l(v)$ is the Dynkin index of the representation.  Dynkin indices of various representations can be found in tables such as those in \cite{Slansky}. 

Let $V$ be the vector representation of $\mf{sl}_n$ so that $\wedge^k V$, $k = 1,\dots,n-1$ are the fundamental representations.  The Dynkin index of $\wedge^k V$ is $\binom{n-2}{k-1}$, so that
\begin{equation}
c(\wedge^k V) = \frac{\binom{n-2}{k-1}}{\binom{n}{k}} = \frac{k (n-k) }{n(n-1)} \;.
\end{equation} 
Consider three fundamental representations $\wedge^{k_i} V$ where $k_1 + k_2 + k_3 = n$. There is an invariant element of the tensor product of these representations coming from the map
\begin{equation}
\wedge^{k_1} V  \otimes \wedge^{k_2} V \otimes \wedge^{k_3} V \to \wedge^{n} V = \C \;. 
\end{equation}
The adjoint representation only appears twice in the tensor product of the three representations $\wedge^{k_i} V$. It follows that the vertex corresponding to the invariant tensor can be quantized. From formulae (\ref{theta_ans}), we can determine the angles.  

We set
\begin{equation} 
\beta_1 = c(\wedge^{k_1} V) - c(\wedge^{k_2} V) - c(\wedge^{k_2}V)
\end{equation}
plus cyclic permutation.  Then
\begin{align} \notag
 \beta_1 &=  \frac{1}{n(n-1)} \left( k_1 (n-k_1) - k_2(n-k_2) - k_3(n-k_3) \right) 
= \frac{-2 k_2 k_3}{n(n-1)}
\end{align}
plus cyclic permutations. Bearing in mind that the a scaling of all the $\beta_i$ will not affect the angles, we find
\begin{align}
\begin{split}
\theta_{12} &=  \pi \frac{k_1 + k_2 }{n}\;,  \\
 \theta_{23} &= \pi \frac{k_2 + k_3 }{n} \;, \\ 
\theta_{31} &=  \pi \frac{k_1  + k_3 }{n}  \;.
\end{split}
\end{align} 

\subsection{\texorpdfstring{Vertices Related to $\Gamma$-Matrices}{Vertices Related to Gamma-Matrices}}\label{gamma}
Consider the group $\op{Spin}(n)$ where $n$ is even. Let $V$ denote the vector representation, and $S_+$, $S_-$ the two irreducible (complex) spin representations of opposite chirality. These are both of dimension $2^{\tfrac{n}{2} - 1}$.   We will use criterion $(^\dagger)$ to show that these representations
can be quantized.

If $n = 2 \mod  4$, then the vector representation appears in $S_+ \otimes S_+$ and in $S_- \otimes S_-$.  If $n= 0 \mod 4$, then the vector representation appears in $S_+ \otimes S_-$.  In each case, we can try to quantize the vertex linking the vector representation to two spin representations. 
 
In the calculation that follows, we will assume $n \ge 8$ to avoid any low-dimensional coincidences. 

Note that the endomorphisms of the direct sum $S_+ \oplus S_-$ make up the Clifford algebra $\op{Cl}_n$ built from the vector representation $V$.  As a representation of $\mf{so}(n)$, the Clifford algebra is isomorphic to direct sum of the exterior powers of $V$. 

To use condition $(\dagger)$ to show that the spinor representations can be quantized,  we need to classify maps
\begin{equation}
\wedge^2  \mf{so}(n)  \to \wedge^k V
\end{equation} 
for all values of $k$.   An exercise in classical invariant theory tells us that the only such maps that can possibly exist are when $k = 2,4,n-4,n-2$.  If $k = 2,n-2$ then $\wedge^k V$ is the adjoint representation, so such maps will not contribute anomalies. If $k = 4,n-4$ then the only map $\mf{so}(n)^{\otimes 2} \to \wedge^k V$ comes from the wedge product map
\begin{equation}
\mf{so}(n) \otimes \mf{so}(n) = \wedge^2 V \otimes \wedge^2 V \to \wedge^4 V
\end{equation} 
(noting that $\wedge^4 V = \wedge^{n-4} V$).  This map is symmetric, not antisymmetric, and so is not relevant to condition $(\dagger)$.
We conclude that  in all cases, the Wilson lines attached to the spin representations exist in the quantum theory. 

Next, let us analyze whether the vertex linking the vector representation $V$ with two copies of a spin representation can be quantized.  We will start with the case when $n = 2 \mod 4$, in which case the vector representation appears in $S_+ \otimes S_+$.  We need to show that there are exactly two copies of the adjoint representation in $V \otimes S_+ \otimes S_+$. 

To do this, we need to recall how to describe $S_+ \otimes S_+$ in terms of exterior powers of $V$ (this computation is familiar from the study of central extensions of supersymmetry algebras).  Since $n = 2 \mod 4$, $S_+$ and $S_-$ are dual representations.  The endomorphisms of $S_+ \oplus S_-$ which reverse the parity of a spinor are the odd elements of the Clifford algebra built from $V$.  This space of endomorphisms is $S_+ \otimes S_+ \oplus S_- \otimes S_-$. Thus, 
\begin{equation}
S_+ \otimes S_+ \oplus S_- \otimes S_- = \oplus_{k \text{ odd} } \wedge^k V \;.
\end{equation}
We can further decompose the right hand side of this equation to find
\begin{equation}
S_+ \otimes S_+ = V \oplus \wedge^3 V \oplus \dots \oplus \wedge^{n/2}_+ V \;,
\end{equation}
where $\wedge^{n/2}_+ V $ indicates those elements which are self-dual under the Hodge star operator. 

In $V\otimes S_+\otimes S_+=V\otimes (V \oplus \wedge^3 V \oplus \dots \oplus \wedge^{n/2}_+ V)$, the only copies of the adjoint representation 
are those  in $V \otimes V$ and in $V \otimes \wedge^3 V$, each of which contains one copy of the adjoint.  So overall there are precisely two copies, and the vertex connecting $V$, $S_+$ and $S_+$ can be quantized.

Next let us check whether the vertex can be quantized in the case that $n = 0 \mod 4$.  In this case $V$ appears in $S_+ \otimes S_-$ and the representations $S_+$, $S_-$ are self-dual.  We again find that $S_+ \otimes S_- \oplus S_- \otimes S_+$ is the space of odd elements in the Clifford algebra, so that
\begin{equation}
S_+ \otimes S_- = V \oplus \wedge^3 V \oplus \dots \oplus \wedge^{n/2-1} V\;. \label{equation_spm} 
\end{equation}   
The only copies of the adjoint that appear in $V \otimes S_+ \otimes S_-$ are those in $V \otimes V$ and $V \otimes \wedge^3 V$, so again the vertex can be quantized.

Now we know that the spinor representations and the vertices connecting spinor and vector representations can be quantized.  The final step is to calculate the angles between the lines at the vertex. To do this, we need to know the Dynkin indices of the representations $V, S_{\pm}$. The Dynkin index of $V$ is $2$.   We can calculate the Dynkin indices of $S_{\pm}$ as follows.

Recall that the Dynkin index is additive under direct sums of representations, and under tensor product there is the following formula:
\begin{equation}
l (R_1 \otimes R_2) = l(R_1) \op{dim} (R_2) + \op{dim} (R_1) l(R_2)\;,
\end{equation} 
where $R_1$, $R_2$ are representations, $l(R_i)$ is the Dynkin index and $\op{dim}(R_i)$ is the dimension.  Further, if $R$ is any representation, then
\begin{equation}
l(\wedge^k R) = \binom{\op{dim} (R) - 2}{k-1} l(R)\;.
\end{equation} 
The representations $S_+$ and $S_-$ are related by a diagram automorphism of the Dynkin diagram $D_{n/2}$ of the Lie algebra $\mf{so}(n)$.  Therefore they have the same Dynkin index.  Using the fact that $(S_+ \oplus S_-)^{\otimes 2}$ is the sum of all the exterior powers of $V$, we find 
\begin{align}
2^{n/2+2} l (S_{\pm}) &= 2\sum_{k = 1}^{n-1} \binom{n-2}{k-1} 
= 2^{n-1}\;.  
\end{align}
Thus, 
\begin{equation}
l(S_{\pm}) = 2^{\frac{n}{2}-3}\;. 
\end{equation}

Now let us compute the angles between the three Wilson lines. We label the lines where $V_1 = S_+$, $V_2 = S_-$ and $V_3 = V$ is the vector representation.   The quadratic Casimirs in each representation are the ratios of the Dynkin index to the dimension. They are
\begin{align}\begin{split}
c(V) &= \frac{2}{n} \;, \\
c(S_{\pm}) &= \frac{2^{\frac{n}{2} - 3}}{2^{\frac{n}{2} - 1}} = 2^{-2}\;. 
\end{split}
\end{align} 
According to formula (\ref{formula_weights_angles_symmetric}), we have
\begin{align}
\begin{split}
\theta_{12} &= \pi \frac{n-4}{n -  2   } \;,\\  
\theta_{31} &= \pi \frac{n }{ 2n - 4}\;,\\  
\theta_{23} &= \pi \frac{n }{   2n - 4}\;. 
\end{split}
\end{align}
Note that when $n  = 8$, the angle between any two lines is $2 \pi / 3$, which is consistent with what we determined earlier using the triality symmetry of the vertex in this case. When $n = 6$, the vertex we are considering is that relating two copies of the vector representation of $\mf{sl}(4)$ with the exterior square of the fundamental.  The formula for the angle for a vertex connecting three fundamental representations of $\mf{sl}_n$ agrees, in this case, with the formula given here. 

\subsection{\texorpdfstring{A Vertex Connecting Representations of $E_6$}{A Vertex Connecting Representations of E(6)}} \label{vesix}
Just for fun, let us use our formula to calculate the angles in a vertex associated to representations of the exceptional group $E_6$.  The fundamental representation of $E_6$ will be denoted by $\mbf{27}$, and its dual by $\br{\mbf{27}}$.  There are four representations of dimension $351$, which come in dual pairs. We will use the conventions of \cite{Slansky} and denote them by $\mbf{351}$, $\br{\mbf{351}}$, $\mbf{351}'$, $\br{\mbf{351}}'$.    The $\mbf{27}^3$ vertex was already considered in section 
\ref{consym}, so here we primarily consider a more elaborate example. (We also will complete the discussion of the $\mbf{27}^3$ by showing that the $\mbf{27}$ can
be quantized.)

According to table 48 of \cite{Slansky}, $\br{\mbf{351}}'$ appears once in $\mbf{27} \otimes \mbf{27}$.  Thus there is an invariant tensor in $\mbf{27} \otimes \mbf{27} \otimes \mbf{351}'$. 

We would like to quantize this to a vertex connecting three Wilson lines. To do this, we first need to show that the Wilson lines themselves quantize.  It is sufficient, according to condition
$(\dagger)$,  to show that any map from the exterior square of the adjoint representation to the endomorphisms of the $\mbf{27}$ or $\mbf{351}'$ factors through the adjoint representation.

According to table 48 of \cite{Slansky}, the exterior square of the adjoint representation decomposes as 
\begin{equation}
\wedge^2 \mbf{78} = \mbf{78} \oplus \mbf{2925} \;.
\end{equation}
To show that the $\mbf{27}$ and $\mbf{351}'$ quantize, we need to show that $\mbf{2925}$ does not appear in $\br{\mbf{27}} \otimes \mbf{27}$ or in $\br{\mbf{351}}'\otimes \mbf{351}'$.  Table 48 of \cite{Slansky} shows that it does not, so these representations quantize.

Next, to show that the vertex in $\mbf{27} \otimes \mbf{27} \otimes \mbf{351}'$ quantizes, we need to show that the adjoint representation appears precisely twice in this tensor product. Table 48 of \cite{Slansky} tells us that
\begin{equation}
\mbf{27}\otimes \mbf{27} = \br{\mbf{27}} \oplus \br{\mbf{351}} \oplus \br{\mbf{351}}' \;.
\end{equation}
If we tensor this with $\mbf{351}'$, table 48 of \cite{Slansky} tells us that the adjoint appears once in $\br{\mbf{351}}' \otimes \mbf{351}'$, once in $\br{\mbf{351}} \otimes \mbf{351}'$, and not at all in $\br{\mbf{27}} \otimes \mbf{351}$.  Therefore the vertex quantizes.

Next, let us compute the angles. Let us label the representations as $V_1 = \mbf{27}$, $V_2 = \mbf{27}$, $V_3 = \mbf{351}$.  Table 47 of \cite{Slansky} tells us that the Dynkin index of $\mbf{27}$ is $6$ and that of $\mbf{351}'$ is $6 \times 28$. The values of the quadratic Casimirs are $6/27$ and $6 \times 28/351$.  We can change the normalization so that the values of the quadratic Casimirs are $1$ and $28/13$.  The angles are 
\begin{align} 
\begin{split}
&\theta_{12} = \pi - \pi \frac{\frac{28}{13}}{4 - \frac{28}{13}} 
=  - \pi \frac{1}{6} \;, \\
&\theta_{23} = \theta_{31}  = \pi \frac{2 }{4 - \frac{28}{13}}
= \pi \tfrac{13}{12}\;.  
\end{split}
\end{align}
Since, in this example, $\theta_{12} < 0$, we have to shift the Wilson lines in the $z$-plane instead of just placing them at angles in the topological plane.

Many more examples can be analyzed in a similar way.

\subsection{\texorpdfstring{A $4$-Valent Vertex for the $\mathbf{56}$ of $E_7$}{A 4-Valent Vertex for the 56 of E(7)}}

The smallest representation of $E_7$ is the $\mbf{56}$.  It is a pseudoreal or symplectic representation, so there is an invariant antisymmetric form $\omega\in \wedge^2\mbf{56}$.
In addition, there is a completely symmetric quartic invariant $\psi\in \mathrm{Sym}^4\mbf{56}$.    It is natural to ask whether this vertex can be quantized, like the 
$\mbf{27}^3$ of $E_6$, which was one of our examples in section \ref{consym}.  The answer is that it can, though the proof is not as simple as for the $\mbf{27}^3$ vertex.
Both of these examples will be useful in \cite{Part2}.

A configuration of four Wilson lines cannot have full $S_4$ permutation symmetry.   The maximum possible symmetry is a dihedral subgroup $D_4$.
This is simply the symmetry group of four Wilson lines with equal relative angles $\pi/2$, say running along the $\pm x$ and $\pm y$ axes.  Dihedral symmetry was
discussed in section \ref{consym}.  In this case, the group $D_4$ is generated by rotations of the $xy$ plane by an angle $\pi/2$, along with a reflection that preserves
the given configuration of Wilson lines. Thus in all $D_4$ has eight elements.

We would like to understand the possible vertices that are $D_4$ invariant, and the possible anomalies that are compatible with the $D_4$ symmetry.    Let us first enumerate the 
$E_7$-invariant elements in $\mathbf{56}^{\otimes 4}$ that are also $D_4$-invariant.
Note that we can identify $\mathbf{56}$ with its dual, using the $E_7$-invariant symplectic form $\omega$.  We can therefore identify $E_7$-invariant elements of $\mathbf{56}^{\otimes 4}$ with maps of $E_7$ representations
\begin{equation} 
\mathbf{56}^{\otimes 2} \to \mathbf{56}^{\otimes 2} \;. 
\end{equation} 
According to the tables of \cite{Slansky} or \cite{McKay},  $\mathbf{56}^{\otimes 2}$ decomposes as a sum of $4$ distinct irreducible representations.  Therefore, there are $4$ $E_7$ invariant linear operators on $\mathbf{56}^{\otimes 2}$, given by the projectors onto these $4$ irreducible subrepresentations. Correspondingly, there are $4$ invariant tensors in $\mathbf{56}^{\otimes 4}$.

We can enumerate these  tensors as follows.  One of them is the completely symmetric invariant $\psi$ with which we began.  The other three are more elementary.
Let $e_i$ be a basis of $\mathbf{56}$;  in this basis, the antisymmetric form $\omega$ corresponds to a matrix $\omega^{ij}$.   The other three $E_7$ invariants in $\mbf{56}^{\otimes 4}$
 are given by the formulas
\begin{align} 
\begin{split}
& \omega^{ij} \omega^{kl} e_i \otimes e_j \otimes e_k \otimes e_l \;,\\
&\omega^{ik} \omega^{lj} e_i \otimes e_j \otimes e_k \otimes e_l \;, \\ 
&\omega^{il} \omega^{jk} e_i \otimes e_j \otimes e_k \otimes e_l  \;.
\end{split}
\end{align} 
Among these three tensors, there is a single linear combination which is invariant under $D_4$, namely
\begin{equation} \label{porm}
(\omega^{ij} \omega^{kl} + \omega^{jk} \omega^{li} ) e_i  \otimes e_j \otimes e_k \otimes e_l  \;.
\end{equation}
We conclude that there are a total of two dihedrally invariant tensors in $\mathbf{56}^{\otimes 4}$.   

Given the explicit form of the invariant (\ref{porm}), an elementary computation using the general formula (\ref{equation_network_anomaly}) for the anomaly
shows that a vertex constructed using this invariant has an anomaly.  This means that it will be possible to use this invariant as a counterterm to help
in canceling an anomaly.

We will prove that there is a unique linear combination of these two dihedrally invariant tensors which quantizes to a vertex linking the Wilson lines.   To show this,
we have to show that there is precisely one possible anomaly.

\subsubsection{Anomalies}
  Anomalies are proportional to $\partial_z \c$.    Since an element of $D_4$ that acts
as a reflection of the $xy$ plane also acts as $z\to -z$, changing the sign of $\partial_z\c$, the group theory invariant that multiplies $\partial_z \c$ in an anomaly
 is not $D_4$-invariant.  Rather, it is $D_4$ anti-invariant,
that is,  invariant under rotations in $D_4$ but odd under reflections.   We will show that the adjoint representation of $E_7$ (which is the $\mbf{133}$) occurs
precisely once in the $D_4$ anti-invariant part of $\mbf{56}^{\otimes 4}$.  Since the invariant (\ref{porm}) does have a nonzero anomaly, this implies that by adding a multiple
of (\ref{porm}) to the invariant $\psi$, one can construct an anomaly-free $\mbf{56}^{\otimes 4}$ vertex.  

As a first step, let us compute the part of $\mathbf{56}^{\otimes 4}$ that is anti-invariant under a dihedral subgroup $D_2 \subset D_4$.    There actually are two possible
embeddings of $D_2$ in $D_4$.  We make the following choice.  If the four Wilson lines run along the $\pm x$ and $\pm y$ axes, we consider a subgroup $D_2\cong
\Z_2\times \Z_2$ that is generated by a reflection that acts by $(x,y)\to (x,-y)$ and one that acts as $(x,y)\to (-x,y)$.  Thus if we label the four Wilson lines in
cyclic order as 1,2,3, and 4, one reflection acts by exchanging 1 and 3, keeping fixed 2 and 4, and the other exchanges 2 and 4, keeping fixed 1 and 3.
Thus, the part of $\mbf{56}^{\otimes 4}$ that is odd under each reflection is $\wedge^2\mbf{56}\otimes \wedge^2\mbf{56}$.  The two factors are associated to the pairs 13 and 24.

From the tables in \cite{Slansky} or \cite{McKay}, one has $\wedge^2\mbf{56}\cong \mbf{1}\oplus \mbf{1539}$.  So the $D_2$ anti-invariant part of 
$\mbf{56}^{\otimes 4}$ is $(\mbf{1}\oplus \mbf{1539})\otimes(\mbf{1}\oplus \mbf{1539})$.

Now we want to identify the part of this that is anti-invariant under $D_4$, not just under $D_2$.  So we have to consider the action of a $\pi/2$ rotation.   The $D_4$
anti-invariants are simply the $D_2$ anti-invariants that are invariant under a $\pi/2$ rotation.  However, there is a small surprise when we try to impose invariance
under a $\pi/2$ rotation on the above description of the $D_2$ anti-invariants.  

A $\pi/2$ rotation exchanges the two factors of $\wedge^2\mbf{56}$ that we used in the above analysis, but with an important minus sign.  This happens as follows.
We recall that the two factors of $\wedge^2\mbf{56}$ are associated respectively to the pair of Wilson lines 13 and 24.  A $\pi/2$ rotation maps 13 to 24, but it maps
24 to 31; replacing 31 with 13 acts as $-1$ on one of the two copies of $\wedge^2\mbf{56}$.   Thus the $D_4$ anti-invariant part of $\mbf{56}^{\otimes 4}$ is the antisymmetric part of
$(\mbf{1}\oplus \mbf{1539})\otimes(\mbf{1}\oplus \mbf{1539})$, or more explicitly it is\footnote{As a check on this, the dimension of 
 $\mbf{1539}\oplus \wedge^2\mbf{1539}$ is 1185030.  This is the right dimension for the $D_4$ anti-invariant part of $\mbf{56}^{\otimes 4}$.  For example, an exercise
 using the character table of $D_4$ (or based on elementary considerations)
 tells us that the $D_4$ anti-invariants in $(\C^d)^{\otimes 4}$ are of dimension $\frac{1}{8}\left(d^4-2d^3-d^2+2d\right)$.}
 $\mbf{1539}\oplus \wedge^2\mbf{1539}$.  
 
 From the tables of \cite{McKay}, one learns that the adjoint or $\mbf{133}$ of $E_7$ occurs precisely once in $\mbf{1539}^{\otimes 2}$.  This one occurrence actually is
 in $\wedge^2\mbf{1539}$, not in $\mathrm{Sym}^2\mbf{1539}$, because, as $\mbf{1539}$ is a real representation of $E_7$, the adjoint must occur at least once
 in $\wedge^2\mbf{1539}$.   So as claimed above, there is precisely one possible anomaly.

\section{Two-Loop Correction To Gauge Invariance}\label{section_2loop} 

\subsection{Preliminaries}\label{prelimones}

Consider a Wilson line in our four-dimensional theory in a general representation $V$.  If we choose a basis $t_a$ of the Lie algebra $\mf{g}$, the Wilson line is characterized classically by matrices 
\begin{equation}
t_{a,k} : V \to V \;,
\end{equation}
where, as before, the matrices $t_{a,k}$ tell us how $\partial_z^k A$ is coupled to the Wilson line.   If $V$ is a representation of $G$ and not of $\g[[z]]$, then $t_{a,k}=0$ for $k>0$.

At the classical level, gauge-invariance requires that these matrices must satisfy the commutation relations
\begin{equation}  \label{classones}
[ t_{a,k}, t_{b, l} ] = f_{ab}{}^c t_{c, k+l} \;.
\end{equation}
These relations receive quantum corrections; the condition for a Wilson operator to be anomaly-free at the quantum level is different
from eqn. (\ref{classones}).    We have seen in section \ref{interpretation} that a correction must occur at order $\hbar^2$. We gave one derivation of this statement, which relied on an analysis of the fusion of parallel Wilson lines. In a companion paper \cite{Part2}, we will give another derivation based on the 
RTT presentation of the Yangian algebra. 

Both derivations are a little indirect, and one might wish for a more direct one.  
In this section we will provide a direct derivation via Feynman diagrams.

\begin{figure}[htbp]
\centering{\includegraphics[scale=0.7]{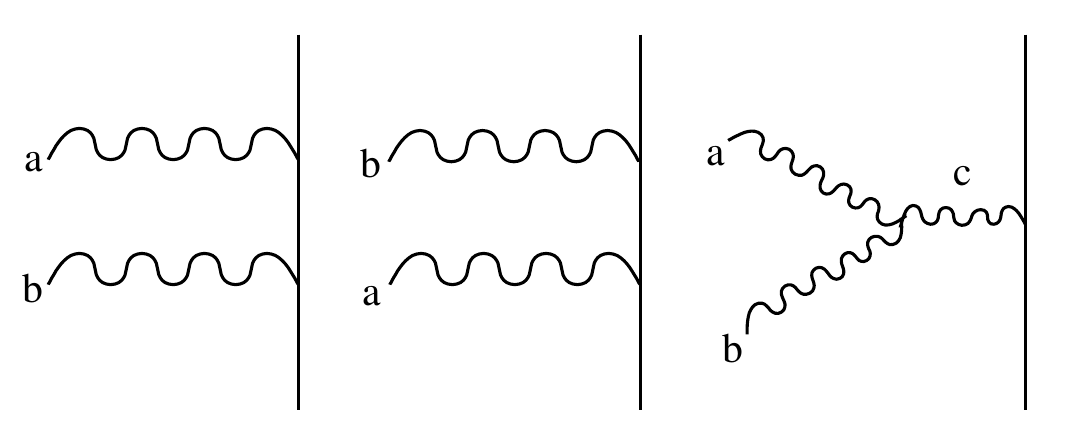}}
\caption{\small{The three tree-level diagrams describing coupling of a pair of gauge bosons to a charged particle or a Wilson line operator.
} }\label{tred}
\end{figure}

Some elementary remarks may help one anticipate what sort of diagrams will be relevant.  In an ordinary gauge theory with a finite dimensional Lie algebra $\mathfrak g$, how
does one usually see in perturbation theory that the matrices $t_a$ by which gauge fields $A^a$ couple to a charged particle or a Wilson line must satisfy the expected commutation relations?
At tree level, there are three diagrams (F   ig. \ref{tred}).  They are proportional respectively to $t_a t_b$, $t_b t_a$, and $f_{ab}^c t_c$, where $f_{abc}$ are the structure
constants that appear in the bulk $A^3$ vertex.  One tests these diagrams for gauge-invariance or BRST invariance
 by making a gauge transformation $A^a\to A^a+\d \c^a$ for one of the 
external gauge bosons, with $\c^a$ the ghost field.  A standard calculation shows that the sum of the three diagrams is gauge-invariant and BRST-invariant if and only if $[t_a,t_b]=f_{ab}^ct_c$.  The violation of BRST-invariance if this condition is not satisfied is bilinear in $A$ and $\c$, because the diagrams have two external bosons, one of
which is replaced by $\d\c^a$ when one tests for BRST-invariance.  Likewise,
in the analysis that follows, the anomalies come from diagrams with two external gauge bosons, and are bilinear in $A$ and $\c$.

After verifying gauge invariance at tree level, one should go on to determine whether quantum corrections to the coupling of two gauge bosons to a charged particle
or to a Wilson line preserve gauge invariance.    In conventional four-dimensional quantum field theory, the answer in general is that there is a problem: in certain theories,
one-loop triangle diagrams have an anomaly that irreparably spoils gauge invariance.  In the model that we will study, there is no problem in order $\hbar$, but we will
find an anomaly in order $\hbar^2$ -- roughly speaking, in two-loop order.  This anomaly, however, will not represent a complete breakdown of gauge invariance.  Rather,
it will represent a deformation of the gauge symmetry algebra -- a quantum correction to the classical commutation relations (\ref{classones}). 

Such a quantum deformation as the outcome of an analysis of anomalies is perhaps unfamiliar, for the following reason.  It 
does not occur in the case of a semi-simple gauge group, because the classification of
semi-simple groups is discrete.  The infinite-dimensional gauge algebra $\g[[z]]$ is, however, susceptible to continuous deformation.
 (Actually, because the anomalous
diagrams have more than one gauge boson attached to the Wilson line, what we will get is  a deformation of the universal enveloping algebra of $\g[[z]]$, not
a deformation of $\g[[z]]$ as a Lie algebra.)

In our theory, because $\g[[z]]$ will be deformed, the gauge invariance of a Wilson line depends on modifying the matrices $t_{a,n}$
 to provide a representation of the deformed algebra rather than
of $\g[[z]]$.  
If this is possible for a given Wilson line, we say that this Wilson line (or the associated representation of $\g[[z]]$) ``quantizes.''   A general Wilson line does not quantize
in this sense.  In section \ref{exam}, we describe an explicit counterexample.

\subsection{Anomalies on Wilson Lines}\label{anomwi}
We will first provide a general analysis of  what anomalies can occur in the coupling of gauge fields to a given Wilson line operator.   We will derive a cohomological interpretation of the anomaly. Combined with a cohomology calculation given in appendix \ref{app.anomaly_Wilson}, this  allows us to prove the following results. 
\begin{theorem}
Let  $V$ be a representation of a simple Lie algebra $\g$ which has no Abelian factors and which is not $\mf{sl}_2$.  
\begin{enumerate}
\item The first possible anomaly to quantizing the Wilson line arises at $2$ loops, that is in order $\hbar^2$.
\item Let $\wedge^2_0 \g$ be the kernel of the bracket map $\wedge^2 \g \to \g$.  If there are no non-trivial $G$-invariant maps  $\wedge^2_0 \g \to \op{End}(V)$, then there are no anomalies.
\item At two loops, the number of possible anomalies is exactly the dimension of the space of $G$-invariant maps  $\wedge^2_0 \g \to \op{End}(V)$.
\end{enumerate}

\end{theorem}
These results imply, for instance, that the vector representations of the classical groups $SL_n$, $SO_n$ and $Sp_{2n}$ all quantize.    

Let us analyze the possible forms of the anomalies to quantizing a Wilson line.   By the anomaly of a Wilson line, we mean its BRST variation.  This will always be
the integral over the Wilson line of a local operator of ghost number 1.  Thus, this operator will be linear in the ghost field $\c$ (and its derivatives) with {\it a priori}
an arbitrary polynomial dependence on the gauge field $A$ (and its derivatives).

As usual, we work on $\R^2 \times C$, with the usual real coordinates $x,y$ and complex coordinate $z$.  We take a Wilson line to be parametrized by $x$, at $y=z=0$.  Such a Wilson line is invariant under the classical symmetries that independently rescale $x$ and $y$.  The anomaly will have the same invariances.
This implies that the anomaly must be linear in $A_x$ and independent of $A_y$, and cannot have any $x$ or $y$ derivatives.

Because the classical symmetry which scales the $x$ and $y$ directions is unbroken by the coupling of a Wilson line, any anomaly\footnote{In general, given a quantum system defined up to order $k$ in $\hbar$, the anomaly to quantizing the system to order $k+1$ in $\hbar$  must respect any symmetries present in the system at order $k$. To see this, we note that if we change the order $k+1$ counter-terms, the anomaly changes by a BRST-exact term. Therefore the cohomology class of the anomaly only depends on the theory up to order $k$, and is preserved by any symmetries present in the system up to order $k$.}  must respect this symmetry.    Since the anomaly must be an integral over the Wilson line, it must involve only the $x$-component $A_x$ of the connection, and can not involve the $y$-component or any $x$ and $y$ derivatives.  Otherwise, it would not respect the scaling symmetries in the $xy$ plane.  For the same reason, the anomaly must be at most a linear function of $A_x$.  

This argument does not exclude the possibility that the anomaly can include some polynomial in $A_{\zbar}$, or its $\zbar$ derivatives.  Let us now see why these can not occur (in section \ref{vanishing} we have performed a similar analysis for the anomaly to the existence of a vertex).  The classical action functional of the theory is invariant under the symmetry which scales $z$ $\zbar$, and $\hbar$  by a real parameter $c$. It is also invariant under the symmetry by which $z$ and $\hbar$ are rotated through an angle $e^{2 \pi i\theta}$,  and where $\zbar$ is therefore rotated by $e^{-2\pi i \theta}$.  Any anomaly must be invariant under these symmetries (where we include the factor of $\hbar^k$ naturally present in a $k$-loop anomaly).  

The most general anomaly at $k$ loops will involve $n_1$ copies of $A_{\zbar}$, $n_2$ $\zbar$ derivatives, and $n_3$ $z$-derivatives.  A copy of $A_{\zbar}$ in the anomaly behaves the same with respect to these symmetries as a $\zbar$-derivative, as it involves contracting the connection $A$ with the vector field $\partial_{\zbar}$.  The symmetries under scaling and rotation in the $z$-plane tell us that
\begin{align}
\begin{split}
n_1 + n_2 + n_3 - k &= 0 \;, \\
-n_1 - n_2 + n_3 - k &= 0 \;.
\end{split}
\end{align}
From this we see that $n_1 = n_2 = 0$ and $n_3 = k$.  Therefore, the anomaly can not involve the $\zbar$-component of the connection or any $\zbar$ derivatives, and the number
of $z$ derivatives in a $k$-loop anomaly must be precisely $k$.

As a further constraint, note that if the anomaly does not depend at all on $A_x$, and so depends only on the ghost field $\c$, then it must be given by an integral of $\partial_x \c$ (in view
of the scaling symmetry in the $x$ direction). We can write  $\partial_x \c=D_x\c-[A_x,\c]$, where $D_x\c$ is the covariant derivative of $\c$. Since insertion of $D_x\c$ in a Wilson operator gives a total derivative that would not contribute, we can replace $\partial_x c$
with $-[A_x,\c]$, and thus it is not necessary to consider terms that are independent of $A_x$.

The constraints we have considered so far tell us that the most general $k$-loop anomaly is of the form
\begin{equation}
\sum_{k_1 + k_2 = k} \Theta_{a, k_1 , b, k_2} \int_{y = z = 0} \partial_z^{k_1} A^a_x  \partial_z^{k_2} \c^b  \d x \;, \label{equation_anomaly} 
\end{equation}
where for each value of $k_1,k_2$,  $\Theta^{a,k_1,b,k_2}$ is a $G$-invariant linear map
$$
\mf{g} \otimes \mf{g} \to \op{End}(V)\;,
$$ 
where $V$ is the representation of $G$ from which we build the Wilson line.

There is one more constraint. Since the anomaly is the BRST variation of the Wilson line, it is itself BRST-invariant.   Applying this constraint leads us to the equation
\begin{align}
\begin{split}
&\sum_{k_1 + k_2 = k} \Theta_{a, k_1 , b, k_2} \int_{y = z = 0} \partial_z^{k_1} (\d \c^a + f^a{}_{cd} \c^c A^d_x) \left( \partial_z^{k_2} \c^b \right) \d x \\ 
&+\sum_{k_1 + k_2 = k} \Theta_{d, k_1 , a, k_2} \int_{y = z = 0} \partial_z^{k_1} A_x^d\partial_z^{k_2}  \left( \frac{1}{2} f^{a}{}_{cb} \c^c \c^b \right) \d x  \\
&\qquad+ \sum_{k_1 + k_2 = k} [\rho_c, \Theta_{d, k_1 , b, k_2}]  \int_{y = z = 0} \c^c\left( \partial_z^{k_1}   A^d_x \right)\left( \partial_z^{k_2} \c^b \right) \d x = 0 \;
\label{equation_anomaly_gauge_invariant} 
\end{split}
\end{align}
In the third term in this equation, the operator $\rho_{a} : V \to V$ indicates the classical action of the operator $t_a$ on $V$.  The appearance of this term was be explained in detail in a similar context in section \ref{vanishing}.    

This equation can be separated into a term which is linear in $A$  and one which has no $A$-dependence.  Both of these terms must vanish.  Let us first analyze the term which is independent of $A$.  If we integrate by parts, and use the fact that the ghost field is an anti-commuting variable (i.e.\ a fermionic field of spin $0$), we find that this equation tells us
\begin{equation}
	\Theta_{a, k_1 , b, k_2}+ \Theta_{b,k_2,a,k_1} = 0 \;. \label{equation_anomaly_antisymmetry} 
\end{equation}

Next, let us analyze the term in eqn. (\ref{equation_anomaly_gauge_invariant}) which is linear in $A$.  We obtain equations by setting the coefficients of $\int \left( \partial_z^{k_0} \c^c\right)\left( \partial_z^{k_1} A^d\right)\left( \partial_z^{k_2} \c^b\right)$ to zero. These equations are 
\be
\begin{split}
\Theta_{a,k_0+k_1,b,k_2} f^a{}_{cd}  \binom{k_0+k_1}{k_0} - \Theta_{a,k_2 + k_1,c,k_0} f^a{}_{bd}\binom{k_1+k_2}{k_1} +\Theta_{d,k_1, a,k_0+k_2}f^{a}{}_{cb} \binom{k_0+k_2}{k_0} 
 \\
       + \delta_{k_0=0}[\rho_c,\Theta_{d,k_1,b,k_2}] - \delta_{k_2=0} [\rho_b,\Theta_{d,k_1,c,k_0}] = 0 \;. \label{equation_anomaly_cocycle} 
\end{split}
\ee
Note that this quantity is  totally antisymmetric under permutation of $(k_0,c),(k_1,d)$ and $(k_2,b)$, as follows from \eqn \eqref{equation_anomaly_antisymmetry}.
 
We can summarize this result as follows. Let us view the collection $\Theta_{a,k_1,b,k_2}$ as a linear map 
\begin{align}
\begin{split}
\Theta : \quad \wedge^2 \g[[z]] \quad &\to \op{End}(V) \;,\\
\vin \qquad \quad &  \quad\quad \vin \\
 (t_a z^{k_1}) \wedge(t_b z^{k_2}) &\mapsto k_1 ! k_2 !\Theta_{a,k_1,b,k_2} \;. \label{equation_cocycle_theta}
\end{split}
\end{align}
Let $C^\ast(\g[[z]],\op{End}(V))$ denote the Chevalley-Eilenberg cochain complex,
which in degree $n$ consists of linear maps $\wedge^n \g[[z]] \to \op{End}(V)$. 

Then,  \eqn \eqref{equation_anomaly_cocycle} is precisely the condition that $\Theta$ is a closed element of this Chevalley-Eilenberg  cochain complex. 

\subsubsection{Cancellation by Counter-Terms.} 
Let us now turn to analyzing when such an anomaly can be cancelled by a counter-term. Suppose that the first anomaly arises at $k$ loops and is given by the expression in \eqn \eqref{equation_anomaly} involving some $\Theta_{a,k_1,b,k_2}$ satisfying \ref{equation_anomaly_antisymmetry}.  

To try to cancel the anomaly, we change the coupling of the gauge field $A$ so that $\partial_z^k A^a$ is coupled by some operator $\rho_{a,k} \in \op{End}(V)$. We will assume that $\rho_{a,k}$ transforms under the adjoint representation of $G$ on $\op{End}(V)$.   Simple Lie algebras and their representations have a discrete classification and cannot
be deformed, so we will assume that $\rho_{a,k}$ vanishes for $k=0$.

The term $ \rho_{a,k}\int \partial_z^k A^a$ that we add to the action may not be gauge invariant at the classical level. The failure to be gauge invariant is given by the expression 
 \begin{equation}
\rho_{a,k} \int \partial_z^k \d \c^a + \rho_{a,k} f_{abc} \int \partial_z^k \left( \c^b A^c\right) +  [\rho_b, \rho_{a,k} ] \int \c^b \partial_z^k A^a \;. 
\end{equation}    
In order to cancel the anomaly at $k$ loops,  we need
\begin{equation}
\sum_{k_1 + k_2 = k} \Theta_{a, k_1 , b, k_2} \int_{y = z = 0} \partial_z^{k_1} A^a_x  \partial_z^{k_2} \c^b  \d x =   \rho_{a,k} f_{abc} \int \partial_z^k \left( \c^b A^c\right) +  [\rho_b, \rho_{a,k} ] \int \c^b \partial_z^k A^a \;. 
\end{equation} 
This can happen if and only if
\begin{equation}
\Theta_{a,k_1,b,k_2} = \rho_{c,k} f^{c}_{ba} \binom{k}{k_1} + [\rho_a, \rho_{b,k}] \;.  
\end{equation}
This is precisely the condition that the cocycle in $C^\ast(\g[[z]], \op{End}(V))$ associated to the anomaly (as in \eqn \eqref{equation_cocycle_theta}) is exact.

In sum, we have found that possible anomalies to coupling to a Wilson line, modulo counter-terms which can cancel these anomalies, are given by the Chevalley-Eilenberg cohomology group
\begin{equation}
H^2(C^\ast(\mf{g}[[z]], \op{End}(V)) \;.
\end{equation} 
 We will repeatedly use the fact that for $G$ a simple Lie group,  the Chevalley-Eilenberg cohomology (of any degree)
 with values in any representation can always be represented by a $G$-invariant cocycle, roughly by averaging
over a maximal compact subgroup of $G$.  (Of course, the expressions that we generate from Feynman diagrams in expanding around the trivial flat connection
will always be $G$-invariant.)  Note that Chevalley-Eilenberg cohomology is also called Lie algebra cohomology, and under that name often appears in analysis
of BRST cohomology in quantum field theory.

One nice consequence of the description of anomalies by Chevalley-Eilenberg cohomology
 is that we can understand what anomalies can exist by applying standard arguments from homological algebra. When translated into field-theory language,
the arguments would be quite complicated.  

For instance, the following is a relatively easy result: 
\begin{lemma} 
	For a Wilson line in a representation $V$ of a simple group $G$ (as opposed to a representation of $\g[[z]]$),  there are no one-loop anomalies which cannot be cancelled by the introduction of a counter-term. 

Further, all two-loop anomalies are equivalent to ones given by expressions of the form 
\begin{equation}
\sum \Theta_{a, 1 , b, 1} \int_{y = z = 0} \partial_z A^a_x  \partial_z \c^b  \d x \;,\label{equation_twoloop_form}
\end{equation}
where, as before, $\Theta_{a,1,b,1}$ is a linear operator on the vector space $V$ and is antisymmetric under the exchange of the $a,b$ indices.  Viewed as a linear map $\wedge^2 \mf{g} \to \op{End}(V)$, $\Theta$ can be assumed to be $G$-invariant.   

	Two-loop anomalies which factor through a map $\wedge^2 \mf{g} \to \mf{g}$ can be cancelled by a counter-term. 
\end{lemma}

\begin{proof}
The proof of the lemma uses the notion of relative Lie algebra cochains of $\g[[z]]$ with respect to the subalgebra $\g$.  Relative Lie algebra cochains are $G$-invariant elements of $C^\ast(z\g[[z]], \op{End}(V))$, that is, $G$-invariant cochains which do not involve $\g$.  The relative cochains are equipped with the same Chevalley-Eilenberg differential as before.  We let $H^\ast((\g[[z]] \, | \, \g), \op{End}(V))$ denote the relative cohomology. General results on the cohomology of simple Lie algebras imply that there is an isomorphism
\begin{equation}
	\oplus_{i + j = k} H^i(\g) \otimes H^j((\g[[z]]\mid \g), \op{End}(V)) \iso H^{k}(\g[[z]], \op{End}(V)) \;.
\end{equation}
	Since $\g$ is simple, $H^\ast(\g)$ is zero in degrees $1$ and $2$, and $H^0(\g) = \C$. Therefore relative and absolute Lie algebra cohomology coincide in degrees $1$ and $2$.  

At weight $1$ under the $\C^\times$ action that scales $z$, the relative cohomology is zero in degrees $2$ and higher. This is simply because we cannot built a bilinear 
of weight $1$ from $z \g[[z]]$.  Therefore there are no anomalies of weight $1$. 

	The cohomology of weight $2$ under this $\C^\times$ action is the part that contributes to the two-loop anomaly. Relative $2$-cocycles of weight 2 must be given by a $G$-invariant bilinear and antisymmetric map $z \g \otimes z \g \to \op{End}(V)$.  Because each tensor factor involves one $z$, the corresponding two-loop anomaly  
is of the form given in \eqn \eqref{equation_twoloop_form} (with one $z$ derivative of $A$ and one of $\c$). So every two-loop anomaly is of this form.

Concerning the last statement in the lemma, we observe first that for simple $G$, any $G$-invariant map $\wedge^2\g \to \g$ is actually a multiple of the commutator map.
We do not know a unified proof of this statement, but for classical groups it follows from invariant theory,\footnote{That is, one explicitly considers an element of $\g$ as
a two-index tensor, and one considers the possible contractions of indices to make a $G$-invariant map from $\wedge^2\g$ to $\g$.} and for exceptional groups, it can
be verified by consulting tables (such as those in \cite{McKay}).
Given this, we need to show that if a two-cocycle
\begin{equation} 
\wedge^2 (z \g) \to \op{End}(V)
\end{equation}
factors through the commutator map $\wedge^2 \g \to \g$, then it is the coboundary of a $G$-invariant map $z^2 \g \to \op{End}(V)$.  But the coboundary of a $G$-invariant
map $z^2\g\to \op{End}(V)$ is given by the composition
	\begin{equation} 
		\wedge^2 (z \g) \to z^2 \g \to \op{End}(V) ,
	\end{equation}
	where the first map is the commutator.
	Therefore any two-cocycle that factors through the commutator is a coboundary.
\end{proof}

This lemma tells us that find that the first possible anomalies can arise at two loops, from diagrams with two external gauge fields.  

In appendix \ref{app.anomaly_Wilson}, we prove the following result (relying heavily on the results of \cite{FGT}). This result generalizes the previous lemma to all loops.
The proof is much more difficult.
\begin{proposition}\label{proposition_cohomology}
Let $\mf{g}$ be a simple Lie algebra which is not\footnote{The case of $\mf{sl}_2$ is slightly different. There we find that $H^2(\g[[z]], \g \otimes \g)$ is isomorphic to $H^2(\g[[z]], \Sym^2 \g)$, which is one-dimensional and in weight $3$ under the $\C^\times$ action scaling $z$. Since it is known by other methods that every representation of $\mf{sl}_2$ quantizes, we will avoid the special case of $\mf{sl}_2$ in what follows.} $\mf{sl}_2$. 
Let $\wedge^2_0 \mf{g}$ be the kernel of the Lie bracket map from $\wedge^2 \mf{g}$ to $\g$. (Note that there are no copies of the adjoint  in the representation $\wedge^2_0 \mf{g}$). Then the inclusion map
\begin{equation}
H^2(\g[[z]], \wedge^2_0 \g) \to H^2(\g[[z], \g \otimes \g)  
\end{equation}
is an isomorphism.
\end{proposition}
In other words, if we decompose $\g \otimes \g$ as $\Sym^2 \g \oplus \g \oplus \wedge^2_0 \g$, then all of the cohomology we are interested in comes from $\wedge^2_0 \g$.

We are interested in $H^2(\g[[z]], \op{End}(V))$ for a representation $V$.  In general, $\op{End}(V)$ is a complicated, highly
reducible representation of $G$. Chevalley-Eilenberg cohomology is defined for every representation,
and  if $\op{End}(V)$ is a direct sum of irreducibles, then the desired cohomology is the direct
sum of cohomology with values in these irreducibles.
  
However, the desired cohomology can be represented by $G$-invariant cocycles and such a cocycle, ignoring the $z$-dependence,
is a $G$-invariant map from $\g\otimes \g$ to $\op{End}(V)$.
Hence, any two-cocycle valued in $\op{End}(V)$ must come from a two-cocycle valued in $\mf{g} \otimes \mf{g}$ under some map of $G$-representations from $\mf{g} \otimes \g$ to $\op{End}(V)$.  The same is true for the corresponding second cohomology classes.  The proposition gives an important refinement of this statement:
any non-trivial second cohomology class must come from a map $\wedge^2_0 \g \to \op{End}(V)$.   This means that when we study Feynman diagrams, we can ignore
terms in which two gauge bosons $A^a$ and $A^b$ are coupled symmetrically in $a$ and $b$, and we can also ignore terms in which they are coupled via $f_{ab}^c$.

As a corollary, we find that if there are no maps of representations  $\wedge^2_0 \g \to \op{End}(V)$, then the Wilson line living in $V$ can not have an anomaly.  This is the criterion 
$(\dagger)$ that was introduced in section \ref{netov} and used in various applications.

\subsection{Enumerating Feynman diagrams}
So far, we have performed a cohomological analysis of possible anomalies that can appear when one tries to quantize a Wilson line.  We have seen that the first such anomalies can appear at two loops.  Next, we will enumerate all possible two-loop diagrams that may contribute an anomaly. We will find that, for various reasons, all but one of the diagrams we enumerate  can not produce an anomaly. 

Throughout our analysis, we will assume our Wilson line is placed at $z = 0$ and is invariant under the symmetry which rescales $z$ and $\hbar$.  This means that $t_{a, k}$ is accompanied by $\hbar^k$, and in particular at the classical level  $t_{a,k} = 0$ for $k > 0$.

Let us now analyze possible two-loop diagrams which can contribute anomalies of this form.  The only diagrams that are relevant are two-loop diagrams which may be attached to the Wilson line at an arbitrary number of points, but which have precisely two external lines to which the gauge field is coupled.

We must also bear in mind an additional subtlety.  In principle one can introduce one-loop counter-terms whereby $\partial_z A$ is coupled to the Wilson line in some way.   It could happen that these are forced on us to cancel any one-loop anomalies,  but even if this is not the case, such one-loop counter-terms do not violate any symmetries and so one is always free to introduce them.   Therefore there are two possible vertices at which a gluon can meet a Wilson line. There is the usual classical interaction, but also the interaction coming from a one-loop counter-term. We will depict these interactions as an ordinary vertex, and a circle labelled $1$, respectively.  Counting of the loop parameter tells us that one-loop diagrams, one of whose vertices on the Wilson line is the one-loop vertex, are counted as two-loop diagrams. 

One can cut down the number of diagrams we need to consider by observing that every diagram must have at least two vertices on the Wilson line. Diagrams with one vertex on the Wilson line can not contribute, because the corresponding anomaly will be given by a two-cocycle valued in some copy of $\mf{g} \subset \op{End}(V)$.  Our cohomological analysis in proposition \ref{proposition_cohomology} tells us that any anomaly of this form can be cancelled by a counter-term.  

In evaluating whether a diagram can contribute to the anomaly, note that the anomaly must have one $z$-derivative at each external line, and no other derivatives.  We can thus calculate whether or not a diagram contributes to the anomaly by taking the two external fields to be the gauge field $A = z \delta_{x = 0} t_a$, and the ghost field $\c = z t_b$.  This tells us, in particular, that it is not possible to have an anomaly from a diagram in which the external lines are connected directly to the Wilson line by a tree-level vertex.  This is because the Wilson line is at $z = 0$, and there are no $z$-derivatives in how the Wilson line is coupled classically.   

A further constraint on the diagrams that can appear is provided by the observation that if we evaluate the propagator with both ends on a straight Wilson line we find zero. Therefore in any diagram which can contribute to the anomaly this configuration can not occur. 

These observations cut down substantially the diagrams that can appear.  We will further exclude a few simple diagrams which are topologically trees, but whose coupling to the Wilson line involves one and two-loop counter-terms.  These diagrams are similar to those in \fig \ref{tred}.  In section \ref{anomaly_correction} we will analyze the effect of these diagrams, and see that they contribute to the two-loop correction to the algebra $\g[[z]]$.  

All remaining diagrams are depicted in \fig \ref{figure_anomaly}. Each figure admits a number of variants where the order of the vertices on the Wilson line is permuted.  These variants are not depicted. 

\begin{figure}
\begin{tikzpicture}
\begin{scope}[ shift={(-1.5,0)}, scale=(0.7)]
\draw[dashed](0,2.7) to (0,-2.7);
\draw (0,-2.2) arc (-90:90:2.2);
\draw(2.2,0) to (0,0); 
\draw (40:2.2) to (40:3.2);
\draw (-40:2.2) to (-40:3.2);
\node at (1,-3.2) {(A1)};
\end{scope}

\begin{scope}[shift={(2.5,0)}, scale=(0.7) ]
\draw[dashed](0,2.7) to (0,-2.7);
\draw (0,-2.2) arc (-90:90:2.2);
\draw(2.2,0) to (0,0); 
\draw (30:2.2) to (30:3.2);
\draw (60:2.2) to (60:3.2);
\node at (1,-3.2) {(A2)};
\end{scope}

\begin{scope}[shift={(6.5,0)}, scale=(0.7)]
\draw[dashed] (0,2.7) to (0,-2.7);
\draw (0,1) to [out=0,in=100](1,0.5) to [out=260,in=0](0,0);
\draw (1,0.5) to (1.5,0.5);
\draw (0,0.5) to [out=180,in=80](-1,0) to [out=260,in=360](0,-0.5);
\draw (-1,0) to (-1.5,0);
\node at (1,-3.2) {(B)};
\end{scope}

\begin{scope}[shift={(10.5,0)},scale=(0.7)]
\draw[dashed] (-2,2.7) to (-2,-2.7);
\draw (0,0) circle (1cm);
\draw (130:1) to [out=130,in=0] (-2,1.3);
\draw (230:1) to [out=230, in=0] (-2,-1.3);
\draw (30:1) to (30:2);
\draw (-30:1) to (-30:2);
\node at (-1,-3.2) {(C1)};
\end{scope}

\begin{scope}[shift={(0,-6)},scale=(0.7)]
\draw[dashed] (-2,2.7) to (-2,-2.7);
\draw (0,0) circle (1cm);
\draw (130:1) to [out=130,in=0] (-2,1.3);
\draw (230:1) to [out=230, in=0] (-2,-1.3);
\draw (-1.1, 1.15) to (-0.8, 2.0);
\draw (-1.6, 1.27) to (-1.55, 2.13);
\node at (-1,-3.2) {(C2)};
\end{scope}

\begin{scope}[shift={(3.5,-6)},scale=(0.7)]
\draw[dashed] (-2,2.7) to (-2,-2.7);
\draw (0,0) circle (1cm);
\draw (130:1) to [out=130,in=0] (-2,1.3);
\draw (230:1) to [out=230, in=0] (-2,-1.3);
\draw (0:1) to (0:2);
\draw (180:1) to (180:0);
\node at (-1,-3.2) {(C3)};
\end{scope}

\begin{scope}[shift={(7,-6)},scale=(0.7)]
\draw[dashed] (-2,2.7) to (-2,-2.7);
\draw (0,0) circle (1cm);
\draw (130:1) to [out=130,in=0] (-2,1.3);
\draw (230:1) to [out=230, in=0] (-2,-1.3);
\draw (-1.5, 1.25) to (-1.4, 2.1);
\draw (-1.5, -1.25) to (-1.4, -2.1);
\node at (-1,-3.2) {(C4)};
\end{scope}

\begin{scope}[shift={(10.5,-6)},scale=(0.7)]
\draw[dashed] (-2,2.7) to (-2,-2.7);
\draw (0,0) circle (1cm);
\draw (130:1) to [out=130,in=0] (-2,1.3);
\draw (230:1) to [out=230, in=0] (-2,-1.3);
\draw (-1.5, 1.25) to (-1.4, 2.1);
\draw (1,0) to (2,0); 
\node at (-1,-3.2) {(C5)};
\end{scope}

\begin{scope}[shift={(2,-12)}, scale=(0.7)]
\draw[dashed] (0,2.7) to (0,-2.7); 
\draw (0,0.8) to [out=0,in=100](1,0) to [out=260,in=0](0,-0.8);
\draw (1,0) to (1.5,0);
\draw (0,0) to (-1,0);
\draw[ fill  = white] (0,0) circle (0.3); 
\node at (0,0) {$1$};
\node at (1,-3.2) {(D1)};
\end{scope}

\begin{scope}[ shift={(6,-12)}, scale=(0.7)]
\draw[dashed](0,2.9) to (0,-2.9);
\draw (0,-2.2) arc (-90:90:2.2);
\draw (40:2.2) to (40:3.2);
\draw (-40:2.2) to (-40:3.2);
\draw[ fill  = white] (0,-2.2) circle (0.3); 
\node at (0,-2.2) {$1$};
\node at (1,-3.2) {(D2)};
\end{scope}

\end{tikzpicture}

\caption{\small {Two-loop Feynman diagrams with two external lines which can potentially contribute to the anomaly. The dashed vertical line represents a Wilson line, and the vertices on the Wilson line labelled $1$ indicate a coupling of the gauge field and Wilson line by a one-loop counter-term.  Each diagram admits a number of variants, not shown, in which the positions of the vertices on the Wilson line are permuted.}} \label{figure_anomaly}
\end{figure}
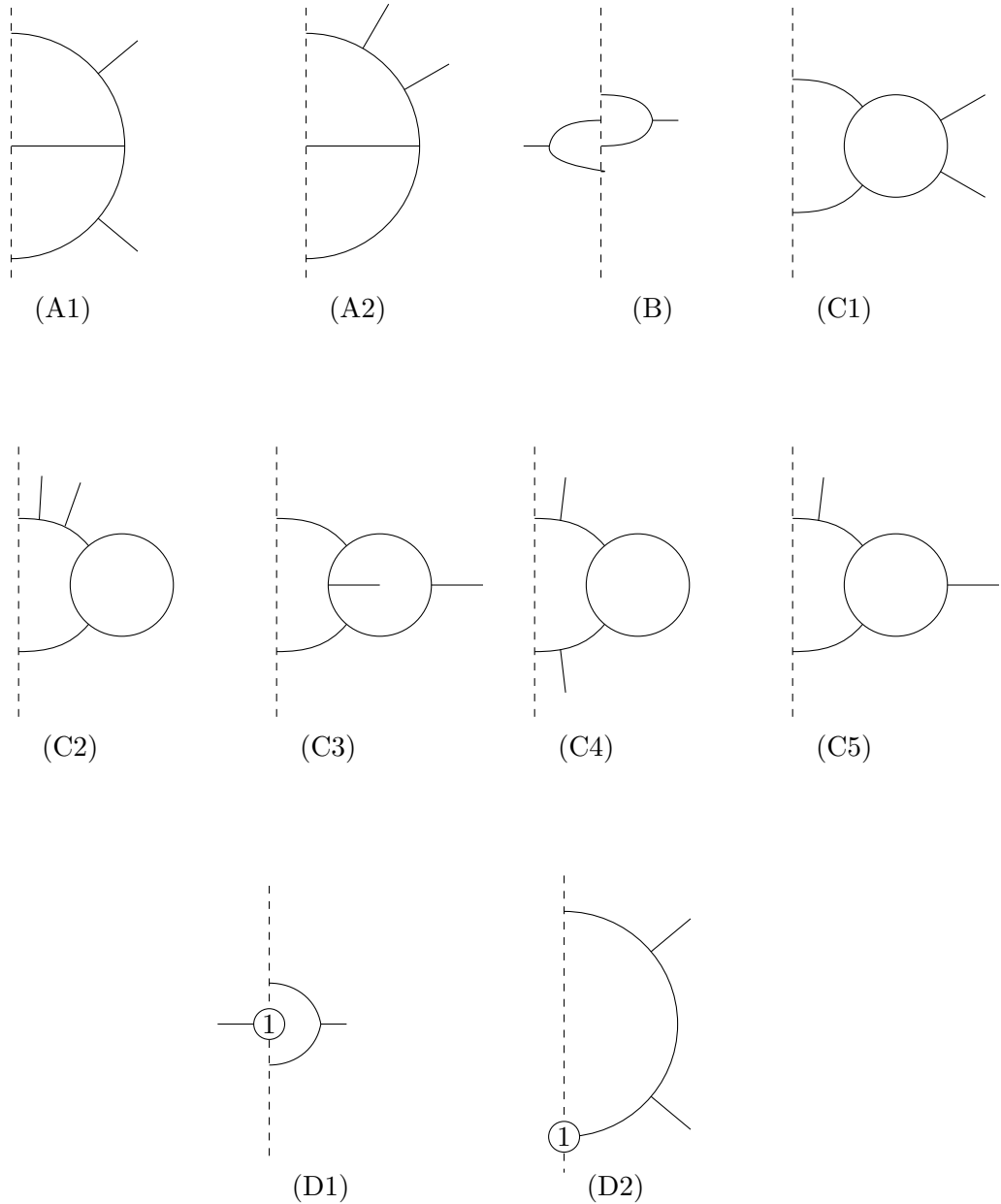

We will find that the only diagram that can contribute is diagram (A1).  All other diagrams will be excluded by a variety of arguments.


\subsubsection{Excluding Diagrams (A2), (C1), (C2) and (D2)}  \label{first case}
Consider any two-loop diagram with the feature that the two vertices attached to the external lines are connected by a propagator.  The only non-trivial anomalies are antisymmetric in the external lines.  After antisymmetrizing and using the Jacobi identity, the anomaly from such a diagram is of the form
\begin{equation}
\int f^{cab} \partial_z \c_a \partial_z A_b \mu_c
\end{equation}
for some matrices $\mu_c : V \to V$.  Any anomaly of this form can be cancelled, because it comes from a map $\wedge^2 \mf{g} \to \op{End}(V)$ which factors through the adjoint representation.   

This tells us that diagrams (A2), (C1), (C2) and (D2) can not contribute a non-trivial anomaly.

\subsubsection{Excluding Diagrams of Type (C3)}
Diagram of type (C3) can be excluded for Lie algebraic reasons.   Any anomaly that can appear from a diagram of type (C3) is of the form:
\begin{equation}
 \op{Tr}_{\mf{g}} \left(t_a t_c t_b t_d  \right) \int \partial_z \c_a  \partial_z A_b \rho(t_c) \rho(t_d) \;. 
\end{equation}
(Here the trace is taken in the adjoint representation, and the external lines have indices $a,b$). 

This anomaly can be cancelled  by a two-loop counter-term if the tensor $\op{Tr}_{\mf{g}} \left(t_a t_c t_b t_d  \right)$ is symmetric in $a$ and $b$, or if it can be written in terms of the commutator $f^e{}_{ab} t_e$ of $t_a$ and $t_b$.   Since the adjoint representation is equipped with a symmetric non-degenerate pairing, we have the identity
\begin{equation}
\op{Tr}_{\mf{g}} (t_a t_c t_b t_d) = \op{Tr}_{\mf{g}} ( t_d t_b t_c t_a ) \;.
\end{equation}   
Together with cyclic symmetry of the trace, this tells us that $\op{Tr}(t_a t_c t_b t_d)$ is symmetric in $a$ and $b$. This implies that any anomaly associated to the tensor $\op{Tr}_{\mf{g}} (t_a t_c t_b t_d)$ can be cancelled by a counter-term.

\subsubsection{Excluding Diagrams (C4) and (C5)}\label{excluding_C} 

Diagrams (C4) and (C5) contain an internal gauge particle loop.  In (C4), this loop has two external lines, both of which are internal lines in the full diagram,
while (C5) has an internal loop with three external lines, of which in the full diagram two are internal and one is external.

The group theory factor that comes from the internal loop in (C4) is an invariant in $\g\otimes \g$.  Any such invariant is a multiple of  $\delta_{cd}$, that is, 
the Killing form, and accordingly, the overall group theory factor of the diagram in (C4) is a multiple of what it would be if the internal loop were collapsed to a point.  Then in (C4) the
two external gauge fields would be connected by a propagator.  So this diagram can be excluded by the same argument as in section \ref{first case}.

In (C5), since the internal loop has three external lines, it produces an invariant in $\g\otimes\g\otimes \g$.  For any simple Lie algebra other than $SL_N$, $N>2$, such
an invariant actually lies in\footnote{Invariants in $\g\otimes\g\otimes\g$ correspond to $G$-invariant maps $\g\otimes \g\to \g$.  That there is only one such invariant
for most $G$ and two for $SL_N$, $N>2$, can be shown by invariant theory for classical $G$ and by consulting tables such as those in \cite{McKay} for exceptional $G$.
In $SL_N$, $N>2$, the two invariants are $\Tr\,A[B,C]$ and $\Tr\,A\{B,C\}$ and are respectively even and odd under the outer automorphism.  The fact that an invariant in
$\mathrm{Sym}^3\g$ exists only for $SL_N$, $N>2$ is important in particle physics in classifying possible anomalies in gauge theory.} $\wedge^3\g$ and is a multiple of the invariant associated to the structure constants $f_{abc}$.  When this is the case, the group theory factor of (C5)
is the same as if the internal loop were collapsed to a point.  When this is done, the two external gauge fields are connected by a propagator and this diagram can
be excluded as before.

For $SL_N$, $N>2$,  there actually is another invariant in $\g\otimes\g\otimes \g$; in fact, it lies in $\mathrm{Sym}^3\g$.  However, this invariant is odd under the outer automorphism
of $SL_N$, which is a symmetry of the theory under discussion, and therefore cannot appear.


\subsubsection{Excluding Diagrams (B) and (D1)}\label{excluding_amplitude}
Diagrams (B) and (D1) cannot be excluded by Lie algebraic arguments.  Instead we will show by an explicit calculation that the amplitudes for such diagrams are zero, so they can not contribute an anomaly.   The computation involves only the part of the diagram on the right in the figures; this is the part that (B) and (D1) have in common.

To evaluate the amplitudes for these diagrams, we will employ a point-splitting regularization in which any pair of vertices on the Wilson lines must be separated by a distance of at least $\eps$.  Once we employ such a point-splitting regularization, we can compute the regularized amplitude for each connected component of the  diagrams (B) and (D1) separately. The two connected components only interact with each other by how they affect the domain of integration of the vertices on the Wilson lines.  We will get a vanishing that does
not depend on this domain.

We therefore need to compute the amplitude of the diagram
\begin{center}
\begin{tikzpicture}
\draw[dashed] (0,2.7) to (0,0.3);
\draw (0,2) to [out=0,in=100](1,1.5) to [out=260,in=0](0,1);
\draw (1,1.5) to (1.5,1.5);
\node at (-0.7,2) {$v_1$};
\node at (-0.7,1) {$v_2$};
\node at (1.2,1.7) {$v_3$};
\end{tikzpicture}
\end{center}
 
The coordinates for the vertex $v_i$ will be $(x_i,y_i,z_i)$. For vertices $v_1,v_2$ we have $y_i = z_i = 0$.  

Recall  (see eqn.\ \eqref{eq.propagator}) that the propagator two-form is 
\begin{align}
\begin{split}
	P &= 
	\frac{1}{2 \pi} 
	\frac{x \d y\wedge \d \zbar -  y \d x \wedge\d \zbar + 2 \zbar \d x \wedge\d y }{   \left(x^2 + y^2 + \abs{z}^2\right)^{2} }\;.
\end{split}
\end{align}
The amplitude for the diagram is then
\begin{multline}
	c
	 \int_{x_1,x_2,x_3,y_3,z_3}
 A  \d z_3\wedge \frac{ (x_3 - x_1) \d y_3 \d \zbar_3 -  y_3 \d (x_3 - x_1) \d \zbar_3 + 2 \zbar_3 \d (x_3 - x_1) \d y_3   }{\norm{v_3}{v_1}^{4}}\\ 
\wedge \frac{ (x_3 - x_2) \d y_3 \d \zbar_3 -  y_3 \d (x_3 - x_2) \d \zbar_3 + 2 \zbar_3 \d (x_3 - x_2) \d y_3  }{ \norm{v_3}{v_2}^{4}}  \;,
\end{multline}
where $\norm{v_i }{v_j}$ indicates the distance between the vertices $v_i$ and $v_j$, and $A$ is the gauge field on the external line.  

Since we are integrating over $x_1,x_2$, we need only keep terms which involve $\d x_1$ and $\d x_2$.  If we retain only those terms, we find that the integrand is $\d x_1 \wedge d x_2$ times the square of the one-form
\begin{equation} 
y_3  \d \zbar_3 - 2 \zbar_3  \d y_3 
 \end{equation}
and is therefore zero. 

\subsection{The Anomaly Associated to Diagram (A1).}

Let us now come to the only remaining diagram, namely (A1) in \fig \ref{figure_anomaly}.   We need to note that we first need to include two other diagrams, as depicted in Fig.\ref{figure_anomaly2}, where we permute the three vertices along the Wilson lines.  

\begin{figure}[htbp]
\begin{center}
\begin{tikzpicture}
\begin{scope}[ shift={(0,0)}, scale=(0.7)]
\draw[dashed](0,2.7) to (0,-2.7);

\draw (0,-2.2) arc (-90:90:2.2);
\draw(2.2,0) to (0,0); 
\draw (40:2.2) to (40:3.2);
\draw (-40:2.2) to (-40:3.2);
\node at (0.3,0.3) {$d$};
\node at (0.3,2.5) {$c$};
\node at (0.3,-1.9) {$e$};
\node at (2.3,0.8) {$f$};
\node at (2.3,-0.8) {$g$};
\node at (2.5,2.3) {$a$};
\node at (2.5,-2.3) {$b$};

\end{scope}

\begin{scope}[ shift={(4,0)}, scale=(0.7)]
\draw[dashed](0,2.7) to (0,-2.7);

\draw (0,-2.2) arc (-90:90:2.2);
\draw(2.2,0) to (0,0); 
\draw (40:2.2) to (40:3.2);
\draw (1,0) to (1,1);

\end{scope}

\begin{scope}[ shift={(8,0)}, scale=(0.7)]
\draw[dashed](0,2.7) to (0,-2.7);
\draw (0,-2.2) arc (-90:90:2.2);
\draw(2.2,0) to (0,0); 

\draw (1,0) to (1,1);
\draw (-40:2.2) to (-40:3.2);

\end{scope}

\end{tikzpicture}

\caption{\small{Permutations of diagram (A1)}}
\label{figure_anomaly2}
\end{center}
\end{figure}
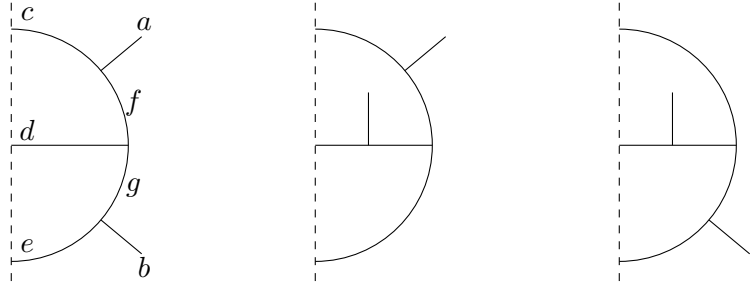

\subsubsection{Color Factor}

It turns out that all these possibilities have the same color factor,
up to those contributions which can be canceled by counterterms.

To see this, let us assign color indices $a, b$ to external lines (ghost field and the gauge field) and 
$c, d, e$ to points on Wilson lines. In the first diagram in Fig. \ref{figure_anomaly2} we have $c,d,e$ along the Wilson line from top to bottom, 
but in the second we have $c,e,d$ and third $d,c,e$.

Let us consider the difference of color factors between first and the second diagram.
Since this involves the change in the relative position of $d$ and $e$,
the difference gives the commutator $[\rho(t^d), \rho(t^e)]=f^{de}{}_{f} \rho(t^f)$.
Hence the difference in color factors can be represented graphically as the diagram 
on the left of Fig. \ref{figure_anomaly3}. Similarly, the difference of the color factor between
the second the third diagram gives the second diagram of Fig. \ref{figure_anomaly3}.

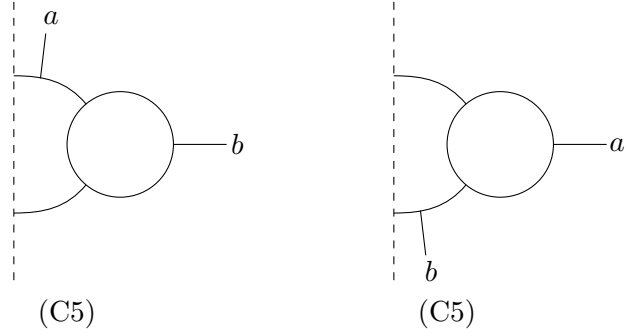
\begin{figure}[htbp]
\begin{center}
\begin{tikzpicture}

\begin{scope}[scale=0.7]
\draw[dashed] (-2,2.7) to (-2,-2.7);
\node at (2.2,0.0){$b$};
\node at (-1.3,2.4) {$a$};
\draw (0,0) circle (1cm);
\draw (130:1) to [out=130,in=0] (-2,1.3);
\draw (230:1) to [out=230, in=0] (-2,-1.3);
\draw (-1.5, 1.25) to (-1.4, 2.1);
\draw (1,0) to (2,0); 
\node at (-1,-3.2) {(C5)};
\end{scope}

\begin{scope}[shift={(5,0)},scale=(0.7)]
\draw[dashed] (-2,2.7) to (-2,-2.7);
\node at (2.2,0.0){$a$};
\node at (-1.3,-2.4) {$b$};
\draw (0,0) circle (1cm);
\draw (130:1) to [out=130,in=0] (-2,1.3);
\draw (230:1) to [out=230, in=0] (-2,-1.3);
\draw (-1.5, -1.25) to (-1.4, -2.1);
\draw (1,0) to (2,0); 
\node at (-1,-3.2) {(C5)};
\end{scope}

\end{tikzpicture}
	
\caption{\small{The differences of the color factors between first and second, and 
first and third diagrams in Fig. \ref{figure_anomaly2} have the color structures 
as shown in this figure.}}
\label{figure_anomaly3}
\end{center}
\end{figure}
We verified above (section \ref{excluding_C}) that the color factor for diagrams of this type is such that any anomaly can always be cancelled by a counter-term.

This tells us that for each of the three diagrams in Fig. \ref{figure_anomaly2}, the anomaly is symmetric under permutation of the order of the vertices connected to the Wilson line, up to anomalies which can be cancelled by a counter-term.  Since the three diagrams in Fig. \ref{figure_anomaly2} are exchanged by such a permutation, we conclude that they all have the same color factor (again, modulo terms which can be cancelled by a counter-term).  

The non-trivial  part of the color factor is (recall the normalization of the Killing form in \eqn \eqref{Killing_normalization})
\begin{align}
 f^{acf} f^{fgd} f^{geb}\{ \rho(t^c), \rho(t^d), \rho(t^e)\}
 =( [[t^a, t^c],t^d],[t^b,t^e] )\{\rho(t^c), \rho(t^d) ,\rho(t^e)\} \;,
 \label{color_factor}
\end{align}
where $a,b$ are indices for the external lines and
\begin{equation} 
\label{triple_def}
	  \{ \rho(t^{a_1}), \rho(t^{a_2}), \rho(t^{a_3})\} = \tfrac{1}{3!}\sum_{\sigma \in S_3} \rho(t^{a_{\sigma(1)}}) \rho(t^{a_{\sigma(2)}}) \rho(t^{a_{\sigma(3)}})\;,
 \end{equation}
 where the sum is over the permutations of the indices $1,2,3$. 

\subsection{Numerical Factor}

Let us next initiate the evaluation of the numerical factor.
We choose the Wilson line to be placed along a straight line 
at $y=z=0$.

We let $p_1,p_2,p_3$ denote the position of the vertices on the Wilson line, and $v_1,v_2,v_3 \in \R \times \C$ be the positions of the internal vertices. The coordinates for vertex $v_i$ are denoted $x_i,y_i,z_i$ for $i = 1,2,3$. The labeling on the vertices for the first diagram in Fig. \ref{figure_anomaly2} is given in  Fig. \ref{figure_diagram_labels}.

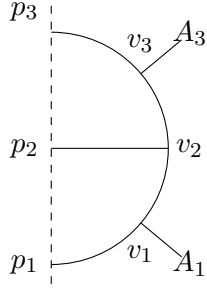
\begin{figure}
\begin{center}
\begin{tikzpicture}
\begin{scope}[ shift={(0,0)}, scale=(0.7)]
\draw[dashed](0,2.7) to (0,-2.7);

\draw (0,-2.2) arc (-90:90:2.2);
\draw(2.2,0) to (0,0); 

\draw (40:2.2) to (40:3.2);
\draw (-40:2.2) to (-40:3.2);
\node at (-0.5,-2.2) {$p_1$};
\node at (-0.5,0) {$p_2$};
\node at (-0.5,2.6) {$p_3$};
\node at (50:2.6) {$v_3$};
\node at (0:2.6) {$v_2$};
\node at (-50:2.6) {$v_1$};
\node at (40:3.4) {$A_3$};
\node at (-40:3.4) {$A_1$};
\end{scope}
\end{tikzpicture}
\caption{\small{The Feynman diagram with vertices labelled. The $A_i$ on the external lines indicate gauge fields.}}
\label{figure_diagram_labels}
\end{center}
\end{figure}

Let us write down the numerical factor of the amplitude, where we include all possible orderings of the vertices on the Wilson lines.  We let $P$ be the propagator viewed as a two-form on $\R^2 \times \C$ (with singularities at the origin).  The space over which we integrate consists of a copy of $\R^2 \times \C$ for each internal vertex, and $\R$ for each vertex on the Wilson line. Every vertex gives a map from this space to $\R^2 \times \C$, where we project on the relevant factor for an internal vertex and project to $\R$ and then include on the $x$-axis for the vertices on the Wilson lines.  For a pair of vertices $(v_i,v_j)$ or $(v_i,p_j)$ we define two-forms by 
\begin{align}
P(v_i,v_j) = (v_i - v_j)^\ast P \;, \quad
P(v_i,p_j) = (v_i - p_j)^\ast P\;,
\end{align}
where we pull-back the two-form $P$ on $\R^2 \times \C$ via the map $v_i - v_j$ or $v_i - p_j$.

The complete amplitude, including all the relevant diagrams and the color factor,  is then
\begin{multline}\label{complete_amplitude}
\tfrac{1}{2} \left(\tfrac{\ii}{2\pi}\right)^3 \int_{p_1 , p_2 , p_3 \in \R}\int_{v_1, v_2, v_3} 
( P(0, v_1)) \d z_1 A^a_1  P(v_1, v_2) \d z_2  P(v_2, p_2) P(v_2, v_3) \d z_3 A^b_3  P(v_3, p_3)  \\  \times f^{afc} f^{fgd} f^{gbe}\{ \rho(t^c), \rho(t^d)\;, \rho(t^e)\} \;,
\end{multline}
where $A_2,A_3$ indicate the external gauge fields.  The prefactor of $\tfrac{1}{2}$ comes from the diagram automorphism.

\subsubsection{A comment on signs}
In the numerical factor, we have implicitly oriented each edge of the graph. If we reverse the orientation of the graph, the propagator changes sign, because $P(v_i,v_j) = - P(v_j,v_i)$.  Similarly, in defining the integral, we have implicitly ordered the set of $6$ vertices (three on the Wilson line and three integrated over the four-dimensional space-time).  If we permute the order of these vertices we change the integral by a sign.  For the vertices on the Wilson line, the sign arises because a permutation can reverse the orientation. For the bulk vertices, since each is accompanied by the one-form $\d z$, permuting them will change the order in which we take the wedge product of these forms, and therefore introduce a sign. 

These signs are also implicit in the definition of the color factor, so there is no overall sign ambiguity.  One way to see this is to treat the elements of the Lie algebra as being anti-commuting variables.  This is a natural thing to do, because the form degree of the forms we are integrating differs by one from the ghost number of the corresponding fields of the physical theory.  For example, a one-form is of odd degree but is a field of ghost number zero.   We can force Lie algebra valued forms to have the correct parity if we declare that the Lie algebra factor is fermionic.

If we treat the Lie algebra elements as fermionic, then the color factor $\delta_{ab} t^a_{v_i} \otimes t^b_{v_j}$ associated to the propagator $P(v_i,v_j)$ changes sign if we reverse the order of an edge, just like the $P(v_i,v_j)$ does. The color factor associated to a vertex is a graded symmetric function of the three Lie algebra elements. Because it is a fermionic function -- since it depends on three fermionic Lie algebra elements -- the overall color factor depends via a sign on the ordering of the set of  vertices.  Similarly, the color factor associated to a vertex on the Wilson line depends on a single Lie algebra element, together with an endomorphism of the bosonic vector space of states on the Wilson line. Therefore the color factor depends on the ordering of the set of vertices on the Wilson line, just as the numerical factor does.

To get the correct signs in the total amplitude, the rule is to choose an orientation of each edge on the graph and an ordering on the set of vertices. Then, compute both the color factor and the numerical factor using this choice. Their product is the amplitude.  Because the color factor and the numerical factor change signs in the same way if we change the ordering of the vertices and the orientation of the edges, there is no ambiguity. The color factor in \eqn \eqref{complete_amplitude} was computed in this way.

\subsubsection{Regularizing the integral}

We will use a point-splitting regulator in which we restrict the domain of integration to the region where $p_3 - p_1 \ge \eps$, if we assume that the three points are ordered so that $p_1 < p_2 < p_3$.   To verify that this is a good regulator, we need to show that the integral converges absolutely in this domain.   We are only concerned with UV divergences, so we will further restrict the domain of integration to a region where all vertices are in some ball around the origin.

To verify convergence, we will use the following bound in absolute value of the propagator:
\begin{equation} 
\abs{P} \le (x^2 + y^2 + z \zbar)^{-3/2} \;. 
\end{equation}
This bound arises because the denominator in $P$ is $(x^2 + y^2 + \abs{z}^2)^{-2}$, while the numerator is a linear function of the variables $x,y,\zbar$.  Any linear function on $\R^4$ is bounded in absolute value by some multiple of $(x^2 + y^2 + \abs{z}^2)^{1/2}$. 
It suffices to verify that the integral converges when each propagator $P(v_i,v_j)$ or $P(p_i,v_j)$ is replaced by $d(v_i,v_j)^{-3}$ or $d(p_i,v_j)^{-3}$, where $d$ is the Euclidean distance.  

We can bound the external gauge fields by a constant, so we can drop them from the integral.  We are thus left with the integral
\begin{equation} 
	\int_{p_2,p_3 > \eps,v_1,v_2,v_3}\frac{1}{ d(p_3,v_3)^{3} d(v_2,v_3)^{3} d(v_2,p_2)^{3} d(v_2,v_1)^{3} d(v_1,p_1=0)^{3}}  \;,
\end{equation}
where by overall translation invariance we set $p_1 = 0$. In our domain of integration, the most divergent region is when $v_1,v_2,v_3,p_2$ are all near $0$ (or all near $p_3$).  Focusing on this region is equivalent to taking $p_3$ to be far away from the other points, which allows us to drop the $d(p_3,v_3)^{-3}$ term since it is non-singular.  We are then considering the integral
\begin{equation} 
	\int_{p_2,v_1,v_2,v_3} \frac{1}{d(0,v_1)^{3} d(v_1,v_2)^{3} d(v_2,p_1)^{3} d(v_2,v_3)^{3}}   \d^4 v_1 \d^4 v_2 \d^4 v_3 \d p_2 \;.  
\end{equation}
There are $13$ integration variables, and the integrand has weight $-12$ under scaling all the variables.  Therefore the integral converges absolutely on a domain where the integration variables $v_i, p_2$ are bounded from above. 

\subsubsection{The Failure of the Amplitude to be Gauge Invariant}
If we include all three diagrams that contribute to the anomaly, we find that the regularized amplitude is 
\begin{align}
\begin{split}
\tfrac{1}{2} \left(\tfrac{\ii}{2\pi}\right)^3
\int_{\substack{p_i,v_j \\ p_{\rm min} < p_{\rm max} - \eps}} P(0, v_1) \d z_1 A^a_1  P(v_1, v_2) \d z_2  P(v_2, p_2) P(v_2, v_3) \d z_3 A^b_3  P(v_3, p_3)  \\
\times f^{afc} f^{fgd} f^{gbe}\{ \rho(t^c), \rho(t^d), \rho(t^e)\} \;,
\end{split}
\end{align}
where $p_{\rm min}, p_{\rm max}$ are the coordinates on the Wilson line with the minimum and maximum value. The factor of $\tfrac{1}{2}$ appears because there are $6$ ways of ordering three points on the line, but only $3$ diagrams we need to consider.   

Let us now investigate the failure of the integral to be gauge invariant.  Let us perform a linearized gauge transformation $A \mapsto A + \d \c$ to the external gauge field $A$.  By integration by parts, the result can be written as a sum of terms where we apply the exterior derivative $\d$  to one of the propagators, or else we integrate over one of the boundary components of the domain of integration.  We will show that all such terms vanish except the term where we integrate over the boundary component in which $p_{\rm max} = p_{\rm min} + \eps$.

We use the identity \eqref{propagator_greens}.
Applying the exterior derivative to a propagator has the effect of yielding a contribution where we integrate over the region where the vertices $v_i,v_j$ (or $p_i, v_j$) at either end of the propagator are identified.

If we contract the propagator connecting $p_3$ and $v_3$, the result must vanish. This is because we can assume that the external gauge field $A$ and the ghost $\c$ are both divisible by $z$. This assumption is justified because the anomaly always involves a $z$-derivative on each external field.  If we integrate over the region where $v_3 = p_3$, then since $z_3 = 0$ on this domain, the integrand vanishes.

If we contract the propagator connecting $p_2$ and $v_2$, the result vanishes by the argument presented in section \ref{excluding_amplitude}.  

If we contract the propagator connecting $v_1$ and $v_2$, or $v_3$ or $v_1$, the results cancel by the Jacobi identity. 

Next, let us consider the boundary components where two adjacent points on the Wilson line can meet. They can meet from above or below, with different signs.  After
using the Jacobi identity, the color factor associated to a boundary component like this is that given by the diagram
\begin{center}
\begin{tikzpicture}
\begin{scope}[scale=0.7]
\draw[dashed] (-2,2.7) to (-2,-2.7);
\draw (0,0) circle (1cm);
\draw (130:1) to [out=130,in=0] (-2,1.3);
\draw (230:1) to [out=230, in=0] (-2,-1.3);
\draw (-1.5, 1.25) to (-1.4, 2.1);
\draw (1,0) to (2,0); 
\end{scope}
\end{tikzpicture}
\end{center}
We have seen in section \ref{excluding_C} that anomalies with a color factor of this form can be cancelled by a counter-term, and so are not relevant. 

The remaining boundary component is the one with $p_{\rm max} = p_{\rm min} + \eps$.  This is the one that will contribute to the anomaly.

The integral describing the anomaly is then
\begin{align}
\begin{split}
\left(\tfrac{\ii}{2\pi}\right)^3
 \int_{\substack{p_i,v_j \\ p_{\rm min}  = p_{\rm max} - \eps}} & P(p_1, v_1) \d z_1  \c^a_1  P(v_1, v_2) \d z_2  P(v_2, p_2) P(v_2, v_3) \d z_3 A^b_3  P(v_3, p_3)   \\
& \times f^{afc} f^{fgd} f^{gbe}\{ \rho(t^c), \rho(t^d), \rho(t^e)\}  \;,
\end{split}
\end{align}
where on vertex $v_1$ the external gauge field has been replaced by the ghost $\c$.  The prefactor of $\tfrac{1}{2}$ has been cancelled by the factor of $2$ coming from the two possible external lines to the ghost field.  The diagram automorphism interchanges these two, so we are left with the single integral above but without the prefactor of $\tfrac{1}{2}$. 

This integral can be written as a sum of six terms, according to the six possible orderings of the points on the Wilson line. A reflection in a plane orthogonal to the Wilson line shows that there are only three independent integrals, which we can take to be the three where the points $p_1,p_2,p_3$ are cyclically ordered.  We find that we need to compute
\begin{align}
\begin{split} 
&2  \left(\tfrac{\ii}{2\pi}\right)^3\int_{p_1 < p_2 < p_3= p_1 + \eps}\int_{v_1,v_2} P(p_1, v_1) \d z_1  \c^a_1  P(v_1, v_2) \d z_2  P(v_2, p_2) P(v_2, v_3) \d z_3 A^b_3  P(v_3, p_3) 
\\ &+  2\left(\tfrac{\ii}{2\pi}\right)^3 \int_{p_2 < p_3 < p_1= p_2 + \eps}\int_{v_1,v_2}  P(p_1, v_1) \d z_1  \c^a_1  P(v_1, v_2) \d z_2  P(v_2, p_2) P(v_2, v_3) \d z_3 A^b_3  P(v_3, p_3) \\
 &+ 2\left(\tfrac{\ii}{2\pi}\right)^3 \int_{p_3 < p_1 < p_2= p_3 + \eps}\int_{v_1,v_2}  P(p_1, v_1) \d z_1  \c^a_1  P(v_1, v_2) \d z_2  P(v_2, p_2) P(v_2, v_3) \d z_3 A^b_3  P(v_3, p_3)  \;. \label{equation_three_integrals} 
 \end{split}
 \end{align}
The factor of $2$ here is because each integral represents the contribution from one of two possible orderings on the points $p_i$.

\subsubsection{Calculating the Anomaly Integral}

Because we know that the anomaly must involve a $z$-derivative for each external line and no other derivatives, we can detect the anomaly by assuming that the external ghost field $\c$ is $z$, and the external gauge field $A$ is $z \delta_{x = 0}$.  

The choice of external gauge field fixes the location of the vertex $v_3$; by translation invariance we choose instead to fix the location of the vertex $p_1$ to be $p_1 = 0$. 

We will first evaluate the integral of the diagram in Fig. \ref{figure_diagram_labels}, where $p_1 < p_2 < p_3$, and we integrate over the region where $p_3 = p_1 + \eps$, and we set $p_1 = 0$.  We denote $p_2$ by $p$.  

The integral for the first diagram of \fig \ref{figure_anomaly2} is 
\begin{align}
\label{13form}
2\left(\tfrac{\ii}{2\pi}\right)^3 \int_{p=0}^{\epsilon} \int_{v_1, v_2, v_3} 
P(0, v_1)  \wedge  z_1 \d z_1  \wedge P(v_1, v_2) \wedge \d z_2 \wedge P(v_2, p)\wedge
P(v_2, v_3) \wedge z_3 \d z_3 \wedge P(v_3, \epsilon) \;.
\end{align}
Now that the integrand has 5 propagator 2-forms, and 1 vertex 1-forms, making a 13-form, which is consistent since we wish to integrate over $1+4\times 3=13$ variables.

We evaluate this integral in Appendix \ref{app.two-loop_computation}.
Reintroducing the color factor and the factor of two present in \eqn \eqref{equation_three_integrals},
we find that the anomaly is given by
\begin{equation} 
\label{13form_result}
	\frac{\hbar^2}{12}   f^{afc} f^{fgd} f^{gbe}\{ \rho(t^c), \rho(t^d), \rho(t^e)\}  \;.
\end{equation}

\subsection{The Anomaly as a Correction to the Algebra}
\label{anomaly_correction}
In this section we will show that for a Wilson line to be defined at the quantum level, there must be quantum corrections to the coupling of the gauge field which satisfy a certain algebraic relation.  This relation forces the Wilson line to be built from a representation of the Yangian, not of the Lie algebra $\g[[z]]$.

Suppose we have a Wilson line in a representation $V$ of our gauge group. Suppose that the level-one generators $t_{a,1}$ of the Yangian act by some operators $\rho(t_{a,1})$.  Consider the four two-loop diagrams in Figure \ref{fig.four_diagrams}.

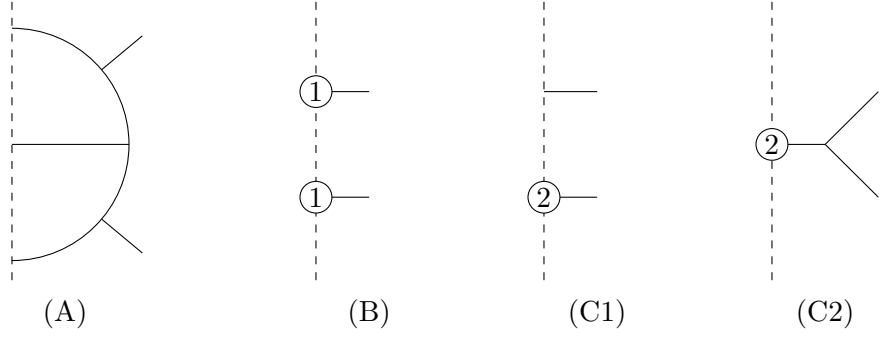
\begin{figure}[htbp]
\begin{center}
\begin{tikzpicture}

\begin{scope}[ shift={(-2,0)}, scale=(0.7)]
\draw[dashed](0,2.7) to (0,-2.7);
\draw (0,-2.2) arc (-90:90:2.2);
\draw(2.2,0) to (0,0); 
\draw (40:2.2) to (40:3.2);
\draw (-40:2.2) to (-40:3.2);
\node at (1,-3.2) {(A)};
\end{scope}
	
\begin{scope}[shift={(2,0)}, scale=(0.7)]
\draw[dashed] (0,2.7) to (0,-2.7); 
\draw (0,-1) to (1,-1);
\draw (0,1) to (1,1);
\draw[ fill  = white] (0,1) circle (0.3); 
\node at (0,1) {$1$};
\draw[ fill  = white] (0,-1) circle (0.3); 
\node at (0,-1) {$1$};
\node at (1,-3.2) {(B)};
\end{scope}

\begin{scope}[shift={(5,0)}, scale=(0.7)]
\draw[dashed] (0,2.7) to (0,-2.7); 
\draw (0,-1) to (1,-1);
\draw (0,1) to (1,1);
\draw[ fill  = white] (0,-1) circle (0.3); 
\node at (0,-1) {$2$};
\node at (1,-3.2) {(C1)};
\end{scope}

\begin{scope}[shift={(8,0)}, scale=(0.7)]
\draw[dashed] (0,2.7) to (0,-2.7); 
\draw (0,0) to (1,0) to (2,1);
\draw (1,0) to (2,-1);
\draw[ fill  = white] (0,0) circle (0.3); 
\node at (0,0) {$2$};
\node at (1,-3.2) {(C2)};
\end{scope}

\end{tikzpicture}
\label{figure_commutator}
\end{center}
\caption{\small{Four Feynman diagrams contributing to the anomaly.}}
\label{fig.four_diagrams}
\end{figure}

As before,  a vertex on the Wilson line labelled with the number $1$ or $2$ indicates a coupling of the first or second $z$-derivative of the gauge field.  (We tacitly include the variations on these diagrams where the ordering of the vertices on the Wilson lines has been permuted). 

The amplitude of each of these diagrams can fail to  be gauge invariant. For the first diagram (A), this involves the two-loop computation that we have been describing.  For the other diagrams, the failure to be gauge invariant is much more straightforward, as we will now see.

Suppose that at the quantum level, the first and second derivative of the gauge field are coupled to the Wilson line by operators $\rho_{a,1}:V\to V$ and $\tfrac{1}{2} \rho_{a,2}:V\to V$, where $a$ is an adjoint index. 

Then, the amplitude for diagram (B) is given by 
\begin{equation}
\int_{p_1< p_2 \in \R} \partial_z A^a(p_1)\partial_z A^b (p_2) \rho_{a,1}\circ\rho_{b,1} \;. 
\end{equation}
If we change the gauge field by a gauge transformation $A^a \mapsto A^a+ \d \c^a$ we find (by integration by parts and Stokes' theorem) that the integral becomes  
 \begin{equation}
\int_{p \in \R} \partial_z \c^a(p)\partial_z A^b (p) [\rho_{a,1}, \rho_{b,1}] \;.
\end{equation}
Similarly, for diagram (C1), the amplitude is 
\begin{equation}
	\tfrac{1}{2} \int_{p_1 < p_2\in \R}( \partial_z^2 A^a(p_1)) A^b(p_2) \rho_{a,2} \circ \rho_{b,0} + 
	\tfrac{1}{2} \int_{p_1 < p_2 \in \R} (A^a(p_1)) \partial_z^2 A^b(p_2) \rho_{a,0} \circ \rho_{b,2} \;. 
\end{equation}
The failure of this to be gauge invariant is given by 
\begin{equation} 
	\tfrac{1}{2} \int_{p \in \R} (\partial_z^2 \c^a(p)) A^b(p) [\rho_{a,2} , \rho_{b,0}] 
	+ \tfrac{1}{2} \int_{p \in \R}  \c^a(p) \partial_z^2 A^b(p) [\rho_{a,0} , \rho_{b,2}]  \;.
\end{equation}

Next, diagram (C2) has amplitude 
\begin{equation}
\tfrac{1}{4} \int_{p \in \R, v = (x,y,z)} \d z \partial_{z_p}^2  P(p,v) A^a(v) A^b(v) f_{ab}{}^c \rho_{c,2} \;. 
\end{equation}
Here $v$ indicates the position of the interior vertex, and $\partial_{z_p}$ indicates we apply a $z$-derivative to the $p$-coordinate of the propagator.  
The failure of this to be gauge invariant is, using the fact that $\d z_v \d P(p,v) = -\delta_{p = v}$ and imposing the equations of motion for $A$, 
\begin{equation}
-\tfrac{1}{2}\int_{p \in \R} \partial_z^2 \left( \c^a(p) A^b(b)   \right) f_{ab}{}^c \rho_{c,2} \;. 
\end{equation}

Finally, the anomaly for diagram (A) is, as we have seen, 
\begin{equation}
	\tfrac{\hbar^2}{12} \int_{p \in \R} \partial_z \c^a(p) \partial_z A^b(p) f^{afc} f^{fgd} f^{gbe}\{ \rho(t^c), \rho(t^d), \rho(t^e)\}  \;.
\end{equation}
The sum of the anomalies for the diagrams (A), (B) and (C) must vanish.  This gives rise to the equations
\begin{align} 
[\rho_{a,2}, \rho_{b,0}] &= f_{ab}{}^c \rho_{c,2} \;,\\
[\rho_{a,1}, \rho_{b,1}] &= f_{ab}{}^c \rho_{c,2} - \tfrac{\hbar^2}{12} Q_{ab}( \rho_{\bullet,0}) \;, \label{level1_relations}
 \end{align}
 where the first equation is from the cancellation of the coefficients of 
 $(\partial_z^2 \c^a(p)) A^b(p)$ and  $\c^a(p) \partial_z^2 A^b(p)$,
 and the second from those of  $(\partial_z \c^a(p)) \partial_z A^b(p)$.
In the equations above we used the short-hand notation $\rho_{a,n}=\rho(t_{a,n})$ and
we defined 
\begin{align}
Q_{ab}( \rho_{\bullet,0}) = f^{afc} f^{fgd} f^{gbe}\{ \rho(t^c), \rho(t^d), \rho(t^e)\} \;.
\end{align} 
The first equation is the commutation relation we found classically.  The second equation tells us that the operators $\rho_{a,1}$ do not commute to give $\rho_{c,2}$, as we would find classically, but a linear combination of $\rho_{c,2}$ and a certain cubic polynomial in the level $0$ generators $\rho_{d,0}$.

We can decompose the exterior square $\wedge^2 \g$ into $\g \oplus \wedge^2_0 \g$, where $\wedge^2_0 \g$ is the kernel of the Lie bracket map from $\wedge^2 \g \to \g$.  By using this decomposition, we can write the second relation as a sum of two independent relations:
\begin{align} 
f^{ab}{}_c [\rho_{a,1}, \rho_{b,1}] &=  \sh^\vee \rho_{c,2} -  \tfrac{\hbar^2}{12} f^{ab}{}_c Q_{ab}(\rho_{\bullet,0})\;, \\
\Lambda^{ab}[\rho_{a,1}, \rho_{b,1}] &= -\tfrac{\hbar^2}{12} \Lambda^{ab} Q_{ab}(\rho_{\bullet,0})\label{relation_new} \;,
 \end{align}
where $\Lambda^{ab} \in \wedge^2 \g$. 
The first relation can always be satisfied by redefining the operator $\rho_{c,2}$ to
\begin{equation} 
\rho'_{c,2} = \rho_{c,2} - \frac{\hbar^2}{12 \sh^\vee}  f^{ab}{}_c Q_{ab}(\rho_{\bullet,0}) \; . 
 \end{equation}
The second relation, however, can not satisfied in such a trivial way. 

What we have found from this analysis is that, in order for a classical Wilson line to quantize modulo $\hbar^3$, we must be able to define the operators $\rho_{a,1}$ so that they satisfy the relation \eqref{relation_new}.  This relation is one of the relations in the Yangian algebra which makes it into a non-trivial deformation of the universal enveloping algebra of $\g[[z]]$. 

The reader will note that the diagrams in Fig. \ref{fig.four_diagrams} are simply the classical diagrams of Fig. \ref{tred}, with quantum corrections to the vertices, plus
the anomaly diagram (A).  Note that we did not add to diagrams (C2) in the figure additional contributions involving a quantum correction to the bulk
vertex.  Such contributions can be omitted, using Proposition \ref{proposition_cohomology}, as they involve a cohomology class that factors through a map to $\g$.

\subsubsection{Matching with Standard Relations in the Yangian}\label{matching}
Let us explain how to match the relation we have found with one of the known descriptions of the Yangian.

For any simple Lie algebra $\mf{g}$, the sequence
\begin{equation} 
\wedge^3 \g \to \wedge^2 \g \to \g 
 \end{equation}
is exact in the middle, where each map in the sequence is obtained by applying the Lie bracket to two entries in the exterior power.  This sequence is part of the Chevalley-Eilenberg homology complex of $\g$.  For any simple Lie algebra, $H_2(\g,\C) = 0$, which is why this sequence is exact in the middle.

Therefore, the kernel of the map from $\wedge^2 \g \to \g$ is the image of the map from $\wedge^3 \g \to \wedge^2 \g$. The antisymmetric tensor $\Lambda^{ab} \in \wedge^2 \g$ appearing in \eqn \eqref{relation_new} is in the kernel of this map. Therefore, we lose no information by considering the relation in \eqn \eqref{relation_new} when $\Lambda^{ab}$ is taken to be in the image of the map from $\wedge^3 \g$, that is, when $\Lambda^{ab}$ is of the form
\begin{equation} 
\Lambda^{ab} =   \Gamma^{ade}f_{de}{}^b  
 \end{equation}
for some antisymmetric tensor $\Gamma^{ade} \in \wedge^3 \g$.

Relation \eqref{relation_new} takes the form, in this notation,
\begin{equation} 
\Gamma^{ade} f_{de}{}^b[\rho_{a,1}, \rho_{b,1}] = - \tfrac{\hbar^2}{12}   \Gamma^{ade} f_{de}{}^b Q_{ab}(\rho_{\bullet,0})  \;.
\end{equation}
For an element $t \in \g$, we let $\rho(t)$ denote the action of $t$ on our representation, and $\rho(J(t))$ denote the action of the corresponding level one generator on our representation.  Then our relation is 
\begin{equation} 
\label{Serre-tmp}
\Gamma^{abc} [  \rho (J( t_a) ), \rho ( J( [t_b,t_c]) ) ] =  - \tfrac{\hbar^2}{12} \Gamma^{abc} ([t_e, [t_a,t_d]], [t_f,[t_b,t_c]] )  \{\rho(t_d), \rho(t_e), \rho(t_f)\}  \;,
 \end{equation}
where as before $\{\rho(t_d), \rho(t_e), \rho(t_f)\}$ is  $1/6$  times the sum of the products of these operators in each of the six possible orders  (see \eqn \eqref{triple_def}).
The
Jacobi identity and the $\mathfrak{g}$-invariance of the Killing form allows us to rewrite
\begin{align*} 
 ([t_e, [t_a,t_d]], [t_f,[t_b,t_c]] )  
 &= ([t_e, [t_a,t_d]], [[t_f,t_b],t_c] )  +  ([t_e, [t_a,t_d]], [t_b,[t_f,t_c]] )   \\
 &= ([t_c, [t_e, [t_a,t_d]]], [t_f,t_b] ) -  ([t_b,[t_e, [t_a,t_d]]], [t_f,t_c] )  \;,
 \end{align*}
Thus anti-symmetrizing in  $a,b,c$ and again using the Jacobi identity, we obtain
\begin{align*} 
 & \Gamma^{abc} ([t_e, [t_a,t_d]], [t_f,[t_b,t_c]] )  \\
  &\qquad =  2 \Gamma^{abc} ([t_c, [t_e, [t_a,t_d]]], [t_f,t_b] )  \\
  &\qquad=  2 \Gamma^{abc} ( [ [t_c,t_e], [t_a,t_d]], [t_f,t_b] ) + 2 \Gamma^{abc} ( [t_e, [t_c, [t_a,t_d]], [t_f,t_b])  \\
  &\qquad = -2 \Gamma^{abc} ( [ [t_c,t_e], [t_a,t_d]], [t_b,t_f] )  +\Gamma^{abc} ( [t_e, [[t_c, t_a],t_d]], [t_f,t_b])  \;,
 \end{align*}
where on the last line we have again used anti-symmetry of $\Gamma^{abc}$ to conclude that we can replace $[t_c,[t_a,t_d]]$ by $\tfrac{1}{2}[[t_c,t_a],t_d]$. The second factor in the final line, after combining with the factor $\{\rho(t_d), \rho(t_e), \rho(t_f)\}$,
gives back the expression we started with, with an overall minus sign:
\begin{align*} 
 &\Gamma^{abc} ( [t_e, [[t_c, t_a],t_d]], [t_f,t_b])  \{\rho(t_d), \rho(t_e), \rho(t_f)\} \\
 &\qquad= -\Gamma^{abc} (   [t_e, [t_f,t_b]], [[t_c, t_a],t_d]])  \{\rho(t_d), \rho(t_e), \rho(t_f)\} \\
& \qquad= -  \Gamma^{abc} ([t_e, [t_a,t_d]], [t_f,[t_b,t_c]] )  \{\rho(t_d), \rho(t_e), \rho(t_f)\} \;,
 \end{align*}
 where we used $\mathfrak{g}$-invariance of the Killing form, anti-symmetry in $a,b,c$ and symmetry in $d,e,f$.
 This allows us to conclude that
\begin{align}
\begin{split} 
 &\Gamma^{abc} ([t_e, [t_a,t_d]], [t_f,[t_b,t_c]] )  \{\rho(t_d), \rho(t_e), \rho(t_f)\} \\
 & \qquad=  -  \Gamma^{abc} ( [ [t_c,t_e], [t_a,t_d]], [t_b,t_f] ) \{\rho(t_d), \rho(t_e), \rho(t_f)\}\;. 
\end{split}
 \end{align}
Thus the quantum-corrected commutation relation \eqref{Serre-tmp} becomes
\begin{equation}
 \Gamma^{abc} [  \rho (J( t_a) ), \rho ( J ([t_b,t_c]) ) ] = \frac{\hbar^2}{12} \Gamma^{abc} ( [ [t_c,t_e], [t_a,t_d]], [t_b,t_f] ) \{\rho(t_d), \rho(t_e), \rho(t_f)\}\;. 
\end{equation}
Note that by anti-symmetry of $\Gamma^{abc}$ and symmetry of $ \{\rho(t_d), \rho(t_e), \rho(t_f)\}$ we have
\begin{align}
\begin{split}
 \Gamma^{abc} &[  \rho (J( t_a) ), \rho ( J [t_b,t_c] ) ]\\
 & =
\tfrac{1}{3} \Gamma^{abc}  \left( [  \rho (J( t_a) ), \rho ( J ([t_b,t_c]) ) ] +   [  \rho ( J(t_b) ), \rho ( J ([t_c,t_a]) ) ]+  [  \rho ( J(t_c) ), \rho ( J([t_a,t_b]) ) ] \right)  \;.
\end{split}
\end{align}
Also, $ ( [ [t_c,t_e], [t_a,t_d]], [t_f,t_b] ) \{\rho(t_d), \rho(t_e), \rho(t_f)\} $
is antisymmetric in $a,b,c$.  Therefore our relation is equivalent to
\begin{align}
\begin{split} 
&   [  \rho (J( t_a) ), \rho ( J ([t_b,t_c]) ) ] +   [  \rho ( J(t_b) ), \rho ( J ([t_c,t_a]) ) ]+  [  \rho ( J(t_c) ), \rho ( J([t_a,t_b]) ) ]  \\
& \qquad \qquad =  \tfrac{\hbar^2}{4} ( [ [t_c,t_e], [t_a,t_d]], [t_b,t_f] ) \{\rho(t_d), \rho(t_e), \rho(t_f)\} \\
& \qquad \qquad =  \tfrac{\hbar^2}{4} ( [t_a,t_d], [[t_b,t_e]], [t_c,t_f] ]) \{\rho(t_d), \rho(t_e), \rho(t_f)\}\;. 
 \end{split}
 \end{align}
This is the same as the known relation for the Yangian,
see e.g.\ \eqn (4) in \cite[Theorem 12.1.1]{Chari-Pressley}. 

\subsection{An Example of an Anomalous Wilson Line}\label{exam}
It is known that the adjoint representation of any simple Lie algebra not of type $A$ does not lift to a representation of the Yangian.   This indicates that a Wilson line in the adjoint representation should be anomalous.  In this section we will see explicitly that our two-loop anomaly cannot be cancelled for the adjoint representation of $SO_N$. 

Let us start with some group-theory background.  If $M_1,\dots,M_k\in \mf{so}_N$ (for $N$ sufficiently large), the ring of invariant function of these $k$ elements of $\mf{so}_N$ is generated by the invariant functions given by traces, that is, the functions $\op{Tr}(M_{i_1} \dots M_{i_r})$ where $i_1,\dots,i_r$ are in the set $\{1,\dots,k\}$.  Repeated indices are allowed and the trace is taken in the fundamental representation.   

For finite $N$, there are trace relations among these functions. For $N$ large compared to the degree of the polynomial in the variables $M_i$, the only relation that survives is the relation that says
\begin{equation} 
	\op{Tr}(M_{i_1} \dots M_{i_r})= (-1)^{r}\op{Tr}(M_{i_r}\dots M_{i_1}) \;.  
\end{equation}

We will need to understand the decomposition of $\mf{so}_N \otimes \mf{so}_N$ into irreducible representations.  The number of irreducible representations is the same as the number of invariant tensors in $\mf{so}_N^{\otimes 4}$.  Let $M_1,M_2,N_1,N_2$ denote elements of $\mf{so}_N$. The analysis above allows us to enumerate the invariant functions of these matrices which are linear in each variable.  We find that there are $6$ invariant tensors, three with a single trace and three with two traces.  Of these, $4$ of them are symmetric when we exchange $M_1$ and $M_2$, or when we exchange $N_1$ and $N_2$.  These correspond to projectors onto irreducible representations inside $\Sym^2 \mf{so}_N$. Two of them are antisymmetric under the same permutations, and correspond to projectors onto irreducible representations in $\wedge^2 \mf{so}_N$.  The two antisymmetric tensors are
\begin{align} 
\begin{split}
& \op{Tr}([M_1,M_2] [N_1,N_2]) \;, \\  
&\op{Tr}(M_1 N_1) \op{Tr}(M_2 N_2) - \op{Tr}(M_1 N_2) \op{Tr}(M_2  N_1) \;.   
\end{split}
 \end{align}
The first antisymmetric tensor gives rise to the projection from $\wedge^2 \mf{so}_N$ onto the adjoint representation, and the second to the projection onto the irreducible representation $\wedge^2_0 \mf{so}_N$.

Note also that there is only one copy of the adjoint representation in $\mf{so}_N\otimes \mf{so}_N$. Any copy of the adjoint representation is associated to an invariant element in $\mf{so}_N^{\otimes 3}$,  and so an invariant function of $M_1,M_2,M_3$ which is linear in each $M_i$.  The discussion above tells us that the only such invariant function is $\op{Tr}(M_1 M_2 M_3)$.

If we have a Wilson line in the adjoint representation of $SO_N$, it follows from this that the level one generators must be coupled by a multiple of the level $0$ generators. Therefore, the commutator of the level $1$ generators must be proportional to the level $0$ generators.  Relation \ref{level1_relations} can never hold unless 
\begin{equation} 
	 Q_{ab}( \rho_{\bullet,0}) = f^{afc} f^{fgd} f^{gbe}\{ \rho(t^c), \rho(t^d), \rho(t^e)\} 
\end{equation}
is a multiple of some $\rho(t^c)$.  We will show the projection of the tensor $Q_{ab}$ onto $\wedge^2_0(\mf{so}_N)$ is non-zero, so the anomaly is non-zero. 

Let us introduce incoming and outgoing states $t_{\rm in}$, $t_{\rm out}$ on the adjoint Wilson line.  The three diagrams which give us an anomaly all, as we have seen, have the same color factor. Since $\rho$ is in the adjoint representation, we can compute the color factor using a diagram in which the Wilson line is placed on the same footing as the other lines in the Feynman diagram. If we do this, the diagram looks like
\begin{center}
\begin{tikzpicture}
\draw (0:0) circle (1);
\draw (0:1) to (180:1);
\node (N1) at (120:2) {$t_{\rm in}$};
\node (N2) at (60:2) {$t_a$};
\node (N3) at (240:2){$t_{\rm out}$};
\node (N4)at (300:2) {$t_b$};
\draw (120:1) to (N1);
\draw (60:1) to (N2);
\draw (240:1) to (N3);
\draw (300:1) to (N4);
\end{tikzpicture}		
\end{center}
where the segment connecting $t_{in}$, $t_{out}$ is what was the Wilson line.   Here we are implicitly anti-symmetrizing in $a$ and $b$, since the anomaly is always antisymmetric in the external lines. 

This makes it clear that the color factor of the anomaly in the adjoint representation is 
\begin{equation} 
	(t_{\rm out}, Q_{ab}( \rho_{\bullet,0})t_{\rm in}) = \op{Tr} (t_{\rm in} t_{a} t_f t_{b} t_{\rm out} t_f ) -  \op{Tr} (t_{\rm in} t_{b} t_f t_{a} t_{\rm out} t_f )\;,   \label{eqn_anomaly_adjoint}  
\end{equation}
where the trace is taken in the adjoint representation, and we sum over $f$.  

Let us view the anomaly \eqref{eqn_anomaly_adjoint} as a linear operator from $\mf{so}_N^{\otimes 2} \to \mf{so}_N^{\otimes 2}$, where the first two copies of $\mf{so}_N$ are given by $t_a,t_b$ and the second two by $t_{\rm in}, t_{\rm out}$.  Since this is an $SO(N)$ invariant operator, it is a linear combination of the projection onto the $6$ irreducible representations in $\mf{so}_N^{\otimes 2}$.  The anomaly is non-zero as long as the coefficient of the projection onto $\wedge^2_0 \mf{so}_N$ is non-zero.

Equivalently, we can expand the expression \eqref{eqn_anomaly_adjoint} as a linear combination of the $6$ invariant tensors in $\mf{so}_N^{\otimes 4}$. To show that the anomaly is non-zero it suffices to show that, in this expansion, the coefficient of 
\begin{equation} 
\op{Tr} (t_{\rm in} t_a) \op{Tr}(t_{\rm out} t_b) - \op{Tr}(t_{\rm in} t_b) \op{Tr}(t_{\rm out} t_a) 
 \end{equation}
is non-zero. This is because this invariant tensor corresponds to projection onto $\wedge^2_0 \mf{so}_N$.  Since the anomaly is antisymmetric in $a$ and $b$ we need to verify that the coefficient of $\op{Tr}(t_{\rm in} t_a) \op{Tr}(t_{\rm out} t_b)$ is non-zero.

We can evaluate the expression in equation \eqref{eqn_anomaly_adjoint} for $\mf{so}_N$ using the double-line notation familiar from evaluation of the color factors for $\mf{gl}_N$ gauge theories. Since we are using $\mf{so}_N$ instead of $\mf{gl}_N$, the double-line technique works a little differently.  There are two types of double-line edges, which we can think of as a flat ribbon and a ribbon with a half-twist. This is because the Casimir for $\mf{so}_N$ is represented as 
\begin{equation} 
	\sum E_{ij} \otimes E_{ji} - \sum E_{ij} \otimes E_{ij}  \;,
\end{equation}
where $E_{ij}$ is the elementary matrix.  Only the first term appears in the Casimir of $\mf{gl}_N$, which is why only the first type of double-line edge appears when we study $\mf{gl}_N$ color factors.

Using the double line notation, we  find that the color factor we are computing is given by a sum over connected unoriented surfaces of Euler characteristic $-1$, with $4$ marked points on the boundary. Since we are interested in the double-trace terms in the anomaly, we need to consider surfaces where two of the marked points (labelled by $t_{in}, t_a$) are on one boundary component, and the other two or on a different component.  

The order $N$ two-trace term in the anomaly is given by the planar diagram  
\begin{center}
\begin{tikzpicture}[scale=1.3]
\draw (0:0) circle (1);
\draw (-0.8,0.1) [out = 0, in = 180] to (0.8, 0.1);
\draw (-0.8, -0.1) to [out = 0, in = 180] (0.8,-0.1);
\draw (-0.8,0.1) to [out = 90, in = 180] (0,0.8) to [out=0, in = 90] (0.8, 0.1);
\draw (-0.8, -0.1) to [out = 270, in = 180] (0,-0.8) to [out=0, in = 270] (0.8,-0.1);
\node (N1) at (120:0.55) {$t_{\rm in}$};
\node (N2) at (60:0.55) {$t_a$};
\node (N3) at (240:0.55){$t_{\rm out}$};
\node (N4) at (300:0.55) {$t_b$};
\draw[ fill = black] (120:0.8) circle (0.05);
\draw[ fill = black] (60:0.8) circle (0.05);
\draw[ fill = black] (240:0.8) circle (0.05);
\draw[ fill = black] (300:0.8) circle (0.05); 
\end{tikzpicture}
\end{center}
which contributes 
\begin{equation} 
N \op{Tr} (t_{\rm in} t_a)  \op{Tr}( t_b t_{\rm out} ) \;.
 \end{equation}
 (We have not drawn the similar diagram which contributes the same expression with $a$ and $b$ exchanged and a different sign).

There is also an order $1$ term, given by a double-line diagram which has the topology of a M\"obius band with a circle removed. For $N$ sufficiently large ($N=10$ certainly suffices) the coefficient of this term is much smaller than that of the order $N$ term, so the anomaly is non-zero.

\section{Trigonometric Solutions of the Yang-Baxter Equations}\label{trigonometric}

\subsection{Preliminaries}\label{prelims}
Our goal in this section will be to understand from the present perspective the trigonometric solutions of the Yang-Baxter equation.
For this we have to consider the case that $C=\C^\times$, or equivalently $C=\CP^1$ with a differential $\omega = \d z/z$ that has two
simple poles at 0 and $\infty$.

The basic consequence of a pole in $\omega$ is the following.  One has
\be\label{zifo}
\bar\partial\frac {\d z}{z}=2\pi \i \delta^2(z)\;,
\ee
where $\delta^2(z)$ is a delta function normalized by $\int|\d^2z|\delta^2(z)=1$.
Accordingly, when we vary the action (\ref{eq.action}) to derive the equations of motion, we pick up a ``boundary  term'' supported at $z=0$:
\be\label{iffto}
\delta S=\dots -\frac{\i}{\hbar}\int_{\Sigma\times \{0\}}\Tr\, A\wedge \delta A\;, 
\ee
where bulk terms have been omitted.
The equations of motion of the classical theory tell us to set to zero  the bulk terms 
and also the boundary term $\Tr\,A\wedge \delta A|_{\Sigma\times \{0\}}$.  We will abbreviate the latter as $\Tr\,A\wedge \delta A|$ (and we use
a similar notation in general for restriction to a singularity or boundary).

Actually, the boundary terms in the equations of motion play a distinguished role.   Quantum field theory is constructed by integrating over fields
that do not necessarily satisfy the equations of motion in bulk.  The bulk equations of motion are only satisfied in the classical limit.  But boundary terms in the
equations of motion have to vanish exactly.   This is needed in order to get in bulk a symmetric
propagator, as expected for bosons (or an antisymmetric one in an analogous problem with fermions).
In general, in quantum field theory in the presence of boundaries or defects, one typically runs into ``boundary terms'' analogous to $\Tr\,A\wedge \delta A$.
As a starting point in quantization, one always needs to impose a minimal condition that sets the boundary terms to 0.   For example, if one quantizes on a manifold $X$ an ordinary scalar field with the usual
Lagrangian $\frac{1}{2} |\d\phi|^2$, the boundary term is $\int_{\partial X}\delta\phi \partial_n\phi$ (where $\partial_n$ is the normal derivative).
The simplest way to dispose of this boundary term is to set either $\phi|=0$ (Dirichlet boundary conditions) or $\partial_n\phi|=0$ (Neumann boundary conditions).
With either of these conditions, the theory can be quantized.  By contrast,  a stronger condition such as $\phi=\partial_n\phi=0$ is too strong and does not lead to
a quantum theory of a scalar field. 

What condition on $A$ will we use to set $\Tr\,A\wedge\delta A|=0$?  The most obvious condition might seem
to be  $A|=0$.  
However, this condition  is too strong, analogous to $\phi|=\partial_n\phi|=0$ for the scalar field.  If we impose it, a suitable propagator will not exist and
we will not be able to do perturbation theory.\footnote{Later, we will impose a weaker condition and find a unique solution for the $r$-matrix.
A similar analysis assuming that $A|=0$  will show that no possible $r$-matrix exists.    This is actually clear from the fact that
the unique $r$-matrix we get with a weaker condition does not obey $A|=0$.}

We can gain some intuition by considering the case of a double pole.  Of course, we have already studied double poles in the context
of rational solutions of the Yang-Baxter equation.  In that analysis, we considered the case that $\omega=\d z$ with a double pole
at infinity.  For our present purposes, it is more convenient to place the pole at a finite point, so we take $\omega=\d z/z^2$, with a double
pole at the origin.  Since 
\be\label{zifffo}\bar\partial\frac {\d z}{z^2}=-2\pi \i \partial_z\delta^2(z)\;,\ee
the vanishing of the surface term now requires
\be\label{miffo}0=\left.\left(\partial_z\Tr\, A\wedge \delta A\right)\right|\;. \ee
A completely natural way to satisfy this condition is to set $A|=0$.   Varying $A$ with the constraint $A|=0$, we have also $\delta A|=0$,
so that $\Tr\,A\wedge \delta A$ has a double zero at $z=0$ and eqn.\ (\ref{miffo}) is obeyed.  The condition $A|=0$ is gauge-invariant if
we likewise constrain the generator $\varepsilon$ of a gauge transformation to obey $\varepsilon|=0$.  

Not only is $A|=0$ a natural way to satisfy eqn. (\ref{miffo}), it is what we have actually done in studying rational solutions
of the Yang-Baxter equation.  Starting in section \ref{lowestorder}, we used a propagator which vanishes at infinity, which amounts
to taking $A$ to vanish at infinity (where we placed the double pole).   As explained in section \ref{ybeu}, this choice has the advantage
of ensuring that  up to gauge transformation, the only classical solution is the trivial one $A=0$, a condition that
is needed if we want to obtain something as simple as the usual Yang-Baxter equation.  As we will explain below, there are other reasonable ways 
to satisfy $(\partial_z\Tr\,A\wedge\delta A)|=0$, but they lead to something more complicated.   

If $A|=0$ is a satisfactory condition in the presence of a double pole, it can hardly be the right condition for a simple pole. If we start with a
differential $\omega=\d z/z^2$ with a double pole, and perturb it to, say, $\omega'=\d z/(z-z_1)(z-z_2)$, with small $z_1,z_2$, the condition
$A|=0$  at a double pole must somehow be split between the two simple poles at $z_1$ and $z_2$.

As remarked above,  when one studies field theory in the presence of a boundary or defect, the right  procedure is always to impose a minimally restrictive condition that sets this term to zero. 
But the  condition $A|=0$ is not a minimal condition to ensure that $\left.\Tr\,A\wedge\delta A\right|=0$.  For this, it suffices to pick a middle-dimensional
complex subspace $\l_0\subset \g$ that is ``isotropic'' (or ``Lagrangian'') for the quadratic form $\Tr$, in the sense that for $a,b\in \l_0$, $\Tr\,ab=0$.  (We take $\l_0$
to be middle-dimensional because this is the maximum possible dimension for an isotropic subspace, leading to the weakest possible
condition on $A$.)  Then we require that $A|$ is $\l_0$-valued,
that is, that it is an $\l_0$-valued 1-form along $\Sigma$.  Having imposed this condition on $A$, we impose it also on $\delta A$, and then we see
that with $A|$ and $\delta A|$ being $\l_0$-valued 1-forms, $\left.\Tr\,A\wedge \delta A\right|=0$.  

One may object that the condition for $A|$ to be $\l_0$-valued is not gauge-invariant.  However, what we actually need is not to maintain the full
gauge symmetry along the locus of the pole, but only to maintain enough gauge symmetry so that the usual ``longitudinal'' part of $A|$ can be
gauged away.   For this, we simply ask that $\l_0$ should be a subalgebra of $\g$ (and not just a subspace) and we ask that the generator
$\varepsilon$ of a gauge transformation should satisfy the condition that $\varepsilon|$ is $\l_0$-valued.  These conditions on $A$ and $\varepsilon$
are compatible in the sense that, with both $A|$ and $\varepsilon|$ constrained to be $\l_0$-valued, the usual gauge transformation law 
$\delta A=D\varepsilon$ makes sense, and moreover the action (\ref{eq.action}) is gauge-invariant.   In effect, what has happened is the following.
Along the locus of the pole, we have set to zero some components of $A$, and the other components have their usual gauge-invariance.

We are really interested in a situation in which $\omega$ has two simple poles, say $\omega=\d z/z$ with poles at 0 and at $\infty$. Treating
each simple pole as above, we pick two middle-dimensional isotropic subalgebras of $\mathfrak{g}$, say $\l_0$ and $\l_\infty$, in general with no
relation between them.  We require $A$ and $\varepsilon$ to be $\l_0$-valued when restricted to $\Sigma\times \{0\}$, and $\l_\infty$-valued
when restricted to $\Sigma\times \{\infty\}$.  

To get from this construction a solution of the usual Yang-Baxter equation (as opposed to the ``dynamical Yang-Baxter equation,'' which we study
in section \ref{dybe}), we need a further condition that is familiar from section \ref{ybeu}:  
the trivial solution $A=0$ should have no deformations and no continuous unbroken gauge
symmetries.  The two conditions are equivalent for the following reason.  Let $\g_{0,\infty}$ be the sheaf of holomorphic $\g$-valued
functions on $\CP^1$ that are $\l_0$-valued at 0 and $\l_\infty$-valued at $\infty$.   The Lie algebra of the group of gauge symmetries of the trivial
solution $A=0$ is $H^0(\CP^1,\g_{0,\infty})$, and the tangent space to $A=0$ in the moduli space of classical solutions of the theory
is $H^1(\CP^1,\g_{0,\infty})$.  Thus the condition that $A=0$ has no deformations is $H^1(\CP^1,\g_{0,\infty})=0$, and the
condition that it has no continuous gauge symmetries is $H^0(\CP^1,\g_{0,\infty})=0$.  These two conditions are equivalent because of the
Riemann-Roch theorem, which in the present situation implies that
\be\label{merz}
\mathrm{dim} \,H^0(\CP^1,\g_{0,\infty})-\mathrm{dim}\,H^1(\CP^1,\g_{0,\infty})=0\;.
\ee

Concretely, $H^0(\CP^1,\g_{0,\infty})$ is simply\footnote{A global holomorphic section of $\g_{0,\infty}$ is a $\g$-valued constant
(or it would have singularities somewhere) that must be valued in $\l_0\cap \l_\infty$ because of the conditions at 0 and at $\infty$.} 
$\l_0\cap \l_\infty$.  Thus, our two conditions are equivalent to
\be\label{erz}\l_0\cap \l_\infty =0\;. \ee
On dimensional grounds, this is equivalent to
\be\label{meorz} \l_0+\l_\infty=\g\;. \ee
We can understand directly the role of this last condition.  If $\g'=\l_0+\l_\infty$ and $\g=\g'\oplus \g''$, then the $\g''$-valued part of $A$ cannot
be gauged away\footnote{In trying to do so, one runs into the fact that $H^1(\CP^1,\O(-p-q))\not=0$, where $\O(-p-q)$ is the sheaf of holomorphic
functions that vanish at two points $p$ and $q$ (here those two points are $z=0$ and $\infty$).} and therefore the trivial solution $A=0$ has
deformations.   So absence of deformations means that $\g''=0$ and $\g=\l_0+\l_\infty$ (since $\l_0\cap \l_\infty=0$, this is equivalent to $\g=\l_0\oplus \l_\infty$).

Accordingly, to get a trigonometric solution of the Yang-Baxter equation, rather than its ``dynamical'' generalization, we need $\l_0+\l_\infty=\g$.
We have arrived precisely at the notion of a ``Manin triple'' (see for example \cite[p.\ 26]{Chari-Pressley}).   A Manin triple is a complex Lie algebra $\g$ with an invariant, nondegenerate
quadratic form $\Tr$, and a decomposition $\g=\l_0+\l_\infty$, where $\l_0$ and $\l_\infty$ are middle-dimensional isotropic subalgebras.   So
trigonometric solutions of the Yang-Baxter equation, or at least the ones that we will study, are associated to Manin triples.\footnote{\label{truform} The construction as we have described
it involves a Manin triple of $\g$, but one could actually in a somewhat similar way use a Manin triple of $\g[z,z^{-1}]$ (endowed with the nondegenerate
quadratic form $(a,b)=\oint \frac{\d z}{z} \Tr\, a(z)b(z)$).  This would involve a construction somewhat  like that above,
with a more complicated set of conditions on $A$ and $\varepsilon$.
The relation between Manin triples of $\g$ and of $\g[z,z^{-1}]$ is  that a Manin triple of $\g$ determines a Manin triple of $\g[z,z^{-1}]$, namely $\g[z,z^{-1}]=\h\l_0\oplus \h\l_\infty$,
with $\h\l_0=z \g[z]\oplus\l_0$, $\h\l_\infty=z^{-1}\g[z^{-1}]\oplus\l_\infty$.}

At this stage, it is perhaps also clear that in contrast to the rational solutions of Yang-Baxter, the trigonometric ones do not have $G$ as a group
of global symmetries.  In the rational case, given the differential $\omega =\d z/z^2$,  the condition on the generator of a gauge transformation
at $z=0$ was $\varepsilon|=0$.  This condition leaves constant gauge transformations at $z=0$ as global symmetries.  In the trigonometric case,
a constant unbroken gauge symmetry, to be compatible with the conditions at 0 and at $\infty$, has to be an element of $G$ that conjugates
$\l_0$ to itself, and also conjugates $\l_\infty$ to itself.  In practice, in the main example that will be introduced in section \ref{example}, this means that the group of global symmetries
is the maximal torus of $G$.   (This symmetry accounts for the block diagonal form of eqn.\ (\ref{goodform}).)

Before leaving this subject, let us note that we stated the condition $\l_0+\l_\infty=\g$ in a somewhat naive way.  In  gauge
theory, any comparison between $\l_0$ and $\l_\infty$ involves parallel transport from 0 to $\infty$.  What the condition $\l_0+\l_\infty=\g$
really  means is that, after conjugating $\l_0$ and $\l_\infty$ into general position, $\l_0+\l_\infty=\g$.  The case that would lead to a
solution of the dynamical Yang-Baxter equation rather than the ordinary one is that $\l_0$ and $\l_\infty$ are such that even after conjugating
them into general position, $\l_0+\l_\infty\not=\g$.  

\subsection{Example}\label{example}

Given $\g$, can we pick $\l_0$ and $\l_\infty$ to make a Manin triple?  In general, the answer is certainly ``no,'' since if $\g$ is of odd dimension -- for example $\mathfrak{sl}_2$ -- it has no middle-dimensional
isotropic subspaces.  However, the following is a useful construction of examples that are related to the usual trigonometric solutions of the Yang-Baxter equation.

For any simple Lie algebra $\mf{g}$, add to $\mf{g}$ another copy $\til{\mf{h}}$ of the Cartan 
subalgebra $\mf{h}$ of $\mf{g}$ to make a Lie algebra $\til {\g}=\g\oplus \til{\mathfrak h}$.  Equip $\til{\mf{h}}$ with a quadratic form given by the restriction of the Killing form on $\mf{g}$. 
In this way $\til\g$ acquires an invariant nondegenerate pairing.  We will construct a Manin triple for $\til\g$.

Choose a  decomposition $\mf{g} = \mf{n}_- \oplus \mf{h} \oplus \mf{n}_+$ into nilpotent and Cartan subalgebras.   Set
\begin{align}
\begin{split}
\mf{h}_+ &= \{(X,\i X ) \mid X \in \mf{h} \} \subset \mf{h} \oplus \til{\mf{h}} \;,  \\
\mf{h}_- &= \{(X,- \i X ) \mid X \in \mf{h} \} \subset \mf{h} \oplus \til{\mf{h}} \;. 
\end{split}
\end{align}
Then we can choose a pair of complementary Lagrangian subalgebras  of $\til{\g}$ by
\begin{align}
\begin{split}
\mf{l}_0 & = \mf{n}_+ \oplus \mf{h}_+\;, \\
\mf{l}_{\infty} &= \mf{n}_- \oplus \mf{h}_-\;.
\end{split}
\end{align}
This gives a Manin triple for $\til\g$. 

\subsection{\texorpdfstring{The $r$-Matrix}{The r-Matrix}}\label{rma}
Now we consider our four-dimensional theory on $\R^2 \times \mbb{P}^1$, where we use the $1$-form $\d z /z$ on $\mbb{P}^1$ and the Manin triple just described.  Let us place Wilson lines at $z_1,z_2 \in \C^\times$, and have them cross in the topological plane as in \fig \ref{FigureGluon_Exchange}.  Suppose the Wilson lines are in representations $V,W$ of $\mf{g}$.   We will show that the result of this crossing will be the trigonometric $R$-matrix.

Let us first calculate explicitly what happens to leading order in $\hbar$.  To leading order in $\hbar$, the result of the crossing Wilson lines is  described by $R=1+\hbar r$,
where $r$ is a  $\til{\g} \otimes \til{\g}$-valued meromorphic function.  
To obtain $r(z_1,z_2)$,  we can pick a gauge and then repeat the Feynman diagram computations of the previous section.
Instead let us here note that $r(z_1,z_2)$
must satisfy the  following properties:
\begin{enumerate} 
\item $r(z_1,z_2)$ has a first-order pole at $z_1 = z_2$, and is regular elsewhere.  The residue of this pole is $\hbar c$ where $c \in \til{\g} \otimes \til{\g}$ is the Casimir element, dual to the chosen invariant pairing. This is the same singular behavior as in the rational case, since 
OPE singularities are local and are not affected by global topology.
\item At $z_1 = 0$, $r(z_1,z_2)$ is in $\mf{l}_0 \otimes \mf{l}_\infty$, and at $z_1 = \infty$ it is in $\mf{l}_\infty \otimes \mf{l}_0$.   This just reflects the corresponding conditions on $A$.
\item $r(z_1,z_2)$ is sent to $-r(z_1,z_2)$ if we simultaneously swap $z_1$ and $z_2$ and exchange the tensor factors in $\til{\g} \otimes \til{\g}$.  (This reflects the fact that
the action of the theory -- and likewise the propagator and the Feynman diagram used to compute $r(z_1,z_2)$ -- are all odd under an orientation-reversing symmetry of the
topological plane that exchanges the two oriented Wilson lines that are crossing.  Note that the rational $r$-matrix $r= t_a\otimes t_a/(z_1-z_2)$ has the same property.)  
\item Finally, $r(z_1,z_2)$ is a function only of the ratio $z_1/z_2$, since  the differential $\omega=\d z/z$
the conditions we placed at $0$ and $\infty$ are all invariant under the $\C^\times $ action on $C=\C^\times$.  
\end{enumerate}
There is a unique function satisfying these properties.  To write it down, we observe that the Casimir  $c=\sum_a t_a\otimes t_a \in \til{\g} \otimes \til{\g}$ can be written in a unique way 
as a sum
\be 
c = c(\mf{l}_0, \mf{l}_\infty) + c(\mf{l}_\infty,\mf{l}_0) \;,
\ee 
where $c(\mf{l}_0,\mf{l}_\infty) \in \mf{l}_0 \otimes \mf{l}_\infty$, and similarly $c(\mf{l}_\infty,\mf{l}_0)\in \l_\infty\otimes \l_0$. 

Then the unique function $r(z_1,z_2)$ satisfying the desired conditions is
\be 
2\pi \i \, r(z_1,z_2) = \frac{c (\mf{l}_0, \mf{l}_\infty)}{1 - \frac{z_1}{z_2}} - \frac{ c(\mf{l}_\infty,\mf{l}_0) } {1 - \frac{z_2}{z_1}}.  
\ee 
Let us write this out more explicitly. Let us choose a basis $X_\alpha^{\pm}, H_i$ of $\mf{g}$, so that $X_\alpha^{\pm} \in \mf{n}_{\pm}$ and $H_i$ form a basis of $\mf{h}$.  We assume that in this basis the chosen pairing on $\g$ is such that $\ip{X_{\alpha^{+}}, X_{\beta^{-}}} = \delta_{\alpha \beta}$, and $\ip{H_i, H_j} = \delta_{ij}$.  Let $\til{H}_i$ be the corresponding basis of the other copy $\til{h}$ of the Cartan.  Then, in this basis, we have 
\begin{align*}  
2 \pi \i \, r(z_1,z_2) & = \frac{1}{1 - \frac{z_1}{z_2}} \sum_\alpha X_\alpha^+ \otimes X_\alpha^-  + \frac{1}{1 - \frac{z_1}{z_2}} \frac{1}{2} \sum (H_j + \i \til{H}_j) \otimes (H_j - \i \til{H}_j ) \\ 
& - \frac{1}{1 - \frac{z_2}{z_1}} \sum_\alpha X_\alpha^- \otimes X_\alpha^+  - \frac{1}{1 - \frac{z_2}{z_1}} \frac{1}{2} \sum (H_j - \i \til{H}_j) \otimes (H_j + \i \til{H}_j ) \\
&=  \frac{1}{1 - \frac{z_1}{z_2}} \sum_\alpha X_\alpha^+ \otimes X_\alpha^-   - \frac{1}{1 - \frac{z_2}{z_1}} \sum_\alpha X_\alpha^- \otimes X_\alpha^+ \\ 
&+ \tfrac{1}{2} \frac{z_2 + z_1}{z_2 - z_1} \left(
 \sum H_j  \otimes H_j + \sum \til{H}_j \otimes \til{H}_j  
\right)
 + \frac{i}{2}\sum \left( \til{H}_j \otimes H_j -H_j \otimes \til{H}_j\right)  .   
\end{align*}
To evaluate what happens when Wilson lines in particular representations $V,W$ of $\til{\g}$ cross, one applies the homomorphism 
\be 
\rho_V \otimes \rho_W: \g \otimes \g \to \op{End}(V) \otimes \op{End}(W) 
\ee 
to the function $r(z_1,z_2)$ (where $\rho_V,\rho_W$ indicate the maps coming from the $\g$-action on $V$ and $W$). 

\subsection{\texorpdfstring{Specializing to $\mathfrak{sl}_2$}{Specializing to sl(2)}} \label{specializing_sl2}
As an example, let us consider the case that $\mf{g} = \mf{sl}_2$. Let us use the standard basis $e,f,h$ of $\mf{sl}_2$ where $\ip{e,f} = 1$ and $\ip{h,h} =2$.   In this basis,  
\begin{align} 
\begin{split}
 2 \pi \i \, r(z_1,z_2)
= &  \frac{1}{1 - \frac{z_1}{z_2}} e \otimes f   - \frac{1}{1 -\frac{z_2}{z_1}} f \otimes e  \\ 
&+ \tfrac{1}{4} \frac{z_2 + z_1}{z_2 - z_1} \left(\sum h  \otimes h+\sum \til{h} \otimes \til{h}   \right)
 + \frac{\i}{4}\sum \left( \til{h} \otimes h -h \otimes \til{h}\right)  .   
 \end{split}
\end{align}
Let us see what this looks like if our representations $V,W$ are both the spin $1/2$ representation of $\mf{sl}_2$, and  the  basis element $\til{H}$ of the second copy of the Cartan of $\mf{sl}_2$ acts with constants $s_1,s_2$ in the two representations.  Let us choose a basis $e_+, e_-$ of the spin $1/2$ representation of $\mf{sl}_2$, and a corresponding basis
$e_+\otimes e_+$, etc., of the tensor product $V\otimes W$.  In this basis, the matrix $r(z_1,z_2)$ looks like     
\be \label{goodform}
\frac{1}{2\pi \i} \left( \begin{array}{ c c c c }
r^{++}_{++} & & & \\
& r^{+-}_{+-} & r^{-+}_{+-} & \\
& r^{+-}_{-+} & r^{-+}_{-+} & \\
& & &  r^{--}_{--}  
\end{array}\right)\;,
\ee 
where
\begin{align} 
\begin{split}
 r^{++}_{++} &= (1 + s_1 s_2) \frac{z_2 + z_1}{z_2 - z_1}+ \i (s_1 - s_2)\;,\\
 r^{+-}_{+-} &=   (-1+s_1 s_2 )\frac{z_2 + z_1}{z_2 - z_1} +\ i (-s_1 - s_2)\;,  \\
 r^{-+}_{+-} &=   \frac{4}{1-\frac{z_1}{z_2}}\;,\\
 r^{+-}_{-+} &=  \frac{-4}{1-\frac{z_2}{z_1}}\;,\\
 r^{-+}_{-+} &=   (-1+s_1 s_2  ) \frac{z_2 + z_1}{z_2 - z_1} + \i (s_1 + s_2) \;,\\
 r^{--}_{--} &=  (1+s_1 s_2  )\frac{z_2  + z_1}{z_2 - z_1} + \i (s_2 - s_1)\;.
\end{split}  
\end{align} 

If we set $s_1=s_2=0$, we get the usual trigonometric $r$-matrix for $\mathfrak{sl}_2$, associated
to the six-vertex model of statistical mechanics.
For comparison with the literature, one might want to symmetrize our $r$-matrix by conjugation, 
replacing the off-diagonal components by 
\begin{align} 
\begin{split}
 r^{-+}_{+-} =   \frac{4\left(\frac{z_1}{z_2}\right)^{\frac{1}{2}}} {1-\frac{z_1}{z_2}}=
 r^{+-}_{-+} \;.
\end{split}  
\end{align}
without changing the diagonal entries.

In the above basis, the $s$-dependent terms only contribute to the diagonal matrix elements of $r$.  They come
 in two groups. First, we have  $s_1s_2\frac{z_2  + z_1}{z_2 - z_1}$ times the identity operator. This  we can absorb into the definition of the 
overall constant normalization factor of the $R$-matrix. The remaining terms are nontrivial.  They  reproduce the known generalization of the 
six-vertex model  to include horizontal and vertical fields given by $\i s_1$ and $-\i s_2$ (see \eqn \eqref{6v_w_field} in Appendix). 

\section{Elliptic Solutions Of the Yang-Baxter Equation}\label{elliptic}

\subsection{Preliminaries}\label{Preliminaries}

In this section we discuss the elliptic case, where the holomorphic curve $C$
is an elliptic curve $E$. Since the differential $\d z$ does not have any poles we do not need to consider boundary terms like those encountered in section 
\ref{trigonometric}. However, there is an important topological subtlety to consider.

In formulating the theory with gauge group $G$, we begin by considering a topological $G$-bundle over $\Sigma\times E$.  Our considerations
leading to a solution of the Yang-Baxter equation are local along $\Sigma$, so we can take $\Sigma= \R^2$, in which case the choice of a topological
$G$-bundle  $\V\to \Sigma\times E$ amounts to the choice of topological $G$-bundle  $\V\to E$.   However, in general there is a choice to be made,
because, for a connected group $G$,  a $G$-bundle over  $E$ is classified topologically by an invariant\footnote{For example, for $G=SO(3)$, this invariant
is the second Stiefel-Whitney class $w_2(\V)$.  To define it in general, we observe that for $G$ connected and $p$ a point in $E$, the restriction of
$\V$ to $E\backslash p$ is trivial, so $\V$ can be constructed by gluing a trivial bundle over $E\backslash p$ to a trivial bundle over a small disc $D$ containing  $p$.
Since $D\backslash p$ is homotopic to a circle, the gluing function that is used here is classified up to homotopy by
 a class $\zeta\in \pi_1(G)$.} $\zeta\in \pi_1(G)$.  

Therefore, in setting up the theory, we have the freedom to make an arbitrary choice of the element $\zeta$.  Once we  make this choice, a
 classical solution will be defined by a $\bar\partial$ operator $\partial_{\bar z}+[A_{\bar z},\dots]$ on $\V$, modulo gauge transformations.  
 Such an operator gives a holomorphic structure
 to the bundle $\V$ and the gauge-invariant data is precisely the holomorphic equivalence class of this bundle.  
 
 The moduli space of classical solutions is therefore the moduli space of holomorphic $G$-bundles over $E$ with topological class $\zeta$.  The tangent space to the moduli
 space is $H^1(E,\mathrm{ad}(\V))$ (where $\ad(\V)$ is the adjoint bundle associated to $\V$).   Thus the condition of  the theory having a unique classical solution  -- and therefore leading to straightforward perturbation theory and a solution of
 the classical Yang-Baxter equation -- is equivalent to $H^1(E,\ad(\V))=0$.
 
 By Serre duality, $H^1(E,\ad(\V))$ is dual to $H^0(E,\ad(\V))$ for an elliptic curve $E$.  
 Here $H^0(E,\ad(\V))$ is the Lie algebra of the automorphism group of a holomorphic
 $G_\C$ bundle $\V$.  Thus if and only if $\V$ has no infinitesimal deformations, the Lie 
 algebra of its automorphism group will be trivial and $\V$ has only a finite group of
 automorphisms.
(Instead of Serre duality, we could have invoked here the Riemann-Roch theorem, as we did at a similar point in section \ref{prelims}.)

To find, therefore, an elliptic solution of the Yang-Baxter equation (as opposed to its ``dynamical'' generalization), we need to find a complex Lie group $G$
and a holomorphic $G$-bundle $\V\to E$ such that $H^0(E,\ad(\V))=H^1(E,\ad(\V))=0.$  However, the options for such a $G$ and $\V$ are very limited.
The only cases are $G=PGL_N=GL_N/GL_1$ with $N\geq 2$, with  $\V$ chosen so that $\zeta$ is a generator of the finite
group $\pi_1(G)\cong \Z_N$.  For such $G$ and $\zeta$, there is a unique stable holomorphic vector bundle 
$\V$, and it does obey $H^0(E,\ad(\V))=H^1(E,\ad(\V))=0.$

Thus elliptic solutions of the Yang-Baxter equation (or at least those that we can construct) are classified by a choice of $N\geq 2$ -- determining the group 
$PGL_N$ -- and a generator $\zeta$ of $\Z_N$, or equivalently a primitive $N^{th}$ root of 1.  The automorphism group of such a solution is the automorphism
group of $\V$, which in all cases is $\Z_N\times \Z_N$.   

\subsection{Rigid Holomorphic Bundles}\label{rigid}

We will now describe concretely the rigid holomorphic bundles $\V\to E$ associated a choice of $N$ and $\zeta$.

It is well-known that there exist pairs of $N\times N$ matrices $A,B$ obeying
\be\label{comrel}AB=BA \, e^{2\pi\i/N}\;. 
\ee

Moreover, $A$ and $B$ are unique up to conjugation and multiplication by scalars.  For example, we can pick $A=\mathrm{diag}(1,e^{2\pi \i/N},e^{4\pi \i/N},
\dots,e^{2\pi\i(N-1)/N})$, and $B$ a matrix that cyclically permutes the eigenspaces of $A$.  In the opposite direction, the equation shows that (after
possibly multiplying $A$ by a scalar) the eigenvalues of $A$ are the $N^{th}$ roots of 1, each with multiplicity 1.  So $A$ can be put in the claimed
form, and then it is not hard to see that (up to conjugation and multiplication by a scalar) $B$ must be as claimed.

Now given $\zeta\in \Z_N$, the matrices $A^\zeta$ and $B$ do not commute, but obey $A^\zeta B=BA^\zeta \exp(2\pi \i\zeta/N)$.  But this means that
$A^\zeta$ and $B$ commute if projected to $PGL_N$.  So we can define a flat $PGL_N$ bundle $\V\to E$ whose monodromies around a pair of generators
of $\pi_1(E)=\Z\oplus \Z$ are $A^\zeta$ and $B$.   Being flat, this bundle is automatically stable and holomorphic.  

The automorphism group of this bundle always includes a subgroup $\Z_N\times \Z_N$ generated by $A$ and $B$.  This is the full automorphism group
if $\zeta$ is a generator of $\Z_N$.  (One may prove this starting with the fact that if $\zeta$ is a generator of $\Z_N$,
then the eigenvalues of the monodromy $A^\zeta$ are nondegenerate.)  Accordingly, if $\zeta$ is a generator, then $H^0(E,\V)=H^1(E,\V)=0$, and we are in the
favorable situation that will lead to a solution of the Yang-Baxter equation.

If $\zeta$ is not a generator, then this automorphism group has a strictly positive
dimension and likewise the bundle $\V$ can be deformed (as a holomorphic bundle and even as a flat one).  

To conclude this section, we will describe $\ad(\V)$ as a holomorphic bundle over $E$.  We consider $A^\zeta$ and $B$ as matrices acting by conjugation
on $\g$, the Lie algebra of $PGL_N$.  The eigenvalues are pairs $\mu,\mu'$ of $N^{th}$ roots of 1, with each pair occurring exactly once except the trivial
pair $\mu=\mu'=1$.  Each joint eigenspace of $A^\zeta$ and $B$ corresponds to a holomorphic line bundle over $E$, and these lines bundles are all of order
$N$ because the eigenvalues of $A^\zeta$ and $B$ are of order $N$.  There are $N^2$ equivalence classes of holomorphic line bundle over $E$ of order $N$. 
The trivial line bundle is of order $N$, but there are  $N^2-1$ nontrivial ones, and   $\ad(\V)\to E$ is the direct sum of all $N^2-1$ nontrivial line bundles
of order $N$.  (Thus as a holomorphic bundle, $\ad(\V)$ does not depend on the generator $\zeta$ of $\Z_N$, though its Lie algebra structure does depend on $\zeta$.)

\subsection{\texorpdfstring{Specializing to $N=2$}{Specializing to N=2}}\label{specializing}

To achieve some minor simplifications, we will specialize to $N=2$.  There is now only one choice of $\zeta$.  Accordingly, $\V$ is unique up to isomorphism.
As a holomorphic vector bundle over $E$, $\ad(\V)$ is the direct sum of the three non-trivial line bundles over $E$.  However, we would like to describe
the Lie algebra structure on $\ad(V)$.  

We identify $E$ with the complex $z$-plane modulo $z\cong z+1$ and $z\cong z+\tau$, where the complex number $\tau$ is constrained to have $\mathrm{Im}\,
\tau>0$ and is the modulus of $E$.   A flat line bundle can be described by its monodromies $\alpha$ and $\beta$ under $z\to z+1$ and $z\to z+\tau$.   A flat
line bundle is of order 2 if $\alpha^2=\beta^2=1$.  To describe the 3 non-trivial line bundles of order 2, we take
\be\label{wetak}
(\alpha_i,\beta_i) =\begin{cases} (1,-1)& {\mathrm{for}}~\L_1\;,\cr
(-1,1)&{\mathrm{for}}~\L_2 \;,\cr 
(-1,-1)&{\mathrm{for}}~\L_3 \;.\end{cases}
\ee
Since $(\alpha_1\alpha_2,\beta_1\beta_2)=(\alpha_3,\beta_3)$, there is a natural isomorphism $\phi_{12}:\L_1\otimes \L_2\cong \L_3$.  The same holds
with any cyclic permutations of the labels $123$.   The Lie algebra structure on $\V$ can be defined as follows.  If $s_1$ and $s_2$ are local sections of $\L_1$ and $\L_2$,
then $[s_1,s_2]$ is the local section of $\L_3$ defined by
\be\label{etak}  
[s_1,s_2]=\phi_{12}(s_1\otimes s_2)\;.
\ee
This statement and its cyclic permutations define the Lie algebra structure.

However, it may be useful to describe this structure in a slightly more explicit way.  
The matrices $A$ and $B$ of section \ref{rigid} commute as matrices acting on the Lie
algebra $\mathfrak{sl}_2$.  Their joint eigenfunctions, up to a choice of basis, are the standard generators $t_1,t_2,t_3$ of $\mathfrak{sl}_2$, obeying $[t_1,t_2]=t_3$
and cyclic permutations.  Thus  one might prefer to think of a section of $\ad(\V)$ more explicitly as $\sum_{i=1}^3 s_i t_i$, where $s_i$ is a section of $\L_i$ and $t_i$ is an element of a standard basis of $\mathfrak{sl}(2)$.
In this description, the commutator of $s=\sum_{i=1}^3 s_i t_i$  with   $s'=\sum_{i=1}^3 s'_i t_i$ is $s''=\sum_{i=1}^3 s''_i t_i$ where
\be\label{wetakk} 
s_3''=\phi_{12}(s_1\otimes s_2'-s_2\otimes s_1')\;, \ee
and cyclic permutations.   Thus, one can think of eqn.\ (\ref{etak}) as a formula for the commutator of $s_1t_1$ with $s_2t_2$.

Now let us determine the lowest order nontrivial contribution to the $r$-matrix.   It will be a (meromorphic) section of\footnote{Here $\ad(\V)\boxtimes \ad(\V)$ is simply the tensor product of the $\ad(\V)$ bundle over the first copy of $\Sigma$ with the $\ad(\V)$ bundle
over the second copy; if $\pi_i:\Sigma\times\Sigma\to \Sigma$, $i=1,2$ are the two projections, then $\ad(\V)\boxtimes\ad(\V)=
\pi_1^*(\ad(\V))\otimes \pi_2^*(\ad(\V))$.}   $\ad(\V)\boxtimes \ad(\V)$ over $\Sigma\times \Sigma$.   This section can be determined
by reasoning similar to what we used in section \ref{rma}, though the details are simpler.  

{\it A priori}, a general  form of $r$ would be $r(z_1,z_2)=\sum_{i,j=1}^3 w_{i,j}(z_1,z_2) t_i\otimes t_j$, where
$w_{i,j}(z_1,z_2)$ is a section of $\L_i\boxtimes \L_j$.  However, the $\Z_2\times \Z_2$ automorphism group 
ensures that $w_{i,j}=0$ for $i\not=j$.  Thus we reduce to $r(z_1,z_2)=\sum_{i=1}^3 w_i(z_1,z_2) t_i\otimes t_i$.  Here the functions
$w_i(z_1,z_2)$ are determined by the following (somewhat redundant) set of conditions:
\begin{enumerate} 
\item $r(z_1,z_2)$ has a first-order pole at $z_1 = z_2$,  with residue $\hbar c=\hbar \sum_k t_k\otimes t_k$ 
as in the rational case, since OPE singularities
are local. 
\item $r(z_1,z_2)=-r(z_2,z_1)$ (we need not combine this with an exchange of the two tensor factors of $\mathfrak{sl}_2\otimes \mathfrak{sl}_2$,
since $r=\sum_i w_i t_i\otimes t_i$ is invariant under this switch).
\item Finally, $r(z_1,z_2)$ is a function only of $z_1-z_2$, because of the translation symmetry of
the elliptic curve $E$ and the differential $\omega=\d z$. \end{enumerate}

The functions $w_k(z_1,z_2)$ that satisfy these conditions actually have a simple interpretation.  Since the canonical bundle $K$ of $E$ is
trivial, its possible square roots $K^{1/2}$ are line bundles of order 2, and thus we can think of $\L_k$, for $k=1,2$ or 3, as
 one of the three even spin bundles
of $\Sigma$, that is, as representing a possible $K^{1/2}$.  Let $\psi_k$ be a (holomorphic) fermion field on $\Sigma$ valued in this spin
bundle, with action
\be\label{flaction} I=\frac{1}{2\pi}\int_E \psi_k\bar\partial \psi_k \;.\ee
Then the two-point function $\langle \psi_k(z_1)\psi_k(z_2)\rangle$ satisfies precisely the conditions of the functions $w_k(z_1,z_2)$.

With or without this interpretation, it is straightforward to write a formula for $w_k$:
\be\label{naction}
w_k(z_1,z_2)=\sum_{n,m\in\Z \setminus{\{(0,0)\}}}\frac{ \alpha_k^n\beta_k^m}{z_1-z_2-n-m\tau}\;. 
\ee
The factors of $\alpha_k^n\beta_k^m$ ensure that $w_k(z_1,z_2)$ is a section of $\L_k\boxtimes \L_k$, and the desired
properties are all obvious.
We have defined $w_k$ so that for $z_1\to z_2$ it behaves as $1/(z_1-z_2)$.  Thus, taking the normalization from the rational case, 
the $r$-matrix is 
\be\label{merifo}
r=\hbar \sum_{k=1}^3 w_k(z_1-z_2)t_k \;. 
\ee
This coincides with the known expression in the literature, see e.g.\ \cite[p.\ 539]{Faddeev-Takhtajan}.

\section{The Dynamical Yang-Baxter Equation}\label{dybe}

\subsection{Reduction to an Abelian Subgroup}\label{abred}

Our goal in this section is to take a first look at the case that $G$ and $C$ and the other relevant choices are such as to lead to a moduli space
of classical solutions, not just a single isolated classical solution.   This makes matters fundamentally more difficult and our analysis here will be preliminary.

A basic case to bear in mind is that $C$ is an elliptic curve $E$, $G$ is simple,  and the topological 
invariant $\zeta$ introduced in section  \ref{Preliminaries} 
vanishes,
so that  the $G$-bundle $\V\to E$ is topologically trivial.  Under this condition,
 the structure group of a stable holomorphic $G$-bundle $\V\to E$ reduces always to the (complex) maximal torus $T\subset G$.
The moduli space $\M$ of stable holomorphic $G$-bundles over $E$ is then
simply the corresponding moduli space $\M'$ for $T$, divided by the Weyl group  $\W$: $\M=\M'/\W$.
The automorphism group of a generic stable holomorphic $G$-bundle is simply $T$, though the symmetry is enhanced at points in $\M'$ at which the
Weyl action is not free.  For our purposes, we can just identify the moduli space of classical solutions with $\M'$ and ignore the Weyl 
group action and possible symmetry enhancement.
(That is because we will eventually impose a constraint that keeps us away from special points.)

For $C$ an elliptic curve, the story is similar\footnote{More exotic examples with moduli, such as 
the case $\l_0+\l_\infty\not=\g$ mentioned in section \ref{Preliminaries},
might be more complicated.} for any $\zeta$, with $T$ replaced by another torus $T_\zeta \subset T$, $\M'$ replaced with the moduli space $\M'_\zeta$
of $T_\zeta$-bundles over $E$, and the Weyl group replaced by a subgroup.  Rather than explain these details, we will simply continue with the case
that $\zeta=0$. 

A complex line bundle over $E$ can be represented by a very simple gauge field $A_{\bar z}=b$, with $b$ a complex constant.  Here $b$ is subject
to some equivalences, but they will not be important in what follows.  Similarly, for gauge group $G$, a classical solution whose structure group reduces to 
$T$ can be represented by $A_{\bar z}=b$, where now $b$ is a $\mathfrak t$-valued constant.  For example, if $G=SL_N$, this means that
$A_{\bar z}=\mathrm{diag}(b_1,b_2,\dots,b_N)$, with $\sum_{i=1}^N b_i=0$.   In particular, this classical solution is invariant under the translation
symmetries of $E$.

In general in nonabelian gauge theory with gauge group $G$, when we expand around a classical solution whose structure group reduces to a proper
subgroup $H$ of $G$, we should look for a low energy description in the form of an $H$ gauge theory.  In the present case, we have the further fact
that to study the Yang-Baxter equation or its analog, we can take the topological plane $\Sigma$ to be simply a plane $\R^2$, and thus we work
on $\R^2\times E$ where $E$ is compact and $\R^2$ is not.  In this situation, in quantum field theory in general, we can look for a low energy effective
field theory on $\R^2$.  Our present problem is actually diffeomorphism-invariant in the $\R^2$ direction, so the low energy description is really  a topological field theory.

Combining these facts, we can aim to find an effective description in the form of a two-dimensional topological gauge theory with structure group $T$.
To find this description at least at a formal level, we make use of the translation symmetry of the classical solutions and throw away ``massive modes''
that lack this symmetry. Likewise, we keep only the $\mathfrak{t}$-valued part of the four-dimensional gauge field $A$.  This means that $A_{\bar z}$
is reduced to the $\mathfrak{t}$-valued field $b$, which now depends on coordinates $x$ and $y$ of $\R^2$ but not on $z$ and $\bar z$,
and likewise the rest of the gauge field, namely $A_x \d x +A_y\d y$, becomes a purely two-dimensional gauge field that we will denote simply as $A$.

Specializing the underlying action (\ref{eq.action}) to this situation, we get an effective two-dimensional abelian action
\be\label{abac} S=- \kappa\int_\Sigma \Tr\, b F \;, \ee
where $\Tr $ is the restriction to $\mathfrak t$ of the quadratic form of the same name on $\mathfrak g$, $F=\d A$ is the two-dimensional gauge
curvature, and $\kappa=\i\mathrm{Im}\,\tau/\pi \hbar$.   This is the action that we will use
somewhat formally in analyzing the present problem.

The reader may object that in
quantum field theory in general, it is not correct to simply set the massive modes to zero, as we have done.  Instead, one has to integrate them out,
producing in general corrections to the effective action of the fields that are retained in the effective description.   Moreover, if we wish to include
Wilson loops and their crossings, we have to integrate out the massive fields in the presence of those operators.  All this is true, but the implications for 
our problem are limited because two-dimensional diffeomorphism symmetry severely constrains what couplings can be generated by integrating out massive modes.
If $\Sigma$ is not $\R^2$ but is a curved two-manifold with scalar curvature $R$, then it is possible to generate a coupling $\int_\Sigma\d^2x \sqrt g R f(b)$,
for some function $f(b)$.  One should expect such a term, but it will not be important for our purposes because in analyzing the crossing of line operators,
one can assume that $\Sigma$ is flat.  In the presence of line operators, integrating out the massive modes will generate the framing anomaly of $\mathfrak{g}$.
That is an important effect, but we are already familiar with it and will not discuss it further in the present section.  Finally, integrating out the massive
modes certainly affects the $R$-matrix that governs crossing of two line operators.  But here we will just discuss the formal properties of this
$R$-matrix, and will not try to calculate it.  So we will not have to explicitly discuss the contribution of the massive modes to the $R$-matrix.

If $A$ is regarded as a two-dimensional gauge field with structure group a {\it compact} torus $T$, and similarly $b$ is real-valued, taking
values in the corresponding Lie algebra $\mathfrak t_c$, then
eqn. (\ref{abac}) becomes the action of a simple but much-studied two-dimensional topological field theory, often called $BF$ 
theory.\footnote{In that context, $\kappa$ is real, assuming we want a unitary theory.
  In the context of reduction from the four-dimensional theory, $\kappa$ is a complex number, since $\hbar$ is complex.}  We are
not quite in that situation, since in \eqn (\ref{abac}), $b$, $A_x$, and $A_y$ are all complex-valued fields and the action is a holomorphic function of those
fields.  Moreover, we are in a situation in which there is a nontrivial moduli space of classical solutions (parameterized locally by $b\in \mathfrak t$), so if we want
to define correlation functions,
we cannot just do perturbation theory in a formal way; we need some recipe for what to do with the moduli.  One point of view that may seem somewhat formal
but that seems to provide a satisfactory framework for our considerations below (for instance, the derivation of the dynamical Yang-Baxter equation makes sense in this framework)
is to view the theory as a machine that generates a differential form on the moduli space of classical solutions, without worrying about how to integrate it.  Alternatively,
it might be possible to define an operator  that fixes the values of the moduli.  This idea is explained at the end of section \ref{critique}.   Finally, in principle one may use
the D4-NS5 system to define a nonperturbative integration cycle for the underlying four-dimensional theory (\ref{eq.action}).  Specializing this
to the present situation, one would then learn in principle what to do with the moduli in the effective theory (\ref{abac}).  

In practice, we will simply draw inspiration from conventional two-dimensional $BF$ theory.\footnote{The purely two-dimensional
analysis that follows is similar to section 2.5 of \cite{WittenTwoDim}.}  Based on this, we will
suggest a simple procedure  to make contact with the dynamical Yang-Baxter equation \cite{GN,Felder,FelderTwo,Etinghof}.    Given that we assume
 $\zeta=0$, the solution
of the dynamical Baxter equation that arises in this way will have $T$ as a group of automorphisms, since a generic classical solution with 
$\zeta=0$ has automorphism group $T$, and we will be avoiding the exceptional cases.

\subsection{\texorpdfstring{$BF$ Theory and the Dynamical Yang-Baxter Equation}{BF Theory and the Dynamical Yang-Baxter Equation}}\label{bf}

The classical equation of motion of the gauge field $A$ in $BF$ theory, in the absence of line operators, simply says that $\d b=0$.  Thus $b$ is simply constant.

What happens in the presence of a Wilson line operator?  To start with, we take the gauge group to be $G=U(1)$.  An irreducible representation
of $U(1)$ is determined by the choice of an integer $n$, the ``charge.'' The Wilson operator supported on a curve $K$ for the representation
of charge $n$ is $\exp\left(\i n\int_K A\right)$.  Including
this factor, the argument of the path integral becomes $\exp(\i S)\exp\left(\i n\int_K A\right)=\exp(\i S')$, with
\be\label{seld}S'=-\kappa\int_\Sigma \Tr\, b F+ n\int_K A\;.\ee
In other words, the Wilson operator for $U(1)$ effectively contributes an additional term to the classical action.  The classical equation of motion now
becomes
\be\label{eid}\kappa \d b+ n\delta_K=0\;,\ee where $\delta_K$ is a one-form delta function Poincar\'e dual to $K$.  
The import of this is that $b'=\kappa b$ jumps by $n$ in crossing $K$ from right to left (\fig \ref{bfcross}(a)).  

Now we can consider the case of two Wilson operators crossing.   But in doing so, we may as well generalize\footnote{The purpose
of discussing this generalization here is to give the simplest possible motivation for the picture of \fig \ref{bfcross}(b). A similar generalization is
 not
possible in the context of the four-dimensional theory, assuming that the lines that are crossing have different values of the spectral parameter. 
Such a generalization is possible in two-dimensional $BF$ theory with a nonabelian gauge group, though we will not need that case.}
 beyond a simple crossing of two
line operators of charges $n_1$ and $n_2$.  Since we are discussing a purely two-dimensional theory, any crossing really does involve a physical
intersection of the two line operators, and charge exchange is possible.  In general, we can consider a case with charges $n_1$ and $n_2$ coming
in and charges $n_3$ and $n_4$ going out, the only constraint being that $n_1+n_2=n_3+n_4$.  The behavior of $b'=\kappa b$ is then as sketched
in \fig \ref{bfcross}(b).

\begin{figure}[htbp]
\centering{\includegraphics[scale=1]{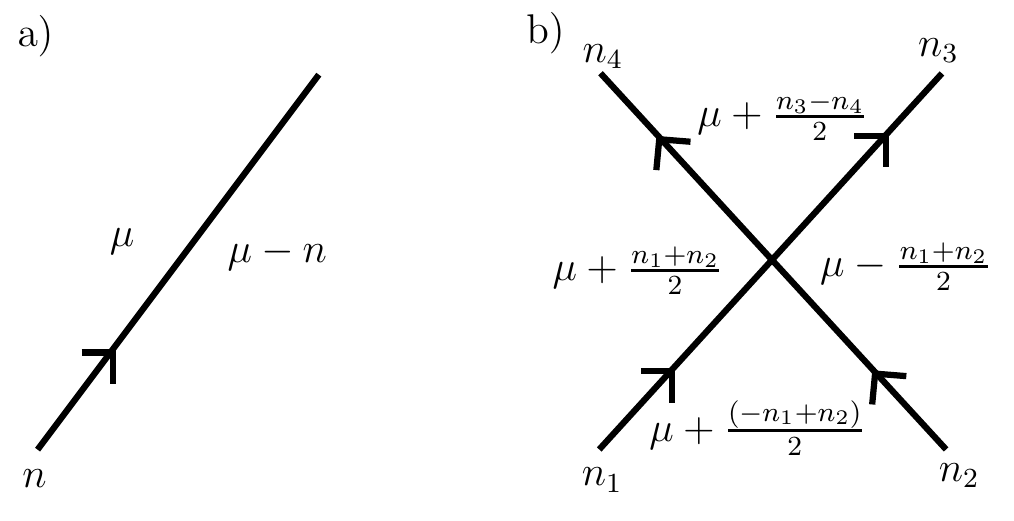}}
\caption{\small{(a) The field $b'=\kappa b$ is constant away from Wilson operators and jumps by $n$ in crossing a charge $n$ Wilson operator from
right  to left.
In the example shown, $b'$ jumps from $\mu-n$ to $\mu$.
(b) Two Wilson operators of charges $n_1$ and $n_2$ come in to this intersection and  two of charges $n_3$ and $n_4$ go out
(here $n_1+n_2=n_3+n_4$).  The field $b'$ jumps as shown.}}
\label{bfcross}
\end{figure}

\begin{landscape}
\begin{figure}[htbp]
\centering{\includegraphics[scale=1]{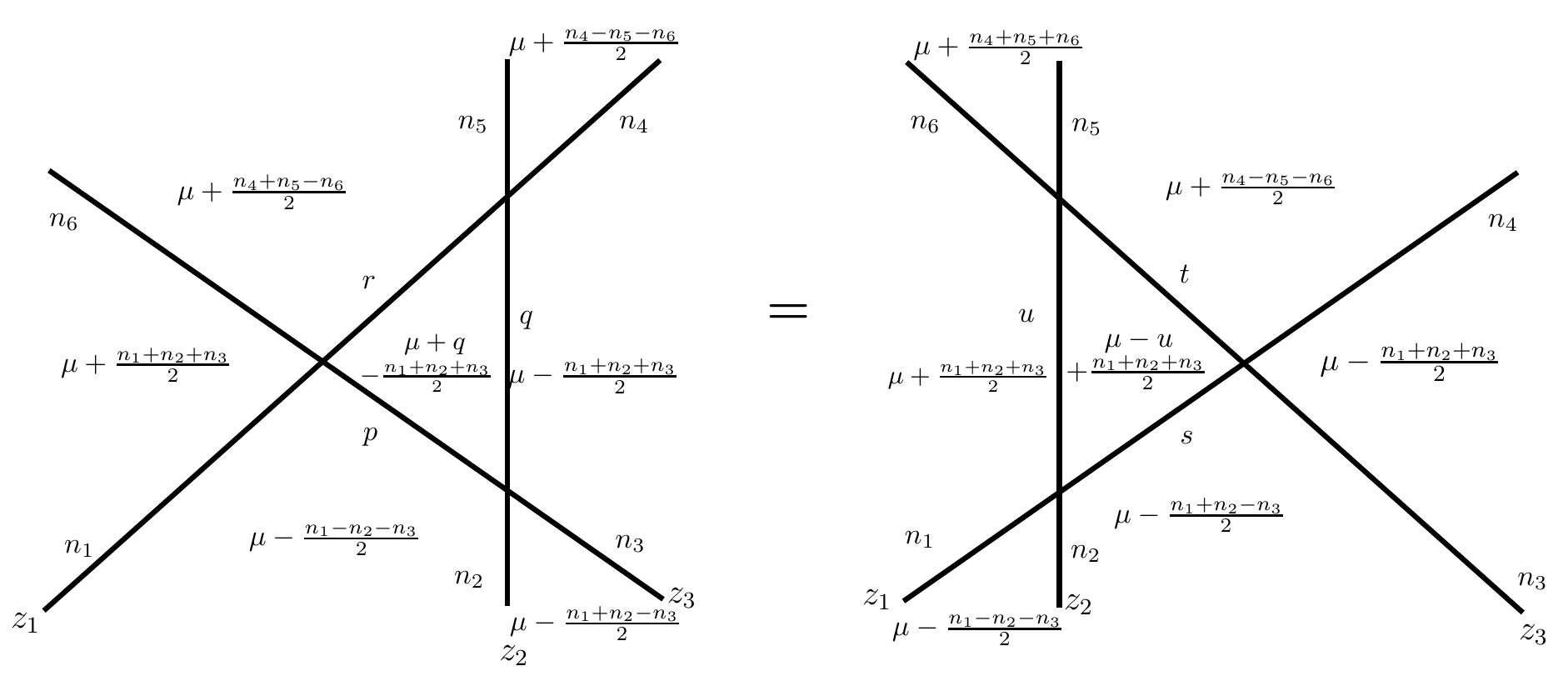}}
\caption{The dynamical Yang-Baxter equation.  
Both line segments and bulk regions are labeled.
A line segment is  labeled by a basis vector in a representation of $G$ attached to that line, and a bulk  region is labeled
by a weight of $G$.  The picture is drawn here for $G=SU(2)$, so that a basis vector is determined by its integer weight.
 Two-dimensional $BF$ theory has all the ingredients
to draw such a picture (without spectral parameters) but there is no reason to expect this equivalence
to hold.  When we go to four dimensions, the lines
are labeled by spectral parameters  in the familiar way (as shown here) and the usual arguments show an equivalence
between the two pictures.}

\label{dYB}
\end{figure}
\end{landscape}

Now we consider the case of two-dimensional $BF$ theory with a nonabelian gauge group.  First let us take the gauge group to be $SU(2)$.
As long as the field $b$ is nonzero, the gauge symmetry is spontaneously broken from $SU(2)$ to $U(1)$.  As a result, apart from possible subtleties
when $b$ vanishes, we can aim for a description of $BF$ theory of $SU(2)$ in terms of $BF$ theory of $U(1)$.
It is actually possible (see section 2.3 of \cite{BT}) to compute rather precisely the effective $U(1)$ action that arises from integrating out
the charged modes.\footnote{The relevant procedure also has an analog in three-dimensional Chern-Simons theory on a Seifert manifold \cite{BTtwo}.}
We omit this, as the results are not very essential for us, as explained in section \ref{abred}.  (However, we describe some aspects of a precise
quantum treatment in section \ref{critique}.) What is important for us is that a line operator
of the $SU(2)$ theory can be written as a sum of line operators of the effective $U(1)$ theory.  For Wilson line operators, this decomposition is
fairly evident.  
  A Wilson operator associated to an irreducible representation $\rho$ of $SU(2)$ decomposes in the $U(1)$ description
as a sum of line operators  corresponding to the weight spaces of $\rho$.  
Let us denote a Wilson operator for the representation $\rho$ of $SU(2)$ as $W_\rho$ and one for the charge $n$ representation of $U(1)$ as $W'_n$.
 Then, for example, if $\rho$ is the two-dimensional 
representation of $SU(2)$,  it decomposes under $U(1)$ as the sum of one-dimensional weight spaces of charges $1$ and $-1$, so the corresponding
formula for Wilson operators is $W_\rho=W'_1+W'_{-1}$.

Now we can reconsider \fig \ref{bfcross}(b), with a slightly different interpretation.  We consider the crossing of two Wilson operators of $SU(2)$
in some representations $\rho$ and $\rho'$.   On either of the two Wilson operators, on either side of the crossing, we can replace $W_\rho$ or
$W_{\rho'}$ by a sum of $U(1)$ Wilson operators, labeled by the weight spaces in $\rho$ or $\rho'$.  Thus the integers $n_1,\dots,n_4$ in \fig
 \ref{bfcross}(b) are now not arbitrary but denote weights of $\rho$ or $\rho'$.    In passing through the point at which the two Wilson
 operators cross, the charges obey $n_1+n_2=n_3+n_4$ as before, because the effective description has $U(1)$ symmetry.

The case of a compact nonabelian gauge group $G$ of any rank $r$ is similar. Away from special values of $b$, the $G$
gauge symmetry is broken to the maximal torus $T$.  The charges of an irreducible  representation of $T$ are now an $r$-plet of integers
$\vec n$.  Likewise $b'=\kappa b$ becomes an $r$-plet $\vec b'$.  Its jumping in crossing a Wilson operator is $\vec b'\to \vec b'+\vec n$,
just as before.  One detail is slightly different.
For $r>1$, the weight spaces of an irreducible representation may have multiplicity greater
than 1.  (For example, the adjoint representation of $SU(3)$ has six weight spaces each of multiplicity 1 and one of multiplicity 2.)   Accordingly,
when we write $W_\rho$ as a sum of Wilson operators of an effective $T$ gauge theory, some charges may appear more than once.  So in general
away from crossings, an effective Wilson operator of the low energy theory carries some additional labels in addition to its ``charges'' $\vec n$.

Finally we come to the main point.
In actually calculating a path integral appropriate to the situation of \fig \ref{bfcross}(b),
 we would run into a factor associated to the crossing.  This factor would depend on the charges that label the
incoming and outgoing lines in the effective abelian theory 
and on the parameter $\mu$ that determines the labels of the bulk regions in the figure.  We can denote this
factor as a generalized $R$-matrix element $R_{n_1n_2}^{n_3n_4}(\mu)$.    (For simplicity, in writing the $R$-matrix in detail as a matrix,
 we take $G=SU(2)$ so that each line is simply labeled by an integer.)

Having defined the generalized $R$-matrix of two-dimensional $BF$ theory, 
one can ask if it obeys a generalized Yang-Baxter equation, now with bulk labels as well as
labels for line segments between crossings (\fig \ref{dYB}).
However, there is no reason to expect this.  In a purely two-dimensional theory, one has to pass through a singularity
to interpolate between the left and right hand sides of the figure.  As far as we know, the generalized Yang-Baxter equation is not
satisfied in this situation.

Hopefully the reader can anticipate what comes next.  We consider not a purely two-dimensional $BF$ theory, but an effective abelian $BF$ theory
in two dimensions that arises as in the discussion of eqn.\ (\ref{seld}) from a four-dimensional theory on $\R^2\times E$ with complex gauge group $G$
and with $\zeta=0$.  A Wilson operator is now labeled by a representation $\rho$ of $G$ (or more generally by a representation of
a quantum deformation of $\g[[z]]$)
and also by a spectral parameter $z\in E$.  Between crossings, a Wilson operator is further labeled by a basis vector of $\rho$.
The bulk parameter $b'=\kappa b$ jumps in crossing a Wilson operator of the low energy theory.
Now, we associate to a crossing a generalized $R$-matrix element $R_{n_1n_2}^{n_3n_4}(z_1-z_2;\mu)$.   The difference from before as that as long
as $z_1\not= z_2$, there is no singularity associated to a crossing.
Therefore the two sides of \fig \ref{dYB} (with spectral parameters now included) are equivalent.  The equivalence is known
as the dynamical Yang-Baxter equation.

\subsection{Quantum Treatment}\label{critique}

We have treated the parameter $b'$ classically, starting with purely two-dimensional $BF$ theory.  Actually, the locally constant value of $b'$
has a natural meaning in the quantum theory.

Let us ask what are the quantum states when $BF$ theory -- to begin with for $G=U(1)$ -- is quantized on a circle $S$.  A physical state is a 
gauge-invariant function of the connection $A$ restricted to $S$.  Such a function is $\Psi_n=\exp(\i n \oint_S A)$ for some integer $n$.
On the other hand, $b'$ is the momentum conjugate to $A$, so (for $p$ a point in $S$), $b'(p)$ can be identified with
$-\i \delta/\delta A(p)$.  Explicitly acting with this on $\Psi_n$, we find that $b'(p)\Psi_n=n\Psi_n$, for any $p$, independent of $p$.

Quantum mechanically, the constant value of $b'(p)$ should be interpreted as an eigenvalue of this operator, and (for $U(1)$) the eigenvalues are
integers, as we have just seen.    We can also understand in this language why crossing a Wilson line operator has the effect of shifting
$b'$ by an integer.  For example, we can regard the Wilson line operator $W_m=\exp(\i m\oint_S A)$ as an operator that acts on physical
states on $S$.  Since $W_m\Psi_n=\Psi_{n+m}$, acting with $W_m$ shifts the value of $b'$ by $m$.  

Now let us repeat this analysis for a compact but possibly nonabelian gauge group $G$.  Of course,
the analog of $\Psi_n$ is what we might call $\Psi_\rho$, the trace of the holonomy around $S$ in the representation $\rho$:
\be\label{zorb}\Psi_\rho =\Tr_\rho \,U,~~U=P\exp\left(\oint_S A\right)\;. \ee
However, we would prefer to express this in a language that is better suited for the reduction to an effective abelian description.  For this,
we observe that a gauge-invariant function of $A$ is precisely a function of the holonomy $U$ that is invariant under conjugation.   In other words,
gauge-invariant functions of $A$ are functions on $G/G$, where $G$ acts on itself by conjugation.

The quotient $G/G$ is the same as $T/\W$, the quotient of the maximal torus $T$ by the Weyl group $\W$.  This might lead one to expect
that gauge-invariants functions of $A$ would correspond to Weyl-invariant functions on $T$, but actually they correspond in a natural way
to Weyl {\it anti}-invariant functions.  We say that a function on $T$ is Weyl anti-invariant if it is odd under each of the elementary reflections
that generate $\W$.  The association of a representation $\rho$ of $G$ with a Weyl anti-invariant function on $T$ is given by the Weyl character
formula; the Weyl anti-invariant function corresponding to $\rho$ is the numerator of the usual Weyl character formula for $\rho$.  Let us just
explain what this means for $SU(2)$.  A maximal torus of $SU(2)$ is the $U(1)$ subgroup $\mathrm{diag}(e^{\i\theta},e^{-\i\theta})$.  The character
of the $n$-dimensional representation $\rho_n$ of $SU(2)$ is
\be\label{worb} e^{\i(n-1)\theta}+e^{\i(n-3)\theta}+\dots + e^{-\i(n-1)\theta}=\frac{\sin( n\theta)}{\sin\theta}\;. \ee
The Weyl anti-invariant function corresponding to $\rho_n$ is the numerator, or 
\be\label{porb}\sin(n\theta)=\frac{1}{2\i}\left(e^{\i n\theta}-e^{-\i n\theta}\right)\;. \ee

The functions $\sin(n\theta)$, $n=1,2,3,\dots$ are a basis for the Hilbert space of Weyl anti-invariant functions on $T$.  This is the Hilbert space
of $BF$ theory of $SU(2)$, quantized on a circle.  In this description, we can conveniently see the effective $U(1)$ description of $SU(2)$ $BF$ theory.
In this effective $U(1)$ description, the holonomy is $e^{\i\theta}$ and $b'$ is the canonical momentum $b'=-\i\partial/\partial\theta$.  
We see that $\sin(n\theta)$ is not an eigenstate of $b'$ but rather is a linear combination of eigenstates $e^{\i n\theta}$ and $e^{-\i n\theta}$ with
eigenvalues $\pm n$.  These values are Weyl conjugate, since the Weyl group of $SU(2)$ is $\Z_2$, acting as $-1$ on the Lie algebra $\mathfrak t$
of $T$.  The implication is that $b'$ should not really be regarded as an integer (or a real
number, as in the classical description) but as a Weyl orbit of nonzero integers.  It is fairly natural to pick from each Weyl orbit $n,-n$ the positive
representative, and if we do this then the values of $b'$ in the effective abelian description are positive integers.   For any compact simple
$G$, the analog is that the values of $b'$ are dominant weights in the interior of a positive Weyl chamber.

Now let us see what happens to the value of $b'$ in crossing a Wilson operator, say the operator $W_m$ associated to the $m$-dimensional
representation $\rho_m$.  The character of $\rho_m$ is $F_m(\theta)=e^{\i(m-1)\theta}+\dots + e^{-\i (m-1)\theta}.$   Crossing $W_m$ has the
effect of multiplying the quantum state by $F_m(\theta)$.  We have
\be\label{incov}F_m(\theta)\sin(n\theta)=\sum_{j=-m+1,-m+3,\dots, m-1}\sin((n+j)\theta)\;,  \ee
and if $m\leq n$, then $n+j$ is always positive.  
This means that in crossing $W_m$, $b'$ can jump by $j$ for any $j=-m+1,-m+3,\dots,m-1$, that is,  any weight of the representation $\rho_m$.
This is the result that was claimed in section \ref{bf}.  For $m>n$, it is possible for $n+j$ to be nonpositive, and some terms on the right hand
side of eqn.\ (\ref{incov}) vanish or cancel.  This leads to some modification of the formalism when $b'$ is not large.

These subtleties do not really affect the discussion in section \ref{bf} very much.  In that discussion, $b'$ was treated as a generic real number,
but it would not have been much different to regard $b'$ as a generic positive integer, where here ``generic'' is equivalent to ``sufficiently large.''
Thus, in studying any concrete collection of Wilson operators, associated with representations of dimensions $m_1,\dots, m_s$, the
reasoning in section \ref{bf} is valid if $b'$ is sufficiently large compared to those dimensions.  

This quantum treatment of purely two-dimensional $BF$ theory, however, highlights what is missing in our understanding of the four-dimensional
theory.  We certainly do not have available a quantum treatment that would identify definite allowed values of $b'$ in the
theory obtained by compactification from four dimensions.  
In fact, since the four-dimensional theory is a somewhat formal construction with a holomorphic action, it is not clear to what extent one should
expect to have such a quantum treatment.

Neither -- at least at first sight -- do we wish to integrate over the possible values of $b'$, as a natural cycle for such an integral does not present itself.
In section \ref{abred}, we already mentioned several possible ways to deal with this issue.  Here we just elaborate on one possibility.
In purely two-dimensional $BF$ theory on a two-manifold $\Sigma$, it is fairly natural to pick a point $p\in\Sigma$, not in the support
of any Wilson operator, and specify the value of $b'$ there.  The values of $b'$ elsewhere would then be determined by the crossing rules of
section \ref{bf}.  It is tempting to believe that a similar constraint should be imposed in the four-dimensional theory, to spare ourselves
from having to sum or integrate over $b'$. Moreover, a constraint setting $b'$ to a generic value at some particular point $p$,
when supplemented with the crossing rules, will ensure that $b'$ never takes a value at which the automorphism group is enhanced
and the effective abelian description breaks down. However, we have not seriously tried to study a quantum field theory operator that would impose
this constraint.


\section*{Acknowledgments}
The authors would like to thank Davide Gaiotto, Robbert Dijkgraaf, Michio Jimbo, Cumrun Vafa, and especially Jacques H.~H.~Perk for discussions.
K.~C.~ is supported by the NSERC Discovery Grant program and by the Perimeter Institute for Theoretical Physics. Research at Perimeter Institute is supported by the Government of Canada through Industry Canada and by the Province of Ontario through the Ministry of Research and Innovation. 
E.~W.~ is partially supported by National Science Foundation grant NSF Grant PHY-1606531.
M.~Y.~ is partially supported by WPI program (MEXT, Japan), by JSPS Program for Advancing Strategic International Networks to Accelerate the Circulation of Talented Researchers, by JSPS KAKENHI Grant No.\ 15K17634, and by JSPS-NRF Joint Research Project.  


\appendix

\section{\texorpdfstring{Rational $R$-matrix for $SO_N$}{Rational R-matrix for SO(N)}}\label{app.R_SO}

In this Appendix we discuss the rational R-matrix for  the
fundamental representation of $G=SO_N$, extending the similar discussion for 
fundamental and anti-fundamental representations for $G=GL_N$ (or $SL_N$) in section \ref{elementary}.
While the result in itself is known since the old work of \cite{Zamolodchikov*2},
in the framework of this paper it is an illuminating exercise to check the consistency between unitarity, crossing and the framing anomaly.

Let us consider the R-matrix $R(z_1-z_2): V\otimes V \to V\otimes V$
for the fundamental representation $V$ of $G=SO_N$.   Since this representation
is equivalent to its own conjugate, a Wilson line in this representation carries no orientation.

As we discussed in the main text, the rational R-matrix has $G$ as a symmetry.
Consequently, the R-matrix is constrained to be of the form
\begin{align}
R_{ij}^{i' j'}(z)=  \delta_{i}^{i'} \delta_j^{j'} E(z) +
\delta_i^{j'} \delta_j^{i'}F(z)  +  \delta_{ij} \delta^{i'j'} G(z) \;,
\label{R_SON}
\end{align}
with three unknown functions $E(z), F(z)$ and $G(z)$. 
In this literature this is also written as 
\begin{align}
R_{ij}^{i' j'}(z)=  E(z) I+F(z) P + G(z) Q \;,
\label{RPQ}
\end{align}
where $P$ is a permutation, and $Q=^t\!\!P$ its transposition.
Compared to the case of $GL_N$ in 
eqn.\ (\ref{zelf}), we here have one extra structure  $\delta_{ij} \delta^{i'j'}$ consistent with $G$-symmetry.

The Yang-Baxter equation constrains only the ratios of the three functions:
$X(z)=F(z)/E(z), Y(z)=G(z)/E(z)$. As we can see from \fig \ref{fig.SO_YBE},
 the resulting constraints takes precisely the same form
as in \eqref{welf} and \eqref{nelf}, where the role of $U(z)$ and $W(z)$ are played by 
$X(z)$ and $Y(z)$ respectively:
\begin{align}
\begin{split}
X(z_1-z_3)X(z_2-z_3)+X(z_1-z_2)X(z_1-z_3) &= X(z_1-z_2)X(z_2-z_3) \;  ,\\
Y(z_1-z_3)X(z_2-z_3)=Y(z_1-z_2)Y(z_1-z_3) &+ Y(z_1-z_2)X(z_2-z_3)  \;.
\end{split}
\label{SO_XY}
\end{align}
By the similar reasoning as before 
we obtain
\begin{align}
X(z)=\frac{\hbar}{z}\; ,  \quad
Y(z)=-\frac{\hbar}{z-c\hbar }\;,
\end{align}
where $c$ is some undetermined constant.

\begin{figure}[htbp]
\centering\includegraphics[scale=0.43]{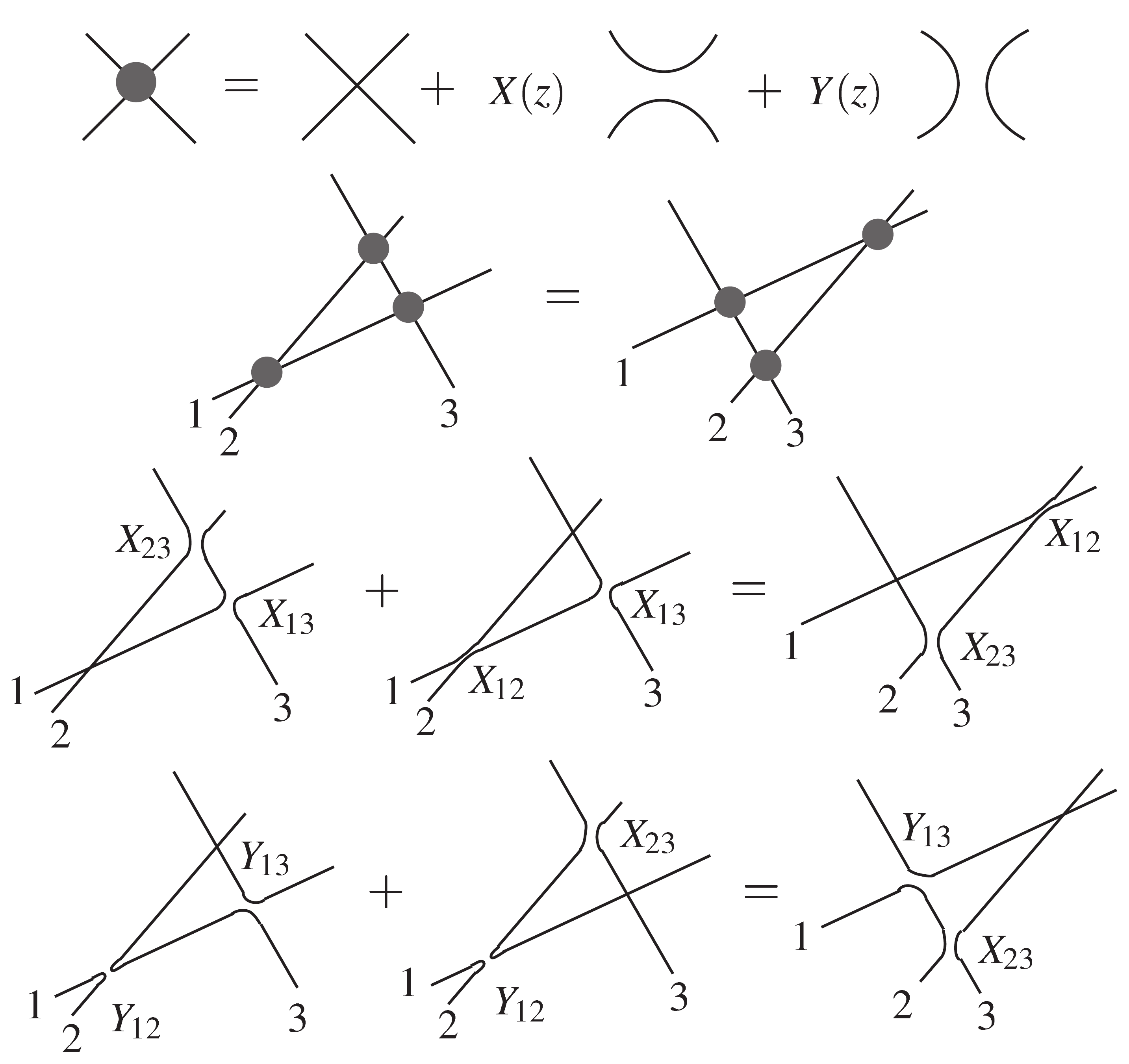}
\caption{\small{Graphical representation for the Yang-Baxter equation for the $SO_N$ $R$-matrix
given in \eqn \eqref{R_SON}. Here we used the shorthand notation $X_{ij}=X(z_i-z_j)$ and $Y_{ij}=Y(z_i-z_j)$, 
with $z_i$ associated with the $i$-th line. The two equations here give eqns.\ \eqref{SO_XY}.}}
\label{fig.SO_YBE}
\end{figure}

The undetermined constant $c$ is constrained by unitarity (recall \eqn \eqref{zunitarity}), which gives the relation
(see \fig \ref{fig.SO_YBE}; note the difference from eqn. \eqref{lefot}):
\begin{align}
N Y(z) Y(-z) +Y(z)+Y(-z)+ X(z) Y(-z) +  X(-z) Y(z) =0 .
\end{align}
This leads to
\begin{align}
X(z)=\frac{\hbar}{z} \;,  \quad
Y(z)=-\frac{\hbar}{z+\frac{N-2}{2}\hbar }\;.
\end{align}

\begin{figure}[htbp]
\centering\includegraphics[scale=0.4]{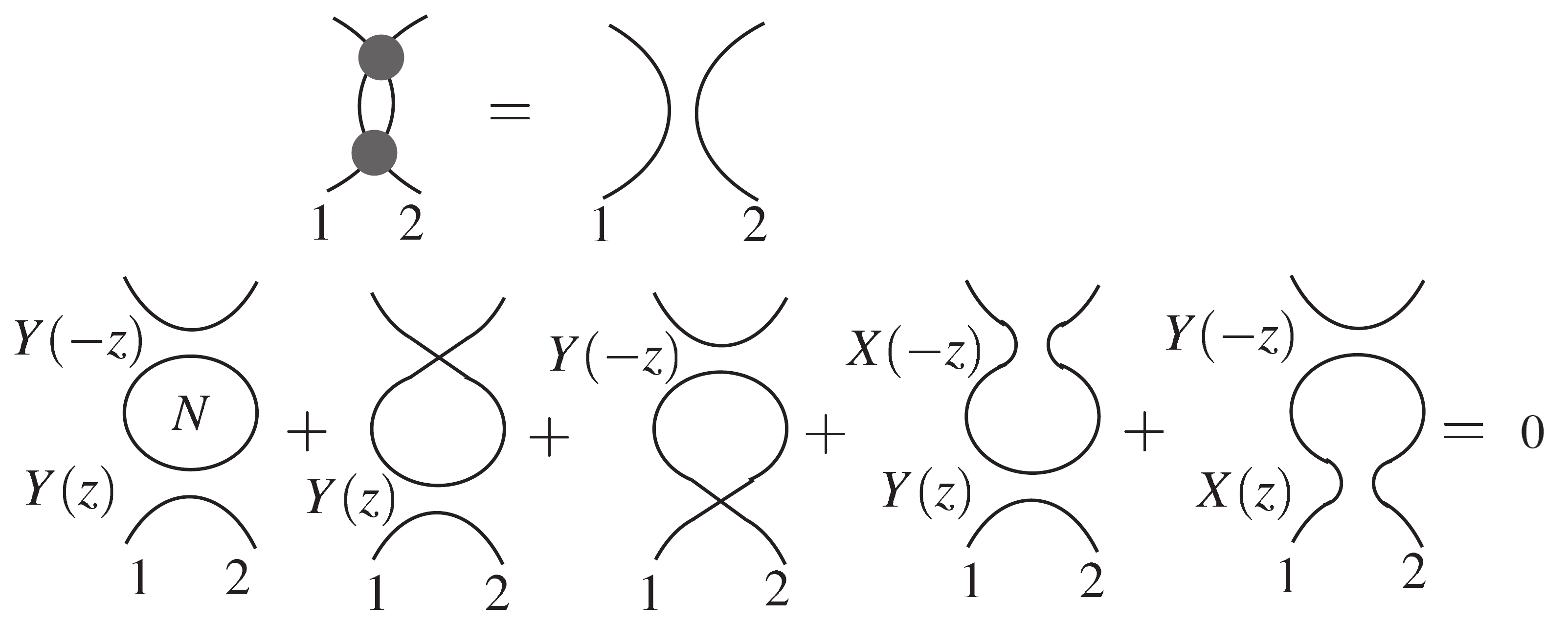}
\caption{\small{Graphical representation for the unitarity relation for the $SO_N$ $R$-matrix
given in \eqn \eqref{R_SON}, with $z=z_1-z_2$. The factor of $N$ comes from the 
color factor for the bubble in the first term of the equation.}}
\label{fig.SO_unitarity}
\end{figure}

These two functions satisfy
\begin{align}
X\left(z-\hbar \frac{N-2}{2}\right)=Y(-z) \;,
\end{align}
which is consistent with the crossing and framing anomaly, since
$N-2$ is the 
 dual Coxeter number of $SO_N$.

We can also try to determine the overall factor $E(z)$.
The diagonal component of the unitarity relation gives
\begin{align}
E(z) E(-z)+ F(z) F(-z) = E(z) E(-z) (1+ X(z) X(-z))=1 \;,
\end{align}
or equivalently 
\begin{align}
E(z) E(-z)=\frac{z^2}{z^2-\hbar^2} \;.
\label{Ez_1}
\end{align}
This solution can easily be solved, for example, by
\begin{align}
E(z)=\frac{z}{z \pm \hbar}\;.
\end{align}
However, this is actually not the solution we want, since 
we also want to impose the condition that this factor does not spoil the unitarity relation:
\begin{align}
E\left(z-\hbar \frac{N-2}{2}\right)=E(-z) \;,\quad {\rm or} \quad
E\left(-\hbar \frac{N-2}{2}-z\right)=E(z)\;.
\label{Ez_2}
\end{align}

The minimal solutions for \eqref{Ez_1} and \eqref{Ez_2} are
\begin{align}
E(z)=\frac{ Q_{\pm \hbar, -\hbar(N-2)/2}(z) }{ Q_{0, -\hbar(N-2)/2}(z)}\;,
\label{minimal_E}
\end{align}
where the function $Q_{\alpha, \beta}(z)$, 
satisfying the functional identities
\begin{align}
Q_{\alpha, \beta}(z) Q_{\alpha, \beta} (-z)=\frac{1}{\frac{(z+\alpha)}{2\beta}\frac{(-z+\alpha)}{2\beta}} \;,
\quad
Q_{\alpha, \beta}(z)=Q_{\alpha, \beta}(\beta-z)\;,
\end{align}
is written as
\begin{align}
&Q_{\alpha, \beta}(z)=
\frac
{ \Gamma\left( \frac{z+\alpha}{2\beta} \right)  \Gamma\left( \frac{-z+\alpha+\beta}{2\beta} \right)}
{ \Gamma\left( \frac{z+\alpha+\beta}{2\beta} \right) \Gamma\left( \frac{-z+\alpha+2\beta}{2\beta} \right)}\;.
\end{align}
Note that both expressions in \eqref{minimal_E}, with either plus or minus sign, are minimal.
In the literature, one is associated with the S-matrix for the $O(N)$ sigma model,
while another to that for the Gross-Neveu model \cite{Zamolodchikov*2}.

To study the behavior at $\hbar\sim 0$, let us first note that 
\begin{align}
&Q_{\gamma \delta \hbar, \delta \hbar}(z)=
\frac
{ \Gamma\left( \frac{z}{2\hbar \delta}+\frac{1}{2} \right)  \Gamma\left( - \frac{z}{2\hbar \delta}+\frac{1+\gamma}{2}  \right)}
{ \Gamma\left(   \frac{z}{2\hbar \delta}+\frac{1+\gamma}{2}\right) \Gamma\left(  - \frac{z}{2\hbar \delta}+\frac{2+\gamma}{2}  \right)}\;.
\end{align}
From the Stirling's formula
\begin{align}
\log \Gamma(z+c)\sim z \log z-z+c \log z-\frac{1}{2} \log z +\frac{1}{2} \log 2 \pi +\frac{c^2-c}{z}+\mathcal{O}\left(\frac{1}{z^2}\right)
\end{align}
($c$ being constant), we obtain 
\begin{align}
&Q_{\gamma \delta \hbar, \delta \hbar}(z)
\sim e^{-\log \frac{z}{2\hbar \delta}-\hbar \frac{(\gamma^2-2\gamma-1)\delta}{2z} +\mathcal{O}\left(\left(\frac{\hbar}{z}\right)^2\right)}\;,
\end{align}
and hence we learn that the overall factor has a perturbative expansion in powers of $\hbar/z$,
starting with identity:
\begin{align}
E(z)=\frac{ Q_{\pm \hbar, -\hbar(N-2)/2}(z) }{ Q_{0, -\hbar(N-2)/2}(z)}
=
\begin{cases}
1 - \frac{N^2-8}{4(N-2)} \frac{\hbar}{z}+\mathcal{O}\left(\left(\frac{\hbar}{z}\right)^2\right) & \textrm{(plus sign)}\;, \\
1 - \frac{N^2-8N+8}{4(N-2)}\frac{\hbar}{z}+\mathcal{O}\left(\left(\frac{\hbar}{z}\right)^2\right) & \textrm{(minus sign)}\; .
\end{cases}
\end{align}

In perturbation theory, these formulas appear satisfactory, but nonperturbatively, one would like to understand the meaning of the poles 
of the function $E(z)$.  
As remarked in the introduction, the D4-NS5 system is likely to provide a nonperturbative framework in which such questions could be addressed, but
  we will not pursue that in the present paper.

We can repeat a similar exercise for the gauge group $G=Sp_{2N}$.
We again find that the constraint from the Yang-Baxter equation and unitarity are consistent
with crossing symmetry and the framing anomaly, with the dual Coxeter number given by $N+1$.

\section{\texorpdfstring{Examples of Trigonometric and Elliptic $R$-matrix}{Examples of R-matrix}}\label{app.R-matrix}

In section \eqref{elementary} and appendix \ref{app.R_SO}, we discussed several examples of rational $R$-matrices.
In this appendix, we present examples of trigonometric and elliptic $R$-matrices known in the literature,
for the simplest case of the fundamental representation of $G=SL_2$. We verify that their classical limit
reproduces the classical $r$-matrices discussed in the sections \ref{specializing_sl2} and \ref{specializing}.
The $R$-matrices below are quasi-classical, and are normalized to be $R_{\hbar=0}=I$, which is the canonical normalization in gauge theory.  If expanded in perturbation 
theory, they lead to $r$-matrices that we computed in sections \ref{trigonometric} and \ref{elliptic}.

As discussed in appendix \ref{app.R_SO} the complete understanding of the overall factor of the $R$-matrix goes beyond the 
perturbative analysis of this paper. Expressions below should be understood modulo an overall scalar factor ambiguities for the $R$-matrix, and hence modulo the shift by identity matrices for the classical $r$-matrix.

\subsection{Rational}

As discussed in the main text, there are three types of quasi-classical $R$-matrix known in the literature: rational, trigonometric and elliptic. 
Let us begin with the rational $R$-matrix, which we already discussed in section \ref{elementary}:
\begin{align}
\label{rational_rmatrix}
R_{\hbar}^{\rm rational}(z)
=\left(z+\frac{\hbar}{2}\right) I + \hbar \mathop{c} 
=
zI +  \hbar P
\;,
\end{align}
where $z\in \mathbb{C}$, $c=(\vec{\sigma}\cdot \vec{\sigma})/2$ is the quadratic Casimir \eqref{eq.color} for $\mathfrak{g}=\mathfrak{sl}_2$, and $P=\mathop{c}+I/2$ is the permutation operator of the two spins.

This $R$-matrix has manifest $SL_2$ symmetry. 
Choosing a basis, the $R$-matrix reads
\begin{align}
\begin{split}
R_{\hbar}^{\rm rational}(z)
=
\bordermatrix{
&|\uparrow \uparrow \rangle &|\uparrow \downarrow \rangle &|\downarrow \uparrow \rangle& |\downarrow \downarrow \rangle\\
|\uparrow \uparrow \rangle &z+\hbar  & & & \\
|\uparrow \downarrow \rangle &  & z & \hbar& \\
|\downarrow \uparrow \rangle &  & \hbar& z& \\
|\downarrow \downarrow \rangle&  & & &z+\hbar
} \;.
\end{split}
\end{align} 
Note that up to normalization only the ratio of 
$z$ and $\hbar$ enters into this solution, as expected on dimensional grounds.

\subsection{Trigonometric Solution}

The trigonometric solution is
\begin{align}
R^{\rm trig}_{\hbar}(z)= 
\left(
\begin{array}{cccc}
a(z, \hbar) & & & \\
  & b(z, \hbar)&c(z, \hbar)& \\
  &c(z, \hbar)& b(z,  \hbar)& \\
& & &a(z, \hbar)\\ 
\end{array}
\right)\;,
\end{align}
where
\begin{align}
\begin{split}
a(z, \hbar)&=
\sinh\left(w+\hbar\right)\;,
 \\
b(z, \hbar)&
=\sinh\left(w \right) \;,
\\
c(z, \hbar)&
= \sinh(\hbar) \;,
\end{split}
\end{align}
with $z=e^w\in \mathbb{C}^{\times}$ and $q:=e^{\hbar}$.

In terms of the quasi-classical parameter $\hbar$ and the cylindrical coordinate $w$, this $R$-matrix is written in terms of sine functions, as the name 
``trigonometric'' suggests. 
When we impose unitarity, this $R$-matrix should be multiplied by a suitable overall normalization factor.

The six-vertex model allows for an integrable deformation, corresponding to 
an inclusion of the ``electric field'' $h, v$,  in vertical and horizontal directions:
\begin{align}
\label{6v_w_field} 
R^{\rm trig}_{\hbar, \theta}(z)= 
\left(
\begin{array}{cccc}
a(z, \hbar) e^{h+v} & & & \\
  & b(z, \hbar)  e^{-h+v} &c(z, \hbar) & \\
  &c(z, \hbar) & b(z, \hbar) e^{h-v}& \\
& & &a(z, \hbar) e^{-h-v}\\ 
\end{array}
\right)\;.
\end{align}
The parameters $h,v$ correspond to $s_1,s_2$ in section \ref{specializing_sl2}: $h_1=\i s_1, v=-\i s_2$.  

When this $R$-matrix is used to construct a model in lattice statistical mechanics in the usual way, its six non-zero matrix elements become the Boltzmann weights for
six allowed configurations at a vertex.  The  model  accordingly is generally called the six-vertex model.   
Among the six matrix elements of the $R$-matrix, two combinations are inessential; one parameter is the overall normalization of the $R$-matrix and one can be removed
by conjugation.
This leaves four parameters,
which we have here denoted $z$, $\hbar$ and $h, v$.    The one-dimensional quantum spin chain associated with this model is the XXZ spin chain.

\subsection{Elliptic Solution}
The elliptic $R$-matrix has a new parameter, namely the modulus of the torus $\tau$:
\begin{align}
R^{\rm elliptic}_{\hbar, \tau}(z)= 
e^{i \pi \hbar F(z, \tau)}
\left(
\begin{array}{cccc}
a(z, \tau, \hbar) & & & d(z, \tau, \hbar)\\
  & b(z, \tau, \hbar)&c(z, \tau, \hbar)& \\
  &c(z, \tau, \hbar)& b(z, \tau, \hbar)& \\
 d(z, \tau, \hbar) & & &a(z, \tau, \hbar)\\ 
\end{array}
\right)\;,
\end{align}
where we defined
\begin{align}
\begin{split}
a(z, \tau, \hbar)&=
\frac{\vartheta_0 (\pi \hbar| 2\tau)}  {\vartheta_0 (0| 2\tau)} \frac{ \vartheta_1( \pi(z+\hbar)| 2\tau)}  {\vartheta_1(\pi z| 2\tau)}\;, \\
b(z, \tau, \hbar)&=
\frac{\vartheta_0 (\pi \hbar| 2\tau)}{\vartheta_0 (0| 2\tau) } \frac{ \vartheta_0(\pi(z+\hbar)| 2\tau)}{ \vartheta_0(\pi z| 2\tau)}\;,\\
c(z, \tau, \hbar)&=
\frac{\vartheta_1 (\pi \hbar| 2\tau) }{\vartheta_0 (0| 2\tau)}  \frac{\vartheta_0(\pi(z+\hbar)| 2\tau) } { \vartheta_1(\pi z| 2\tau)}\;, \\
d(z, \tau, \hbar)&=\frac{\vartheta_1 (\pi \hbar| 2\tau)}{\vartheta_0 (0| 2\tau)}  \frac{ \vartheta_1(\pi(z+\hbar)| 2\tau) } {\vartheta_0(\pi z| 2\tau)}.
\end{split}
\end{align}
Here $\vartheta_0(z|\tau)$ and $\vartheta_1(z|\tau)$ are Jacobi theta functions
\begin{align}
\begin{split}
\vartheta_0(z|\tau)&=\prod_{n=1}^{\infty} (1-2e^{(2n-1)\pi i \tau} \cos 2 z +e^{(4n-2)\pi i \tau} )
(1-e^{2n\pi i \tau}) \;,\\
\vartheta_1(z|\tau)&=2 e^{\pi i \tau/2} \sin u \prod_{n=1}^{\infty} (1- 2e^{2n\pi i \tau} \cos 2 z +e^{4n\pi i \tau} )
(1-e^{2n\pi i \tau}) \;,
\end{split}
\end{align}
and $F(z, \tau)$ is a normalization factor on which we will comment later.
Note also that, for a better match with the main text, we have used the Jacobi theta function with elliptic modulus $2\tau$ (not $\tau$); care is need for comparison with the literature.

This $R$-matrix has eight non-zero matrix elements and  represents the  eight-vertex model solved by Baxter \cite{Baxter:1971cr}.
The associated one-dimensional quantum spin chain is the XYZ spin chain.

The periodicity of the Jacobi theta functions are given by
\begin{align}
\label{theta_periodicity}
\begin{split}
&\vartheta_1(\pi (z+1) |2\tau)=-\vartheta_1(\pi z|2\tau)\;,
 \quad \vartheta_1\left( \pi \left(z+  \tau\right)|2\tau\right)=i e^{-\frac{i \pi \tau}{2}} e^{-i \pi z}\vartheta_0(\pi z|2\tau) \;,\\
&\vartheta_0(\pi (z+1)|2\tau )=\vartheta_0(\pi z|2\tau)\;,
 \quad \vartheta_0\left(\pi \left(z+  \tau\right)|2\tau\right)=i e^{- \frac{i \pi \tau}{2}} e^{-i \pi z}\vartheta_1(\pi z|2\tau).
\end{split}
\end{align}
This means if we choose the normalization factor $F(z, \hbar)$ to be 
a function on the two-torus satisfying 
\begin{align}
F(z+1, \tau)=F(z, \tau)\;, \quad
F(z+\tau, \tau)=F(z, \tau)+1\;,
\end{align}
then all the entries of the $R$-matrix are periodic under the shift $z\to z+1, z\to z+\tau$, 
up to a possible overall sign factor; for such $F$ we can for example choose
\begin{align}
F(z, \tau)=\frac{i}{\pi} \log \vartheta_0\left(\pi \left(z-\frac{\tau-1}{2}\right)\big|2\tau \right)\; .
\end{align}
Note also that the normalization of the $R$-matrix here is chosen so that $a(z, \tau)=b(z, \tau)=1$ when $\hbar=0$.

The classical $r$-matrix is then (note $\vartheta_1(0)=0$)
\begin{align}
r^{\rm elliptic}_{\tau}(z, \tau)=
\sum_{k=1}^3 w_k(z)\, \sigma_k \otimes \sigma_k
\end{align}
with
\begin{align}
\begin{split}
w_1(z, \tau)&
=\pi\frac{\vartheta_1' (0|2\tau)}{\vartheta_0 (0|2\tau) } \left[\frac{ \vartheta_0(\pi z|2\tau)} {\vartheta_1(\pi z|2\tau)}
+\frac{\vartheta_1(\pi z|2\tau)} { \vartheta_0(\pi z|2\tau)}  \right]\;,\\
w_2(z, \tau)&
=\pi\frac{\vartheta_1' (0|2\tau)}{\vartheta_0 (0|2\tau) } \left[\frac{ \vartheta_0(\pi z|2\tau)} {\vartheta_1(\pi z|2\tau)}
-\frac{\vartheta_1(\pi z|2\tau)} { \vartheta_0(\pi z|2\tau)}  \right]\;,\\
w_3(z, \tau)&
=\partial_z \log \frac{\vartheta_1(\pi z|2\tau)}{\vartheta_0(\pi z|2\tau)}\;.
\end{split}
\end{align}
From (\ref{theta_periodicity}) we find that\footnote{In particular, we have from (\ref{theta_periodicity})
\begin{align} 
\frac{\vartheta_1(\pi (z+1)|2\tau )}{\vartheta_0(\pi (z+1) |2\tau)}=-\frac{\vartheta_1(\pi z |2\tau)}{\vartheta_0(\pi z|2\tau )} ,\quad \frac{\vartheta_1\left(\pi (z+\tau|2\tau) \right)}{\vartheta_0\left(\pi (z+\tau|2\tau) \right)}
=\frac{\vartheta_0\left(\pi (z|2\tau) \right)}{\vartheta_1\left(\pi (z|2\tau) \right)}.
\end{align}
}
\begin{align}
\begin{split}
&w_1(z+1)=- w_1(z) \;, \quad w_1\left(z+\tau\right)= w_1(z)\;,\\
&w_2(z+1)=- w_2(z) \;, \quad w_2\left(z+\tau\right)=- w_2(z)\;,\\
&w_3(z+1)= w_3(z) \;, \quad w_3\left(z+\tau\right)= -w_3(z)\;.
\end{split}
\end{align}
This is consistent with eqn. \eqref{merifo}.

\subsection{Relations}

Among the three solutions for $G=SL_2$, 
the elliptic solution is the most general, and the other two can be obtained by limits
(together with suitable overall rescaling):
\begin{align}
R^{\rm elliptic}_{\hbar, \tau}(z) \xrightarrow[]{\tau\to i\infty}  R^{\rm trig}_{\hbar}(z) \xrightarrow[]{\hbar \to 0} R^{\rm rational}(z/\hbar)\;.
\end{align}
However, the trigonometric $R$-matrix has a generalization  with a non-zero ``field'' turned on; this generalization 
cannot be obtained from a reduction of the eight-vertex model.
In section \ref{specializing}, we explained this fact from an analysis of  the four-dimensional gauge theory on $\R^2\times \mathbb{C}^\vee$.

\section{Anomalies to Quantizing Wilson lines}\label{app.anomaly_Wilson}

In this section we will perform the cohomology calculation referenced in section \ref{section_2loop}.

The main result we will show is the following.
\begin{theorem} 
Let $\mf{g}$ be a simple Lie algebra with no Abelian factors which is not $\mf{sl}_2$.  Let $\wedge^2_0 \mf{g}$ denote the kernel of the Lie bracket map from $\wedge^2 \mf{g} \to \mf{g}$.  Let $V$ be a representation of $\mf{g}$.  Then
\begin{enumerate} 
\item If there are no $G$-invariant maps from $\wedge^2_0 \mf{g} \to \op{End}(V)$ then
\begin{equation} 
H^2(\g[[z]], \op{End}(V)) = 0 \;.  
 \end{equation}
 \item Let $H^2_{(k)}(\g[[z]], \op{End}(V))$ denote the part of the cohomology of weight $k$ under the scaling of $z$ (this group appears in the study  of $k$-loop anomalies).  Then, there is an isomorphism
 \begin{equation} 
 H^2_{(2)}(\g[[z]], \op{End}(V)) \iso \op{Hom}_G ( \wedge^2_0 \g, \op{End}(V))   \;,
  \end{equation}
  where on the right hand side we have the space of $G$-invariant linear operators from $\wedge^2_0 \g$ to $\op{End}(V)$. 
 \end{enumerate}
 If $\mf{g} = \mf{sl}_2$, then $H^2_{(k)}(\g[[z]], \op{End}(V)) = 0$ unless $k = 3$, and the dimension of $H^2_{(3)}(\g[[z]], \op{End}(V))$ is the number of copies of the $5$-dimensional irreducible representation of $\mf{sl}_2$ in $\op{End}(V)$.  
 \end{theorem}
The proof will take a number of steps, which we will write as separate propositions.   The first step is the following. 
\begin{proposition}
\label{proposition_wilson_line_quantization}
	Let $\mf{g}$ be a simple Lie algebra with no Abelian factors. Let $V$ be a representation of $\mf{g}$ with the feature that the image of any $G$-invariant linear map $\mf{g}^{\otimes 2} \to \op{End}(V)$ consists of copies of the trivial and adjoint representations of $\mf{g}$. Then,
\begin{equation}
H^2(\g[[z]], \op{End}(V)) = 0 \;.
\end{equation}
\end{proposition}
The vanishing of this cohomology group implies that there are no anomalies to quantizing the Wilson line in the representation $V$. 

\begin{proof}
In this calculation, we will use the Lie algebra $\g[z]$ of polynomials with coefficients in $\g$, instead of the Lie algebra $\g[[z]]$ of series with coefficients in $\g$. As long as we focus on cohomology classes of a fixed weight under the $\C^\times$ action which scales $z$, the cohomology groups are the same in each case. This is because a cochain of weight $k$ under this $\C^\times$ action involves $z^m$ for $m \le k$.

Let $R \subset \op{End}(V)$ be the largest sub-representation which contains all the representations occurring in the decomposition of $\mf{g} \otimes \mf{g}$ into irreducible components.  There is a natural injective map 
\begin{equation}
C^\ast(\mf{g}[z], R) \to C^\ast(\mf{g}[z], \op{End}(V)) \;. 
\end{equation}
It is a standard fact of Lie algebra cohomology that we find the same cohomology groups if we use $G$-invariant cocycles.  It follows from the definition of $R$ that, on $G$-invariant cocycles, this map is an isomorphism in degrees $\le 2$. From this we see that this map must be an isomorphism on degree $2$.  

This reduces the problem to showing that $H^2(\g[z], R) = 0$.  The assumptions we state in the proposition are that $R$ is a direct sum of copies of the adjoint and trivial representations.  We therefore need to show that $H^2(\g[z],\C) = 0$ and that $H^2(\g[z], \g) = 0$.

Since we are working in low cohomological degrees, both of these statements can be proved by hand with a direct computation.  However, it is more convenient to invoke a general result of Fishel, Grojnowski, and Teleman \cite{FGT}, who calculate the Lie algebra cohomology of $\mf{g}[z]$ with various coefficients. 

Their results (Theorem B, p.\ 6 of \cite{FGT}) tell us that $H^2(\g[z],\C) = 0$.  We need to show that $H^2(\g[z],\g) = 0$ as well. This follows a little indirectly from their analysis, as we will now explain. 

We write $ \mf{g}[z^{-1}]$ for the space of $\mf{g}$-valued functions of $z$ of the form $$\alpha_{a_0} t_{a_0} + z^{-1} \alpha_{a_1}t_{a_1} + \dots .$$ These are meromorphic functions on $\C$, which have a  finite order pole at $0$ and no other poles, and which are regular at infinity.  There is a natural residue pairing
\begin{equation}
\ip{t_a f(z), t_b g(z^{-1}) } = 2 \delta_{ab} \frac{1}{2 \pi i} \oint f(z) g(z^{-1}) z^{-1} \d z 
\end{equation}
between $\mf{g}[z^{-1}]$ and $\mf{g}[z]$.  

The space $\mf{g}[z^{-1}]$ has the structure of module for $\mf{g}[z]$ whereby 
\begin{equation}
(t_a z^{k}) \cdot (t_b z^{-l}) = f^c_{ab} t_c z^{k-l} \delta_{k-l \le 0} \;.
\end{equation}
This identifies $ \mf{g}[z^{-1}]$ with the dual representation to the adjoint representation of $\mf{g}[z]$, that is, with the coadjoint representation.  (Because these are infinite dimensional modules, one has to take care with the meaning of ``dual'': we mean the restricted dual, where we take the direct sum over the weight spaces of the $\C^\times$ action which scales $z$.  This use of restricted duals and the corresponding restricted cohomology elides the difference between $\mf{g}[z]$ and $\g[[z]]$, so we will use the algebra $\g[z]$ in the course of the proof). 

Fishel, Grojnowski and Teleman show (Theorem B, p.\ 6 of \cite{FGT}) that the cohomology of $\mf{g}[z]$ with coefficients in the coadjoint representation $\mf{g}[z^{-1}]$ is a graded tensor product of the cohomology of the compact group $G$ with a graded vector space consists of $ \C[z^{-1}]$ in degree $0$ and $z^{-1} \C[z^{-1}]$ in degree $1$.  We have written these cohomology groups in a way compatible with the $\C^\times$  which scales $z$.  In particular, since the cohomology of the group $G$ lives in degree $0$ and degrees $\ge 3$, we find that there is no second cohomology of $\mf{g}[z]$ with coefficients in $\g[z^{-1}]$. 

To describe the cocycles for these cohomology classes, it is more convenient to use cohomology relative to the subalgebra $\g \subset \g[z]$. A relative $k$-cocycle is given by a $G$-invariant linear map
\begin{equation}
\wedge^k (z \g[z]) \to \g[z^{-1}]\;.
\end{equation}
The absolute cohomology groups are the tensor product of the relative cohomology groups with the cohomology of the group $G$.  Therefore they coincide in degrees less than $3$.

The relative $1$-cocycle associated to an element $z^{-k} \in z^{-1} \C[z^{-1}]$ is a linear map
\begin{align}
\begin{split}
	z \mf{g}[z]\quad &\to \quad  \mf{g}[z^{-1}] \\
\vin \quad \quad& \quad \qquad \vin \\
t_a f(z) \quad&\mapsto t_a \pi_{-} \left( z^{-k} f(z)\right) \;,
\end{split}
\end{align}
	where $\pi_-$ indicates the operation of projection of an element of $\g[z,z^{-1}]$ onto the non-positive powers of $z$. 

There is a short exact sequence of $\mf{g}[z]$-representations
\begin{equation}
0 \to \g \to   \g[z^{-1}] \xto{z}  \g[z^{-1}] \to 0\;,
\end{equation}
where  the second map is given  by multiplying by $z$, using the convention that we only retain negative powers of $z$. 

This short exact sequence leads to a long exact sequence on cohomology
\begin{equation}
\dots \to H^1(\mf{g}[z], \g) \to H^1(\mf{g}[z],  \g[z^{-1}]) \to H^1(\g[z],  \g[z^{-1}]) \to H^2(\g[z], \g) \to 0 \;,
\end{equation}
where we use the fact that $H^2(\g[z],  \g[z^{-1}]) = 0$.

This leads to an exact sequence
\begin{equation}
\dots \to H^1(\mf{g}[z], \g) \to z^{-1} \C[z^{-1}] \xto{z} z^{-1} \C[z^{-1}] \to H^2(\g[z], \g) \to 0 \;.
\end{equation}
From this we conclude that $H^2(\g[z],\g) = 0$, as desired. 
\end{proof}

For any representation $R$ of $\g$, let $H^i_{(k)}(\g[z],R)$ denote the cohomology group of weight $-k$ under the $\C^\times$ action which scales $z$.  This is the cohomology group that will play a role in quantization at $k$ loops.

Next, we will prove a more difficult cohomology vanishing result, also based on the results of \cite{FGT}. 
\begin{proposition}
For any simple Lie algebra with no Abelian factors which is not $\mf{sl}_n$ with $n>2$, we have
\begin{equation}
H^2_{(k)}(\g[z], \Sym^2 \g) = 0 \text{ unless } k = 3\;.
\end{equation}
Further, there is an exact sequence
\begin{equation}
0 \to H^2_{(3)} (\g[z], \Sym^2 \g) \to \C \to H^2_{(2)} (\g[z], \wedge^2 \g) \to H^3_{(3)} (\g[z], \Sym^2 \g) \to 0\;. 
\end{equation}
\label{proposition_symmetric_cohomology}
\end{proposition} 
Note that this proposition applies to $\mf{sl}_2$.
\begin{proof}
The proof again uses the results of Fishel, Grojnowski and Teleman \cite{FGT}. Let us consider the $\g[z]$-module $\Sym^2  \g[z^{-1}]$. The results of \cite{FGT} show that
\begin{align}
\begin{split}
H^2(\g[z], \Sym^2  \g[z^{-1}] ) &= \wedge^2 z^{-1} \C[z^{-1}] \;,\\
H^1(\g[z], \Sym^2  \g[z^{-1}] ) &= (\Sym^3 \g)^G \otimes z^{-1} \C[z^{-1}] \;,\\  
H^0(\g[z], \Sym^2 \g[z^{-1}] &= \C[z^{-1}]\;.  
\end{split}
\end{align} 
We have written this isomorphism in a way compatible with the $\C^\times$-action which scales $z$.  Note that, unless $\g$ is $\mf{sl}_n$ for $n > 2$, $(\Sym^3 \g)^G = 0$. Since we restrict to Lie algebras which are not of this form, we have $H^1(\g[z], \Sym^2 \g[z^{-1}]) = 0$.  

Define operators 
\begin{align}
\begin{split}
z^n : \Sym^2  \g [z^{-1}] & \mapsto \Sym^2  \g [z^{-1}] \\
z^n \left( (z^{-k} t_a) (z^{-l} t_b) \right) &= (z^{n-k} t_a)(z^{-l} t_b) \delta_{k-n\ge 0} + (z^{-k} t_a)(z^{n-l} t_b) \delta_{l-n \ge  0}.
\end{split}
\end{align}
for $n \ge 1$.  Note that the operator $z$ is surjective.  We let $M$ denote the kernel of $z$.   The space $M$ is spanned by the elements
\begin{align}
\begin{split}
&A_{ab} (t_a ) (t_b z^{-2k}) - A_{ab} (t_a z^{-1}) (t_b z^{1-2k}) \pm \dots  
\\ 
&\qquad + (-1)^{k+1} A_{ab} (t_a z^{1-k}) (t_b z^{-k-1})  + \tfrac{1}{2} (-1)^{k} A_{ab} (t_a z^{-k}) (t_b z^{-k})\;,
\end{split}
\end{align}
where $A_{ab}$ is a symmetric tensor, and the elements
\begin{align}
\begin{split}
&B_{ab} (t_a ) (t_b z^{-2k-1}) - B_{ab} (t_a z^{-1}) (t_b z^{-2k}) \pm \dots  
\\ 
&\qquad
+ (-1)^{k+1} B_{ab} (t_a z^{1-k}) (t_b z^{-k-2})  +  (-1)^{k} B_{ab} (t_a z^{-k}) (t_b z^{-k-1})\;,
\end{split}
\end{align}
where $B_{ab}$ is antisymmetric.

Thus, as a representation of $\g$, the module $M$ decomposes as 
\begin{equation}
M = \Sym^2 \g \oplus z^{-1} \wedge^2 \g \oplus z^{-2} \Sym^2 \g \dots \;,
\end{equation}
where the powers of $z$ indicate the $\C^\times$-weights. 

As a $\g[z]$-module, the structure is more complicated: elements in $z \g$ moves us between summands of different weight. 

Now, let us identify the cohomology of $\g[z]$ with coefficients in $M$.  Because we have a short exact sequence
\begin{equation}
0 \to M \to \Sym^2 ( \g[z^{-1}]) \xto{z} \Sym^2 ( \g[z^{-1}]) \to 0 \;,
\end{equation} 
we get a long exact sequence in cohomology
\begin{align}
\begin{split}
0 \to H^0(\g[z],M) &\to  \C[z^{-1}] \xto{z}  \C[z^{-1}] \\
   \to H^1(\g[z],M) & \to 0 \to 0\\ 
 \to H^2(\g[z],M) &\to \wedge^2 (z^{-1}  \C[z^{-1}]) \xto{z} \wedge^2 (z^{-1}  \C[z^{-1}])\\
  & \to H^3(\g[z],M) \to 0 \;.   
  \end{split}
\end{align}
The endomorphism $z$ of $\wedge^2 (z^{-1} \C[z^{-1}])$ is defined by
\begin{equation}
z( z^{-k} \wedge z^{-l} ) = z^{1-k} \wedge z^{-l} \delta_{k-1\ge 1} + z^{-k} \wedge z^{1-l} \delta_{l-1 \ge 1}\;.  
\end{equation}
The kernel of the operator $z$ is spanned by the elements
\begin{equation}
z^{-1} \wedge z^{-2k}  - z^{-2} \wedge z^{1-2k} \pm \dots + (-1)^{k} z^{-k}\wedge z^{-1-k} 
\end{equation}
for $k \ge 1$. Thus, $H^2(\g[z],M)$ is spanned by these elements, in weights $3,5,\dots$.  To sum up, we find 
\begin{align}
\begin{split}
H^0(\g[z],M) &= \C\;, \\
H^1(\g[z],M) &= 0\;,\\  
H^2(\g[z],M) &= \C \cdot z^{-3} \oplus \C \cdot z^{-5} \oplus \dots \;,\\
H^3(\g[z], M) &= 0\;. 
  \end{split}
\end{align}
where the powers of $z$ indicate the weights under the $\C^\times$-action. 

Let $M_0 \subset M$ be the kernel of the operator $z^2$.  This operator is surjective on $H^2(M)$, with kernel $\C$ in weight $-3$, and surjective on $M$, but acts by zero on $H^0(M)$ and $H^1(M)$.  We have a long exact sequence
\begin{align}
\begin{split}
0 \to H^0(\g[z], M_0) & \to \C \xto{0} \C \\
 \to H^1(\g[z],M_0) & \to 0 \to 0 \\ 
 \to H^2(\g[z], M_0) & \to z^{-3} \C[z^{-2}] \xto{z^2} z^{-3} \C[z^{-2}] \\
 \to H^3(\g[z], M_0) & \to 0 \;.  
\end{split}
\end{align}
 This sequence allows us to calculate the cohomology groups of $H^\ast(\g[z],M_0)$.  Before we state the answer, we should note that the boundary maps in this long exact sequence shift the weight under the $\C^\times$-action by $-2$.  This is because of the appearance of the operator $z^2$, of weight $2$, in the short exact sequence of modules leading to this long exact sequence of cohomology  groups.  

 From this exact sequence we see that 
\begin{align}
\begin{split}
 H^0(\g[z],M_0) &= \C \;,\\
 H^1(\g[z], M_0) &= \C \cdot z^{-2}\; ,\\ 
 H^2(\g[z],M_0) &= \C \cdot z^{-3}\; . \\ 
\end{split}
\end{align}
Since $M_0$ is the intersection of the kernel of $z$ and the kernel of $z^2$ in $\Sym^2 \g[z^{-1}]$, one can show that there is an exact sequence
\begin{equation}
0 \to \Sym^2 \g \to M_0 \to z^{-1} \wedge^2 \g \to 0 \;,
\end{equation}  
from which we derive a long exact sequence
\begin{align}
\begin{split}
0 \to H^0(\g[z], \Sym^2 \g) &\to H^0(\g[z], M_0) \to H^0(\g[z], z^{-1} \wedge^2 \g) \\
\to H^1(\g[z], \Sym^2 \g) & \to H^1(\g[z], M_0) \to H^1 (\g[z], z^{-1} \wedge^2 \g)\\
\to H^2( \g[z], \Sym^2 \g) & \to H^2(\g[z], M_0) \to H^2(\g[z], z^{-1} \wedge^2 \g) \to H^3(\g[z], \Sym^2 \g) \to 0 \;.
\end{split}
\end{align}
Now, $H^0(\g[z], \wedge^2 \g) = 0$ because the are no $G$-invariant elements in $\wedge^2 \g$ for any simple Lie algebra $\g$.  Further $H^1(\g[z], \Sym^2 \g)$ is zero because there are no copies of the adjoint representation in $\Sym^2 \g$ (recall that we are assuming that $\g$ is not $\mf{sl}_n$ for $n>2$). 

We also have $H^1(\g[z], z^{-1} \wedge^2 \g) = z^{-2} \C$, because there is one copy of the adjoint in $\wedge^2 \g$.

From this, we find that we have an exact sequence
\begin{align}
\begin{split}
0  & \to \C \cdot z^{-2}   \to  \C \cdot z^{-2}   \xto{\delta} H^2( \g[z], \Sym^2 \g) \\
& \to H^2(\g[z], M_0) \to H^2(\g[z], z^{-1} \wedge^2 \g) \to H^3(\g[z], \Sym^2 \g) \to 0 \;.
\end{split}
\end{align}
The boundary map 
\begin{equation}
\delta:  \C \cdot z^{-2} \to H^2(\g[z], \Sym^2 \g) 
\end{equation}
must  be zero. 

We conclude that $H^2(\g[z], \Sym^2 \g)$ is a subspace of 
\begin{equation}
H^2(\g[z], M_0) = \C \cdot z^{-3}. 
\end{equation}
In particular, $H^2(\g[z], \Sym^2 \g)$ is only non-zero in weight $3$. 

\end{proof}

Next, let us perform the same analysis in the case that $\mf{g}$ is $\mf{sl}_n$.
\begin{proposition}
Suppose that $\mf{g} = \mf{sl}_n$ for $n > 2$. Then, 
\begin{equation}
H^2_{(k)}(\mf{sl}_n[z], \Sym^2 \mf{sl}_n) = 0 \text{ unless } k = 3\;.
\end{equation}
 As before, there is an exact sequence
\begin{equation}
0 \to H^2_{(3)} (\mf{sl}_n[z], \Sym^2 \mf{sl}_n) \to \C \to H^2_{(2)} (\mf{sl}_n[z], \wedge^2 \mf{sl}_n) \to H^3_{(3)} (\mf{sl}_n[z], \Sym^2 \mf{sl}_n) \to 0\;. 
\end{equation}

\end{proposition}
\begin{proof}
The cohomology groups $H_{(k)}^\ast(\mf{sl}_n[z], \Sym^2 \mf{sl}_n)$ are acted on by  the outer automorphism group $\Z/2$ of $\mf{sl}_n$ (which is the automorphism group of the Dynkin diagram).  We let $H^\ast_{(k), \textrm{even}}$ and $H^\ast_{(k), \textrm{odd}}$ refer to the cohomology groups which are even and odd under this action. We will first prove that the cohomology groups which are even under the outer automorphism satisfy the statement of the proposition, and then show that the groups which are odd under the outer automorphism are zero.

Note that in the proof of proposition \ref{proposition_symmetric_cohomology} we used the fact that $\mf{g}$ is not $\mf{sl}_n$ with $n>2$  in two places.  Firstly, we used the fact that Lie algebras which are not of this type do not have an invariant element in $\Sym^3 \mf{g}$.  An invariant element in $\Sym^3 \g$ contributes, according to the results of \cite{FGT}, an element of $H^1(\g[z], \Sym^2 (\g[z^{-1}])$.  However, because the invariant element in $\Sym^3 \mf{sl}_n$ is odd under the outer automorphism, the corresponding element of $H^1(\mf{sl}_n[z], \Sym^2 (\mf{sl}_n[z^{-1}]))$ is also odd, and so does not contribute when we analyze the cohomology which is even under the outer automorphism.

The other place where we used the assumption that $\g$ is not $\mf{sl}_n$ with $n > 2$ was when we asserted that $H^1(\g[z], \Sym^2 \g)$ must be zero.  This is essentially the same point, as any element of $H^1(\g[z], \Sym^2 \g)$ must come from an invariant element in $\g \otimes \Sym^2 \g$, and the only invariant element (when $\g = \mf{sl}_n$, $n > 2$) is totally symmetric.   Because this element is odd under the outer automorphism of $\mf{sl}_n$, we find again that it can not contribute to the cohomology which is even under this automorphism. 

Therefore, when we restrict to the cohomology which is even under the outer automorphims, proposition \ref{proposition_symmetric_cohomology} holds with the same proof. 

We need to calculate the cohomology which is odd under the outer automorphism.  We are interested in $H^2_{\rm odd}(\mf{sl}_n[z], \Sym^2 \mf{sl}_n)$. Any odd two-cocycle can be represented as a sum of $\mf{sl}_n$-invariant linear operators
\begin{equation} 
(z^k \mf{sl}_n) \otimes (z^l \mf{sl}_n) \to \Sym^2 \mf{sl}_n \;,
 \end{equation}
Any such linear operator can be thought of as an invariant tensor in $\mf{sl}_n^{\otimes 4}$, which is invariant under the permutation of the first two factors.  If $A_1,A_2,B,C$ denote elements of $\mf{sl}_n$ then there are three such tensors, invariant under permutation of $A_1$ and $A_2$:
\begin{align*} 
&\op{Tr} (A_1 A_2 B C ) + \op{Tr}(A_2 A_1 B C ) \;,\\
&\op{Tr} (A_1 A_2 C B ) + \op{Tr}(A_2 A_1 C B ) \;,\\
&\op{Tr} (A_1 B A_2 C  ) + \op{Tr}(A_2 B  A_1 C  ) \;.
 \end{align*}
The only linear combination of these three tensors which is odd under the outer automorphism is
\begin{equation} 
 \op{Tr} (A_1 A_2 [B, C] ) + \op{Tr}(A_2 A_1 [B, C] ) 
\end{equation}
 (recall that the outer automorphism sends $A \in \mf{sl}_n$ to $-A^T$).

This invariant tensor can be viewed as a linear operator 
\begin{equation} 
(z^k \mf{sl}_n) \otimes (z^l \mf{sl}_n) \to \mf{sl}_n \subset \Sym^2 \mf{sl}_n \;.
 \end{equation}
We have just shown that the natural map
 \begin{equation} 
	 H^2(\mf{sl}_n[z], \mf{sl}_n) \to H^2_{\rm odd}(\mf{sl}_n[z], \Sym^2 \mf{sl}_n)  
  \end{equation}
  is surjective.  However, we already know that $H^2(\mf{sl}_n[z], \mf{sl}_n) = 0$, so we conclude that $H^2_{\rm odd}(\mf{sl}_n[z], \Sym^2 \mf{sl}_n) = 0$ also. 
\end{proof}

The next result we need is the following.
\begin{proposition}
Let $\wedge^2_0 \g \subset \wedge^2 \g$ denote the kernel of the Lie bracket map $\wedge^2 \g \to \g$.  Then, 
\begin{equation}
H^2_{(2)}(\g[z], \wedge^2 \g) = \op{Hom}^G( \wedge^2_0\g, \wedge^2_0 \g)\;, 
\end{equation}
where on the right hand side we have the space of $G$-invariant maps from $\wedge^2_0 \g$ to itself. 

Further, the map
\begin{equation}
\C = H^2_{(3)}(\g[z], M_0) \to H^2_{(2)}(\g[z], \wedge^2 \g) 
\end{equation} 
is surjective for $\g \neq \mf{sl}_2$ (where we are using the notation $M_0$ from the proof of proposition \ref{proposition_symmetric_cohomology}).  
\end{proposition}
\begin{proof}
	We can calculate $H^2_{(2)}(\g[z], \wedge^2 \g)$ using Lie algebra cohomology relative to $\g$.  A basis for the relative $2$-cochains is the space of $G$-invariant maps $\wedge^2 (z \g) \to \wedge^2 \g$.    Every such relative two-cochain is closed, because there are no relative three-cochains of weight $2$. Elements of the form
\begin{equation}
f^a{}_{bc} (z t_a \wedge z t_b) \otimes R^a_{ef} (t_e \wedge t_f)
\end{equation}
are exact, where $R^a_{ef}$ defines some $G$-invariant map from $\g$ to $\wedge^2 \g$.  

For all simple Lie algebras, there are no copies of the adjoint representation in $\wedge^2_0(\g)$. This argument tells us that the space of closed, but not exact, $2$-cochains of weight $-2$ is the space of $G$-invariant maps $\wedge^2 \g \to \wedge^2_0 \g$.   

For $G \neq SL_2$, this space is non-trivial. The final thing we need to check is that for $G \neq SL_2$, the map 
\begin{equation}
\C = H^2_{(3)}(\g[z], M_0) \to H^2_{(2)}(\g[z], \wedge^2 \g) 
\end{equation}
is non-zero. This follows from the explicit description of the cocycles of $\g[z]$ with values in $\Sym^\ast \g[z^{-1}]$ given in \cite{FGT}.  
\end{proof}
The following corollary sums up what we have learned about the cohomology groups controlling the obstructions to quantizing a Wilson line.
\begin{corollary}
For a representation $V$ of a simple Lie algebra $\g\neq \mf{sl}_2$, consider the group $H^2_{(k)}(\g[z], \op{End}(V))$ which contains possible anomalies to quantizing the Wilson line associated to $V$ at $k$ loops. Then, $H^2_{(k)}(\g[z], \op{End}(V)) = 0$ unless there exists a non-trivial $G$-invariant map from $\wedge^2_0 \g$ to $\op{End}(V)$. 
\end{corollary}
\begin{proof}
Two-cocycles representing classes in $H^2(\g[z], \op{End}(V))$ are given by $G$-invariant linear maps $\wedge^2 z \g[z] \to \op{End}(V)$.  Any such $G$-invariant map must factor through some copies of $\Sym^2 \g$, $\g$, or $\wedge^2_0 \g$. We have seen that the cohomology with coefficients in $\Sym^2 \g$ and $\g$ is zero. So, assuming there are no non-trivial maps from $\wedge^2_0 \g \to \op{End}(V)$, the cohomology groups are all zero. 
\end{proof}

In fact, one can strengthen this result with some further work. For instance, one can show that $H^2_{(k)}(\g[z], \wedge^2_0 \g) = 0$ unless $k = 2,4,6$. We have already described the rank of the cohomology group when $k = 2$.  Computing explicitly the rank when $k = 4,6$ is a bit involved, and since we do not ultimately need to know the answer, we have not included these calculations.  

\section{\texorpdfstring{Derivation of Eqn.\ \eqref{milp}}{Derivation of Equation \ref{milp}}}\label{apptwo}

In the Feynman diagram of \fig \ref{Diagram3},
we have three vertices, two vertices $u, v$ on the Wilson line $K$
\begin{align}
\left(u, \frac{f u^2}{2}, 0,0\right) \;, \quad
\left(v, \frac{f v^2}{2}, 0, 0\right) \;,
\end{align}
and one vertex $w$ in the bulk, 
which we parametrize as $(x,y, z, \bar{z})$.
At the vertex $u$ of the Wilson line, we have the line element
\begin{align}
\int_{K} A ds= \left(A_x(u)  +  A_y(u) f u \right) du \;.
\end{align}
In the following, to save space 
we sometimes write $A_x(u):=A_x(x(u), y(u), 0, 0)$.
The similar expression applies to the vertex $v$.

The diagram has a symmetry factor of $1/2$, which we can incorporate by restricting the integration region to 
be $u<v$.

Now we connect the vertices by a propagator, to the vertex $w$.
We are interested in linear order in $f$.

One possibility is to take $A_x$ from both, and then obtain a factor of $f$ from the expansion of the propagator.  However, this contributes an expression of the form
\begin{align}
 \pm \langle A_x(u) A_y(w)   \rangle \langle A_x(v)  A_{\bar{z}}(w)  \rangle
A_x(w) 
\end{align}
(or with $u$ and $v$ exchanged), and since we have $A_x(w)$ this does not contribute to $\Lambda$.

Another possibility is to take one $A_x$  and one $A_y$ from $u$ and $v$:
\begin{align}
(fv) A_x(u) A_y(v)  \quad
\textrm{or}  \quad
(fu) A_y(u) A_x(v) \;.
\label{Axy}
\end{align}

When we consider the propagator, there are three different possibilities for \eqref{Axy}:
\begin{align}\label{AAAAA}
\begin{split}
(f(v-u))  \langle A_x(u) A_y(w)   \rangle \langle A_y(v)  A_{\bar{z}}(w)  \rangle
A_x(w)\;, \\
(f(v-u))  
\langle A_x(u) A_{\bar{z}}(w)   \rangle \langle A_y(v)  A_x(w)  \rangle
A_y(w)\;, \\
(f(v-u)) \langle A_x(u) A_y(w)   \rangle \langle A_y(v)  A_x(w)  \rangle
A_{\bar{z}}(w) \;.
\end{split}
\end{align}

Using the expression for the propagator \eqref{eq.propagator},
the expressions \eqref{AAAAA} can be simplified, and when combined with
differential forms we have
\begin{align}
\begin{split}
& \frac{1}{(2\pi)^2}
(f(v-u)) \frac{(2\bar{z})  (x-u)}{d(u,w)^4 d(v,w)^4}
 A_x(w) 
\cdot dz(du dy)(dv d\bar{z}) dx
 \;, \\
& 
\frac{1}{(2\pi)^2} 
(f(v-u)) \frac{ (y-\frac{fu^2}{2})  (2\bar{z})}{d(u,w)^4 d(v,w)^4}
A_y(w)
\cdot dz(du d\bar{z})(dv dx) dy
\;, \\
&
\frac{1}{(2\pi)^2}
(f(v-u))   \frac{(-1)(2\bar{z})  (2\bar{z})}{d(u,w)^4 d(v,w)^4}
A_{\bar{z}}(w)
\cdot dz(du dy)(dv dx) d\bar{z}
 \;,
\end{split}
\end{align}
with distance $d(u,w)$ defined by
\begin{align}
d(u,w)^2=(x-u)^2+\left(y-\frac{fu^2}{2}\right)^2+|z|^2  \;.
\end{align}
We also need to supplement these expressions by a factor $\hbar\, \sh^\vee$,
where $\sh^\vee$ is the color factor explained in the text and $\hbar$ is the loop counting parameter.
In the leading order in $f$, we have $d(u,w)^2\simeq (x-u)^2+y^2+|z|^2$,
and hence we obtain
\begin{align}
\begin{split}
\label{before}
-\frac{1}{(2\pi)^2}\frac{ \ii}{2\pi} \hbar\, \sh^\vee
& \int_{u<v} \d u \d v \int_{\mathbb{R}^4} \d x \d y \d z \d\bar{z} 
\,
 (f(v-u)) \\
&\qquad\times
\frac{2\bar{z}  (x-u) A_x(w) +2 y \bar{z} A_y(w) +4 \bar{z}^2 A_{\bar{z}}(w) }{\left((x-u)^2+y^2+|z|^2 \right)^2 \left((x-v)^2+y^2+|z|^2 \right)^2}\;, 
\end{split}
\end{align}
The equation above can be expressed in the form \eqref{ffo}
if we define
\begin{align}
\label{after}
\begin{split}
\Theta_0&=-\frac{1}{(2\pi)^2}\frac{ \ii}{2\pi} \hbar\, \sh^\vee
\int_{u<v} \d u \d v 
\,
f(v-u)
\\
& \quad \times \frac{ 2(\bar{z})  (x-u) (\d y \wedge \d z \wedge \d\bar{z} ) - 2y \bar{z}  (\d x\wedge \d z \wedge \d\bar{z} ) +4\bar{z}^2 (\d x \wedge \d y \wedge \d z ) }{\left((x-u)^2+y^2+|z|^2 \right)^2 \left((x-v)^2+y^2+|z|^2 \right)^2}\;.
\end{split}
\end{align}
In the notation of the main text, we have $\Theta= z \Theta_0$ and $\Theta=\d x\wedge \Lambda+\Lambda'$ to obtain
\begin{align}
\begin{split}
\Lambda&= -\frac{1}{(2\pi)^2}\frac{2f \ii}{2\pi} \hbar\, \sh^\vee \\
&\times\int_{u<v} \d u \d v \,
 (v-u)
\frac{ -y z \bar{z}  (\d z \wedge \d\bar{z} ) +2 z\bar{z}^2  (\d y \wedge \d z )}{\left((x-u)^2+y^2+|z|^2 \right)^2 \left((x-v)^2+y^2+|z|^2 \right)^2} \;.
\end{split}
\end{align}
We obtain \eqref{milp} after doing the $u$ and $v$ integrals:
\begin{align}
&\int_{-\infty \le u<v \le \infty} \d u \d v \,
(v-u)
\frac{1}{\left(u^2+a^2 \right)^2 \left(v^2+a^2 \right)^2}=\frac{1}{a^5}\frac{3 \pi}{8} \;.
\end{align}

\section{\texorpdfstring{Evaluation of Eqn.\ \eqref{13form}}{Evaluation of Equation \ref{13form}}}\label{app.two-loop_computation}

In this appendix we present details on the evaluation of the integral \eqref{13form}.

Let us first evaluate
the part relevant for the $v_1$ integral inside the expression \eqref{13form}:
\begin{align}
\label{v1_int}
 \int_{v_1} 
P(v_0, v_1)  \wedge \d z_1 (z_1) \wedge P(v_1, v_2)  \;,
\end{align}
where we have introduced, temporarily, a vertex labelled $v_0$with coordinates $v_0 = (x_0,y_0,z_0,\zbar_0)$. Later we will set these coordinates to  zero.

From the explicit expression for the propagator in \eqn \eqref{eq.propagator} we have
\begin{align}
\begin{split}
P(v_i, v_j)
	&=\frac{1}{2\pi} 
	\frac{  x_{ij} \d y_{ij}\wedge \d\bar{z}_{ij}  - y_{ij} \d x_{ij}\wedge \d \zbar_{ij} +  2 \zbar_{ij}\d x_{ij}\wedge \d y_{ij}} {d(v_i, v_j)^4}  \;,
\end{split}
\end{align}
where the distance $d(v_i, v_j)$ between two points $v_i, v_j$ is given by
$d(v_i, v_j)^2=x_{ij}^2+y_{ij}^2+|z_{ij}|^2$, and $x_{ij}=x_i-x_j, z_{ij}=z_i-z_j$, etc.

We can calculate that
\begin{align}
\begin{split} 
	& P(v_0,v_1)\wedge P(v_1,v_2)=
	\frac{1}{(2 \pi)^2} \d x_1 \d y_1 \d \zbar_1    \\
	&\quad \frac{ \left( x_{01} y_{12}\d \zbar_2-2 x_{01}\zbar_{12} \d y_2 +2 y_{01}\zbar_{12}\d x_2 - y_{01}x_{12}\d \zbar_2 + 2 \zbar_{01}x_{12}\d y_2 - 2 \zbar_{01}y_{12} \d x_2  \right) }{ d(v_0,v_1)^{4}d(v_1,v_2)^{4}} \;. 
\end{split}
\end{align}
where we have dropped all terms involving $\d x_0$, $\d y_0$, $\d \zbar_0$ because we will not be integrating over the vertex $v_0$.  	 
From this we see that
\begin{multline} 
	P(v_0,v_1)\wedge z_1 \d z_1  P(v_1,v_2)= 
	\frac{1}{(2 \pi)^2} 
	\d x_1 \d y_1 \d z_1 \d \zbar_1 
	\left[\partial_{\zbar_0} \frac{1}{d(v_0,v_1)^{2}d(v_1,v_2)^{4}}\right]
	 \\ 
	\left( x_{01} y_{12}\d \zbar_2  -2 x_{01}\zbar_{12} \d y_2 +2 y_{01}\zbar_{12}\d x_2 - y_{01}x_{12}\d \zbar_2 + 2 \zbar_{01}x_{12}\d y_2 - 2 \zbar_{01}y_{12} \d x_2  \right) \;. 
\end{multline}
after we set the $v_0$ coordinates to zero.

Moving the position of the derivative in $\zbar_0$, we obtain
\begin{multline} 
	P(v_0,v_1)\wedge z_1 \d z_1  P(v_1,v_2)=
	 \frac{1}{(2  \pi)^2} 
	\d x_1 \d y_1 \d z_1 \d \zbar_1 
	 \\ \times \partial_{\zbar_0}  \left[\frac{\left( x_{01} y_{12}\d \zbar_2  -2 x_{01}\zbar_{12} \d y_2 +2 y_{01}\zbar_{12}\d x_2 - y_{01}x_{12}\d \zbar_2 + 2 \zbar_{01}x_{12}\d y_2 - 2 \zbar_{01}y_{12} \d x_2  \right)}{d(v_0,v_1)^{2}d(v_1,v_2)^{4}}\right]
	 \\
 \frac{1}{(2 \pi^2)^2} \d x_1 \d y_1 \d z_1 \d \zbar_1 
 \frac{\left(  2 x_{12}\d y_2 - 2 y_{12} \d x_2  \right)}{d(v_0,v_1)^{2}d(v_1,v_2)^{4}}
	 \;. 
\end{multline}
By integration by parts, the integral over $x_1,y_1,z_1,\zbar_1$ of all the terms on the first two lines vanishes.  For example, 
\begin{align*} 
	&\int_{v_1} (x_{01}y_{12}- y_{01}x_{12}) \frac{1}{d(v_0,v_1)^{2}}\frac{1}{d(v_1,v_2)^{4}}\\
	&= 
	\tfrac{1}{2}\int_{v_1} (- x_{01}) \frac{1}{d(v_0,v_1)^{2}} \left(\partial_{y_1} \frac{1}{d(v_1,v_2)^{2}} \right)
	+ \tfrac{1}{2}\int_{v_1}y_{01} \frac{1}{d(v_0,v_1)^{2}} \left(\partial_{x_1}\frac{1}{d(v_1,v_2)^{2}} \right)\\
	& = \tfrac{1}{2}\int_{v_1} x_{01}\left(\partial_{y_1}\frac{1}{d(v_0,v_1)^{2}}\right)\frac{1}{d(v_1,v_2)^{2}}
	 - \tfrac{1}{2}\int_{v_1}y_{01} \left( \partial_{x_1} \frac{1}{d(v_0,v_1)^{2}}\right) \frac{1}{d(v_1,v_2)^{2}} \\
	&= \int_{v_1}(x_{01}y_{01}- x_{01}y_{01}) \frac{1}{d(v_0,v_1)^{4}} \frac{1}{d(v_1,v_2)^{2}}=0 \;. 
\end{align*}
We are left with the integral
\begin{align}
	- \frac{1}{(2 \pi^2)^2} 
	\int_{x_1,y_1,z_1}\d x_1 \d y_1 \d z_1 \d \zbar_1 \frac{\left(2 x_{12}\d y_2 - 2 y_{12}\d x_2  \right)}{ d(v_0,v_1)^{2}d(v_1,v_2)^{4}} \;.  
\end{align}

We can evaluate this integral with the help of \eqn \eqref{Feyn_param}.
Choosing $\alpha=2, \beta=1$ and using $\Gamma(3)/(\Gamma(2)\Gamma(1))=2$, we obtain
\begin{align}
- \frac{2}{(2 \pi)^2} 
 \int_0^1 \d t \, t\int \d x_1 \d y_1 \d z_1 \d\bar{z}_1
 \frac{\left(2 x_{12}\d y_2 - 2 y_{12}\d x_2  \right)}{(\star)^3} ,
\end{align}
where
\begin{align}
\begin{split}
(\star)&=(1-t) \left[ x_1^2+y_1^2+|z_1|^2\right]+t \left[x_{12}^2+y_{12}^2+|z_{12}|^2 \right]\\
&=(x_1-tx_2)^2+(y_1-ty_2)^2+|z_1-t z_2|^2+t(1-t)  \left[x_{2}^2+y_{2}^2+|z_{2}|^2 \right] \;.
\end{split}
\end{align}
After shifting integration variables we obtain
\begin{align}
\begin{split}
&- \frac{2}{(2 \pi)^2} 
 \int_0^1 \d t \, t \int \d x_1 \d y_1 \d z_1 \d\bar{z}_1
 \frac{\left(2 (x_{1}-(1-t)x_2) \d y_2 - 2 (y_{1}-(1-t) y_{2}) \d x_2  \right) }{\left(x_1^2+y_1^2+|z_1|^2+t(1-t)  \left[x_{2}^2+y_{2}^2+|z_{2}|^2 \right]\right)^3} \\
 &
 = \frac{4}{(2 \pi)^2} 
 \int_0^1 \d t \, t(1-t)  \int \d x_1 \d y_1 \d z_1 \d\bar{z}_1
 \frac{\left(x_2 \d y_2 -  y_{2} \d x_2  \right) }{\left(x_1^2+y_1^2+|z_1|^2+t(1-t)  \left[x_{2}^2+y_{2}^2+|z_{2}|^2 \right]\right)^3} \;,
 \end{split}
\end{align}
where we dropped the pieces odd in $x_1$ or $y_1$, and hence do not contribute to the integral.
After evaluating the integral over $x_1, y_1$ and then over $z_1, \bar{z}_1$, the integral of $t$ becomes trivial, and we obtain 
(note that $\d z \d\bar{z}=-2 \ii r\d r \d\theta$ for the polar coordinate $z=r e^{i\theta}$) 
\begin{align}
\begin{split}
& \frac{4}{(2 \pi)^2}  \frac{\pi}{2} \left(x_2 \d y_2 -  y_{2} \d x_2  \right)
 \int_0^1 \d t \, t(1-t)  \int\d z_1 \d\bar{z}_1
 \frac{1 }{\left(|z_1|^2+t(1-t)  \left[x_{2}^2+y_{2}^2+|z_{2}|^2 \right]\right)^2}\\
& = \frac{4}{(2 \pi)^2} \frac{\pi}{2}(-2 \ii)  \pi
 \left(x_2 \d y_2 -  y_{2} \d x_2  \right)
 \frac{1 }{\left( x_{2}^2+y_{2}^2+|z_{2}|^2 \right)}
  \;.
\end{split}
\end{align}
We therefore find
\begin{align} 
\label{v1_ans}
	\int_{v_1}  P(v_0,v_1) \wedge z_1  \d z_1\wedge  P(v_1,v_2)  & =  
	\frac{1 }{\ii } \frac{(x_2 \d y_2-y_2 \d x_2)}{ d (v_0,v_2)^{2}} \;. 	
\end{align}
Similarly, the part of \eqref{13form}  relevant for the $v_3$ integral is 
\begin{equation} 
\label{v3_ans}
	\int_{v_3}  P(v_2,v_3) \wedge z_3  \d z_3 \wedge P(v_3,p_3) =  
	-\frac{1 }{\ii } \frac{((x_2-\eps) \d y_2 - y_2 \d x_2)}{   d(v_0,v_2)^{2}} \;. 
\end{equation}

We can now come back to the evaluation of \eqref{13form}.
Using the results \eqref{v1_ans} and \eqref{v3_ans}, 
we obtain
\begin{align}
	2 \left(\tfrac{\ii}{2\pi}\right)^3
	\frac{1 }{2\pi (\ii)(-\ii)} \int_{p=0}^{\epsilon}  \int_{x, y, z, \bar{z}}
\frac{
(x \d y-y \d x)  (\d z)(y \d\bar{z} \d p) ((x-\epsilon) \d y-y \d x) 
}
{d(0,v)^2 d(p,v)^4 d(\epsilon,v)^2} \;,
\end{align}
where we dropped the index $2$ from $v_2=(x_2, y_2, z_2, \bar{z}_2)$ 
to simplify the expressions.
This gives 
\begin{align}
\label{sum1}
	2 \frac{\ii }{(2\pi )^4 } 
	  \int_{p=0}^{\epsilon}  \int_{x, y, z, \bar{z}}
\frac{
\epsilon y^2  \d x \d y \d z \d\bar{z} \d p
} 
{(x^2+y^2+|z|^2) ((x-p)^2+y^2+|z|^2)^2 ((x-\epsilon)^2+y^2+|z|^2)} \;.
\end{align}

We need to take into account two more diagrams in \fig \ref{figure_anomaly2}.
It turns out that the evaluation is rather similar, with the only different exchange of the role of the 
points $0, p$ and $\epsilon$ on the Wilson line.
This means in addition to \eqref{sum1} we have two extra contributions
\begin{align}
\begin{split}
  	2 \frac{\ii }{(2\pi )^4 }  \int_{p=0}^{\epsilon}  \int_{x, y, z, \bar{z}}
\frac{
\epsilon y^2  \d x \d y \d z \d\bar{z} \d p
}
{(x^2+y^2+|z|^2)^2 ((x-p)^2+y^2+|z|^2) ((x-\epsilon)^2+y^2+|z|^2)} \;,\\
	2 \frac{\ii }{(2\pi )^4 }  \int_{p=0}^{\epsilon}  \int_{x, y, z, \bar{z}}
\frac{
\epsilon y^2  \d x \d y \d z \d\bar{z} \d p
}
{(x^2+y^2+|z|^2) ((x-p)^2+y^2+|z|^2) ((x-\epsilon)^2+y^2+|z|^2)^2} \;.
\end{split}
\end{align}
Summing all the three contributions,  we obtain
\begin{align}
- \frac{\ii }{(2\pi )^4 }   \int_{p=0}^{\epsilon}  \int_{x, y, z, \bar{z}}
(\epsilon y) \frac{\partial }{\partial y}
\frac{
  \d x \d y \d z \d\bar{z} \d p
}
{(x^2+y^2+|z|^2) ((x-p)^2+y^2+|z|^2) ((x-\epsilon)^2+y^2+|z|^2)} \;,
\end{align}
which after integrating by parts gives
\begin{align} 
	\frac{\ii }{(2\pi )^4 }   \int_{p=0}^{\epsilon}  \int_{x, y, z, \bar{z}}
\epsilon
\frac{
  \d x \d y \d z \d\bar{z} \d p
}
{(x^2+y^2+|z|^2) ((x-p)^2+y^2+|z|^2) ((x-\epsilon)^2+y^2+|z|^2)}  \;.
\end{align}
After scaling the integration variables by $\epsilon$ (assuming $\epsilon>0$) 
$\epsilon$ depends drops out, as expected:
\begin{align}
\frac{\ii }{(2\pi )^4 }     \int_{p=0}^{1}  \int_{x, y, z, \bar{z}}
\frac{
  \d x \d y \d z \d\bar{z} \d p
}
{(x^2+y^2+|z|^2) ((x-p)^2+y^2+|z|^2) ((x-1)^2+y^2+|z|^2)}  \;.
\end{align}
In angular coordinates in $(y, z, \bar{z})$-plane, the volume form $\d y \d z \d \zbar$ is $-8 \pi \ii r^2 \d r \d \Omega_{S^2}$ where $\d \Omega_{S^2}$ is the volume form on the two-sphere of volume $1$.  Integrating over the two-sphere we get 
\begin{align}	
	\frac{8 \pi }{(2\pi )^4 }   \int_{p=0}^{1} \d p  \int_{x} \d x
\int \d r r^2
\frac{
1
}
{(x^2+r^2) ((x-p)^2+r^2) ((x-1)^2+r^2)} \;.
\end{align}
After the $r$ integral we obtain
\begin{align} 
\frac{8 \pi  }{(2\pi )^4 }  
\frac{\pi}{2}
 \int_{p=0}^{1} \d p  \int_{x} \d x
\frac{1}{
(|x|+|x-p|)(|x|+|x-1|)(|x-p|+|x-1|)
} \;.
\end{align}
We can evaluate this integral by dividing into four cases $x<0, 0<x<p, p<x<1, 1<x$. These integrals are interchanged by the change of coordinates $x \mapsto 1-x$, $p \mapsto 1-p$, so that the only two independent integrals are the cases when $p < x < 1$ and $1 < x$.  We have
\begin{align} 
\int_{0 < p < x < 1} \frac{1}{(2x - p)(1-p)} &= \frac{\pi^2}{8} \;,\\
\int_{0 < p < 1 < x} \frac{1}{(2x - p)(2x - 1)(2x - 1 - p)} &= \frac{\pi^2}{24} \;.
 \end{align}
 Therefore we obtain
\begin{align}    
\frac{8 \pi  }{(2\pi )^4 }  
\frac{\pi}{2}
2\left( \frac{\pi^2}{8} +  \frac{\pi^2}{24}\right)
=\frac{1}{12} \;.
\end{align}
Including the the factors of $\hbar^2$ need for the two-loop diagram, 
we reproduce
the numerical factor of \eqref{13form_result}.

\clearpage

\bibliographystyle{unsrt}

\begin{thebibliography}{99}

\bibitem{Bethe}
H. Bethe, ``Zur Theorie der Metalle. I.
Eigenwerte und Eigenfunktionen der linearen Atomkette,'' Zeit. fur Physik {\bf 71} (1931) 205-226.

\bibitem{Onsager}
L. Onsager, ``Crystal Statistics I: A Two-Dimensional Model With An Order-Disorder Transition,''
Phys. Rev. {\bf 65} (1944) 117-49.

\bibitem{McGuire:1964zt}
  J.~B.~McGuire,
  ``Study of Exactly Soluble One-Dimensional N-Body Problems,''
  J.\ Math.\ Phys.\  {\bf 5}, 
  622 (1964).

\bibitem{Yang:1967bm}
  C.~N.~Yang,
  ``Some Exact Results for the Many Body Problems in One Dimension
With Repulsive Delta Function Interaction,''
Phys.\ Rev.\ Lett.\  {\bf 19}, 1312 (1967).  
  
\bibitem{Baxter:1971cr}
  R.~J.~Baxter,
  ``Eight-Vertex Model in Lattice Statistics,''
  Phys.\ Rev.\ Lett.\  {\bf 26}, 832 (1971).
  
\bibitem{Zamolodchikov*2} 
  A.~B.~Zamolodchikov and A.~B.~Zamolodchikov,
  ``Factorized s Matrices in Two-Dimensions as the Exact Solutions of Certain Relativistic Quantum Field Models,''
  Annals Phys.\  {\bf 120}, 253 (1979).

\bibitem{Drinfeld_ICM}
V.~G.~Drinfeld,
``Quantum Groups,'' in 
{\it Proceedings of the International Congress of Mathematicians, Berkeley, 1986}, 
American Mathematical Society, 1987.

\bibitem{Chari-Pressley}
V.~Chari and A.~Pressley, ``A Guide to Quantum Groups,'' Cambridge University Press, 1994.
 
\bibitem{Jimbo}
M. Jimbo, ed., {\it Yang-Baxter Equation In Integrable Systems} (World-Scientific, 1989).

\bibitem{PY}
J. H. H. Perk and  H. Au-Yang, ``Yang-Baxter Equation,''  in
{\it Encyclopedia of Mathematical Physics}, eds. J.-P. Fran\c{c}oise, G.L. Naber
and S. T. Tsou (Oxford: Elsevier, 2006), Vol. 5, pp. 465-73, with an expanded version in
arXiv:math-ph/0606053.

\bibitem{CostelloA}
K.~Costello, ``Integrable Lattice Models From Four-Dimensional 
Field Theories,''   Proc. Symp. Pure Math. {\bf 88} (2014) 3-24, arXiv:1303.2632.

\bibitem{Costello:2013zra}
  K.~Costello,
  ``Supersymmetric Gauge Theory and the Yangian,''
  arXiv:1303.2632 [hep-th].

\bibitem{Witten:2016spx}
  E.~Witten,
  ``Integrable Lattice Models From Gauge Theory,''
  arXiv:1611.00592 [hep-th].
  
  \bibitem{Atiyah}
M. F. Atiyah, ``New Invariants Of 3- and 4-Dimensional Manifolds,''
in R. O. Wells, ed., {\it The Mathematical Heritage Of Hermann Weyl}
(American Mathematical Society, 1988).
  
\bibitem{AuYang:1987zc}
  H.~Au-Yang, B.~M.~McCoy, J.~H.~H.~Perk, S.~Tang and M.~L.~Yan,
  ``Commuting Transfer Matrices in the Chiral Potts Models: Solutions
of Star Triangle Equations with Genus $>$ 1,''
  Phys.\ Lett.\ A {\bf 123}, 219 (1987).
  
\bibitem{Baxter:1987eq}
  R.~J.~Baxter, J.~H.~H.~Perk and H.~Au-Yang,
  ``New Solutions of the Star Triangle Relations for the Chiral Potts Model,''
  Phys.\ Lett.\ A {\bf 128}, 138 (1988).
  
\bibitem{Bazhanov:2010kz} 
  V.~V.~Bazhanov and S.~M.~Sergeev,
  ``A Master solution of the quantum Yang-Baxter equation and classical discrete integrable equations,''
  Adv.\ Theor.\ Math.\ Phys.\  {\bf 16}, no. 1, 65 (2012),
  arXiv:1006.0651 [math-ph].
  
  \bibitem{Spiridonov:2010em} 
  V.~P.~Spiridonov,
  ``Elliptic beta integrals and solvable models of statistical mechanics,''
  Contemp.\ Math.\  {\bf 563}, 181 (2012)
  [arXiv:1011.3798 [hep-th]].
  
  \bibitem{Yamazaki:2012cp} 
  M.~Yamazaki,
  ``Quivers, YBE and 3-manifolds,''
  JHEP {\bf 1205}, 147 (2012),
  arXiv:1203.5784 [hep-th].
  
  \bibitem{Yamazaki:2013nra} 
  M.~Yamazaki,
  ``New Integrable Models from the Gauge/YBE Correspondence,''
  J.\ Statist.\ Phys.\  {\bf 154}, 895 (2014),
  arXiv:1307.1128 [hep-th].

\bibitem{GN}
J.-L.~Gervais and A.~Neveu, ``Novel Triangle Relation And Absence of Tachyons in Liouville String Field Theory,'' 
Nucl. Phys. {\bf B238} (1984) 125-41.
  
\bibitem{Felder}
G.~Felder, ``Elliptic Quantum Groups,'' arXiv:hep-th/9412207.  
      
\bibitem{FelderTwo}
G.~Felder, 
``Conformal Field Theory and Integrable Systems Associated to Elliptic Curves,'' 
in {\it Proceedings of the International Congress of Mathematicians, Z\"{u}rich 1994}, 
Birkh\"{a}user, 1994; ``Elliptic quantum groups'', hep-th/9412207.

\bibitem{Etinghof}
P.~Etinghof and O.~Schiffman, ``Lectures on the Dynamical Yang-Baxter Equations,''
arXiv:math/9908064.
 
\bibitem{Part2}
K.~Costello, E.~Witten and M.~Yamazaki,
  ``Gauge Theory and Integrability, II,'' {\it to appear}.
  
\bibitem{Witten:2011zz}
 E.~Witten, ``Fivebranes And Knots,''  
 Quantum Topol. {\bf 3} (2012) 1-137, arXiv:1101.3216 [hep-th].   

\bibitem{NS}
N.~A.~Nekrasov and S.~L.~Shatashvili, ``Supersymmetric Vacua and Bethe Ansatz,''
Nucl. Phys. Proc. Suppl. {\bf 192-3} (2009) 91-112, arXiv:0901.4744.

\bibitem{NS2}
N.~A.~Nekrasov and S.~L.~Shatashvili, ``Quantization Of Integrable Systems
And Four Dimensional Gauge Theories,''  in P. Exner, ed. {\it Proceedings, 16th International
Congress on Mathematical Physics}, arXiv:0901.4748.

\bibitem{Belavin-Drinfeld}
A.~A.~Belavin and V.~G.~Drinfeld,
``Solutions of the Classical {Y}ang-{B}axter equation for Simple
              {L}ie Algebras,''
Funktsional.\ Anal.\ i Prilozhen.\  {\bf 16}, 1 (1982).  

\bibitem{Witten:1992fb}
  E.~Witten,
  ``Chern-Simons Gauge Theory as a String Theory,''
  Prog.\ Math.\  {\bf 133}, 637 (1995)
  [hep-th/9207094].

\bibitem{FGT}  
     S.~Fishel, I.~Grojnowski and C.~Teleman, 
     ``The strong {M}acdonald conjecture and {H}odge theory on the
              loop {G}rassmannian'',
   Ann.\ of Math.\ {\bf 168}, 175 (2008).
   
\bibitem{Wittenb}
E.~Witten, ``Gauge Theories, Vertex Models, and Quantum Groups,''  Nucl.Phys.  {\bf B330} (1990) 285-346.

\bibitem{Beisert}
N. Beisert, ``The $SU(2|2)$ Dynamic $R$-Matrix,'' Adv. Theor. Math. Phys. {\bf 12} (2008) 945-79.

\bibitem{Ramond} 
  P.~Ramond,
  ``Group Theory: A Physicist's Survey,''
  Cambridge, UK: Univ. Pr. (2010) 310 p.
   
\bibitem{Slansky} 
  R.~Slansky,
  ``Group Theory for Unified Model Building,''
  Phys.\ Rept.\  {\bf 79}, 1 (1981).

\bibitem{McKay}
W. G. McKay, J. Patera, and D. W. Rand, {\it Tables Of Representations Of Simple Lie Algebras: Exceptional Simple Lie Algebras} (Centre de Recherches Mathematiques,
Univ. of Montr\'{e}al, 1990).

\bibitem{Faddeev-Takhtajan}
  L.~D.~Faddeev and L.~A.~Takhtajan,
  ``Hamiltonian Methods In The Theory Of Solitons,''
  Springer, 1987.

\bibitem{WittenTwoDim}
E. Witten, ``On Quantum Gauge Theories In Two Dimensions,''
Commun. Math. Phys. {\bf 141} (1991) 153-209.


\bibitem{BT}
  M.~Blau and G.~Thompson, ``Lectures on $2-d$ Gauge Theories: Topological Aspects and Path Integral Techniques,'' in {\it 1993 Summer School
  in High-Energy Physics and Cosmology,} ed. E. Gava et. al. (World Scientific, 1994), arXiv:hep-th/9402097.
   
\bibitem{BTtwo}
M.~Blau and G.~Thompson, ``Chern-Simons Theory on $S^1$-Bundles: Abelianization and
$q$-Deformed Yang-Mills Theory,'' arXiv:hep-th/0601068. 

\end{thebibliography}

\end{document}